\documentclass[11pt,a4paper]{article}

\usepackage[utf8]{inputenc}
\usepackage[T1]{fontenc}
\IfFileExists{lmodern.sty}{\usepackage{lmodern}}{}
\IfFileExists{microtype.sty}{\IfFileExists{lmodern.sty}%
  {\usepackage{microtype}}{\usepackage[expansion=false]{microtype}}}{}
\usepackage[british]{babel}

\usepackage[a4paper,margin=2.5cm]{geometry}
\usepackage{amsmath,amssymb,amsthm,mathtools}
\usepackage{bm}
\usepackage{graphicx}

\usepackage{float}                 
\usepackage[section]{placeins}

\usepackage{booktabs}
\usepackage{tabularx}
\usepackage{longtable}
\usepackage{array}
\usepackage{multirow}
\graphicspath{{figs/}{./}}
\usepackage{caption}
\usepackage{enumitem}
\setlist{nosep,leftmargin=1.5em}
\usepackage[numbers,square,sort&compress]{natbib}
\usepackage{url}
\usepackage{authblk}
\usepackage{titlesec}
\titleformat*{\section}{\large\bfseries}
\titleformat*{\subsection}{\normalsize\bfseries}
\usepackage{orcidlink}
\usepackage{tocloft}
\usepackage{hyperref}
\hypersetup{
  colorlinks  = true,
  linkcolor   = blue!50!black,
  citecolor   = blue!50!black,
  urlcolor    = blue!50!black,
  pdftitle    = {Quantifying Compromise Risk in Exceptional Access
                 Architectures Under Sparse and Indirect Evidence},
  pdfauthor   = {Alan Woodward},
  pdfsubject  = {Cybersecurity risk quantification; exceptional access systems;
                 Bayesian hierarchical model; Monte Carlo simulation;
                 prior predictive analysis},
  pdfkeywords = {exceptional access; lawful intercept; quantitative risk;
                 Monte Carlo; Bayesian hierarchical model;
                 prior predictive analysis; attack graph; uncertainty quantification}
}
\usepackage{xcolor}
\usepackage{tikz}
\usetikzlibrary{arrows.meta,fit,positioning,shapes.geometric,calc}

\theoremstyle{definition}
\newtheorem{definition}{Definition}
\newtheorem{assumption}{Assumption}
\theoremstyle{plain}
\newtheorem{proposition}{Proposition}
\newtheorem{remark}{Remark}

\newcommand{\Lk}{\Lambda^{\mathrm{key}}}
\newcommand{\La}{\Lambda^{\mathrm{acc}}}

\newcommand{\pcondkey}{p^{\mathrm{cond,key}}}
\newcommand{\pcondacc}{p^{\mathrm{cond,acc}}}

\newcommand{\wapt}{w^{\mathrm{APT}}}
\newcommand{\Pinsider}{p_{\mathrm{insider}}}
\newcommand{\Pzeroday}{p_{\mathrm{zeroday}}}
\newcommand{\Psupply}{p_{\mathrm{supply}}}
\newcommand{\Pops}{p_{\mathrm{ops}}}
\newcommand{\Pref}{P_{\mathrm{ref}}}
\newcommand{\given}{\,|\,}
\newcommand{\PR}[1]{\Pr\!\left(#1\right)}

\newcommand{\calG}{\mathcal{G}}

\newcommand{\calV}{\mathcal{V}}
\newcommand{\calE}{\mathcal{E}}

\newcommand{\btheta}{\bm{\theta}}

\title{\bfseries Quantifying Compromise Risk in Exceptional Access
Architectures Under Sparse and Indirect Evidence}
\author[1]{Alan Woodward\,\orcidlink{0000-0002-8472-4836}}

\affil[1]{Surrey Centre for Cyber Security, University of Surrey,

  Guildford GU2~7XH, UK\\

  \texttt{alan.woodward@surrey.ac.uk}\\

Approximate word count: 25{,}600 (main text excluding abstract, tables, figures, captions, and references).}
\date{}
\begin{document}
\maketitle
\tableofcontents
\newpage

% ============================================================
\begin{abstract}
\noindent
Lawful exceptional access (EA) systems hold the cryptographic keys
that let authorised parties decrypt protected communications. The
debate over their risks has been long and qualitative, and it is
complicated by two problems: no public dataset of EA-specific
compromise events exists, so risk assessment must proceed from
sparse and indirect evidence, and prior work has treated
structurally different designs as equivalent, although
transmission-layer EA in carrier infrastructure (T-EA) and
over-the-top EA at the platform layer (OTT-EA) differ in how
cryptographic keys relate to ciphertext data.

This paper builds a structured uncertainty framework for evaluating
systemic compromise risk in EA architectures. It does not produce
predictive forecasts, which the available evidence cannot support;
it separates findings robust to assumption choices from findings
that depend on calibration. Four analytical layers are applied in
parallel to T-EA and OTT-EA: three empirical pillars (historical
analogues, a Monte Carlo scenario layer, and a channel-independence
decomposition) plus a Bayesian Structural Risk Model on a
parallel-subgraph attack graph.

The central findings are structural in the model's sense of the
word. \emph{First}, EA-equipped architectures of either class carry
weakly higher modelled risk than their no-EA counterfactual, an
ordering independent of calibration because it is enforced by the
modifier constraint $\delta \geq 1$ rather than found; its substance
is the empirical case for that constraint. \emph{Second}, the classes differ in distribution
\emph{shape}: T-EA risk is dominated by central tendency, OTT-EA by
the tail under correlated campaigns, a divergence driven by the
cross-cutting coupling prior. \emph{Third},
calibration-conditional annual probability ranges span $1.4\%$
to $12.9\%$ for T-EA across the $3{\times}$--$15{\times}$
structured-judgement targeting-premium interval, with Fréchet--Hoeffding intervals of
$[2.2\%, 7.5\%]$ for T-EA and $[1.1\%, 4.0\%]$ for OTT-EA under any
dependence structure, conditional on the Pillar~II per-scenario
probabilities.
Over multi-decade horizons, cumulative compromise is well above
zero, and key-material exfiltration is irreversible, an asymmetry
weighing heavily on OTT-EA's larger user populations. The framework
quantifies compromise \emph{probability}, not expected harm;
consequence modelling and benefit estimation are outside its
scope.

\smallskip
\noindent\textbf{Keywords:} exceptional access; lawful intercept;
quantitative risk; uncertainty quantification; Bayesian hierarchical model;
attack graph; tail risk.
\end{abstract}

% =============================================================================
\section{Introduction}\label{sec:intro}

The debate over lawful exceptional access (EA) systems is now
three decades old. It dates back at least to the Clipper Chip
proposals of the early 1990s \citep{abelson1997}. EA systems let
authorised parties decrypt data or intercept end-to-end encrypted
communications under defined legal conditions. The landmark
analyses by \citet{abelson1997,abelson2015} argued, on structural
grounds, that any EA system is inherently less secure than the
same system without EA. Key-recovery mechanisms introduce attack
surfaces that would otherwise not exist. Both papers made this
case qualitatively. To our knowledge, the peer-reviewed literature
has not yet produced a structured framework for comparative
quantitative reasoning across EA architectures. A framework of that kind should
not be confused with predictive forecasts of specific compromise
probabilities. The available evidence cannot support those.

A second limitation of prior work is that it treats different EA
designs as if they were the same. EA proposals fall into at least
four recognisable architectural classes
(Section~\ref{subsec:taxonomy}). \emph{Transmission-layer EA}
(T-EA) operates in carrier infrastructure. The canonical example
is CALEA-mandated lawful intercept at US telecommunications
carriers. \emph{Over-the-top EA} (OTT-EA) operates at the
messaging-platform layer. It covers server-side decryption
mandates and the Levy--Robinson ``ghost user'' proposal.
\emph{Hybrid client-side scanning} architectures include the EU
Chat Control proposals. \emph{Distributed threshold schemes} are
the fourth class. These four differ in how cryptographic key
material relates to the ciphertext data it can decrypt. That
relationship is a first-order driver of compromise risk, in both
form and magnitude. A risk framework that lumps all four together
cannot support the comparative reasoning the policy debate
requires.

The problem is hard for a structural reason. A forecasting-style
risk analysis would estimate a specific quantity: the annual
probability that an EA system's cryptographic keys are
exfiltrated, or that ciphertext is obtained along with enough key
material to decrypt the target communications. This quantity
cannot be measured directly. No public dataset of EA-specific
compromise events exists. The systems whose risk most warrants
assessment are prospective by definition. The problem therefore
sits in the regime of \emph{deep uncertainty}
\citep{marchau2019,lempert2003}. In a deep-uncertainty setting,
analysts can reasonably disagree about three things at once:
parameter values, the structure of the probability distributions,
and which historical analogues are relevant. Methods designed for
predictive estimation under ordinary parameter uncertainty cannot
simply be imported. They assume a level of empirical
identification that is not available here.

\subsection{Decision Support Objective}

This paper takes a decision support stance rather than an
estimation stance. The framework does not try to predict the
annual compromise probability of a specific deployed EA system. It
supports \emph{structured comparative reasoning} about three
questions. Which conclusions are robust across the range of
defensible assumptions? Which are sensitive to specific
calibration choices? Which architectural features drive the
divergence between EA classes? The intended output is not a
calibrated point projection. It is an explicit account of what
has to be assumed, and what cannot be assumed, for a risk claim
of any given magnitude or shape to follow.

Four goals shape the design. \emph{First}, constrain the
plausible values within each architectural class to a
structurally bounded range. \emph{Second}, identify the
architectural features driving the comparison between classes.
\emph{Third}, make explicit which conclusions are sensitive to
unobservable assumptions and which are not. \emph{Fourth},
characterise the distributional uncertainty that any specific
projection conceals. These goals are different in kind. No single
analytical method delivers all four, and the design of this paper
is shaped by that fact.

\emph{Two scoping choices are explicit from the outset.} First, the
framework addresses the risk side of EA evaluation only. A complete
risk-benefit determination requires comparable quantitative estimates
of operational benefits: counterfactual law-enforcement effectiveness,
the case-fraction in which EA access is decisive, harm reduction
attributable to EA-enabled interceptions, and displacement and
substitution effects. None is estimated here. No published analogue of
the present framework synthesises that literature into comparable
quantitative form. The output is one input to a net-benefit
determination, not a substitute for it. Second, the empirical inputs
draw predominantly from Western liberal-democratic jurisdictions:
Stream~A from US National Security Systems, Five~Eyes infrastructure,
and CALEA-mandated carriers; Stream~B from Western-rooted browser
trust stores; Stream~C from US-based operators disclosed under US/EU
breach-notification regimes. The calibrated central magnitudes are
specific to these jurisdictional contexts. The architectural-ordering,
Fr\'echet--Hoeffding interval, and tail-shape findings rest on
structural arguments (particularly $\gamma_{\mathrm{X}}/\gamma \geq 1$) that
are jurisdiction-independent and survive transfer to other EA regimes,
even though their calibrated magnitudes there would differ.

We use four layers, applied in parallel to T-EA and OTT-EA. The
first three form a \emph{Three-Pillar Empirical Framework} that
approaches the problem through distinct evidence types and
inferential methods. Each pillar produces
architecture-conditional outputs that bound the problem from a
different direction: empirical analogy, scenario-conditional
projection, and architectural-channel decomposition. The fourth layer,
a \emph{Bayesian Structural Risk Model}, replaces those point estimates
with full prior predictive distributions and explicitly models the
dependence structure that the three-pillar approach can only bound. For
OTT-EA, the Bayesian attack graph is extended with parallel key-side and
data-side subgraphs joined by an operational-compromise AND-node, with a
latent campaign variable $Z(t)$ coupled to cross-cutting attack paths.

\subsection{The Four Analytical Layers}

\textbf{Pillar I. Historical Analogy Analysis.} Three data streams provide
empirical anchors covering different architectural classes. Stream~A
applies six statistical methods to a 13-year dataset of major breaches of
analogous high-security government systems, anchoring the T-EA class.
Stream~B applies Poisson cohort methods to Certificate Authority private-key
compromises across 130 trusted CAs over 19 years, providing an
architectural bridge that holds keying material at scale across a
well-defined operator population. Stream~C provides the OTT-leaning anchor
through a documented incident base of platform-level key-or-data
compromises (Storm-0558, LastPass 2022, Okta 2022/2023, Microsoft Midnight
Blizzard 2024, and related incidents), classified by whether keys, data,
both, or neither were obtained. Pillar~I establishes empirical grounds:
an upper bound from the general analogue population ($\sim$36\%
T-EA-leaning), a lower-magnitude base rate from the CA cohort (0.54\%
per system-year), and an architecture-conditional rate for OTT-leaning
analogues from Stream~C.

\textbf{Pillar II. Prospective Scenario Model.} A formally specified Monte
Carlo simulation integrating EPSS, MITRE ATT\&CK, the Verizon DBIR, and
adversary-tier intelligence parameterises seven attack scenarios across a
two-tier adversary model. The conditional key-material probability
$\pcondkey_i$ is decomposed differently for each architectural class:
for T-EA, $p^{\mathrm{cond,key,T}}_i = P_1 \cdot P_2$ (reach key-management
layer; HSM extraction), and for OTT-EA, $p^{\mathrm{cond,key,OTT}}_i = P_1
\cdot P_2 \cdot P_3 \cdot P_4$ (the same two stages plus segregated-data
acquisition $P_3$ and multi-stage detection avoidance $P_4$). Under
independence, central projections fall in the low single-digit
percentage range (around 7\% for T-EA; lower for well-segregated
OTT-EA, with magnitude depending on the segregation parameter and
cross-cutting attack-path fraction). A Fréchet--Hoeffding bounding
analysis establishes the model-consistent intervals under any
dependence structure, and these intervals (not the point
projections) are the framework's primary policy-relevant output
from Pillar~II.

\textbf{Pillar III. Architectural Heuristic Decomposition.} An analysis of
the minimum viable EA architecture identifies four non-eliminable risk
channels (insider, zero-day, supply chain, operational error). Combining
empirical best-in-class component performance under independence
across the four channels gives the \emph{channel-aggregate projection}
($1 - \prod_c (1 - P_c)$), used as a single-number reference within the
plausibility range. For T-EA this yields a value in the low single-digit
percentage range (around 5\% annually). For OTT-EA,
each channel is weighted by its cross-cutting fraction (the proportion of
that channel's risk that simultaneously compromises $\mathcal{K}$ and
$\mathcal{D}$), with non-cross-cutting risk attenuated by the segregation
factor $P_3 \cdot P_4$.

\textbf{Layer IV. Bayesian Structural Risk Model.} A Bayesian hierarchical
model combines a directed attack-graph representation, a non-homogeneous
Poisson hazard process with log-normal priors on edge hazard rates, domain
transfer from analogue systems, an explicit latent-variable dependence model
for coordinated adversarial campaigns, and Monte Carlo simulation. For the
OTT-EA case the attack graph is extended with parallel key-side and
data-side subgraphs joined by an operational-compromise AND-node, with the
campaign variable $Z(t)$ coupled specifically to cross-cutting edges.
The model's priors are calibrated, separately for each architectural class,
to produce prior predictive outputs consistent with the architecture-matched
three-pillar empirical range. Consequently, the per-class median is not
independent corroboration of that range. The model's unique contributions
are (1) tail shape, (2) dependence quantification, and (3) exceedance
probabilities conditional on the calibration. Three key findings
follow. \emph{(a)}~Explicit dependence modelling inflates the
annual 95th percentile substantially for both architectural
classes. The inflation is larger for OTT-EA, where correlated
campaigns collapse the segregation gain. \emph{(b)}~Exceedance
probabilities at low annual thresholds (e.g.\ $1\%$) are above
$0.9$ under the primary calibration for both classes. This is
conditional on prior choice, with full prior-sensitivity reported.
\emph{(c)}~EA-equipped architectures of either class carry
weakly higher modelled risk than the no-EA counterfactual. The
ordering is a coherence property enforced by the modifier
constraint $\delta \geq 1$; its empirical content is the
observation, argued from the incident record, that EA
infrastructure adds attack surface, attracts adversarial targeting,
and concentrates key material above and beyond comparable non-EA
systems.

\subsection{Motivating Cases}

The threat model is not hypothetical for either architectural class.
Real incidents from both, summarised below and discussed at length in
Sections~\ref{sec:pillar1_streamA} and~\ref{sec:pillar1_streamC},
make the modelled scenarios concrete.

\textit{T-EA cases.} The 2024 Salt Typhoon campaign, in which Chinese
state-sponsored actors compromised CALEA-mandated lawful-intercept
infrastructure at nine or more US telecommunications carriers and obtained
the government's active wiretap target list
\citep{volz2024salttyphoon,eff_salttyphoon}, provides the
most recent operational demonstration that T-EA systems constitute
high-value, specifically targeted attack surfaces. The Greek Vodafone
incident of 2004--2005, in which malicious software was installed on the
country's lawful-intercept infrastructure enabling surveillance of over 100
senior government officials \citep{prevelakis2007}, provides an earlier
independently documented precedent. The Crypto~AG disclosures
(Operation~Rubicon, declassified 2020) \citep{cryptoag2020} illustrate a
third class of T-EA-relevant compromise, in which the cryptographic
infrastructure itself was operated as a covert intelligence channel.

\textit{OTT-EA cases.} No public OTT-EA-specific compromise is documented,
because no jurisdiction has yet mandated OTT-EA at scale. The closest
empirical analogues are platform-level key-or-data compromise incidents in
non-EA contexts, which constitute the Stream~C empirical base
(Section~\ref{sec:pillar1_streamC}). The Storm-0558 campaign (Microsoft
2023), in which a Chinese state-sponsored actor obtained a Microsoft consumer
signing key and forged tokens to access Outlook Web Access mailboxes of US
government officials, demonstrates the catastrophic outcome of master-key
compromise at platform scale. The LastPass 2022 incident, in which an
attacker obtained both encrypted password vaults and source code revealing
encryption details across two compromise stages, illustrates the
multi-stage compromise pattern characteristic of OTT-relevant attacks. The
Microsoft Midnight Blizzard incident (2024), in which the same Russian
state-sponsored actor that conducted SolarWinds obtained source code and
internal systems access through password-spray and OAuth abuse, illustrates
the persistence of platform-level targeting against trusted infrastructure
operators. These incidents do not establish frequency in the EA-mandated
case, no such case yet exists, but they document that platform-level
compromise at the relevant scale is operational.

\textbf{Role of these cases in the framework.} The T-EA incidents enter the
quantitative parameter calibration only via Stream~A, with Salt Typhoon
treated as a sensitivity case (Supplementary \S\ref{sec:salt_typhoon_role}); they
are not used as frequency calibration points individually. The OTT-leaning
incidents enter the calibration via Stream~C, with explicit inclusion
criteria documented in Section~\ref{sec:pillar1_streamC}. No incident from
either set is used to fix a probability point estimate. Their function is
to establish scenario plausibility and to provide architecturally matched
prior calibration for the Bayesian layer.

\subsection{Structure of the Paper}

The paper is laid out as follows. Section~\ref{sec:background} situates the
work in the existing literature and lists the empirical sources we draw on.
Section~\ref{sec:threat_model} fixes our system and threat model and gives
formal compromise definitions. Section~\ref{sec:framework} sets out the
four-layer scheme. The three pillars then occupy
Sections~\ref{sec:pillar1}--\ref{sec:pillar3}, in order: historical analogy,
the scenario model, and the channel-aggregate projection. The Bayesian
structural risk model follows in Section~\ref{sec:bayesian}. Results from all
four layers are gathered in Section~\ref{sec:results}, with the
cross-layer synthesis in Section~\ref{sec:synthesis}.
Sections~\ref{sec:limitations} and~\ref{sec:policy} address limitations
and policy implications respectively.

\subsection{Principal Findings at a Glance}

\paragraph{Reader's guide to the principal claims.} Different
findings of the framework carry different assumption-robustness
status, and the synthesis hierarchy (S1 most assumption-robust,
through S6) gives each one an explicit tier. Table~\ref{tab:readers_guide}
maps each principal claim to the section where it is established and
to its tier in the synthesis hierarchy; full definitions of the
tiers are given in \S\ref{sec:synthesis}.

\begin{table}[!htbp]
\centering
\small
\caption{Reader's guide to the principal claims of the framework. The
``Assumption-robustness tier'' column refers to the synthesis hierarchy of
\S\ref{sec:synthesis}, with S1 the most assumption-robust (structural
finding holding under any defensible parameter choice) and S6 the
most calibration-dependent. The tier is conservative: where a claim
spans tiers (typically because its direction is structural but its
magnitude is calibration-conditional) the lower-tier (more
robust) qualifier is shown.}
\label{tab:readers_guide}
\begin{tabularx}{\textwidth}{@{}p{0.50\textwidth}lX@{}}
\toprule
\textbf{Claim} & \textbf{Established in} & \textbf{Assumption-robustness tier} \\
\midrule
EA-equipped architectures carry weakly higher modelled compromise
risk than the no-EA counterfactual within the same architectural
class, conditional on the superset-graph and $\delta \geq 1$
assumptions. A coherence property of the model, not an empirical
finding; the dominance is weak, with prior mass $0.108$ on equality
  & Proposition~\ref{prop:dominance} (\S\ref{sec:bayesian})
  & S1: internal coherence \\
OTT-EA risk distribution has a heavier upper tail than T-EA under
correlated-campaign conditions. The direction is a consequence of the
coupling prior, since $\xi_{\mathrm{OTT}}/\xi_{\mathrm{T}} \approx
\gamma_{\mathrm{X}}/\gamma \geq 1$ by construction, and is not
evidential. The informative content is the goodness-of-fit asymmetry
between the classes
  & \S\ref{subsec:tail_algebra},
    \S\ref{subsec:dependence_results},
    Table~\ref{tab:dependence_comparison}
  & S2: direction assumed, fit asymmetry observed \\
Upper-quantile risk ordering: OTT-EA upper quantiles exceed T-EA's
across the upper $5$--$10\%$ of the prior predictive distribution,
with the difference widening under correlated-campaign dependence
(direct empirical percentile comparison, distribution-free; the
parametric GPD characterization is supported for T-EA but rejected
for OTT-EA, Supplementary \S\ref{sec:supp_gpd_gof})
  & \S\ref{subsec:tail_risk},
    Supplementary \S\ref{sec:supp_gpd_gof}
  & S2--S3: structural mechanism, parametric tail-index claim
    weakened \\
Empirical anchor cohorts (Streams~A, B, C) project baseline
compromise rates on the order of $1\%$ per system-year (pre-EA
targeting premium)
  & \S\ref{sec:pillar1_streamA}--\ref{sec:pillar1_streamC}
  & S4: empirical-anchor plausibility \\
Annual EA-projected compromise probability lies in the
Fr\'echet--Hoeffding interval $[2.2\%, 7.5\%]$ for T-EA and
$[1.1\%, 4.0\%]$ for OTT-EA under any dependence structure
compatible with the seven-scenario model
  & \S\ref{sec:fh_bounds}
  & S4: interval dominance \\
The T-EA-versus-OTT-EA comparison is not first-order stochastic
dominance but a central-tendency-versus-tail-risk trade-off: T-EA
exceeds OTT-EA at the median, OTT-EA exceeds T-EA in the upper tail
  & Remark~\ref{rem:cross_arch}
  & S6: cross-architecture trade-off \\
Multi-decade (10--25 year) cumulative compromise probability is
materially above zero under any defensible parameter choice
  & \S\ref{subsec:dependence_results},
    Table~\ref{tab:cumulative}
  & S1--S4: qualitative direction structural; specific magnitudes
    calibration-conditional \\
\bottomrule
\end{tabularx}
\end{table}

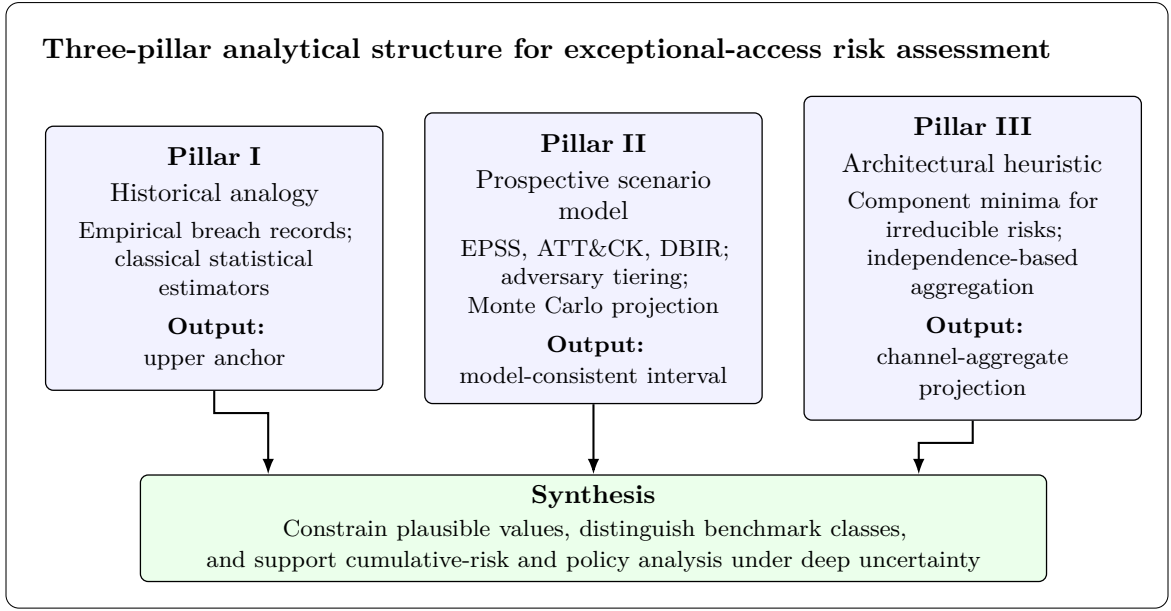
\begin{figure}[!htbp]
\centering
\resizebox{0.96\textwidth}{!}{%
\begin{tikzpicture}[
  font=\small,
  box/.style={draw, rounded corners=3pt, line width=0.5pt, align=center, text width=4.0cm, minimum height=3.0cm, fill=blue!5, inner sep=8pt},
  arr/.style={-{Latex[length=2.1mm]}, thick}
]
\node[align=center, font=\bfseries] (title1) at (6.6,0) {Three-pillar analytical structure for exceptional-access risk assessment};
\node[box, below=0.7cm of title1.south, anchor=north, xshift=-4.5cm] (p1) {\textbf{Pillar I}\\[2pt]Historical analogy\\[3pt]\footnotesize Empirical breach records;\\classical statistical\\estimators\\[4pt]\textbf{Output:}\\upper anchor};
\node[box, right=0.55cm of p1] (p2) {\textbf{Pillar II}\\[2pt]Prospective scenario\\model\\[3pt]\footnotesize EPSS, ATT\&CK, DBIR;\\adversary tiering;\\Monte Carlo projection\\[4pt]\textbf{Output:}\\model-consistent interval};
\node[box, right=0.55cm of p2] (p3) {\textbf{Pillar III}\\[2pt]Architectural heuristic\\[3pt]\footnotesize Component minima for\\irreducible risks;\\independence-based\\aggregation\\[4pt]\textbf{Output:}\\channel-aggregate projection};
\node[draw, rounded corners=3pt, line width=0.5pt, align=center, fill=green!8, text width=12.0cm, minimum height=1.1cm, below=0.95cm of p2] (syn) {\textbf{Synthesis}\\[2pt]\footnotesize Constrain plausible values, distinguish benchmark classes,\\ and support cumulative-risk and policy analysis under deep uncertainty};
\draw[arr] (p1.south) -- ($(p1.south) + (0,-0.3)$) -| ($(syn.north) + (-4.4,0)$);
\draw[arr] (p2.south) -- (syn.north);
\draw[arr] (p3.south) -- ($(p3.south) + (0,-0.3)$) -| ($(syn.north) + (4.4,0)$);
\node[draw, rounded corners=5pt, line width=0.4pt, fit=(title1)(p1)(p2)(p3)(syn), inner sep=10pt] {};
\end{tikzpicture}%
}
\captionsetup{width=0.88\textwidth}
\caption{Overview of the three-pillar framework. Each pillar uses a distinct evidence type and inferential role; their outputs are complementary rather than interchangeable, and are interpreted as decision-support inputs rather than as point estimates.}
\label{fig:pillar_overview}
\par\smallskip\noindent\small\textit{Alt text:} Diagram of the three-pillar framework: three labelled boxes (Pillar I: historical analogy; Pillar II: Monte Carlo scenarios; Pillar III: structural decomposition) feed into a synthesis node, with arrows indicating complementary evidence types and inferential roles.
\end{figure}

% =============================================================================
\section{Background and Related Work}\label{sec:background}

\subsection{Exceptional Access Systems}

Three closely related terms recur in the policy and technical literature
and are worth distinguishing at the outset. \emph{Lawful intercept} (LI)
denotes the authorised real-time interception of communications content
under a specific legal warrant, typically implemented at carrier
infrastructure (the CALEA architecture in the United States is the
canonical instance). \emph{Key escrow} denotes any architectural
arrangement that deposits cryptographic key material with a third party
to make it recoverable independently of the communicating parties. The
1990s Clipper Chip proposal is the historical exemplar.
\emph{Exceptional access} (EA) is the broader category, encompassing
both LI and key-escrow mechanisms as well as later proposals (split-key
schemes, server-side decryption mandates, ghost-user constructions,
client-side scanning). Where this paper analyses a specific
architectural class (transmission-layer EA / T-EA; over-the-top EA /
OTT-EA), the term EA refers to the umbrella concept. Where it analyses
historical incidents at carrier infrastructure (Salt Typhoon, the
Athens affair), LI refers specifically to that subset.

The term \emph{exceptional access} encompasses key escrow, key recovery,
split-key schemes, and server-side decryption mandates
\citep{abelson2015,landau2017,oecd2024}. All such schemes share the structural
property of creating an alternative access path to protected communications
that is available independently of the communicating parties \citep{abelson1997}.
At the protocol-engineering level, the Internet Engineering Task Force
formally adopted a policy in 2000 declining to support wiretap
capabilities in IETF protocols \citep{rfc2804}. The Internet
Architecture Board has since reaffirmed that IETF protocol work
prioritises the interests of end users \citep{rfc8890}, reflecting a
sustained architectural objection from the protocol-design community.
Synthesising the technical literature on exceptional access
\citep{abelson2015,abelson1997}, we distinguish five structural risk classes
that separate EA systems from their non-EA counterparts: (1)~single
point of failure at the key store; (2)~insider threat through privileged EA
roles; (3)~increased complexity correlated with vulnerability density;
(4)~forward-secrecy degradation; and (5)~scale risk from deployment across
diverse operators.

Related proposals include the ``ghost user'' model of Levy and Robinson
\citep{levy2018ghosts}, which attracted substantial cryptographic critique
\citep{ghostletter2019}, and EU-level client-side scanning approaches \citep{eu2022chatcontrol},
assessed technically by Kulshrestha and Mayer \citep{kulshrestha2021identifying}.
Formal cryptographic attacks against key-escrow and malicious cryptography
constructions are surveyed by Young and Yung \citep{young2004malicious}. The
provable-security framework for analysing cryptographic primitives, within which
EA schemes would be formally evaluated, is exemplified by \citet{bellare1996provably}.

\subsection{Cyber Risk Quantification}

Quantitative cyber risk has been approached through several frameworks. The
FAIR methodology \citep{freund_jones_2014} decomposes risk into probability and
magnitude factors estimated from expert judgment and analogical data. The
Gordon-Loeb model \citep{gordon2002economics} addresses security investment
optimisation. Bayesian belief networks have been applied to network intrusion
risk \citep{poolsappasit2012dynamic} and software defect prediction
\citep{fenton2012risk}. Attack-graph approaches are developed by Ou
\emph{et al.}\ \citep{ou2006scalable} (scalable logical-graph
generation) and Wang \emph{et al.}\
\citep{wang2006network} (attack-graph--based alert correlation and
prediction), while Homer \emph{et al.}\ \citep{homer2013aggregating}
contribute probabilistic aggregation of per-vulnerability metrics across attack
graphs. Models for the time to compromise a system, as a function of attacker
skill and visible vulnerabilities, have also been developed \citep{mcqueen2006}.
None of this literature specifically addresses
the structural properties of EA key-management infrastructure.

For the empirical pillars, application of Poisson-process models, Bayesian
Gamma--Poisson conjugate inference, exponential survival analysis, actuarial
modelling, and bootstrap resampling to security incident data follows standard
treatments \citep{lawless2003,gelman2013,efron_tibshirani_1993}. The
rule-of-three zero-event estimator \citep{hanley1983} and actuarial
compound-Poisson models \citep{klugman2012} complete the methodological
repertoire. \citet{allodi_massacci_2014} provide foundational empirical work on
vulnerability severity and exploitation frequency distributions.

\subsection{Empirical Sources and Evidential Status}

Pillar~II draws on four published empirical sources, each
documented and applied in detail in the methodology section. The
Exploit Prediction Scoring System (EPSS)
\citep{jacobs2021epss,jacobs2023epss} is an open framework
assigning daily exploitation-probability scores to published CVEs.
It is used here as an empirically motivated prior for the
opportunistic adversary tier. MITRE ATT\&CK v15
\citep{strom2018attack,mitre2024attack} provides externally
auditable mappings from attack scenarios to documented adversary
techniques. The Verizon DBIR \citep{dbir2024,dbir2025} supplies
breach-frequency statistics and the actor-mix data used to bound
the EA targeting premium: the share of breaches attributable to
state-aligned espionage actors, resolved by target sector. Mandiant
M-Trends 2024 \citep{mtrends2024} provides adversary dwell-time and
detection statistics.
The CrowdStrike Global Threat Report 2024
\citep{crowdstrike2024} and CISA Joint Advisories
\citep{cisa_volt,cisa_salt} provide complementary nation-state
campaign documentation.

The framework treats EA risk estimation as a problem of \emph{deep
uncertainty} \citep{marchau2019,lempert2003}. This is a regime in
which analysts may disagree about parameter values, about the
structure of probability distributions, and even about the relevance
of historical data. Following established approaches in
climate-change economics and infrastructure planning, the goal is
not to resolve deep uncertainty. The goal is to identify
projections robust across defensible assumptions, specify their
conditions of validity, and equip decision-makers to assess
net-benefit given the range of possible outcomes consistent with
current knowledge.

One methodological objection should be answered at the outset.
The decision-under-deep-uncertainty literature holds that where a
single prior distribution is not defensible, analysis should
characterise performance across the space of plausible assumptions
rather than integrate over one \citep{lempert2003,marchau2019}.
The framework adopts exactly that discipline, in two forms. The
Fr\'echet--Hoeffding analysis (\S\ref{sec:fh_bounds}) bounds the
aggregate compromise probability under \emph{every} dependence
structure compatible with the scenario model, a bounded-probability
treatment that requires no dependence prior at all. And the
Bayesian layer is run not under one prior but under a calibration
set (primary, conservative, pessimistic), with findings reported by
whether they hold across the set or only within members of it
(\S\ref{sec:structural_vs_calibration}). Single-prior medians
appear in this paper as labelled reference points inside those
robustness envelopes, never as the analysis itself. Source versions
for the empirical inputs are pinned (MITRE ATT\&CK v15, DBIR 2024
and 2025, M-Trends 2024, with the EPSS score snapshot archived at
analysis date), and the derived parameter values ship with the
reproducibility package.

% =============================================================================
\section{System and Threat Model}\label{sec:threat_model}

\subsection{Abstract EA Architecture}

\begin{definition}[EA System]
\label{def:ea_system}
An \emph{exceptional access system} is a tuple
$\Sigma = (\mathcal{C}, \mathcal{K}, \mathcal{P}, \mathcal{T}, \mathcal{D})$,
where: $\mathcal{C}$ is the set of communications or storage objects subject
to EA; $\mathcal{K}$ is the key management infrastructure, comprising one or
more hardware or software modules holding cryptographic material enabling EA;
$\mathcal{P}$ is the procedural access-control layer governing requests to
$\mathcal{K}$; $\mathcal{T}$ is the trust hierarchy, comprising the set of
entities whose correct operation is necessary for the security of
$\mathcal{K}$; and $\mathcal{D}$ is the ciphertext-data infrastructure
holding the encrypted communications to which EA access pertains.
\end{definition}

The earlier literature on EA risk has typically treated $\mathcal{K}$ and the
ciphertext-data infrastructure as co-located, since for transmission-layer
EA architectures keys and encrypted traffic converge at the intercept point
\citep{abelson2015,abelson1997}. The explicit inclusion of $\mathcal{D}$ as
a separate tuple component is essential for the architectural taxonomy that
follows: in some EA architectures $\mathcal{K}$ and $\mathcal{D}$ share a
trust boundary and access path, while in others they are deliberately
segregated, and this segregation is a first-order determinant of the
operational compromise probability.

\begin{definition}[Architectural Segregation]
\label{def:segregation}
An EA system $\Sigma$ exhibits \emph{architectural segregation} of degree
$s \in [0, 1]$ if, conditional on adversarial compromise of $\mathcal{K}$
(unauthorised access to key material), the probability of also obtaining
sufficient ciphertext from $\mathcal{D}$ to decrypt target communications,
in the absence of additional adversarial effort directed at $\mathcal{D}$,
is $1 - s$. The fully co-located case ($s = 0$) describes systems in which
$\mathcal{K}$ and $\mathcal{D}$ share a trust boundary, access path, and
operational locus. The fully segregated case ($s = 1$) is a theoretical
limit not achievable in operationally viable systems, since some shared
trust dependency is necessary to reconcile keys with ciphertext under
authorised access.
\end{definition}

\subsubsection{Architectural Taxonomy}\label{subsec:taxonomy}

We distinguish four architectural classes of EA system, characterised by the
relationship between $\mathcal{K}$, $\mathcal{D}$, and the operator's
organisational position:

\textit{Class T (Transmission-layer EA, T-EA).} EA mandates implemented in
network or carrier infrastructure, intercepting traffic in transit. The
operator is a telecommunications carrier or equivalent network provider
acting under state mandate. Examples: CALEA-mandated lawful-intercept
facilities at US telecommunications carriers; ETSI TS 103 221-class
retained-data and intercept architectures in European regimes. In Class T,
$\mathcal{K}$ and $\mathcal{D}$ are operationally co-located at the
intercept point: the same orchestration system that holds session-key
material also processes the encrypted traffic flow. Architectural
segregation is structurally low ($s \approx 0$). Motivating compromises:
Salt Typhoon (2024) \citep{volz2024salttyphoon,eff_salttyphoon},
the Athens affair (2004--05) \citep{prevelakis2007}, and the Crypto AG
disclosures (Operation Rubicon, declassified 2020) \citep{cryptoag2020}.

\textit{Class A (Application-layer or Over-The-Top EA, OTT-EA).} EA mandates
implemented at the messaging-platform or storage-service layer, where the
platform operator is required to retain or escrow keying material enabling
decryption of user communications, or to implement server-side decryption
under lawful order. Examples: server-side decryption mandates and the
Levy--Robinson ``ghost user'' proposal \citep{levy2018ghosts}; the
broader policy debate over law-enforcement access to encrypted content,
including key-escrow and related approaches, is surveyed by the OECD
\citep{oecd2024}. (Client-side scanning architectures are a related but
architecturally distinct mechanism, treated separately as Class~H below.)
The storage-service
variant of Class A applies to deployed end-to-end-encrypted
cloud-storage systems such as Apple's Advanced Data
Protection for iCloud \citep{apple_adp}, in which user data is encrypted under keys
held only by the user's devices. A storage-layer EA mandate would
require the operator to retain or escrow access to data that the
architecture is otherwise designed to place beyond operator reach.
In Class A, $\mathcal{K}$ (key
custody) and $\mathcal{D}$ (encrypted message or file storage) are
typically architecturally segregated: keys reside in dedicated
key-management infrastructure with access controls distinct from
the bulk storage layer holding ciphertext. Segregation $s$ is materially greater than zero, with
magnitude governed by platform implementation. No public EA-specific Class
A compromise is documented because no jurisdiction has yet mandated Class A
EA at scale. The closest empirical analogues are platform-level key-or-data
compromise incidents in non-EA contexts (Storm-0558 against Microsoft
2023, LastPass 2022, the Okta 2022 and 2023 incidents, and Microsoft
Midnight Blizzard 2024), which form the basis for the Stream~C empirical
calibration developed in Section~\ref{sec:pillar1_streamC}.

\textit{Class H (Hybrid client-side scanning, CSS).} Architectures in which
detection or content-matching is performed on client devices using
operator-distributed key material, with positive matches escalated to the
operator. The EU Chat Control proposals \citep{eu2022chatcontrol} and
historical client-side perceptual hashing systems are the canonical
examples. The risk profile combines elements of Class A (operator key
custody) with endpoint-distributed key material whose attack surface is
governed by mobile-platform security properties
\citep{kulshrestha2021identifying,abelson2024bugs}. Class H is treated only
qualitatively in this paper (Supplementary \S\ref{sec:hybrid_css}). Its quantitative
analysis requires endpoint risk parameters outside our calibrated set.

\textit{Class D (Distributed threshold).} Architectures in which key
material is split across $N$ independent trustees with a $t$-of-$N$
reconstruction requirement \citep{micali1992fair,bellare1997verifiable,
denning1996taxonomy}. Class D is treated through architectural analysis in
Section~\ref{sec:threshold}.

\begin{table}[!htbp]
\centering
\small
\caption{Architectural taxonomy of exceptional access systems and their
treatment in this paper. The segregation parameter $s$ is defined in
Definition~\ref{def:segregation}.}
\label{tab:architectures}
\begin{tabularx}{\textwidth}{@{}lXXl@{}}
\toprule
\textbf{Class} & \textbf{$\mathcal{K}$/$\mathcal{D}$ relationship} &
\textbf{Operator type} & \textbf{Treatment} \\
\midrule
T-EA & Co-located, $s \approx 0$ & Carrier under state mandate &
  Quantitative \\
OTT-EA & Segregated, $s \in (0, 1)$ & Global platform under state mandate &
  Quantitative \\
Hybrid CSS & Endpoint-distributed & Platform with client co-operation &
  Qualitative (Supp.~\S\ref{sec:hybrid_css}) \\
Threshold & Split across $t$-of-$N$ trustees & Multi-party &
  \S\ref{sec:threshold} \\
\bottomrule
\end{tabularx}
\end{table}

The EA-specific attack surface arises from $\mathcal{K}$, $\mathcal{P}$,
$\mathcal{T}$, and (in segregated architectures) $\mathcal{D}$ jointly.
The quantitative analysis in this paper covers Classes T (transmission-layer)
and OTT (application-layer). These are the two architectural classes that
dominate contemporary EA policy debate, and their parallel treatment is the
principal extension of this paper relative to prior risk-quantification work
that has implicitly conflated them. Hybrid client-side scanning architectures
(Class H) and distributed threshold schemes (Class D) are treated
qualitatively and through dedicated architectural analysis respectively.

\subsection{Adversary Model}

We distinguish three adversary classes:
\begin{description}
  \item[Nation-state APT ($\mathcal{A}_{\mathrm{NS}}$).] Advanced persistent
    threat actors with high capability (sophisticated zero-day exploitation,
    long-horizon campaigns), high motivation (strategic intelligence value),
    and substantial resources.
  \item[Insider threat ($\mathcal{A}_{\mathrm{IN}}$).] Individuals with
    legitimate access to components of $\mathcal{K}$ or $\mathcal{T}$ who
    abuse that access, whether through malice, coercion, or negligence.
  \item[Criminal / opportunistic ($\mathcal{A}_{\mathrm{CR}}$).] Actors
    motivated primarily by financial gain, with moderate capability.
    Relevant because EA infrastructure presents an attractive target for
    ransomware and data exfiltration.
\end{description}

\begin{assumption}[Adversary independence within classes]
Within each adversary class, individual actors are modelled as statistically
independent except as captured by the shared latent campaign variable $Z$
defined in Section~\ref{subsec:dependence}.
\end{assumption}

% --- Figure: Attack-surface schematic ---
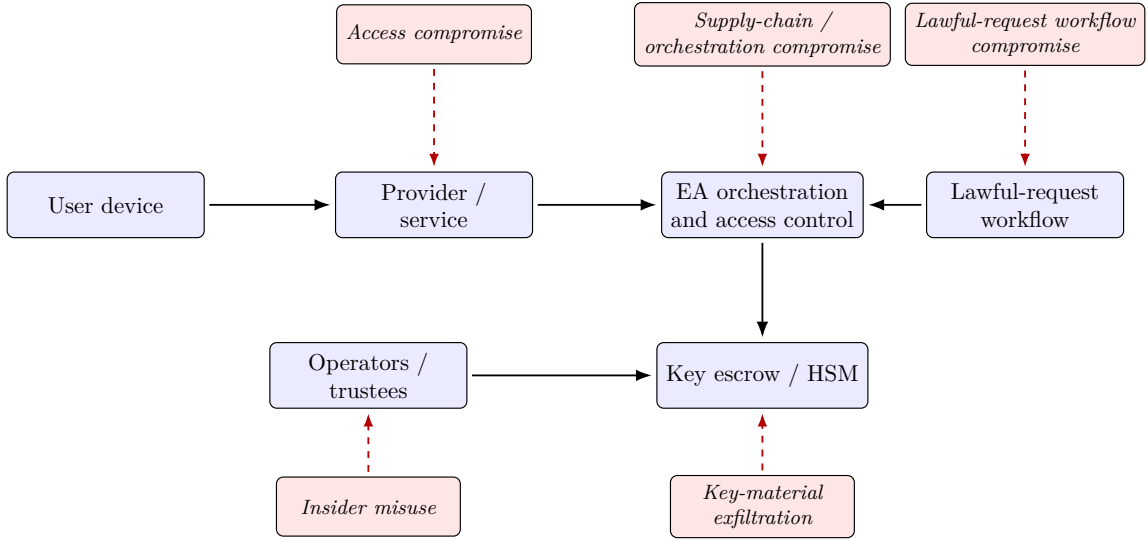
\begin{figure}[!htbp]
\centering
\resizebox{0.94\textwidth}{!}{%
\begin{tikzpicture}[
  font=\small,
  comp/.style={draw, rounded corners=3pt, minimum width=3.0cm, minimum height=1.0cm, align=center, fill=blue!8, line width=0.5pt, inner sep=4pt},
  risk/.style={draw, rounded corners=3pt, minimum width=2.8cm, minimum height=0.9cm, align=center, fill=red!10, line width=0.5pt, inner sep=4pt, font=\footnotesize\itshape},
  arr/.style={-{Latex[length=2.4mm]}, thick, shorten >=2pt, shorten <=2pt},
  riskarr/.style={-{Latex[length=2.0mm]}, thick, dashed, red!70!black, shorten >=2pt, shorten <=2pt}
]
\node[risk] (r_access) at ( 5.0, 2.6) {Access compromise};
\node[risk] (r_supply) at (10.0, 2.6) {Supply-chain /\\orchestration compromise};
\node[comp] (device)   at ( 0.0, 0.0) {User device};
\node[comp] (provider) at ( 5.0, 0.0) {Provider /\\service};
\node[comp] (orch)     at (10.0, 0.0) {EA orchestration\\and access control};
\node[risk] (r_law) at (14.0, 2.6) {Lawful-request workflow\\compromise};
\node[comp] (law)   at (14.0, 0.0) {Lawful-request\\workflow};
\node[comp] (ops) at ( 4.0,-2.6) {Operators /\\trustees};
\node[comp] (hsm) at (10.0,-2.6) {Key escrow / HSM};
\node[risk] (r_insider) at ( 4.0,-4.6) {Insider misuse};
\node[risk] (r_keymat)  at (10.0,-4.6) {Key-material\\exfiltration};
\draw[arr] (device.east)   -- (provider.west);
\draw[arr] (provider.east) -- (orch.west);
\draw[arr] (law.west)      -- (orch.east);
\draw[arr] (orch.south)    -- (hsm.north);
\draw[arr] (ops.east)      -- (hsm.west);
\draw[riskarr] (r_access.south)  -- (provider.north);
\draw[riskarr] (r_supply.south)  -- (orch.north);
\draw[riskarr] (r_law.south)     -- (law.north);
\draw[riskarr] (r_insider.north) -- (ops.south);
\draw[riskarr] (r_keymat.north)  -- (hsm.south);
\end{tikzpicture}%
}
\caption{Simplified attack-surface schematic for a transmission-layer
(Class~T) exceptional-access architecture, in which key material
$\mathcal{K}$ and ciphertext data $\mathcal{D}$ are operationally
co-located. Solid arrows show the primary data-flow paths; dashed red
arrows indicate where each risk class attaches. The mandate creates
compromise pathways concentrated around orchestration, escrow, and
operator trust boundaries that do not exist in the corresponding non-EA
architecture. The corresponding schematic for an OTT-EA (Class~A)
architecture, with $\mathcal{K}$ and $\mathcal{D}$ architecturally
segregated, is given as Figure~\ref{fig:attack_surface_ott} in
Section~\ref{sec:bayesian} where the parallel attack graph is developed.}
\label{fig:attack_surface}
\par\smallskip\noindent\small\textit{Alt text:} Schematic of a transmission-layer EA architecture: cryptographic key material and ciphertext data are operationally co-located; five risk-class attack paths are shown as dashed red arrows attaching at orchestration, escrow, and operational nodes.
\end{figure}

\subsection{Compromise Definitions}\label{subsec:compromise_defs}

For architecturally segregated EA systems, multiple distinct compromise
outcomes must be distinguished, since the operationally relevant outcome is
not equivalent to either of its components alone. We define four annual
compromise probabilities corresponding to qualitatively distinct adversarial
achievements.

\begin{definition}[Technical Compromise]
\label{def:technical_compromise}
A \emph{technical compromise} occurs when an entity not authorised under
$\mathcal{P}$ obtains EA-specific cryptographic key material from
$\mathcal{K}$ sufficient to decrypt one or more communications, were the
corresponding ciphertext also available. We denote the annual probability of
technical compromise by $\Lambda^{\mathrm{key}}$.
\end{definition}

\begin{definition}[Data Compromise]
\label{def:data_compromise}
A \emph{data compromise} occurs when an entity not authorised under
$\mathcal{P}$ obtains EA-pertinent ciphertext from $\mathcal{D}$ in
sufficient volume that, were the corresponding key material also available,
target communications could be decrypted. We denote the annual probability
of data compromise by $\Lambda^{\mathrm{data}}$.
\end{definition}

\begin{definition}[Operational Compromise]
\label{def:op_compromise}
An \emph{operational compromise} occurs when an entity not authorised under
$\mathcal{P}$ jointly obtains key material from $\mathcal{K}$ and ciphertext
from $\mathcal{D}$ sufficient to decrypt one or more target communications.
We denote the annual probability of operational compromise by
$\Lambda^{\mathrm{op}}$. The policy-relevant quantity is this, since
technical or data compromise alone does not enable the adversarial outcome
that motivates concern.
\end{definition}

\begin{definition}[Access-Pathway Compromise]
\label{def:access_compromise}
An \emph{access-pathway compromise} occurs when an entity not authorised
under $\mathcal{P}$ obtains unauthorised access to the EA operational
interface sufficient to issue interception requests through
legitimate-appearing channels. We denote the annual probability by
$\Lambda^{\mathrm{acc}}$. This outcome is distinct from technical, data, or
operational compromise: it involves authorisation abuse rather than
cryptographic-material or ciphertext exfiltration. For policy purposes
$\Lambda^{\mathrm{acc}}$ and $\Lambda^{\mathrm{op}}$ are both
relevant: the former enables targeted decryption under the operator's
own infrastructure; the latter enables out-of-band decryption.
\end{definition}

\begin{remark}[Architectural collapse]
\label{rem:collapse}
For Class T architectures with co-located $\mathcal{K}$ and $\mathcal{D}$,
the conditional probability $\Pr(\text{data} \mid \text{technical})$ is
approximately unity, and consequently $\Lambda^{\mathrm{op}} \approx
\Lambda^{\mathrm{key}}$. The earlier literature's elision of these outcomes
\citep{abelson2015,abelson1997} is correct for the Class T case.
For Class A (OTT) architectures, $\Pr(\text{data} \mid \text{technical}) =
1 - s < 1$ in the absence of correlated adversarial effort, and $1$ under
correlated effort exploiting cross-cutting access paths. The joint
probability $\Lambda^{\mathrm{op}}$ is therefore strictly less than
$\min(\Lambda^{\mathrm{key}}, \Lambda^{\mathrm{data}})$ under independence
and approaches $\min(\Lambda^{\mathrm{key}}, \Lambda^{\mathrm{data}})$
under perfect correlation. The gap is the segregation gain, which is the
quantity that distinguishes the OTT risk profile from the T-EA risk
profile.
\end{remark}

\begin{definition}[Operational Viability]
An EA system is \emph{operationally viable} if it can fulfil its legal
mandate: process authorised lawful-intercept requests, maintain audit
records, operate continuously, and be staffed and maintained over a
multi-year operational lifetime. Any system specification that eliminates a
risk channel at the cost of operational viability is excluded from the
Pillar~III heuristic decomposition.
\end{definition}

\begin{definition}[Compromise Event -- Bayesian Model]
\label{def:compromise}
A \emph{compromise event} $\omega$ occurs at time $t$ if, at time $t$, an
entity not authorised under $\mathcal{P}$ obtains the conditions necessary
to decrypt one or more communications in $\mathcal{C}$ without generating a
detectable audit record. For Class T architectures, those conditions reduce
to obtaining sufficient cryptographic material from $\mathcal{K}$, since
$\mathcal{K}$ and $\mathcal{D}$ are co-located. For Class A architectures,
those conditions require joint compromise of $\mathcal{K}$ and $\mathcal{D}$,
modulated by the segregation parameter $s$
(Definition~\ref{def:segregation}) and the cross-cutting attack-path
fraction (Section~\ref{subsec:dependence}). In either class, sufficient
authority within $\mathcal{T}$ to authorise issuance of intercept material
constitutes a compromise event by the same operational criterion.
\end{definition}

We analyse risk over a primary horizon of $T = 10$ years, reflecting a
plausible operational deployment window. We also report cumulative risk over
$[0, 25]$ years.

\subsection{Notation Summary}\label{sec:notation_main}

Table~\ref{tab:notation} collects the core architectural and
probabilistic notation used throughout. Subscripts $\mathrm{T}$ and
$\mathrm{OTT}$ denote the two configurations where disambiguation is
required; the full glossary, including secondary symbols, is in
Supplementary \S\ref{sec:supp_notation}.

\begin{table}[!htbp]
\centering
\small
\caption{Core notation and recurring terms.}
\label{tab:notation}
\begin{tabular}{@{}lp{0.78\textwidth}@{}}
\toprule
$q_1$ & Annual (year-1) compromise probability of the relevant
        operational outcome (key-material compromise for T-EA; joint
        key-and-data compromise for OTT-EA). \\
$Q_T$ & Cumulative compromise probability over $T$ years. \\
$\La$, $\Lk$ & Pillar~II annual access-pathway and
        technical-compromise probabilities under scenario
        independence. \\
$P_3 P_4$ & OTT-EA segregation factor: probability that an attacker
        holding key material also reaches the encrypted-data store
        before mitigation (central $0.30$). \\
$\xi_c$ & Per-channel cross-cutting fraction: share of channel-$c$
        compromises traversing infrastructure shared between key
        custody and data planes. \\
$Z(t)$ & Latent annual campaign intensity,
        $Z \sim \mathrm{Gamma}(0.6, 2.0)$ at the primary
        calibration. \\
$\gamma$, $\gamma_{\mathrm{X}}$ & Campaign amplification factors
        (T-EA uniform; OTT-EA cross-cutting), with
        $\gamma_{\mathrm{X}}/\gamma \geq 1$ by construction
        (\S\ref{subsec:dependence}). \\
$\delta_{\mathrm{EA}}, \delta_{\mathrm{target}},
 \delta_{\mathrm{concen}}$ & EA exposure, targeting, and
        concentration multipliers, all $\geq 1$. \\
\midrule
Channel-minimum heuristic & Fr\'echet lower-bound floor within the
        four-channel Pillar~III model (\S\ref{sec:pillar3}). \\
Interval dominance & A finding holding at every point of a stated
        interval, here the Fr\'echet--Hoeffding brackets
        (\S\ref{sec:fh_bounds}). \\
\bottomrule
\end{tabular}
\end{table}

% =============================================================================
\section{Framework Architecture and Epistemic Foundations}\label{sec:framework}

\subsection{Rationale for a Multi-Method Approach}

No single analytical method can credibly estimate the compromise probability
of a prospective EA system. Three fundamental challenges preclude any unified
approach:
\begin{enumerate}
  \item \textit{Absence of EA-specific incident data.} Ground-truth compromise
    frequencies for EA systems do not exist in any public dataset. Any estimate
    must be an extrapolation from other system classes, a model
    projection, or a structural argument. Or, as here, a structured
    combination of all three.
  \item \textit{Uncertainty about adversary behaviour and capability
    evolution.} The threat landscape is non-stationary: adversary capabilities,
    tooling, and targeting priorities evolve in ways that historical data can
    at best partially capture.
  \item \textit{Unknown dependence structure between risk channels.} The extent
    to which multiple attack vectors, component vulnerabilities, and operational
    failures co-occur or are independent is a first-order determinant of
    system-level risk, but cannot be directly measured for a system that does
    not yet exist.
\end{enumerate}

Each of the four analytical layers individually is subject to known and
significant limitations. Their value lies in their collective ability to bound
the problem from multiple directions and to make visible the specific
assumptions on which any particular projection depends.

\subsection{The Epistemic Status of Framework Outputs}

The following taxonomy establishes the interpretation that applies to all
numerical results in subsequent sections.

\textit{Pillar~I: Empirical Upper Anchor.} The ${\sim}36\%$
population-average rate is an empirical measurement (with associated
confidence intervals) of the historical compromise rate of a defined
population of analogous systems. The implied multiplicative improvement
required to reach the low single-digit annual range under a
$3$--$10{\times}$ EA targeting premium is substantial; the
projections developed in Section~\ref{sec:pillar1} are the
quantitative form of that requirement.

\textit{Pillar~II: Scenario-Conditional Projections.} The Monte Carlo
outputs are projections conditional on (a)~the model structure, (b)~specific
parameter ranges drawn from adapted empirical sources, and (c)~the independence
assumption embedded in the aggregation formula. The Fr\'echet--Hoeffding
interval [2.2--7.5\%] provides the model-consistent range under any dependence
structure, and we regard this interval as a more honest representation of
Pillar~II's implications than any point estimate.

\textit{Pillar~III: Channel-Aggregate Projection.} A channel-aggregate projection
in the low single-digit percentage range (around 5\% annually for
T-EA) is constructed by combining empirical best-in-class component
performance under the independence assumption. It answers the
conditional question: ``If an EA system can match, but not beat, the
best historically observed performance for each of its four
irreducible risk components, and if those components fail
independently, what combined risk follows?'' The channel aggregate is an
interpretive reference within the plausibility range, not a
best-guess forecast.

\textit{Layer~IV: Prior Predictive Distributions.} The Bayesian model
replaces all point estimates with full prior predictive distributions,
explicitly preserving the epistemic uncertainty inherent in the problem.
Qualitative conclusions stable across wide parameter ranges are
distinguished from conclusions that depend on specific parameter choices. The
stochastic dominance result (Proposition~\ref{prop:dominance}) rests on the
empirically grounded observation that EA system modifiers ($\delta_{\mathrm{EA}}$,
$\delta_{\mathrm{target}}$, $\delta_{\mathrm{concen}}$) are each $\geq 1$,
supported by documented incidents (Salt Typhoon, Athens affair) and structural
analysis. It is not a definitional claim.

\subsection{Preview of the Synthesis Hierarchy}\label{sec:synthesis_preview}

The four method layers produce findings of varying epistemic
weight. Some claims survive any defensible parameter choice;
others depend on the independence assumption that Pillars~I--III
share; still others are conditional on a specific prior
calibration. The synthesis hierarchy of \S\ref{sec:synthesis}
distinguishes six tiers (S1 most assumption-robust through S6
most calibration-dependent) and assigns each principal claim to
its tier. Table~\ref{tab:interpretive_status} below gives the
interpretive status of the framework's principal numerical
outputs as a reading guide for the pillar sections that follow.

\begin{table}[!htbp]
\centering
\small
\caption{Interpretive status of the framework's principal numerical
outputs. The values are reference points for structured comparison,
not predictions. Each quantity is paired with the epistemic regime in
which it should be read.}
\label{tab:interpretive_status}
\begin{tabularx}{\textwidth}{@{}lXX@{}}
\toprule
\textbf{Quantity} & \textbf{Interpretation} & \textbf{Role in argument} \\
\midrule
$\sim$1.5\%   & Heuristic component-minima floor (conditional on four-channel
          exhaustiveness) & Lower reference under stated assumptions
          (not structural) \\
$\sim$5\%   & Channel-aggregate projection under Pillar~III independence
          assumption & Single-number reference within the
          plausibility range \\
$\sim$7\%   & Independence-based central projection from Pillar~II &
          Scenario-model reference value \\
$[2.2\%, 7.5\%]$ & Fr\'echet--Hoeffding interval for Pillar~II &
          Model-consistent range under arbitrary dependence
          (interval dominance) \\
$\sim$4\%   & Bayesian prior predictive central tendency (Layer~IV) &
          Uncertainty-preserving central reference (by prior design,
          not independent estimate) \\
$[1.4\%, 16.5\%]$ & Bayesian 90\% prior predictive interval (Layer~IV)\footnote{Throughout this paper, intervals reported from Layer~IV are \emph{prior predictive} intervals derived from forward sampling under the specified prior distributions, not posterior credible intervals in the formal Bayesian-inference sense. No EA-specific likelihood is available, so no posterior is computed; the priors carry the informational content via their calibration against analogue cohorts (Pillar~I). See \S\ref{sec:bayesian} for the prior-predictive-not-posterior-inference framing.} &
          Full distributional uncertainty range \\
$\sim$0.99   & $\Pr(\text{annual risk}>1\%)$ under Bayesian full model &
          Illustrative exceedance row; full threshold range in
          Table~\ref{tab:exceedance} \\
\bottomrule
\end{tabularx}
\end{table}

% =============================================================================
\section{Pillar~I: Historical Analogy Analysis}\label{sec:pillar1}

Pillar~I establishes architecture-conditional empirical anchors by answering
a circumscribed question: at what rate have systems broadly analogous in
sensitivity, scale, and architectural class to a prospective EA system been
compromised by major breaches in recent history? The pillar comprises three
independent data streams, each architecturally aligned with a different EA
class, and a formal cross-stream meta-analysis.

Stream~A (Section~\ref{sec:pillar1_streamA}) covers government high-security
systems whose architectural profile most closely matches T-EA: tightly
controlled operator population, vetted personnel, classified infrastructure,
and co-location of cryptographic material with the
operationally relevant data. Stream~B (Section~\ref{sec:pillar1_streamB})
covers Certificate Authority private-key infrastructure, providing an
architectural bridge: CAs hold cryptographic material at scale across
defined subscriber populations under unitary operator authority, but
without the data co-location that characterises T-EA. Stream~C
(Section~\ref{sec:pillar1_streamC}) covers platform-level key-or-data
compromise incidents from major cloud, identity, and SaaS operators,
providing the OTT-EA-leaning anchor: large user populations, segregated
key and data infrastructure, and the multi-stage compromise pattern
characteristic of the OTT case.

The three streams are pooled cautiously: each anchors a different
architectural class, and cross-stream meta-analysis is performed only to
the extent that common-rate assumptions can be justified. When the
streams disagree, that disagreement is itself architecturally informative.

\subsection{Stream~A: Government High-Security Systems}\label{sec:pillar1_streamA}

Stream~A is our T-EA anchor. Its constituent incidents share the
architectural property of $\mathcal{K}/\mathcal{D}$ co-location: the
compromised system was, in each case, a single locus where cryptographic
material and the operationally relevant data converged. The Salt Typhoon
incident is the canonical instance, with CALEA-mandated lawful-intercept
infrastructure compromised at the orchestration layer that holds both
session-key material and target-list data. The architectural fit to T-EA
is strong, modulo the inherent imperfection of any analogue.

\subsubsection{Empirical Dataset}

An incident is included in Stream~A if and only if it satisfies all three
conditions: (C1)~the system was a national-scale government system holding
or managing cryptographic key material for critical-infrastructure
communications; (C2)~operated by or under the direct authority of a national
security, signals intelligence, or critical-infrastructure agency; and
(C3)~the compromise is independently documented in a peer-reviewed publication,
government report, or contemporaneous official statement.

The six included incidents over 2011--2024 are listed in
Supplementary Table~\ref{tab:streamA_incidents} with their tier
classification and a
brief justification. Three are classified as Tier~1 (direct compromise
of an EA-architecture system, whether the breach was of its
cryptographic primitives or of its operational infrastructure); one as
Tier~2 (the SolarWinds compromise of cryptographic distribution
infrastructure via build-chain); two as Tier~3 (contextual analogues:
insider exfiltration of cryptographic tooling, and compromise of
equivalent-target-class government systems where the mechanism was not
architecturally specific to EA). The Salt Typhoon sensitivity (its
exclusion from the cohort) is reported in Supplementary
\S\ref{sec:salt_typhoon_role}.
Tier-stratified rate estimates are in Table~\ref{tab:tiered}.

\begin{table}[!htbp]
\centering
\small
\caption{Stratified annual compromise rates by analogical tier.
95\% CIs from Garwood exact Poisson intervals.
\textbf{Population rate}: probability of at least one compromise event
across the population of $N_A$ analogous systems in a given year, computed
as $1 - \exp(-k/T)$. \textbf{Per-system rate}: Poisson rate of compromise
per individual system per year, $k/(N_A T)$. The per-system rate is the
policy-relevant quantity for an individual operator. The population rate
gives the probability that \emph{some} system in the analogue cohort is
compromised in a year. The two answer different questions and should not
be conflated.}
\label{tab:tiered}
\begin{tabular}{@{}lcccc@{}}
\toprule
\textbf{Tier} & $k$ & \textbf{Population rate} & \textbf{95\% CI} &
\textbf{Per-system rate} \\
\midrule
Tier~1 only (direct analogue) & 3 & 20.6\% & $[4.6\%,\ 49.1\%]$ & 0.46\%/yr \\
Tier~1+2 (strong analogue) & 4 & 26.5\% & $[8.0\%,\ 54.5\%]$ & 0.62\%/yr \\
Tier~1+2+3 (all) & 6 & 37.0\% & $[15.6\%,\ 63.4\%]$ & 0.92\%/yr \\
\bottomrule
\end{tabular}
\end{table}

\subsubsection{Maximum Likelihood and Bayesian Estimation}

The denominator $N_A = 50$ comprises the count of national-scale systems
active during 2011--2024 that satisfy inclusion criteria C1--C2 above
and fall within the analogue class motivated by the requirements
analysis of \citet{abelson1997}, which argued that secure third-party
access infrastructure would have to operate at the assurance level of
the most sensitive existing government cryptographic systems.
The cohort is the union of five sovereign cryptographic
sub-populations: US National Security Systems ($\sim$14), Five Eyes
sovereign infrastructures ($\sim$10), EU national frameworks
($\sim$12), NATO collective infrastructure ($\sim$8), and
CALEA-mandated or ETSI-equivalent lawful-intercept carriers
($\sim$6), with the documented construction in Supplementary
\S\ref{sec:supp_streamA}. The enumeration unit is an independently
operated and accredited cryptographic domain: a system with its own
key-management infrastructure, personnel, and accreditation boundary.
The boundary depends on definitional choice in two respects---which
programmes qualify, and at what granularity a sovereign framework is
counted, with programme-level enumeration yielding the coarser count
and installation-level enumeration the finer---so $N_A$ could
plausibly be defended at any value between 35 (coarse) and 75
(fine), and the denominator bracket below is simultaneously a
unit-of-analysis sensitivity. The
qualitative finding (a per-system compromise rate of order
$1\%$/yr) is robust across this range: at $N_A = 35$ the per-system
rate becomes $6/(35\cdot 13) = 1.32\%$/yr (projecting to $\sim 12\%$
at $10\times$), and at $N_A = 75$ it becomes $6/(75\cdot 13) =
0.62\%$/yr (projecting to $\sim 6\%$ at $10\times$). Both bracket the
primary projection.

Breach events are modelled as a Poisson process with rate $\lambda$ across a
population of $N_A = 50$ analogous systems over $T = 13$ years. (The
2011--2024 window is counted as $T = 13$ years of exposure, treating
the endpoint years as partial; counting the window inclusively as 14
calendar years would lower the full-cohort per-system rate from
$0.92\%$ to $0.86\%$/yr without affecting any qualitative finding.
Streams~B and~C count their windows inclusively.) The
\emph{cohort-level} MLE rate (compromises per year across the whole
analogue cohort) is $\hat{\lambda}^{\mathrm{cohort}}_{\mathrm{MLE}} =
k/T = 6/13 = 0.462$ breaches per year, giving the cohort-annual
population rate $P_{\mathrm{MLE}}^{\mathrm{cohort}} = 1-\exp(-k/T) =
37.0\%$ with 95\% CI $[15.6\%, 63.4\%]$. The \emph{per-system} MLE rate
(the policy-relevant quantity for an individual operator) is
$\hat{\lambda}^{\mathrm{sys}}_{\mathrm{MLE}} = k/(N_A T) = 6/650 =
0.92\%$ per system-year. The two answer different questions. Both are
used below.

Under three Bayesian prior specifications (working prior
$\mathrm{Gamma}(2, 5)$; moderately optimistic $\mathrm{Gamma}(2, 100)$; and
extreme strong $\mathrm{Gamma}(2, 400)$ with an MTTF prior of 200 years),
the posterior means are 35.9\%, 6.8\%, and 1.9\% respectively. The structurally
significant finding is that under either parameterisation of an extreme 200-year
MTTF prior, the empirical data prevent the posterior mean from falling below
approximately 2\% (strong prior) or 5.6\% (diffuse prior). The five non-Weibull
methods yield a consensus central estimate of approximately 36\%.

\subsubsection{Summary}

Pillar~I Stream~A produces two distinct headline quantities that must not be
conflated. The \emph{population rate} (probability that at least one system
in the 50-system analogue cohort is compromised in a given year) is $\sim$37\%
under the full cohort. This is not the policy-relevant quantity for an
individual EA operator. The \emph{per-system rate} (Poisson rate per
individual system per year) is the policy-relevant baseline. Its value
depends on which inclusion criterion is applied:
\begin{itemize}
  \item Tier~1 only (direct T-EA architectural analogue, $k = 3$): $0.46\%$/yr
  \item Tier~1+2 (adding build-chain compromise, $k = 4$): $0.62\%$/yr
  \item Tier~1+2+3 (adding contextual analogues, $k = 6$): $0.92\%$/yr
\end{itemize}
Applying the projection formula $1 - \exp(-\lambda \cdot \text{premium})$
across this range and across the defensible EA targeting-premium range
$[3\times, 15\times]$ yields the joint range $[1.4\%, 12.9\%]$ annually for
the per-system Stream~A Pillar~I projection. The central cell of the grid
(full cohort, $10\times$ premium) gives $8.8\%$; the Tier~1 cell at
$10\times$ gives $4.5\%$; the most aggressive cell (full cohort,
$15\times$) gives $12.9\%$. Under the most optimistic Bayesian prior
(Gamma$(2, 400)$ with an MTTF prior of 200 years), the empirical data hold
the posterior mean above $1.9\%$. The Salt Typhoon sensitivity (full cohort
with it excluded, $k = 5$) is reported in
Supplementary \S\ref{sec:salt_typhoon_role} and does not change the qualitative
range.

\subsection{Stream~B: Certificate Authority Cohort}\label{sec:pillar1_streamB}

Stream~B provides an architectural bridge between Stream~A (T-EA-leaning)
and Stream~C (OTT-EA-leaning). Certificate Authorities resemble T-EA
operators in operating under regulatory oversight with unitary corporate
authority and concentrated cryptographic-key custody, but resemble OTT-EA
operators in serving large distributed subscriber populations through
key infrastructure that is operationally segregated from the certificate
issuance and revocation data flows. The architectural fit to either pure
class is partial. We treat Stream~B as anchoring the centralised
single-operator class generally, and use it as a check on the consistency
of Stream~A and Stream~C estimates rather than as a primary anchor for
either architectural extreme.

Stream~B applies Poisson cohort methods to a 19-year record of Certificate
Authority private-key compromises across 130 trusted CAs. This provides a
completely independent dataset against a known population denominator. Key
incidents include DigiNotar (2011) \citep{prins2011diginotar}, the Comodo
reseller breaches (2008, 2011), TURKTRUST (2013) \citep{turktrust2013}, the
WoSign/StartCom back-dating disclosures (2016), and the Symantec
mis-issuance pattern (2017). The complete eleven-incident list with
documentation and inclusion criterion is given in Supplementary
Section~S2. Stream~B's standalone MLE rate is $11/2{,}470 = 0.445\%$ per CA-year (Garwood 95\% CI $[0.222\%, 0.797\%]$). The annual incident counts are consistent with the Poisson
specification: the dispersion index across the 19 cohort-years is
$0.83$ (slightly under-dispersed; $\chi^2_{18}$ dispersion test
$p = 0.67$), and a negative-binomial refit collapses to the Poisson
boundary (dispersion MLE of zero). The test has limited power at
these counts, but the practical point stands either way:
overdispersion-robust intervals would be no wider than the Garwood
intervals reported. Projected through a $10{\times}$ EA targeting premium (a value in the upper-mid of the structured-judgement $3$--$15{\times}$ range; \S\ref{subsec:targeting_premium_limitation} derives the range and reports the range-median projections) this yields an annual range of $[2.20\%, 7.66\%]$, which overlaps with the Pillar~II Fr\'echet--Hoeffding interval of $[2.2\%, 7.5\%]$ derived from per-scenario probabilities. Both anchors place the centralised single-operator class in the low single-digit annual range. Stream~B and Pillar~II share calibrating evidence on cryptographic-infrastructure operator failure modes, so the agreement is partial corroboration within a coherent evidence base rather than independent verification. Pooling with Stream~A under a
common-rate Poisson model gives the fixed-effect estimator
$\hat{\lambda} = k/E = 17/3{,}120 = 0.545\%$ per system-year (Garwood
95\% CI $[0.32\%, 0.87\%]$). Under the common-rate assumption this is the
inverse-variance weighted estimator. We use it as the headline pooled rate
for the centralised single-operator class generally, while noting that
Stream~A and Stream~B individually anchor different architectural sub-cases.
The two streams have visibly different observed per-system rates (Stream~A
$0.92\%$/yr; Stream~B $0.45\%$/yr), and a strict-inclusion sensitivity
restricting Stream~B to clearly adversarial compromises (Supplementary
Section~S2) yields a lower pooled rate. Both choices produce a $10{\times}$
projection in the low-single-digit range.

\subsection{Stream~C: Platform-Level Key/Data Compromise Cohort}\label{sec:pillar1_streamC}

Stream~C is the OTT-EA case's empirical anchor. Its constituent incidents share the
architectural property that the compromised operator is a major platform
provider holding cryptographic key material, authentication infrastructure,
or encrypted user data at scale across a large user population, with
$\mathcal{K}$ and $\mathcal{D}$ operationally segregated within the
operator's infrastructure. The architectural fit to OTT-EA is partial in
the same sense that Stream~A's fit to T-EA is partial. No public
EA-specific compromise of either class exists, but the Stream~C cohort
captures the multi-stage, segregation-traversing compromise pattern that
distinguishes OTT-EA risk from T-EA risk, and provides the empirical
foundation for the $P_3$ (segregated-data acquisition) and $P_4$
(multi-stage detection avoidance) parameters introduced in
Section~\ref{sec:pillar2}.

\subsubsection{Inclusion Criteria}

An incident is included in Stream~C if and only if it satisfies all four
conditions: (D1)~the system was operated by a major platform provider
holding cryptographic key material, authentication tokens, or encrypted
user data on behalf of an external user population numbering in the
millions or larger; (D2)~the compromise involved unauthorised access to,
or exfiltration from, infrastructure components serving the platform's
key custody, identity, or encrypted-data-handling functions; (D3)~the
operator was not a national-security or critical-infrastructure agency in
the sense of Stream~A's C2 criterion (which would route the incident to
Stream~A) and was not a Certificate Authority in the sense of Stream~B
(which would route the incident there); and (D4)~the compromise is
independently documented in a vendor post-mortem or transparency report,
a regulatory disclosure (8-K, ICO, or equivalent), a national CERT or
CISA advisory, or a peer-reviewed publication. Criteria D1--D3 jointly
ensure that Stream~C is disjoint from Streams~A and~B and that included
incidents have plausible architectural relevance to a hypothetical OTT-EA
deployment.

\subsubsection{Empirical Dataset}

Fourteen incidents over $T = 7$ years (2018--2024) satisfy the
inclusion criteria. One further candidate (Cloudflare~2023) is
documented in the near-miss appendix (Supplementary
\S\ref{sec:supp_streamC_nearmiss}) but not counted in any rate calculation, on the basis
that no customer data or service keys were accessed. Each incident is
classified along four binary or ordinal
dimensions: whether key material was obtained ($K \in \{Y, P, N\}$, with
$P$ denoting partial or downstream key acquisition); whether ciphertext
or otherwise-encrypted user data was obtained ($D \in \{Y, P, N\}$);
whether \emph{joint} key-and-data compromise occurred at a level
sufficient to map operationally onto a hypothetical OTT-EA decryption
($J \in \{Y, P, N\}$); and whether the attack path traversed
\emph{cross-cutting} infrastructure shared between the key-management
and data-storage subsystems ($X \in \{Y, P, N\}$). The $J$ classification
is the most consequential for OTT-EA calibration. The $X$ classification
informs the cross-cutting fraction parameter used in
Section~\ref{sec:pillar3}.

The complete classification table is Supplementary
Table~\ref{tab:streamC}; the per-incident documentation, including
the source for each $K/D/J/X$ classification and the rationale for
inclusion or exclusion, is in Supplementary \S\ref{sec:supp_streamC}.

\subsubsection{Cohort Population and Rate Estimation}

The relevant population denominator for Stream~C is the count of platform
operators meeting the D1 criterion: major platform providers holding
cryptographic key material, authentication infrastructure, or encrypted
user data at scale, active during the 2018--2024 window. We estimate this
population at $N_C = 75$ operators, drawing on the union of: top-tier cloud
infrastructure providers ($\sim$8); major SaaS and productivity platforms
($\sim$15); identity and authentication providers ($\sim$8); password and
secret managers ($\sim$6); code hosting and CI/CD platforms ($\sim$6);
content delivery and edge providers ($\sim$5); messaging and
communications platforms ($\sim$8); and other globally significant
platform operators ($\sim$19). The full operator inventory is provided
in Supplementary \S\ref{sec:supp_streamC} and in the reproducibility archive
(Section~\ref{sec:conclusion}). We treat $N_C = 75$ as the central
value and report sensitivity across $N_C \in \{50, 100\}$.

The exposure is $E_C = N_C \cdot T = 75 \cdot 7 = 525$ operator-years.
Stratified rates by classification dimension:

\begin{table}[!htbp]
\centering
\small
\caption{Stream~C rate estimates by classification. Per-system rate
computed as $k/(N_C \cdot T)$; population rate (probability of at
least one qualifying incident across the $N_C$-operator cohort in a
given year) computed as $1 - \exp(-k/T)$. Garwood
exact Poisson 95\% CIs. The recommended headline anchor is the
\textbf{Tier~0 + Tier~1, $J = Y$ strict-joint} row at $0.95\%$/yr,
which is the rate restricted to architecturally clean OTT-EA
analogues where joint key-and-data compromise is documented rather
than interpretive. The Tier~0 row is the very-conservative lower
bracket. The all-tiers, inclusive-joint row is the upper bracket.
The qualitative ordering Stream~B $<$ Stream~A $\approx$ Stream~C at
the order of $1\%$/yr holds across the full range.}
\label{tab:streamC_rates}
\begin{tabular}{@{}lcccc@{}}
\toprule
\textbf{Classification} & $k$ & \textbf{Per-system rate} &
\textbf{95\% CI} & \textbf{Population rate} \\
\midrule
\multicolumn{5}{l}{\emph{Architectural-fit stratification}} \\
Tier~0 only (gold standard)        &  3 & 0.57\%/yr & $[0.12\%, 1.67\%]$ & 34.9\% \\
Tier~0 + Tier~1                    &  5 & 0.95\%/yr & $[0.31\%, 2.22\%]$ & 51.0\% \\
Tier~0 + Tier~1 + Tier~2           & 12 & 2.29\%/yr & $[1.18\%, 4.00\%]$ & 82.0\% \\
All tiers (incl.\ Tier~3)          & 14 & 2.67\%/yr & $[1.46\%, 4.48\%]$ & 86.5\% \\
\midrule
\multicolumn{5}{l}{\emph{Joint-compromise stratification ($J$)}} \\
$J = Y$ strict, all tiers          &  7 & 1.33\%/yr & $[0.54\%, 2.75\%]$ & 63.2\% \\
$J = Y$ strict, Tier~0 + 1         &  5 & 0.95\%/yr & $[0.31\%, 2.22\%]$ & 51.0\% \\
$J \in \{Y, P\}$ inclusive          & 13 & 2.48\%/yr & $[1.32\%, 4.23\%]$ & 84.4\% \\
\midrule
\multicolumn{5}{l}{\emph{Robustness checks}} \\
$J = Y$ strict, no Microsoft       &  6 & 1.14\%/yr & $[0.42\%, 2.49\%]$ & 57.6\% \\
$J = Y$, Tier~0+1, no Microsoft    &  4 & 0.76\%/yr & $[0.21\%, 1.95\%]$ & 43.5\% \\
$J = Y$ strict, $N_C = 50$         &  7 & 2.00\%/yr & $[0.80\%, 4.12\%]$ & 63.2\% \\
$J = Y$ strict, $N_C = 100$        &  7 & 1.00\%/yr & $[0.40\%, 2.06\%]$ & 63.2\% \\
\bottomrule
\end{tabular}
\end{table}

The headline Stream~C estimate is the architectural-fit-restricted
strict-joint rate $\hat\lambda_C^{\mathrm{T0+T1, J=Y}} = 0.95\%$ per
operator-year ($k = 5$ over $E_C = 525$ operator-years, Garwood
$95\%$ CI $[0.31\%, 2.22\%]$). This is the rate restricted to
architecturally clean OTT-EA analogues: incidents where the
operator's key-custody role is documented (Tier~0 or 1) and the
compromise involved \emph{joint} key-and-data acquisition at a level
operationally sufficient to map onto OTT-EA decryption ($J = Y$). We
also report a very-conservative Tier~0 lower bracket of $0.57\%$/yr
($k = 3$) and an all-tiers inclusive-$J$ upper bracket of $2.48\%$/yr
($k = 13$). The substantive comparison to Stream~A
($\sim\!0.92\%$/yr) and Stream~B ($\sim\!0.45\%$/yr) is robust to
this interval. The all-tiers strict-joint rate of $1.33\%$/yr
($J = Y$ across all tiers; coincidentally also the rate implied by
the seven cross-cutting incidents, $X = Y$) provides a closely
related upper anchor specifically for the correlated-campaign
analysis in Section~\ref{sec:bayesian}.

\subsubsection{Architectural Caveats and Cross-Stream Comparison}

\paragraph{Caveats.} \emph{(i)~Denominator imprecision.} $N_C = 75$
is constructed from rough subcategory estimates and is less precisely
defined than $N_A = 50$ or $N_B = 130$. The central value could
plausibly be defended at any value in $[50, 100]$, with rate
sensitivity reported in Table~\ref{tab:streamC_rates}.
\emph{(ii)~Operator clustering.} Three of the fourteen incidents
involve Microsoft. The no-Microsoft sensitivity rows ($k = 6$ across
all tiers, $1.14\%$/yr; $k = 4$ restricted to Tier~0+1, $0.76\%$/yr)
bracket the alternative interpretation that Microsoft's incident
clustering reflects within-operator concentration rather than a
population-level rate. \emph{(iii)~Disclosure visibility bias.}
Selection is driven by disclosure visibility, and the bias runs in
both directions. Toward overstatement: catastrophic compromises of
prominent operators are far more visible than the quiet operating
history of the broader cohort, so an incident-led literature can
inflate perceived base rates. Toward understatement: SEC-reporting
US-based public companies are over-represented relative to private
operators and operators in jurisdictions without mandatory
cyber-incident disclosure, so incidents outside that visibility
regime are missing from the numerator while the denominator counts
the operators. The construction here controls the first direction by
fixing the denominator to an explicit operator inventory ($N_C = 75$,
Supplementary \S\ref{sec:supp_streamC}) rather than letting
prominence define the cohort, which leaves the
missing-incident direction as the residual bias. On that basis the
headline rate is treated as a lower bound, while the two-sided
character of the selection problem is acknowledged rather than
assumed away. \emph{(iv)~Time-window selection.} The 2018--2024
window is bounded below by EPSS introduction (2018) and post-2018
maturation of cloud-platform architectures matching the OTT-EA
reference class; earlier candidates (Yahoo 2013--2014, RSA SecurID
2011, Adobe 2013) would shift the central rate toward the lower
bracket without changing the qualitative ordering.

\paragraph{Cross-stream ordering.} The Stream~C strict-joint Tier-0+1
rate ($0.95\%$/yr) exceeds Stream~B's CA-cohort rate ($0.45\%$/yr)
and is close to Stream~A's per-system rate ($0.92\%$/yr
at the primary inclusion criterion). Two architectural factors are
consistent: platform operators present larger and more dynamic
attack surfaces than tightly controlled CA infrastructure, and the
joint-compromise classification specifically counts the multi-stage
compromise pattern characteristic of OTT-relevant attacks. The
ordering Stream~B $<$ Stream~A $\approx$ Stream~C, all on the order
of $1\%$/yr, is robust across the full range of Tier and $J$
stratifications, the no-Microsoft sensitivity, and the $N_C$
population sensitivity.

\subsection{Cross-Stream Meta-Analysis}

The three streams provide architecture-conditional empirical anchors:
Stream~A for T-EA, Stream~C for OTT-EA, and Stream~B as the centralised
single-operator bridge. We do \emph{not} pool all three streams into a
single rate, because the architectural classes they anchor differ in
ways that the common-rate Poisson model would obscure. Instead, we
report stream-conditional projections and use the cross-stream
disagreement itself as architectural information.

The appropriate targeting premium is uncertain and likely
architecture-conditional. The most directly relevant external
calibration is the Verizon DBIR's actor-mix data: the share of
breaches attributable to state-aligned espionage actors is several
times higher for high-value and government target classes than the
cross-sector base rate, a disparity of roughly $3{\times}$ to
$10{\times}$ depending on sector and reporting year
\citep{dbir2024,dbir2025}. The Salt Typhoon incident demonstrates
successful targeting of CALEA infrastructure, and the Storm-0558
incident equivalently successful targeting of platform-level
signing-key infrastructure; neither incident provides a
frequency multiplier on its own. We instead present projections across
a range of premiums ($3{\times}$, $5{\times}$, $10{\times}$, $15{\times}$),
applied separately to each stream's empirical base rate.

\textit{T-EA projection (Stream~A anchor).} Applying targeting premiums
of $3{\times}$ to $15{\times}$ across the three concentric inclusion
criteria (Tier~1, Tier~1+2, full cohort) for Stream~A's per-system rate
yields a joint projection range of $1.4\%$ to $12.9\%$ annually. The
central cell (full cohort, $10{\times}$ premium) gives $8.8\%$; the
most conservative cell (Tier~1, $3{\times}$) gives $1.4\%$; the most
aggressive (full cohort, $15{\times}$) gives $12.9\%$. The Tier~1 cell
at $10{\times}$ gives $4.5\%$. The defensible T-EA empirical-projection
interval is therefore $1.4$--$12.9\%$, with the centre of the grid in
the low single-digit range.

\textit{Centralised single-operator projection (Stream~B anchor).} At a
$10{\times}$ premium applied to Stream~B's $0.45\%$/yr, the projection is
$4.4\%$ annually. At $3{\times}$, $1.3\%$; at $15{\times}$, $6.5\%$.
(All premium projections in this section use the transform
$1 - \exp(-\mathrm{premium} \cdot \lambda)$.)

\textit{OTT-EA projection (Stream~C anchor).} At a $10{\times}$ premium
applied to Stream~C's architectural-fit-restricted strict-joint rate of
$0.95\%$/yr, the projection is $9.1\%$ annually (or $12.5\%$ at the
all-tiers strict-joint upper bracket of $1.33\%$/yr). However, the appropriate
targeting premium for OTT-EA
relative to the Stream~C base may be \emph{lower} than the T-EA
analogue, because the platform-operator base rate already reflects
substantial existing targeting (these operators are already
high-priority adversarial targets), whereas Stream~A's
government-system base rate reflects systems that are largely shielded
from non-EA-related routine targeting. Treating the EA-specific premium
as $3$--$5\times$ over the Stream~C headline base gives an OTT-EA
projection of $2.8$--$4.7\%$ annually under independence
($3.9$--$6.4\%$ at the $1.33\%$/yr all-tiers upper bracket), before applying the
segregation discount that distinguishes OTT-EA from T-EA at the
operational-compromise level (Section~\ref{sec:pillar2}).

Figure~\ref{fig:premium_sensitivity_arch} illustrates the premium
sensitivity across all three Pillar~I streams (Stream~A T-EA-anchored,
the pooled Stream~A+B centralised single-operator class, and Stream~C
OTT-EA-anchored), with the applicable premium range differing across
architectures.

\begin{figure}[!htbp]
\centering
\includegraphics[width=\textwidth]{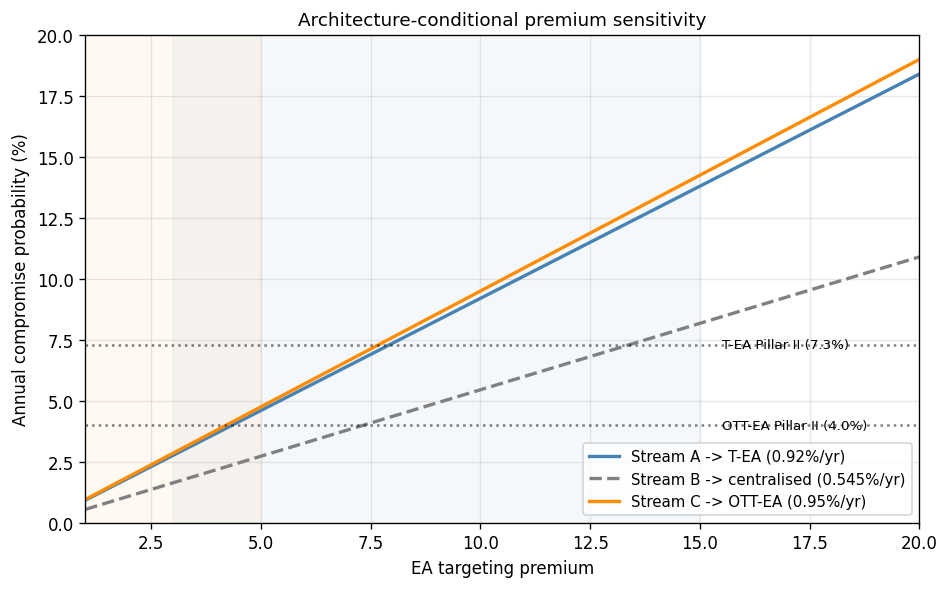}
\caption{Architecture-conditional sensitivity of the meta-analysis
projection to the EA targeting premium. The empirical base rates
(pooled centralised single-operator: $0.545\%$/yr, $k = 17$,
$E = 3{,}120$ system-years from Streams~A and~B; Stream~A T-EA-anchored:
$0.92\%$/yr; Stream~C OTT-EA-anchored: $0.95\%$/yr) are multiplied by
a premium representing elevated EA targeting. Applicable premium
ranges differ across architectures: T-EA and the pooled centralised
class span $3\times$--$15\times$; OTT-EA's defensible range is
$1\times$--$5\times$ (because Stream~C's base rate already reflects
substantial existing platform targeting). At the centralised
$10\times$ choice the projection is $5.3\%$, close to
the Pillar~III channel-aggregate projection of $5.4\%$;
$5\times$ gives $2.7\%$ and $15\times$ gives $7.8\%$. All projections
use the transform $1-\exp(-\mathrm{premium}\cdot\lambda)$. At their respective central premium choices,
the per-stream projections converge in the low-single-digit annual
range. Shaded bands: 95\% CIs transformed through each premium.
\textbf{No single premium is empirically identified}: the
$10\times$ value used in subsequent text is an illustrative
upper-central choice rather than a measurement.}
\label{fig:premium_sensitivity_arch}
\par\smallskip\noindent\small\textit{Alt text:} Line plot of projected annual EA compromise probability as a function of the EA targeting premium (1x to 20x) for three base-rate anchors. Pooled centralised single-operator at 0.545\%/yr (Streams A+B), Stream A T-EA-anchored at 0.92\%/yr, and Stream C OTT-EA-anchored at 0.95\%/yr; projections remain in low single-digit percent across the 3x to 15x defensible range.
\end{figure}

% =============================================================================
\section{Pillar~II: Prospective Scenario Model}\label{sec:pillar2}

Pillar~II uses Monte Carlo simulation to project the annual EA
compromise probability under a defined set of threat scenarios. The
model fixes seven scenarios spanning insider misuse, supply-chain
attacks, zero-day exploitation, advanced persistent threats,
operational failures, cryptographic compromise, and architectural
exploits. For each scenario, parameter values are drawn from
probability distributions calibrated to public threat-intelligence
data (Mandiant M-Trends, EPSS). The pillar produces both an annual
compromise probability under independence between scenarios and a
Fr\'echet--Hoeffding bracket of the rate under any dependence
structure compatible with the marginal distributions. Both T-EA and
OTT-EA configurations are evaluated with a common scenario set and
a common adversary parameterisation, so divergence between
projections reflects architectural difference rather than
calibration choices. The reader may take the headline FH brackets
($[2.2\%, 7.5\%]$ for T-EA and $[1.1\%, 4.0\%]$ for OTT-EA) as the
Pillar~II contribution to the framework's calibration-conditional
output ranges, and skip the technical details of scenario
calibration and Monte Carlo implementation.

\subsection{Model Architecture and Epistemic Framing}

Pillar~II is an explicit scenario-projection model, not a calibrated
predictive model: no EA compromise events exist against which to validate
it. Its contributions are the decomposition of risk into specific
parameterised threat vectors, transparent propagation of input uncertainty,
and quantification of the sensitivity of projections to the independence
assumption through Fr\'echet--Hoeffding bounding.

The Pillar~II model is structured to produce architecture-conditional
projections. A common scenario set $S_1, \ldots, S_7$ and a common
two-tier adversary parameterisation are applied across both architectural
classes, with the architectural distinction entering through the
conditional outcome decomposition (Section~\ref{sec:pillar2_outcome}):
T-EA outcomes follow a two-stage decomposition $\pcondkey_i =
P_1 \cdot P_2$, while OTT-EA outcomes follow a four-stage decomposition
incorporating the segregated-data acquisition and multi-stage detection
parameters. This structure ensures that the architecture-conditional
projections share their underlying threat-model parameterisation, so
that the divergence between T-EA and OTT-EA projections is interpretable
as an architectural effect rather than a parameterisation choice.

\subsection{Two-Tier Adversary Model and Parameter Structure}

Each scenario $S_i$ ($i = 1, \ldots, 7$) has a mixing weight
$\wapt_i \in [0, 1]$ representing the fraction of attempts from nation-state
actors. For iteration $k$:
\begin{align}
T_{ki} &\sim \mathrm{Bernoulli}(\wapt_i), \label{eq:tier_indicator} \\
P^{\mathrm{EPSS}}_{ki} &= T_{ki} \cdot X^{\mathrm{APT}}_{ki}
                       + (1 - T_{ki}) \cdot X^{\mathrm{opp}}_{ki},
\label{eq:tier_mixing}
\end{align}
where $X^{\mathrm{opp}}_{ki} \sim \mathrm{Beta}(\alpha^{\mathrm{opp}}_i,
\beta^{\mathrm{opp}}_i)$ is calibrated to the EPSS score $P_{10}$--$P_{90}$
range for CVEs in the relevant ATT\&CK technique class, and
$X^{\mathrm{APT}}_{ki}$ applies an APT-tier exploitation uplift to the
opportunistic midpoint, modelled as a log-normal multiplier with median
$2{\times}$ and examined over a $1.2{\times}$--$3.0{\times}$ sensitivity range.
Scenarios $S_1$, $S_3$, and $S_4$ form a coordinated-campaign group with
campaign-activation probability $p^{\mathrm{camp}} \in [0.10, 0.20]$ and boost
factor $b \in [1.30, 1.80]$.

\subsection{Architecture-Conditional Outcome Decomposition}\label{sec:pillar2_outcome}

Each scenario is parameterised with sequential conditional probabilities
that decompose differently for the two architectural classes. The first
stage, in which a scenario attempt $E_i$ produces access-pathway
compromise, is common across architectures:
\begin{equation}
C^{\mathrm{acc}}_i \mid E_i = 1 \sim \mathrm{Bernoulli}(\pcondacc_i),
\label{eq:cacc}
\end{equation}
representing the probability of obtaining unauthorised operational access
to the EA system given an attempt via scenario~$i$.

\subsubsection{T-EA Outcome Decomposition}

For T-EA, the second stage, technical compromise (key-material
exfiltration) given access, uses a two-factor decomposition:
\begin{equation}
C^{\mathrm{key},\mathrm{T}}_i \mid C^{\mathrm{acc}}_i = 1
  \sim \mathrm{Bernoulli}(\pcondkey_i),
\quad \pcondkey_i = P_1 \cdot P_2,
\label{eq:pkey_T}
\end{equation}
where $P_1$ is the probability of reaching the key-management layer given
system access, and $P_2$ is the HSM-tier-conditional probability of key
extraction. Because $\mathcal{K}$ and $\mathcal{D}$ are co-located in
T-EA architectures (Definition~\ref{def:segregation},
Remark~\ref{rem:collapse}), operational compromise collapses to technical
compromise: $C^{\mathrm{op},\mathrm{T}}_i = C^{\mathrm{key},\mathrm{T}}_i$.

\subsubsection{OTT-EA Outcome Decomposition}

For OTT-EA, the technical compromise stage is identical. Obtaining key
material from $\mathcal{K}$ has the same mechanics regardless of whether
$\mathcal{D}$ is co-located, but operational compromise additionally
requires obtaining ciphertext from architecturally segregated $\mathcal{D}$
and completing the multi-stage operation without detection-and-block:
\begin{align}
C^{\mathrm{key},\mathrm{OTT}}_i \mid C^{\mathrm{acc}}_i = 1
  &\sim \mathrm{Bernoulli}(\pcondkey_i),
\label{eq:pkey_OTT} \\
C^{\mathrm{op},\mathrm{OTT}}_i \mid C^{\mathrm{key},\mathrm{OTT}}_i = 1
  &\sim \mathrm{Bernoulli}\!\left(p^{\mathrm{op,seg}}_i\right),
\label{eq:pop_OTT}
\end{align}
where $p^{\mathrm{op,seg}}_i$ is the conditional probability of completing
operational compromise given technical compromise. This factor decomposes
into two terms reflecting the two ways an attacker reaches operational
compromise from a position of having obtained keys:
\begin{equation}
p^{\mathrm{op,seg}}_i = \xi_i + (1 - \xi_i) \cdot P_3 \cdot P_4,
\label{eq:pop_seg}
\end{equation}
where:
\begin{itemize}[itemsep=2pt,topsep=2pt]
\item $\xi_i \in [0, 1]$ is the \emph{cross-cutting fraction} for
scenario~$i$: the probability that a scenario-$i$ attack path
traverses shared infrastructure (identity provider, build pipeline,
privileged-role credentials) so that segregation between
$\mathcal{K}$ and $\mathcal{D}$ does not present an additional
barrier;
\item $P_3 \in [0, 1]$ is the conditional probability of obtaining
ciphertext from $\mathcal{D}$ given key access in the
\emph{non-cross-cutting} case;
\item $P_4 \in [0, 1]$ is the conditional probability of completing
the multi-stage operation without detection-and-block in the
non-cross-cutting case.
\end{itemize}
A fraction $\xi_i$ of attack paths inherently bypass segregation
(operational factor $= 1$); the remaining $(1 - \xi_i)$ must traverse
a multi-stage chain with conditional success $P_3 \cdot P_4$.

\subsubsection{System-Level Aggregation}

System-level outcomes aggregate using the complement-of-product formula:
\begin{align}
\La &= 1 - \prod_{i=1}^{7} \!\left(1 - C^{\mathrm{acc}}_i\right),
\label{eq:cop_acc} \\
\Lk &= 1 - \prod_{i=1}^{7} \!\left(1 - C^{\mathrm{key},\mathrm{T}}_i\right),
\label{eq:cop_key} \\
\Lambda^{\mathrm{op},\mathrm{T}} &= \Lk
  \quad \text{(architectural collapse, Remark~\ref{rem:collapse})},
\label{eq:lop_T} \\
\Lambda^{\mathrm{op},\mathrm{OTT}}
  &= 1 - \prod_{i=1}^{7} \!\left(1 - C^{\mathrm{op},\mathrm{OTT}}_i\right).
\label{eq:lop_OTT}
\end{align}
This formula is exact under mutual independence of scenario outcomes.
That is a strong assumption, and its implications are bounded in
Section~\ref{sec:fh_bounds} and addressed structurally in
Section~\ref{sec:bayesian}.

\subsubsection{Calibration of $\xi_i$, $P_3$, and $P_4$}

The cross-cutting fractions $\xi_i$ are calibrated against the Stream~C
incident base (Section~\ref{sec:pillar1_streamC}; Supplementary Table~\ref{tab:streamC}),
mapping each Stream~C incident's documented attack path onto the
nearest of the seven scenarios. The Stream~C aggregate cross-cutting
fraction $X = Y$ is $7/14 = 0.50$. Per-scenario values reflect the
architectural specificity of each scenario:

\begin{table}[!htbp]
\centering
\small
\caption{Per-scenario cross-cutting fractions for the OTT-EA outcome
decomposition. Values calibrated against Stream~C
(Section~\ref{sec:pillar1}) with sensitivity range reflecting
classification uncertainty. These are Pillar~II scenario-level
quantities, distinct from the Pillar~III per-channel fractions of
Table~\ref{tab:xi_channels}.}
\label{tab:xi_scenarios}
\begin{tabular}{@{}llccl@{}}
\toprule
\textbf{Scenario} & \textbf{Description} &
\textbf{Central $\xi_i$} & \textbf{Range} &
\textbf{Calibration} \\
\midrule
$S_1$ & Stolen credentials      & 0.30 & $[0.20, 0.45]$
  & Shared identity provider risk \\
$S_2$ & External app exploit    & 0.20 & $[0.10, 0.30]$
  & Typically subsystem-localised \\
$S_3$ & Edge exploitation       & 0.45 & $[0.30, 0.60]$
  & Edge appliances commonly cross-cutting \\
$S_4$ & Lateral / supply chain  & 0.55 & $[0.40, 0.70]$
  & Cross-cutting by design \\
$S_5$ & Insider                 & 0.30 & $[0.15, 0.50]$
  & Privileged insider variability \\
$S_6$ & HSM extraction          & 0.10 & $[0.05, 0.20]$
  & Highly key-specific \\
$S_7$ & Cryptographic protocol  & 0.10 & $[0.05, 0.20]$
  & Highly key-specific \\
\bottomrule
\end{tabular}
\end{table}

The conditional probabilities $P_3$ and $P_4$ are calibrated against the
Stream~C incidents that did not exhibit cross-cutting attack paths, and
against Mandiant M-Trends 2024 dwell-time and detection-block data
respectively. Central values: $P_3 = 0.50$ (range $[0.30, 0.80]$),
$P_4 = 0.60$ (range $[0.40, 0.80]$). The product $P_3 \cdot P_4$
has central value $0.30$ with range $[0.12, 0.64]$.

These calibrations are deliberately wide. The Stream~C incident base of
14 incidents is too small to support tight $P_3$ or $P_4$ estimates, and
the $\xi_i$ values above are primarily judgement-based mappings of
Stream~C incident classifications onto the seven scenarios. The
sensitivity analysis in Section~\ref{sec:pillar2_sensitivity} reports
how the OTT-EA projection varies across the full joint range, and
Supplementary \S\ref{sec:supp_ott_pillar2_grid} tabulates the
$P_3 P_4$-by-$\xi$ grid in full.

\subsection{System-Level Annual Projections}\label{sec:pillar2_proj}

The central independence-based projection aggregates the seven
per-scenario technical-compromise probabilities at their central
values; computed in closed form, it carries no prior-width parameter.
A Monte Carlo with $N = 200{,}000$ iterations propagates parameter
uncertainty around this projection, drawing the per-scenario rates and
the shared aggregation parameters from the priors specified in
Supplementary~\S\ref{sec:supp_pillar2_priors}. The Monte Carlo median
agrees with the closed-form central to within rounding and, because
each prior is centred on the corresponding central value, the central
figures reported below are governed by those central values rather
than by the prior \emph{widths}. The widths affect only the dispersion
of the simulated distribution. That dispersion is examined in the
sensitivity analysis of Section~\ref{sec:pillar2_sensitivity} and is
not reported as a confidence interval on the central projection. The
projections are reported separately for each architectural class.

\subsubsection{T-EA projections}

For T-EA, the central independence-based projections are:
EA access-pathway $\La = 16.8\%$; EA technical compromise $\Lk = 7.3\%$;
operational compromise $\Lambda^{\mathrm{op},\mathrm{T}} = \Lk = 7.3\%$
(architectural collapse); non-EA baseline access-pathway 12.3\%; non-EA
baseline technical compromise 0\% by construction. The dominant
contributors to system-level risk are $S_1$ (stolen credentials,
$\sim$2.2\% technical compromise per scenario), $S_3$ (edge exploitation,
$\sim$1.8\%), and $S_5$ (insider, $\sim$1.7\%). EA-specific scenarios
$S_6$ (HSM extraction) and $S_7$ (cryptographic protocol) each contribute
less than 0.5\% on average, reflecting the rarity of the underlying
technique classes while marking the existence of risk channels unique to
EA systems.

\subsubsection{OTT-EA projections}

For OTT-EA, the same per-scenario technical compromise probabilities
$\Pr(C^{\mathrm{key}}_i)$ are propagated through
Equations~\eqref{eq:pop_OTT}--\eqref{eq:pop_seg}, with the per-scenario
cross-cutting fractions from Table~\ref{tab:xi_scenarios} and the central
values $P_3 = 0.50$, $P_4 = 0.60$ (so $P_3 \cdot P_4 = 0.30$).

The per-scenario operational-compromise factor
$p^{\mathrm{op,seg}}_i = \xi_i + (1-\xi_i) \cdot 0.30$ takes values
ranging from $0.37$ ($S_6, S_7$, where the architectural specificity
limits cross-cutting) to $0.69$ ($S_4$, where supply-chain attack paths
are cross-cutting by design). Applied to the per-scenario technical
compromise probabilities and aggregated, this gives:
\begin{equation}
\Lambda^{\mathrm{op},\mathrm{OTT}}_{\mathrm{central}} \approx 3.9\%.
\end{equation}

The dominant contributors shift modestly relative to T-EA: $S_3$ and
$S_4$ retain or strengthen their relative contribution under OTT
because their higher cross-cutting fractions reduce the segregation
discount, while $S_6$ and $S_7$ become less significant because their
high architectural specificity attracts the full segregation discount.
$S_1$ and $S_5$ remain prominent contributors with intermediate
cross-cutting fractions.

\subsubsection{T-EA and OTT-EA central comparison}

The central projection ratio $\Lambda^{\mathrm{op},\mathrm{OTT}} /
\Lambda^{\mathrm{op},\mathrm{T}} \approx 0.54$ at the central
parameterisation expresses the segregation gain as a near-halving of
the T-EA central operational compromise probability for OTT-EA under
independence. The interpretation is architectural: well-segregated OTT
systems benefit from the multi-stage AND-chain at the central
tendency, but, as the next two subsections show, this gain is
strongly conditional on the cross-cutting fraction and substantially
attenuated at the tail.

\begin{figure}[!htbp]
\centering
\includegraphics[width=0.95\textwidth]{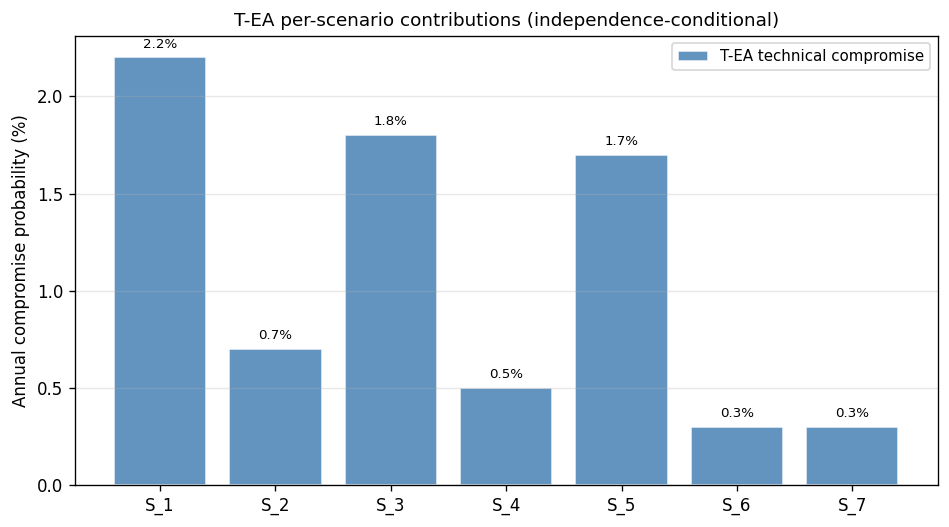}
\caption{Per-scenario mean annual compromise contribution under the
Pillar~II central parameterisation ($10^5$ iterations, seed 2024). Error
bars show 5th--95th percentile range on access-pathway probability.
Hatched bars give mean technical compromise probability per scenario.
The dominant contributors are $S_1$ (stolen credentials, $\sim$2.2\%),
$S_3$ (edge exploitation, $\sim$1.8\%), and $S_5$ (insider,
$\sim$1.7\%). EA-specific specialist scenarios $S_6$ (HSM extraction)
and $S_7$ (cryptographic protocol) each contribute less than 0.5\%.
\textbf{Independence-conditional values for the T-EA projection.} The
parallel OTT-EA contributions are given in
Figure~\ref{fig:scenario_contributions_ott} (Section~\ref{sec:results}),
where each scenario's contribution is reweighted by its
operational-compromise factor $p^{\mathrm{op,seg}}_i$.}
\label{fig:scenario_contributions}
\par\smallskip\noindent\small\textit{Alt text:} Bar chart of mean annual compromise contribution per scenario S1 through S7 under the Pillar II T-EA central parameterisation. Three scenarios dominate: S1 (stolen credentials, approximately 2.2\%), S3 (edge exploitation, approximately 1.8\%), and S5 (insider); remaining scenarios contribute under 1\%. Error bars show 5th to 95th percentile ranges.
\end{figure}

\begin{figure}[!htbp]
\centering
\includegraphics[width=0.95\textwidth]{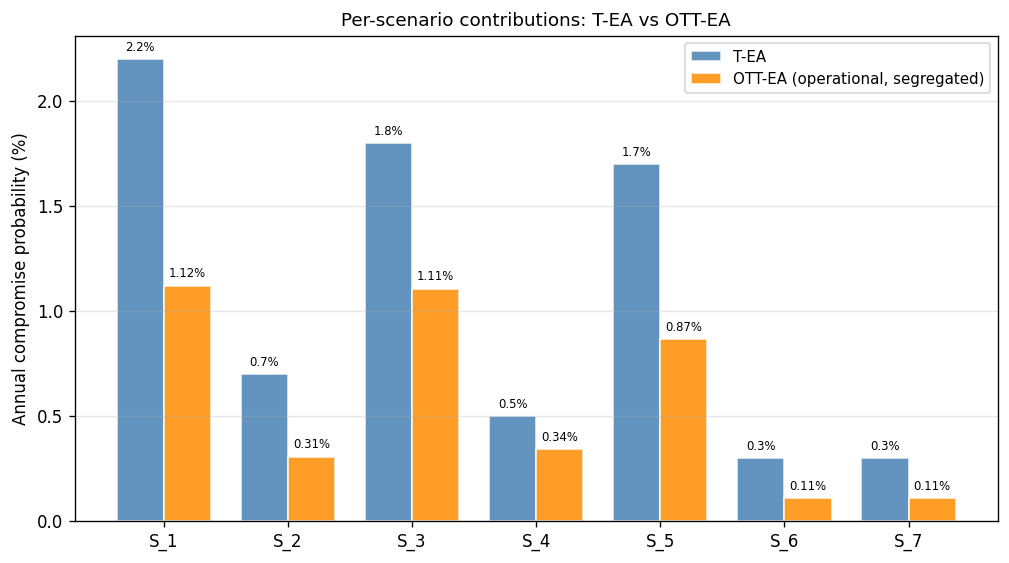}
\caption{Per-scenario mean annual operational compromise contribution
under the Pillar~II OTT-EA central parameterisation, with each scenario's
contribution weighted by its operational-compromise factor
$p^{\mathrm{op,seg}}_i = \xi_i + (1-\xi_i) P_3 P_4$. Side-by-side
comparison with T-EA contributions shows that $S_4$ (lateral / supply
chain) retains its relative contribution because of its high
cross-cutting fraction ($\xi = 0.55$), while $S_6$ (HSM extraction) and
$S_7$ (cryptographic protocol) are further attenuated because of their
low cross-cutting fractions and consequent near-full segregation
discount. $10^5$ iterations, seed 2024.}
\label{fig:scenario_contributions_ott}
\par\smallskip\noindent\small\textit{Alt text:} Bar chart of per-scenario operational-compromise contributions under the Pillar II OTT-EA central parameterisation, with each scenario weighted by its operational-compromise factor. S4 (lateral / supply chain) retains its relative contribution because of high cross-cutting fraction; remaining scenarios are attenuated by the multi-stage detection structure.
\end{figure}

\subsection{Fr\'echet--Hoeffding Bounding Analysis}\label{sec:fh_bounds}

The independence-based projections are conditional on the
complement-of-product formula. To bound the effect of this assumption,
we apply the Fr\'echet--Hoeffding limits for the union of dependent
events, separately to each architectural class. The Fr\'echet--Hoeffding
bounds are mathematically valid for any dependence structure given the
per-scenario marginal probabilities estimated from the model. The
qualifier ``model-consistent'' acknowledges that the marginals
themselves are model outputs, not known constants. The bounds are sharp
with respect to the dependence assumption but inherit any uncertainty
in the marginal calibration.

\subsubsection{T-EA bounds}

For T-EA, the Fr\'echet--Hoeffding bounds take the standard form:
\begin{equation}
\max_i \Pr(C^{\mathrm{key},\mathrm{T}}_i)
\leq \Pr\!\left(\bigcup_i C^{\mathrm{key},\mathrm{T}}_i\right)
\leq \min\!\left(1,\ \sum_i \Pr(C^{\mathrm{key},\mathrm{T}}_i)\right).
\label{eq:fh_T}
\end{equation}
Substituting per-scenario technical compromise means: lower bound
(maximum individual, $S_1$) $\approx 2.2\%$; upper bound (sum of means)
$\approx 7.5\%$. The independence-based projection of 7.3\% sits near
the upper bound. The model-consistent T-EA interval under any
dependence structure is therefore approximately $[2.2\%, 7.5\%]$.

\subsubsection{OTT-EA bounds}

For OTT-EA, the same bounding formula applies to the per-scenario
operational compromise events $C^{\mathrm{op},\mathrm{OTT}}_i$, whose
marginal probabilities satisfy $\Pr(C^{\mathrm{op},\mathrm{OTT}}_i) =
\Pr(C^{\mathrm{key}}_i) \cdot p^{\mathrm{op,seg}}_i$:
\begin{equation}
\max_i \Pr(C^{\mathrm{op},\mathrm{OTT}}_i)
\leq \Pr\!\left(\bigcup_i C^{\mathrm{op},\mathrm{OTT}}_i\right)
\leq \min\!\left(1,\ \sum_i \Pr(C^{\mathrm{op},\mathrm{OTT}}_i)\right).
\label{eq:fh_OTT}
\end{equation}

Because each $\Pr(C^{\mathrm{op},\mathrm{OTT}}_i) =
\Pr(C^{\mathrm{key}}_i) \cdot p^{\mathrm{op,seg}}_i$ and
$p^{\mathrm{op,seg}}_i \leq 1$, both bounds for OTT-EA are
\emph{less than or equal to} the corresponding T-EA bounds
under independence. Substituting per-scenario means at the central
parameterisation:
\begin{itemize}[itemsep=2pt]
  \item Lower bound (maximum individual): $\Pr(C^{\mathrm{op},
        \mathrm{OTT}}_{S_1}) \approx 2.2\% \times 0.51 \approx 1.1\%$
        (where $0.51 = \xi_{S_1} + (1-\xi_{S_1}) \cdot 0.30 = 0.30 +
        0.70 \cdot 0.30$);
  \item Upper bound (sum of per-scenario means):
        $\sum_i p^{\mathrm{op,seg}}_i \cdot \Pr(C^{\mathrm{key}}_i)
        \approx 4.0\%$
        (equivalently, the $\Pr(C^{\mathrm{key}}_i)$-weighted average of
        $p^{\mathrm{op,seg}}_i$ is approximately $0.53$, so the upper
        bound is approximately $0.53 \times 7.5\% \approx 4.0\%$).
\end{itemize}
The model-consistent OTT-EA interval under any dependence structure is
therefore approximately $[1.1\%, 4.0\%]$ at the central
$P_3, P_4, \xi_i$ parameterisation.

\subsubsection{Architectural comparison and the segregation gain}

The architecturally interesting comparison is between the T-EA interval
$[2.2\%, 7.5\%]$ and the OTT-EA interval $[1.1\%, 4.0\%]$. The OTT-EA
interval is shifted downward, reflecting the segregation gain. The
\emph{interval width} is also smaller for OTT, reflecting that the
segregation factor compresses the effect of underlying parameter
variation. However, this comparison is at independent risk channels
within each class. The architecturally distinctive feature of OTT-EA
is that under correlated campaigns the segregation gain partially or
fully collapses (Section~\ref{sec:bayesian}).

\textit{Caveat on the OTT-EA lower bound.} The Fr\'echet--Hoeffding
lower bound is the comonotonic case, in which the seven scenarios are
treated as if perfectly positively correlated. For OTT-EA, perfect
correlation across scenarios is the worst case from the segregation
perspective: it corresponds to a single APT campaign exploiting all
seven attack vectors simultaneously, which is precisely the condition
under which cross-cutting infrastructure (large $\xi_i$ effective
values) would be activated. The OTT-EA Fr\'echet lower bound therefore
\emph{understates} the worst-case dependence risk, because the
$p^{\mathrm{op,seg}}_i$ values used to derive it are calibrated against
the average cross-cutting fraction rather than the campaign-correlated
worst case. It is this limitation that motivates the explicit
correlated-campaign treatment in Section~\ref{sec:bayesian}, where
$Z(t)$ is coupled specifically to cross-cutting edges.

\textit{Caveat on the upper bound.} As in the T-EA case, the
Fr\'echet--Hoeffding upper bound is attained when the seven scenario
events are pairwise disjoint, and it is informative only under the
assumption that the seven scenarios exhaust the attack surface.
(Countermonotonicity has no analogue above two dimensions, so the
upper bound is characterised by disjointness rather than by a
negative-dependence extreme.) The seven scenarios cover the dominant attack vectors
documented in DBIR, ATT\&CK, and CISA advisory sources but do not
constitute a mathematically complete enumeration. The bound is
``model-consistent'' rather than ``model-free''.

The OTT-EA interval $[1.1\%, 4.0\%]$ is the honest representation of
what Pillar~II tells us about the OTT-EA central tendency under the
calibrated parameterisation, with the caveat that tail behaviour
under correlated campaigns is governed by the cross-coupling structure
addressed in Section~\ref{sec:bayesian} rather than by Pillar~II
alone.

\subsection{Sensitivity Analysis}\label{sec:pillar2_sensitivity}

Sensitivity analysis is reported separately for each architectural class,
and also for the OTT-EA-specific parameters $P_3$, $P_4$, and
$\xi_i$ that have no T-EA analogue.

\subsubsection{Common parameters: APT mixing and uplift}

A two-way grid varying global APT mixing weight (10--50\%) and APT
EPSS uplift factor ($1.2{\times}$--$3.0{\times}$) yields T-EA
technical-compromise central projections spanning $5.5$--$7.8\%$ (the
grid minimum $5.5\%$ is the Pillar~II parametric floor under
independence). The corresponding OTT-EA grid spans $3.0$--$4.3\%$,
shifted downward by roughly the central segregation factor
($\sim$0.55). Dominant sensitivity drivers in both grids are the
conditional compromise probabilities for $S_1$, $S_3$, and $S_5$ and
the global APT mixing weight, each contributing more than
$1.5$~percentage points of swing.

\subsubsection{OTT-specific parameters: $P_3$, $P_4$, and $\xi_i$}

The OTT-EA-specific parameters introduce a sensitivity dimension with
no T-EA analogue. A three-way analysis over
$P_3 \in [0.30, 0.80]$, $P_4 \in [0.40, 0.80]$, and a global multiplier
$m_\xi \in [0.5, 1.5]$ on the central $\xi_i$ values from
Table~\ref{tab:xi_scenarios} yields an OTT-EA central projection
spanning $2.0\%$--$6.0\%$ across the full joint range. The defensible
central interval at intermediate $m_\xi \in [0.75, 1.25]$ and central
$P_3 \cdot P_4$ is $[3.5\%, 4.3\%]$. The full $(P_3 \cdot P_4, m_\xi)$
sensitivity grid is reported in Supplementary
\S\ref{sec:supp_ott_pillar2_grid}, Table~\ref{tab:supp_ott_pillar2_sens}.

\subsubsection{Scenario ordering and sensitivity reporting}

The relative ordering of the seven
scenarios should be treated as an informed modelling choice rather
than an externally validated ranking. The same caveat applies to the
per-scenario $\xi_i$ values. The two-way sensitivity grid and one-way
tornado for the Pillar~II central projection are reported in
Supplementary \S\ref{sec:supp_pillar2_sens_figs}.

% =============================================================================
\section{Pillar~III: Architectural Heuristic Decomposition}\label{sec:pillar3}

Pillar~III provides a structural lower-bound check on the Pillar~II
projections by treating compromise risk as the union of risk across
four irreducible channels (insider, zero-day, supply-chain, and
operations) and computing the architecture-conditional compromise
probability under the assumption that channels operate
independently. The point is to show that EA compromise probability
cannot be reduced below a structural floor determined by these
four channels, regardless of the calibration choices made in the
more detailed Pillar~II scenario model. Pillar~III is not a
separate empirical analysis. It is an internal consistency check
that the Pillar~II projections sit above (or near) a defensible
floor. The reader may take the channel-aggregate projections
(approximately $5.4\%$ per system-year for T-EA and $3.0\%$ for
OTT-EA under independence, with channel-minimum heuristic floors of
$1.5\%$ and $1.03\%$ respectively) as the Pillar~III contribution,
and skip the per-channel parameter calibration. The architectural ordering
(T-EA $>$ OTT-EA at the central tendency under independence) holds
under this construction, reflecting the additional operational
steps OTT-EA requires for a full compromise.

\subsection{Rationale and Approach}

Pillars~I and II both work \emph{outwards} from data: historical incidents
(Pillar~I) and parameterised scenarios (Pillar~II). Pillar~III works
\emph{inwards}, from architecture. The question it asks is this: given the
components any operationally viable EA system must contain, and given the
best historically observed security performance for each component class in
analogous domains, what is the lowest combined risk we can prudently project?
The output is a heuristic aggregate, not a structural law, and the contrast
with the other pillars is itself useful: agreement across all three is more
than three times one method's confirmation, but disagreement tells us which
assumptions are doing the work.

The decomposition uses the same cross-cutting weighting machinery introduced
for Pillar~II (Section~\ref{sec:pillar2_outcome}). The four non-eliminable
risk channels are architecture-independent: any operationally viable EA
system must employ cleared personnel, contain non-trivial code, depend on
supply chain, and admit operational error, regardless of whether it is
T-EA or OTT-EA. The architectural distinction enters at conversion. We
distinguish \emph{channel-triggered compromise} (the channel produces
unauthorised access to some part of the system) from \emph{operational
compromise} (the access is sufficient to decrypt target communications).
For T-EA the conversion is automatic by architectural collapse
(Remark~\ref{rem:collapse}); for OTT-EA it is modulated by a per-channel
cross-cutting fraction.

\subsection{Four Non-Eliminable Component Risk Channels}

Any operationally viable EA system must maintain a key escrow database,
access-control mechanisms, cryptographic operations, authorised personnel, network
connectivity, and audit systems. The presence of each component class establishes
an irreducible attack surface: the component cannot be removed without violating
operational viability, and creates at least one exploitable entry point.

\textit{1. Insider threat ($\Pinsider \geq 0.015$).} Any operational EA system
must employ cleared personnel. The CERT \emph{Common Sense Guide to
Mitigating Insider Threats} \citep{cert2023} documents insider-threat
incidents as a recurring exposure across critical-infrastructure and
government-sector organisations. The OPM exfiltration, CIA Vault~7, and
Snowden disclosures all occurred within high-security clearance environments.
The Stream~C platform-incident base
(Section~\ref{sec:pillar1_streamC}) does not lower this estimate. The LastPass
2022, Microsoft Midnight Blizzard 2024, and several Okta incidents involve
insider-mediated initial access at platform operators of comparable scale.
We adopt 1.5\% as a conservative floor applicable to both architectural classes.

\textit{2. Zero-day vulnerabilities ($\Pzeroday \geq 0.015$).} Any system of
non-trivial complexity contains latent unpatched vulnerabilities. A minimum
viable EA system requires at least $\sim$100{,}000 lines of code. Even at
1 defect per 1{,}000 LOC, far below the industry-average range of 15--50
per 1{,}000 lines of delivered code documented by \citet{mcconnell2004}
and comparable to the best released-software rates reported there, this
implies on the order of 100 latent defects.
\citet{bilge_dumitras_2012} document an average zero-day
attack window of 312 days before public disclosure. The CISA KEV database
confirms that HSM firmware, cryptographic libraries, and access-control
middleware are all represented in actively exploited CVEs. OTT platform
codebases are typically larger by several orders of magnitude (millions of
LOC), but benefit from faster patch cadence and better continuous monitoring;
the net effect on zero-day exposure rate is unresolved by the available
evidence. We adopt 1.5\% as a conservative floor applicable to both
architectural classes, with sensitivity to the OTT-specific upward
adjustment discussed in Section~\ref{sec:pillar3_sensitivity}.

\textit{3. Supply-chain exposure ($\Psupply \geq 0.015$).} No EA operator can
achieve complete vertical integration. The SolarWinds (2020) and XZ Utils (2024)
incidents demonstrate that supply-chain attacks against specifically targeted
security infrastructure are operational. The DBIR 2024 records supply-chain
attacks in approximately 4\% of breaches of critical-infrastructure
organisations. The Stream~C incident base contains four documented
supply-chain or supply-chain-adjacent compromises (CodeCov 2021, MOVEit 2023,
Storm-0558 via consumer-key exposure, Cloudflare via Atlassian) at platform
operators, consistent with the floor. We adopt 1.5\% as a conservative
floor applicable to both architectural classes.

\textit{4. Operational errors ($\Pops \geq 0.010$).} Operational
errors (misconfigured access controls, key-management procedure
deviations, certificate mismanagement) are distinct from software
defects. The Mozilla {CA} incident dashboard
\citep{mozilla_ca_incidents} documents formal misissuance incident reports
at a rate of 1--3\% annually among the largest tier-one CAs
(population denominator from CCADB \citep{ccadb2024}). The
lawful-intercept infrastructure compromise in the Athens affair arose from a
procedure failure (rogue software installation
\citep{prevelakis2007}). We adopt 1.0\%.

\subsection{Per-Channel Cross-Cutting Fractions}\label{sec:pillar3_xcut}

For the OTT-EA case, each channel-triggered compromise additionally requires
that the resulting access support operational compromise of architecturally
segregated $\mathcal{K}$ and $\mathcal{D}$. We define the per-channel
cross-cutting fraction $\xi_c$ as the probability that the channel's
realised compromise traverses cross-cutting infrastructure (shared identity
provider, shared build pipeline, common privileged-role credentials, or
analogous mechanism), so that the segregation between $\mathcal{K}$ and
$\mathcal{D}$ does not present an additional barrier. The remaining
$(1 - \xi_c)$ fraction must traverse the multi-stage segregation chain
whose conditional success probability is the product $P_3 \cdot P_4$
(Section~\ref{sec:pillar2_outcome}).

\begin{table}[!htbp]
\centering
\small
\caption{Per-channel cross-cutting fractions for the OTT-EA Pillar~III
decomposition. Calibrated against the Stream~C incident base
(Section~\ref{sec:pillar2_outcome}) and the per-scenario $\xi_i$ values
of Pillar~II (Table~\ref{tab:xi_scenarios}). Architecturally consistent
calibration: each channel's $\xi_i$ is the cross-cutting fraction
appropriate for the dominant attack mechanism in that channel.}
\label{tab:xi_channels}
\begin{tabularx}{\textwidth}{@{}lp{0.18\textwidth}c X@{}}
\toprule
\textbf{Channel} & \textbf{Dominant mechanism} &
\textbf{$\xi_c$} & \textbf{Calibration} \\
\midrule
Insider                & Privileged-role abuse  & 0.30
  & SRE / security-engineer cross-cutting variability \\
Zero-day               & Software defect        & 0.30
  & Some shared-component zero-days, most subsystem-localised \\
Supply chain           & Build/distribution     & 0.55
  & Build pipelines and SDKs cross-cutting by design \\
Operational error      & Configuration drift    & 0.20
  & Misconfigurations typically subsystem-localised \\
\bottomrule
\end{tabularx}
\end{table}

The effective per-channel OTT-EA rate is
\begin{equation}
P^{\mathrm{OTT,eff}}_c = P_c \cdot \left[\xi_c + (1 - \xi_c) \cdot P_3 \cdot P_4\right],
\label{eq:pillar3_eff}
\end{equation}
which reduces to $P_c$ in the limit $\xi_c \to 1$ (perfect cross-cutting,
no segregation gain) or $P_3 \cdot P_4 \to 1$ (no segregation effect),
recovering the T-EA case. At the central calibration ($P_3 \cdot P_4 =
0.30$):
\begin{itemize}[itemsep=2pt]
  \item Insider: $0.015 \cdot [0.30 + 0.70 \cdot 0.30] = 0.015 \cdot 0.51 = 0.77\%$
  \item Zero-day: $0.015 \cdot [0.30 + 0.70 \cdot 0.30] = 0.015 \cdot 0.51 = 0.77\%$
  \item Supply chain: $0.015 \cdot [0.55 + 0.45 \cdot 0.30] = 0.015 \cdot 0.685 = 1.03\%$
  \item Operational error: $0.010 \cdot [0.20 + 0.80 \cdot 0.30] = 0.010 \cdot 0.44 = 0.44\%$
\end{itemize}

\subsection{Architecture-Conditional Channel-Aggregate Projection Estimates}

Assuming independence among the four component risk channels (an assumption
examined in Section~\ref{sec:independence_treatment}), the channel-aggregate
projection for each architectural class follows from the
complement-of-product formula applied to the architecture-conditional
effective rates.

\subsubsection{T-EA channel-aggregate projection}

For T-EA, $\xi_c \cdot 1 + (1 - \xi_c) \cdot P_3 \cdot P_4 \to 1$ for all
channels under the architectural collapse (Remark~\ref{rem:collapse}), so
$P^{\mathrm{T,eff}}_c = P_c$:
\begin{align}
\Pref^{\mathrm{T}} &= 1 - (1 - \Pinsider)(1 - \Pzeroday)(1 - \Psupply)(1 - \Pops)
\notag \\
&= 1 - (0.985)(0.985)(0.985)(0.990) = 1 - 0.946 = 5.4\%.
\label{eq:pillar3_T}
\end{align}
This is the channel-aggregate projection for the centralised single-operator
class, applied here to T-EA without architectural modification.

\subsubsection{OTT-EA channel-aggregate projection}

For OTT-EA, applying the effective rates from
Equation~\eqref{eq:pillar3_eff}:
\begin{align}
\Pref^{\mathrm{OTT}}
  &= 1 - (1 - 0.00765)(1 - 0.00765)(1 - 0.010275)(1 - 0.0044)
\notag \\
  &= 1 - 0.97035 = 0.02965 \approx 3.0\%.
\label{eq:pillar3_OTT}
\end{align}
The OTT-EA channel-aggregate projection of $\sim$3.0\% reflects the segregation
gain at the central parameterisation of $P_3 \cdot P_4$, $\xi_c$ values.
The ratio $\Pref^{\mathrm{OTT}} / \Pref^{\mathrm{T}} \approx 0.55$ is
quantitatively consistent with the Pillar~II ratio of $\sim 0.54$,
because both pillars apply the same architectural correction
mechanism through the segregation factor.

\begin{center}
\fbox{\parbox{0.93\linewidth}{\small\textbf{Independence-conditional only,
both architectures.} These estimates assume the four risk channels fail
independently within each architectural class. The Bayesian model
(Section~\ref{subsec:dependence_results}) shows that realistic dependence
inflates the annual 95th percentile substantially for both classes, with
larger inflation for OTT-EA because correlated campaigns activate
cross-cutting infrastructure and collapse the segregation gain. Under
dependence, the system-level median may be higher or lower than the
independence-based aggregates, but the tail is materially heavier in
both cases. The independence-based aggregates should be read as
lower-anchors for central tendency under independence, not as robust
central estimates. The architecture-conditional Fr\'echet--Hoeffding
intervals from Pillar~II ($[2.2\%, 7.5\%]$ for T-EA; $[1.1\%, 4.0\%]$
for OTT-EA) are preferred for policy use.}}
\end{center}

\subsubsection{Channel-minimum heuristics}

The T-EA \emph{channel-minimum heuristic} of 1.5\% requires two assumptions that are
not structurally guaranteed: (a) the four channels are the \textbf{dominant}
risk sources, and (b) they are collectively \textbf{exhaustive} of all
plausible compromise paths. If additional risk channels exist (e.g.,
compromised hardware supply chain, side-channel attacks on HSMs), the true
floor would be higher than 1.5\%. The figure follows from the Fr\'echet
lower bound \emph{within the four-channel model}: the probability of the
union cannot be less than the maximum of the individual component
probabilities so long as the four channels span the relevant compromise
space.

The OTT-EA channel-minimum heuristic under the same logic is the maximum of the
effective per-channel rates: $\max_c P^{\mathrm{OTT,eff}}_c = 1.03\%$
(supply chain). The result is approximately one-third lower than the T-EA
channel-minimum heuristic under the same assumptions. The same exhaustiveness caveat
applies, with the additional caveat that the OTT-specific calibration
parameters ($P_3$, $P_4$, $\xi_c$) are themselves only weakly empirically
identified.

\subsection{Sensitivity Analysis}\label{sec:pillar3_sensitivity}

The OTT-EA channel-aggregate projection is sensitive to two clusters
of parameters that have no T-EA analogue: the segregation factor
$P_3 \cdot P_4$ and the per-channel cross-cutting fractions $\xi_c$.
The full joint $(P_3 \cdot P_4, m_\xi)$ sensitivity grid is reported
in Supplementary \S\ref{sec:supp_ott_pillar3_grid},
Table~\ref{tab:supp_ott_pillar3_sens}. The OTT-EA Pillar~III aggregate
spans $1.5$--$4.5\%$ across the full joint sensitivity range, with the
defensible central interval at intermediate $\xi$-multiplier values
approximately $[2.6\%, 3.3\%]$. An alternative calibration in which OTT
channel base rates are taken higher than T-EA rates (doubling
$\Pzeroday$ and $\Psupply$) gives an OTT-EA channel-aggregate
projection of approximately $4.7\%$, narrowing the gap to T-EA. This is
reported as a sensitivity rather than the headline because the
empirical evidence on whether platform zero-day and supply-chain rates
exceed those of T-EA infrastructure is mixed.

\subsection{Independence Assumption and Its Status}\label{sec:independence_treatment}

The single most powerful driver of the within-class agreement across
Pillars~I--III is the shared independence assumption. Positive correlation
between risk channels places the combined risk closer to the Fr\'echet
lower bound (the maximum individual channel) than to the
independence-based aggregate. This is the plausible case in which a
single APT campaign simultaneously exploits insider access, a zero-day
vulnerability, and supply-chain positioning.

For T-EA, this positive correlation moves the combined risk from $5.4\%$
toward $1.5\%$ (the maximum individual channel rate). For OTT-EA, the
behaviour under correlation is qualitatively different: positive
correlation activates cross-cutting infrastructure, the very mechanism
through which $\xi_c$ is realised at above-average values, so the
correlated case both moves the combined risk toward the maximum
individual rate \emph{and} inflates the effective $\xi_c$ toward 1.
The combined effect is that under correlated campaigns, the OTT-EA
combined risk approaches the T-EA combined risk, eroding the segregation
gain that distinguishes them under independence. As the Bayesian model
demonstrates in Section~\ref{sec:bayesian}, this is the structural
mechanism by which OTT-EA tail risk grows faster than T-EA tail risk
under correlated-campaign conditions.

The independence assumption therefore produces results that are neither
systematically conservative nor liberal: it underestimates tail risk
in both architectural classes, while potentially overestimating the
median for OTT-EA more than for T-EA. The Bayesian model in the next
section is specifically designed to address this structural limitation
with explicit cross-cutting coupling.

% ============================================================
% PART 2: Bayesian model through conclusion
% ============================================================
\section{The Bayesian Structural Risk Model (Layer~IV)}\label{sec:bayesian}

The Bayesian Structural Risk Model (Layer~IV) provides the formal
probabilistic framework underlying the paper's tail-risk and
dependence-related findings. It represents EA compromise as a
path-based attack-graph problem with hazard rates on individual
edges, common-mode correlated-campaign coupling, and EA-specific
modifier factors that adjust baseline edge hazards by
mechanism-specific multipliers ($\geq 1$). The OTT-EA model adds
a parallel-subgraph structure: a key-side subgraph for
cryptographic-key access, a data-side subgraph for ciphertext
access, joined at an operational-compromise AND-node, with
cross-cutting edges that traverse both subgraphs and a
cross-cutting amplification factor
$\gamma_{\mathrm{X}}/\gamma \geq 1$. This structure produces the
OTT-EA tail-divergence finding: under correlated campaigns,
cross-cutting edges amplify more than subgraph-internal edges,
inflating the upper-tail compromise probability. The reader may
take the central output, a prior predictive distribution on
annual compromise probability with 90\% prior predictive intervals
$[1.4\%, 16.5\%]$ for T-EA and $[0.8\%, 17.4\%]$ for OTT-EA, and
skip the formal hierarchical-prior specification, hyperprior
structure, and dependence-coupling derivations.

\subsection{Overview and Rationale}

The three-pillar framework produces empirically grounded estimates
under stated assumptions, but its outputs are conditional on their
assumptions, in particular the independence assumption, which cannot
be validated for a prospective system. The Bayesian Structural Risk Model (Layer~IV)
produces full prior predictive distributions that propagate uncertainty
from all model parameters through to the final output. It is
implemented in two parallel architectural configurations, T-EA and
OTT-EA, sharing common infrastructure (hierarchical priors, domain
transfer, latent campaign variable) but differing in attack-graph
topology and in the coupling of $Z(t)$ to cross-cutting edges.

\paragraph{Prior predictive, not posterior inference.}
No public dataset of EA-specific compromise events exists, so the
Bayesian model is evaluated by forward sampling from the joint prior
rather than by updating against an EA-specific likelihood. The outputs
$\{q_1^{(r)}\}$ and $\{Q_T^{(r)}\}$ are prior predictive distributions
whose informativeness comes from the empirical content of the priors
(calibrated via domain transfer from Pillar~I analogues), not from
posterior conditioning on direct observations. ``Bayesian'' here
refers to hierarchically structured priors, latent variables, and
probabilistic forward modelling. No MCMC posterior inference is
performed. The contribution is not specific numerical estimates but a
structured framework for uncertainty-aware deliberation in which
qualitative conclusions robust across wide parameter ranges are
explicitly distinguished from those that are not.

\paragraph{Status of the Bayesian medians: a calibration disclosure.}
The edge-hazard priors are calibrated, separately for each
architectural class, to produce prior predictive outputs consistent
with the architecture-matched three-pillar empirical range (3--6\%
annually for T-EA, 1.5--4\% annually for OTT-EA). The Bayesian medians
($4.0\%$ T-EA; $\sim$2.6\% OTT-EA) therefore \emph{cannot} be
presented as independent corroboration of three-pillar
near-equivalence. They are consistent with the within-class agreement
by prior design. The genuinely
additive contributions of this layer are: (i)~tail shape (architecture-
conditional), (ii)~quantification of dependence uplift (which differs
sharply between architectures because of cross-cutting coupling),
(iii)~exceedance probabilities at policy thresholds (conditional on
calibration), and (iv)~parameter uncertainty decomposition. Each is
developed in \S\S\ref{sec:results}--\ref{sec:synthesis}.

The model comprises five layers, applied in parallel for both
architectures: \emph{(1)~Attack-graph layer}: a directed graph
representing multi-step adversarial paths. For T-EA a single graph
terminating at a unitary compromise node, for OTT-EA a compound graph
with parallel key-side and data-side subgraphs joined at an
operational-compromise AND-node with explicit cross-cutting edges.
\emph{(2)~Bayesian hierarchical layer}: hierarchically specified
priors over edge hazards, informed by analogous incident data via
domain transfer (Stream~A for T-EA, Stream~C for OTT-EA).
\emph{(3)~Dependence layer}: a latent campaign variable $Z(t)$
inducing positive correlation among attack events. For T-EA $Z(t)$
couples uniformly to all edges, for OTT-EA preferentially to
cross-cutting edges. \emph{(4)~Adversarial targeting layer}: a
utility-based model for adversary attention. \emph{(5)~Simulation
layer}: Monte Carlo sampling from the joint prior, propagating
uncertainty to produce prior predictive distributions over annual and
cumulative compromise probabilities.

\subsection{Attack Graph Model}

\subsubsection{T-EA attack graph}

For T-EA, the attack graph is the standard directed-acyclic-graph
formulation:

\begin{definition}[T-EA Attack graph]
\label{def:graph_T}
Let $\calG^{\mathrm{T}} = (\calV^{\mathrm{T}}, \calE^{\mathrm{T}})$ be a
directed acyclic graph where $\calV^{\mathrm{T}} = \{v_0, v_1, \ldots, v_n\}$
is the set of nodes, each representing a security-relevant system state
($v_0$: adversary prior to any compromise; $v_n$: full operational compromise
under architectural collapse, Definition~\ref{def:compromise} and
Remark~\ref{rem:collapse}), and $\calE^{\mathrm{T}} \subseteq
\calV^{\mathrm{T}} \times \calV^{\mathrm{T}}$ is the set of directed edges,
each representing a feasible attack step.
\end{definition}

\subsubsection{OTT-EA compound attack graph}

For OTT-EA, the architectural segregation between $\mathcal{K}$ and
$\mathcal{D}$ is represented through a compound graph structure with
parallel key-side and data-side subgraphs joined at an
operational-compromise AND-node:

\begin{definition}[OTT-EA Compound attack graph]
\label{def:graph_OTT}
Let $\calG^{\mathrm{OTT}} = (\calV^{\mathrm{OTT}}, \calE^{\mathrm{OTT}})$
be a compound directed acyclic graph with node set
\[
\calV^{\mathrm{OTT}} =
  \{v_0\} \cup \calV^{\mathrm{X}} \cup \calV^{\mathrm{K}}
  \cup \calV^{\mathrm{D}} \cup \{v_{\mathrm{op}}\},
\]
where $v_0$ is the adversary's initial state; $\calV^{\mathrm{X}}$ is the
set of \emph{cross-cutting} nodes representing shared infrastructure
(identity provider, build pipeline, privileged-role credentials, shared
network perimeter); $\calV^{\mathrm{K}}$ is the set of \emph{key-side}
nodes representing the key-management subsystem; $\calV^{\mathrm{D}}$ is
the set of \emph{data-side} nodes representing the encrypted-data
subsystem; and $v_{\mathrm{op}}$ is the operational-compromise terminal
state. Edge set
\[
\calE^{\mathrm{OTT}} = \calE_{0\to X} \cup \calE_{0\to K} \cup
  \calE_{0\to D} \cup \calE_{X\to K} \cup \calE_{X\to D} \cup
  \calE_{K\to \mathrm{op}} \cup \calE_{D\to \mathrm{op}},
\]
with $v_{\mathrm{op}}$ reached only through the joint condition: the
adversary has reached at least one node in $\calV^{\mathrm{K}}$
\emph{and} at least one node in $\calV^{\mathrm{D}}$.
\end{definition}

The AND-condition at $v_{\mathrm{op}}$ is the formal representation of
operational compromise (Definition~\ref{def:op_compromise}): unauthorised
key access is necessary but not sufficient. Data access is necessary but
not sufficient. Both are required for decryption of target communications.
The cross-cutting nodes $\calV^{\mathrm{X}}$ are the architectural feature
that allows the AND-condition to be satisfied through a single
infrastructure compromise rather than through two distinct attack
campaigns: a node in $\calV^{\mathrm{X}}$ has outgoing edges to nodes in
both $\calV^{\mathrm{K}}$ and $\calV^{\mathrm{D}}$, so that traversing
through $\calV^{\mathrm{X}}$ provides simultaneous access to both
subgraphs.

\begin{figure}[!htbp]
\centering
\resizebox{0.96\textwidth}{!}{%
\begin{tikzpicture}[
  font=\small,
  comp/.style={draw, rounded corners=3pt, minimum width=3.0cm, minimum height=1.0cm, align=center, fill=blue!8, line width=0.5pt, inner sep=4pt},
  xcut/.style={draw, rounded corners=3pt, minimum width=3.4cm, minimum height=1.0cm, align=center, fill=orange!15, line width=0.5pt, inner sep=4pt},
  andnode/.style={draw, diamond, aspect=2, minimum width=3.2cm, align=center, fill=red!15, line width=0.5pt, inner sep=2pt},
  arr/.style={-{Latex[length=2.4mm]}, thick, shorten >=2pt, shorten <=2pt},
  xarr/.style={-{Latex[length=2.4mm]}, thick, dashed, orange!70!black, shorten >=2pt, shorten <=2pt}
]
% Top row: cross-cutting infrastructure
\node[xcut] (idp)   at ( 2.0, 4.0) {Shared identity\\provider};
\node[xcut] (build) at ( 7.0, 4.0) {Build / CI\\pipeline};
\node[xcut] (sre)   at (12.0, 4.0) {Privileged SRE\\credentials};
% Middle-left: key-side subgraph V^K
\node[comp] (kvault) at ( 2.0, 1.0) {Key vault\\(HSM)};
\node[comp] (kapi)   at ( 5.0, 1.0) {Key-management\\API};
% Middle-right: data-side subgraph V^D
\node[comp] (dstore) at ( 9.0, 1.0) {Encrypted\\message store};
\node[comp] (dapi)   at (12.0, 1.0) {Data-access\\API};
% Bottom: operational compromise AND-node
\node[andnode] (op)  at ( 7.0,-2.0) {$v_{\mathrm{op}}$: operational\\compromise (AND)};
% Initial state
\node[draw, circle, minimum width=1.0cm] (v0) at ( 7.0, 6.0) {$v_0$};
% Edges from v0
\draw[arr] (v0.south) -- ($(idp.north) + (0.3,0)$);
\draw[arr] (v0.south) -- (build.north);
\draw[arr] (v0.south) -- ($(sre.north) + (-0.3,0)$);
% Cross-cutting fan-out edges (orange dashed)
\draw[xarr] (idp.south) -- (kvault.north);
\draw[xarr] (idp.south) -- (kapi.north);
\draw[xarr] (build.south) -- ($(kapi.north) + (0.5, 0)$);
\draw[xarr] (build.south) -- ($(dstore.north) + (-0.3, 0)$);
\draw[xarr] (sre.south) -- (dstore.north);
\draw[xarr] (sre.south) -- (dapi.north);
% Internal subgraph edges (key-side)
\draw[arr] (kvault.east) -- (kapi.west);
% Internal subgraph edges (data-side)
\draw[arr] (dstore.east) -- (dapi.west);
% Edges to op (AND-node)
\draw[arr] (kapi.south) -- (op.north west);
\draw[arr] (dstore.south) -- (op.north east);
% Labels
\node[align=center, font=\bfseries\footnotesize] at ( 3.5, 2.0) {Key-side subgraph $\calV^{\mathrm{K}}$};
\node[align=center, font=\bfseries\footnotesize] at (10.5, 2.0) {Data-side subgraph $\calV^{\mathrm{D}}$};
\node[align=center, font=\bfseries\footnotesize, orange!70!black] at ( 7.0, 5.2) {Cross-cutting infrastructure $\calV^{\mathrm{X}}$};
% Legend
\node[align=left, font=\footnotesize] at (4.5,-3.8) {Solid arrow: subgraph-internal edge};
\node[align=left, font=\footnotesize] at (10.5,-3.8) {Dashed orange: cross-cutting edge ($Z(t)$-coupled)};
\end{tikzpicture}%
}
\caption{Compound attack-graph schematic for an OTT-EA (Class~A)
architecture. Cross-cutting infrastructure nodes (top, orange) have
outgoing edges to both the key-side subgraph
$\mathcal{G}^{\rm OTT}_{\rm key}$ and the data-side subgraph $\mathcal{G}^{\rm OTT}_{\rm data}$
(dashed orange edges). Reaching the operational-compromise AND-node
requires either (i) traversing both subgraphs
independently through their respective subgraph-internal edges, or (ii)
traversing a cross-cutting node that provides simultaneous access to
both. The latent campaign variable $Z(t)$ couples preferentially to
cross-cutting edges (Section~\ref{sec:bayesian}), so under
correlated-campaign conditions the segregation gain that distinguishes
OTT-EA from T-EA at the median is eroded.}
\label{fig:attack_surface_ott}
\par\smallskip\noindent\small\textit{Alt text:} Compound attack-graph schematic for OTT-EA showing a key-side subgraph (left) and a data-side subgraph (right), with cross-cutting infrastructure nodes at top whose edges (dashed orange) reach both subgraphs. The operational-compromise AND-node requires either traversing both subgraphs independently or using the cross-cutting path.
\end{figure}

\subsubsection{Hazard rates and attack-path aggregation}\label{sec:cutsets}

For both architectures, each edge $(v_i, v_j) \in \calE$ is assigned a
hazard rate function $\lambda_{ij}(t) \geq 0$ under a piecewise-constant
approximation within annual periods. The system-level hazard rate is
derived by combining minimal attack-path sets $\{M_1, \ldots, M_K\}$
(minimal sets of edges whose joint traversal achieves the compromise
condition) via inclusion--exclusion. For computational tractability
and dimensional consistency, we adopt a \emph{weakest-link
(bottleneck) approximation} for the per-path hazard,
\begin{equation}
\lambda_{M_k}(t) = \min_{(i,j)\in M_k} \lambda_{ij}(t),
\label{eq:cutset_minrate}
\end{equation}
which is exact in the limit where one edge in the path has a hazard rate
substantially smaller than the others. In the present model, edge log-hazards
differ by ${\sim}3$ log-units across adversary tiers (Supplementary
Table~\ref{tab:edge_hazards}), so the approximation error from this step is small relative to
other modelling uncertainties.

For T-EA, minimal $v_0 \to v_n$ attack paths are enumerated following
standard reliability-graph path-set methodology. For OTT-EA, the
AND-condition at $v_{\mathrm{op}}$ requires a generalised path-set
construction: a minimal attack-path set must \emph{either} (a)
comprise at least one $v_0 \to V^{\mathrm{K}}$ path and at least one
$v_0 \to V^{\mathrm{D}}$ path \emph{jointly} (traversing the two
subgraphs separately), \emph{or} (b) traverse at least one
cross-cutting node in $\calV^{\mathrm{X}}$, which provides
simultaneous access to both subgraphs. The second class dominates
when cross-cutting edges have high hazard rates: compromising a
single cross-cutting node simultaneously elevates the traversal
hazard on both subgraphs.

The system hazard for either architecture is then
\begin{equation}
\lambda_{\mathrm{sys}}(t) = \sum_{k=1}^{K} \lambda_{M_k}(t)
  - \sum_{k < \ell} \lambda_{M_k \cup M_\ell}(t)
  + \cdots,
\label{eq:sys_hazard_exact}
\end{equation}
where $\lambda_{M_k \cup M_\ell}(t) = \min_{(i,j) \in M_k \cup M_\ell}
\lambda_{ij}(t)$ and any negative result from inclusion--exclusion truncation
is clamped to zero (occasional sign reversals arise from the min-rate
heuristic; the clamped fraction is below~$1\%$ in our runs). The exact
reliability formulation in terms of path-traversal probabilities
$\prod_{(i,j) \in M_k}(1-\exp(-\lambda_{ij}\Delta t))$ would not reduce to a
closed-form hazard rate. The topology-sensitivity analysis
(Section~\ref{sec:limitations}) bounds the consequence of the weakest-link
approximation at approximately ${\pm}10\%$ on median estimates, and the
qualitative dependence-uplift findings are reliable across these variations.
Under the piecewise-constant hazard assumption, the probability of at least
one compromise in $[0, t]$ follows a non-homogeneous Poisson process:
\begin{equation}
\PR{T_{\mathrm{comp}} \leq t}
= 1 - \exp\!\left(-\int_0^t \lambda_{\mathrm{sys}}(s)\, ds\right).
\label{eq:time_to_compromise}
\end{equation}

The graphs are stylised representations. Sensitivity analysis
(Section~\ref{sec:limitations}) shows that plausible topological variations
change median estimates by approximately 10\%. The dependence-uplift findings
are robust to these variations.

\subsection{Bayesian Hierarchical Model}\label{sec:bayesian_priors_arch}

\subsubsection{Priors for edge hazard rates}

For each edge $(i,j) \in \calE^{\mathrm{A}}$ (architecture
$\mathrm{A} \in \{\mathrm{T}, \mathrm{OTT}\}$), we adopt a log-normal
prior:
\begin{equation}
\log \lambda_{ij}^{\mathrm{A}} \sim
  \mathcal{N}(\mu_{ij}^{\mathrm{A}}, (\sigma_{ij}^{\mathrm{A}})^2),
\label{eq:edge_prior}
\end{equation}
chosen for support on $\mathbb{R}_{>0}$ and widespread use in
reliability modelling. The architecture-conditional location
parameters $\mu_{ij}^{\mathrm{A}}$ are derived from
architecture-matched domain transfer
(\S\ref{sec:domain_transfer_arch}): T-EA priors are calibrated
against Streams~A and~B; OTT-EA priors against Streams~B and~C,
weighted by architectural fit.

\subsubsection{Hyperpriors}

The hyperparameters $\mu_{ij}^{\mathrm{A}}$ and
$\sigma_{ij}^{\mathrm{A}}$ are themselves assigned priors:
\begin{align}
\mu_{ij}^{\mathrm{A}} &\sim
  \mathcal{N}(\hat\mu_{ij}^{\mathrm{transfer},\mathrm{A}}, \tau_\mu^2),
\label{eq:hyperprior_mu}\\
\sigma_{ij}^{\mathrm{A}} &\sim \mathrm{HalfNormal}(\tau_\sigma),
\label{eq:hyperprior_sigma}
\end{align}
where $\hat\mu_{ij}^{\mathrm{transfer},\mathrm{A}}$ is the
architecture-conditional log-scale point estimate from domain
transfer, and $\tau_\mu, \tau_\sigma > 0$ govern additional variance
from domain mismatch. Defaults $\tau_\mu = 1.0$, $\tau_\sigma = 0.5$
in the primary analysis; sensitivity to $\tau_\mu \in [0.5, 1.5]$
reported for both configurations.

\subsubsection{EA-specific system modifiers}

Three multiplicative modifiers are applied to all baseline hazard
rates:
\begin{equation}
\tilde\lambda_{ij}^{\mathrm{A}} = \delta_{\mathrm{EA}}^{\mathrm{A}} \cdot
  \delta_{\mathrm{target}}^{\mathrm{A}} \cdot
  \delta_{\mathrm{concen}}^{\mathrm{A}} \cdot \lambda_{ij}^{\mathrm{A}},
\label{eq:modified_hazard}
\end{equation}
where $\delta_{\mathrm{EA}}^{\mathrm{A}} \geq 1$ represents
increased attack surface from the EA mechanism;
$\delta_{\mathrm{target}}^{\mathrm{A}} \geq 1$ captures elevated
adversarial targeting; and $\delta_{\mathrm{concen}}^{\mathrm{A}}
\geq 1$ reflects key-aggregation concentration risk.

The three modifiers are architecture-independent. The same prior, and
in the simulation the same draw, applies to T-EA and OTT-EA. Each is a
log-censored lognormal, $\delta_k = \exp(\max(\mathcal{N}(\mu_k,
\sigma_k^2), 0))$, which enforces $\delta_k \geq 1$ and places an atom
at exactly unity: the joint modifier equals $1$ with prior probability
$0.108$, so on that mass the EA system adds no hazard and the dominance
of Proposition~\ref{prop:dominance} is weak rather than strict. Full
parameters are in Supplementary \S\ref{sec:supp_modifier_priors}.

Version~1 of this paper described these priors as
architecture-conditional, with the OTT-EA median for
$\delta_{\mathrm{concen}}$ shifted upward by approximately $0.8$
natural-log units on the basis of platform user-population scale. That
specification was not implemented and no result depends on it. The
claim is withdrawn. All architectural difference in Layer~IV enters
through the channel composition of
\S\ref{subsec:dependence}, namely the cross-cutting fractions
$\xi_c$, the coupling ratio $\gamma_{\mathrm{X}}/\gamma$, and the
segregation factor $P_3 P_4$. Readers should attribute the
architectural ordering of the results to that mechanism alone.

The constraint that all three modifiers are $\geq 1$ is a prior
restriction rather than a definitional convenience. Each is defended
on empirical grounds in \S\ref{sec:dominance}, and a reader who
rejects Proposition~\ref{prop:dominance} must identify which modifier
can credibly fall below unity.

\subsection{Domain Transfer}\label{sec:domain_transfer_arch}

In the absence of EA-specific incident data, baseline hazard rates
are estimated by transferring information from architecturally
analogous systems. The streams of Pillar~I provide the appropriate
analogue classes: T-EA transfer draws primarily on Stream~A
(government high-security systems) with Stream~B (CA cohort) as
secondary, reflecting fit to $\mathcal{K}/\mathcal{D}$ co-located
operator-bounded infrastructure; OTT-EA transfer draws primarily on
Stream~C (platform-level key/data compromises) with Stream~B as
secondary, reflecting fit to segregated infrastructure with
large user populations and multi-stage compromise patterns. The
Stream~C-to-OTT-EA transfer is the most analogically direct mapping
in the framework, since Stream~C incidents themselves come from the
architectural class to which OTT-EA belongs.

The domain transfer applies a transport function:
\begin{equation}
\hat\lambda^{\mathrm{transfer},\mathrm{A}}_{ij}
= \hat\lambda'_{ij} \cdot
  \exp\!\left(\sum_r \phi_{ij,r}^{\mathrm{A}} \cdot \Delta_r^{\mathrm{A}} \right),
\label{eq:transfer_fn}
\end{equation}
where $\Delta_r^{\mathrm{A}}$ are signed log-scale adjustment factors
for structural dimension $r$ (architecture-conditional), and
$\phi_{ij,r}^{\mathrm{A}} \sim \mathcal{N}(0, \tau_\phi^2)$ are
random weighting coefficients propagating transport uncertainty. This
formalises what Pillar~I does informally through its analogical
distance criteria and makes the uncertainty in that mapping explicit.
The full set of structural adjustment dimensions
$\{\Delta_r^{\mathrm{A}}\}$, their per-dimension calibration against
the Stream~C incident base, and the dimension-by-dimension
sensitivity analysis are in Supplementary
\S\ref{sec:supp_structural}.

\subsection{Dependence Modelling}\label{subsec:dependence}

A fundamental structural limitation of the complement-of-product
aggregation used in Pillars~I--III is the failure to capture correlated
attack campaigns. An advanced adversary conducting a coordinated operation
does not attack each component independently. Reconnaissance, tooling, and
timing are shared across attack vectors. A single latent campaign variable
$Z(t) \geq 0$ induces positive correlation across attack events.

\subsubsection{T-EA dependence: uniform coupling}

For T-EA, $Z(t)$ couples uniformly to all edges:
\begin{equation}
\lambda_{ij}^{\mathrm{eff,T}}(t \given Z(t)) =
  \tilde\lambda_{ij} \cdot e^{\gamma Z(t)},
  \quad (i,j) \in \calE^{\mathrm{T}},
\label{eq:campaign_amplification_T}
\end{equation}
where $\gamma > 0$ is the campaign amplification factor governing the
sensitivity of edge hazard rates to campaign intensity, and $Z(t)$ evolves
as a non-negative process with prior $Z(t) \sim \mathrm{Gamma}(\alpha_Z,
\beta_Z)$. In high-campaign-intensity periods (large $Z$), all edge
hazards are elevated simultaneously, producing the correlated tail risk
that independence-based models cannot describe.

\subsubsection{OTT-EA dependence: cross-cutting coupling}

For OTT-EA, the architectural distinction between subgraph-internal edges
and cross-cutting edges motivates a structured coupling specification.
Edges into $\calV^{\mathrm{X}}$ and from $\calV^{\mathrm{X}}$ to either
subgraph receive a higher amplification factor than subgraph-internal
edges, reflecting the architectural feature that correlated campaigns
preferentially exploit shared infrastructure:
\begin{equation}
\lambda_{ij}^{\mathrm{eff,OTT}}(t \given Z(t)) =
  \begin{cases}
    \tilde\lambda_{ij} \cdot e^{\gamma_{\mathrm{X}} Z(t)} &
      \text{if } (i,j) \text{ is a cross-cutting edge}, \\
    \tilde\lambda_{ij} \cdot e^{\gamma Z(t)} &
      \text{otherwise},
  \end{cases}
\label{eq:campaign_amplification_OTT}
\end{equation}
where $\gamma_{\mathrm{X}} \geq \gamma > 0$. The constraint
$\gamma_{\mathrm{X}} \geq \gamma$ formalises the architectural
argument that correlated campaigns preferentially exploit
cross-cutting infrastructure: an APT seeking simultaneous access
to key custody and data storage will preferentially compromise
shared identity providers, build pipelines, or privileged-role
credentials over mounting two distinct subsystem-specific
campaigns.

The cross-cutting coupling ratio $\gamma_{\mathrm{X}}/\gamma$ is
treated as a random variable with prior
\begin{equation}
  \gamma_{\mathrm{X}}/\gamma  \;=\;  1  \;+\;  \mathrm{LogNormal}(\log(m - 1),\, \sigma^2),
\label{eq:gamma_X_prior}
\end{equation}
where $m = 2$ is the prior median and $\sigma = 0.40$ is the
log-space spread. The shift by~$1$ enforces the structural
constraint $\gamma_{\mathrm{X}}/\gamma \geq 1$ by construction.
Because the log-normal component is strictly positive, the prior in
fact places probability one on $\gamma_{\mathrm{X}}/\gamma > 1$: the
weak constraint guarantees that the OTT-EA cross-cutting
amplification is never below T-EA's uniform coupling, and the strict
inequality, which holds almost surely under the prior, is what
produces the tail divergence. The
default $(m, \sigma) = (2, 0.40)$ gives a 90\% credible range of
approximately $[1.52, 2.93]$.

\paragraph{Calibration anchor.} The prior median $m = 2$ is
\emph{derived} from, not merely consistent with, the Stream~C
cross-cutting fraction $\hat p_X = 7/14 = 0.50$ (seven of fourteen
incidents exhibited documented cross-cutting attack paths, $X = Y$
in Supplementary Table~\ref{tab:streamC}): under the attack-graph specification, a
cross-cutting incident share of one-half at central campaign
intensity maps to an amplification ratio of approximately~$2$, with
the derivation given in Supplementary
\S\ref{sec:supp_xcut_prior_calibration}. The Salt Typhoon record, in
which a single campaign traversed lawful-intercept and adjacent
administrative infrastructure together, independently rules out
values close to~$1$. The $\sigma = 0.40$ choice
matches the prior width used in the variance-decomposition and
tornado analyses (\S\ref{subsec:evppi}) to keep sensitivity
treatments internally consistent. The architectural-divergence
finding (OTT-EA upper-tail inflation exceeding T-EA's) is robust
across $\sigma = 0.25$, $0.40$, $0.50$, $0.60$, $0.75$
(Table~\ref{tab:gamma_X_sigma_sensitivity}). The specific
magnitude is calibration-dependent and reported as the
prior-integrated value at $\sigma = 0.40$ with explicit
sensitivity. Full prior-calibration derivation and the
empirical-anchor mapping from Stream~C are in Supplementary
\S\ref{sec:supp_xcut_prior_calibration}. Empirical considerations
constraining $\sigma$ more tightly are discussed in
\S\ref{subsec:xcut_coupling_limitation}.

Treating $\gamma_{\mathrm{X}}/\gamma$ as sampled rather than fixed
has two substantive consequences: the headline distributional
results absorb the uncertainty in $\gamma_{\mathrm{X}}/\gamma$
directly (no separate sensitivity number needed to interpret a
published quantile), and the prior-integrated upper tail is
heavier than the tail at the prior median in isolation, because
exponential amplification $e^{\gamma_{\mathrm{X}} Z}$ is convex in
$\gamma_{\mathrm{X}}$ and high realisations dominate the upper
tail.

\subsubsection{Architectural consequence: tail expansion under correlation}

The substantive consequence of the cross-cutting coupling specification is
that under high-$Z(t)$ realisations, which dominate the upper tail of the
prior predictive distribution, the segregation gain that distinguishes
OTT-EA from T-EA at the median is partially or fully eroded. Concretely,
at $Z(t) = 0$ (no campaign activity), cross-cutting edges have the same
nominal hazard as subgraph-internal edges, and the effective $\xi_c$
values realised through the graph structure are at their architectural
defaults. At large $Z(t)$, cross-cutting edges receive amplification
$e^{\gamma_{\mathrm{X}} Z(t)} \geq e^{\gamma Z(t)}$, raising the
proportion of attack-graph traversals that pass through
$\calV^{\mathrm{X}}$ and effectively inflating the realised cross-cutting
fraction toward 1. The OTT-EA tail therefore behaves like the T-EA tail
with the segregation gain progressively erased, while the T-EA tail
behaves under uniform coupling. This is the
structural mechanism by which the OTT-EA distribution exhibits a lower
median but a heavier right tail than the T-EA distribution
(Section~\ref{sec:results}).

\subsubsection{What the specification determines before any
data}\label{subsec:tail_algebra}

The preceding mechanism is a property of the prior, and the paper
should not present its consequences as findings. Two results follow
analytically from the choice of $Z \sim \mathrm{Gamma}(\alpha_Z,
\beta_Z)$ with amplification $e^{\gamma Z}$, and they should be read
before any tail estimate in \S\ref{sec:results}.

For a gamma-distributed $Z$ with rate $\beta_Z$,
\begin{equation}
\Pr\!\left(e^{\gamma Z} > x\right) = \Pr\!\left(Z > \tfrac{\log x}{\gamma}\right)
 \;\sim\; C\,(\log x)^{\alpha_Z - 1}\, x^{-\beta_Z/\gamma},
\label{eq:tail_index}
\end{equation}
so the amplification factor is regularly varying with tail index
$\beta_Z/\gamma$. A positive fitted generalised-Pareto shape is
therefore guaranteed by construction, and is not evidence of
anything. Further, the fitted shape approximates $\gamma/\beta_Z$, so
$\xi_{\mathrm{OTT}}/\xi_{\mathrm{T}} \approx
\gamma_{\mathrm{X}}/\gamma$. Since the model imposes
$\gamma_{\mathrm{X}} \geq \gamma$, the ordering
$\xi_{\mathrm{OTT}} > \xi_{\mathrm{T}}$ is a theorem about the prior.
The fitted values bear this out exactly: $0.518/0.271 = 1.91$ against
a prior median $\gamma_{\mathrm{X}}/\gamma = 2$, and a direct sweep
over pinned coupling ratios recovers shapes tracking
$\gamma_{\mathrm{X}}/\beta_Z$
(\texttt{tail\_index\_diagnostic.py} in the archive).

A second property constrains how far into the tail the model can be
read. The mean $\mathbb{E}[e^{gZ}] = (1 - g/\beta_Z)^{-\alpha_Z}$
exists only for $g < \beta_Z$. With $\beta_Z = 2.0$ and $\gamma = 1.0$
the T-EA amplification has mean $1.516$, but the OTT-EA exponent is
$\gamma_{\mathrm{X}}$, whose prior median is exactly~$2.0$. Half the
prior draws therefore give an amplification factor with no finite
mean. Simulated outputs remain bounded only because $q = 1 -
e^{-\lambda} \leq 1$. Capping the amplification leaves the medians
unchanged to two decimal places while pulling the OTT-EA 99th
percentile from $57\%$ to $44\%$ at a fivefold cap
(\texttt{moment\_condition.py}). Central estimates thus do not rest on
the divergent region, and percentiles above roughly the 95th should be
read as prior artefacts rather than as risk estimates.

Neither property is patched away here. Raising $\beta_Z$ above the
$\gamma_{\mathrm{X}}/\gamma$ 99th percentile of $3.54$ would restore a
finite mean throughout, but it would move every reported number, and
disclosure is the more useful course for a corrections release. The
consequence for interpretation is stated in
\S\ref{sec:synthesis}: the \emph{direction} of the tail contrast is
assumed, not found, and only its magnitude is informative, conditional
on the coupling prior.

\subsection{Adversarial Targeting Layer}

The probability that each adversary class directs resources toward the EA system
in a given period is modelled through a utility-based framework. The expected
utility of targeting an EA system for nation-state actors ($\mathcal{A}_{\mathrm{NS}}$)
is proportional to the intelligence value of the system's content, its
vulnerability to the actor's tooling, and the probability of successful
exfiltration without attribution. Operationally, this is captured by the system
modifier $\delta_{\mathrm{target}}$, with the constraint $\delta_{\mathrm{target}} \geq 1$
reflecting the structural observation that EA systems are, by design, higher-value
targets than their non-EA analogues.

\subsection{Simulation Procedure}\label{sec:simulation_procedure}

The Monte Carlo simulation is run independently for the T-EA and OTT-EA
configurations, sharing common hyperprior structure but differing in
attack-graph topology and in the $Z(t)$-coupling specification. For each
architectural class $\mathrm{A} \in \{\mathrm{T}, \mathrm{OTT}\}$, the
simulation proceeds for each replicate $r = 1, \ldots, R$ ($R = 8{,}000$ in
the primary analysis):
\begin{enumerate}[label=\arabic*.]
  \item \textbf{Sample parameters} $\btheta^{\mathrm{A},(r)}$ from the
    architecture-conditional joint prior over all model parameters,
    including the architecture-conditional hazard hyperpriors
    (Section~\ref{sec:bayesian_priors_arch}), the EA modifiers, and the
    architecture-conditional dependence parameters $\gamma$ and (for OTT)
    $\gamma_{\mathrm{X}}$.
  \item \textbf{Compute effective hazards}: for each year $k$ and edge
    $(i,j) \in \calE^{\mathrm{A}}$, compute
    $\lambda_{ij}^{\mathrm{det},\mathrm{A},(r)}(k)$ via
    Equations~\eqref{eq:modified_hazard} and the architecture-appropriate
    coupling specification (Equation~\eqref{eq:campaign_amplification_T}
    for T-EA; Equation~\eqref{eq:campaign_amplification_OTT} for OTT-EA),
    conditioning on the sampled campaign state $Z^{\mathrm{A},(r)}(k)$
    and applying the detection-adjusted formula
    $\lambda_{ij}^{\mathrm{det}} = (1 - \rho_{ij}) \cdot
    \lambda_{ij}^{\mathrm{eff}}$.
  \item \textbf{Compute annual compromise probability}: For T-EA,
    $q_k^{\mathrm{T},(r)} = 1 - e^{-\lambda_{\mathrm{sys}}^{\mathrm{T},(r)}(k)}$
    via standard path-set aggregation. For OTT-EA, the AND-condition at
    $v_{\mathrm{op}}$ is enforced by computing the operational-compromise
    probability through the generalised path-set construction
    (Section~\ref{def:graph_OTT}):
    \begin{equation}
    q_k^{\mathrm{OTT},(r)} = q_k^{\mathrm{xcut},(r)}
      + \left(1 - q_k^{\mathrm{xcut},(r)}\right) \cdot q_k^{\mathrm{K},(r)}
      \cdot q_k^{\mathrm{D},(r)},
    \label{eq:q_OTT}
    \end{equation}
    where $q_k^{\mathrm{xcut},(r)}$ is the probability that traversal
    reaches $v_{\mathrm{op}}$ via at least one cross-cutting node,
    $q_k^{\mathrm{K},(r)}$ is the probability of subgraph-internal
    traversal of $\calV^{\mathrm{K}}$ in absence of cross-cutting, and
    similarly $q_k^{\mathrm{D},(r)}$ for $\calV^{\mathrm{D}}$. This is
    the Bayesian-model analogue of the Pillar~II/III formula
    $\xi + (1-\xi) P_3 P_4$ (Equation~\eqref{eq:pop_seg}), with the
    architectural feature that under correlated-campaign conditions
    $q_k^{\mathrm{xcut},(r)}$ inflates faster than the subgraph-internal
    probabilities.
  \item \textbf{Compute cumulative probability}: $Q_T^{\mathrm{A},(r)} =
    1 - \prod_{k=1}^{T}(1 - q_k^{\mathrm{A},(r)})$.
\end{enumerate}
The outputs $\{q_1^{\mathrm{T},(r)}\}_{r=1}^R$,
$\{q_1^{\mathrm{OTT},(r)}\}_{r=1}^R$, and the corresponding cumulative
distributions constitute prior predictive samples from which
architecture-conditional summary statistics, credible intervals, and
exceedance probabilities are computed (Section~\ref{sec:results}).

\paragraph{Monte Carlo numerical stability.}
All headline results derive from the $R = 8{,}000$ primary run (seed
2024). A four-stream diagnostic bounds the Monte Carlo error at
$\leq 0.08$~pp on the annual medians, with effective sample sizes
above $7{,}700$; the full diagnostic is in Supplementary
\S\ref{sec:supp_mc_stability}.

\subsection{Structural Lower Bound and Stochastic Dominance}\label{sec:dominance}

A result of the Bayesian model follows from the constraints on the system
modifier parameters, once those constraints are empirically grounded.

\subsubsection*{Empirical basis for the modifier constraints}

The model posits three EA-specific multipliers $\delta_{\mathrm{EA}}^{\mathrm{A}},
\delta_{\mathrm{target}}^{\mathrm{A}}, \delta_{\mathrm{concen}}^{\mathrm{A}}
\geq 1$ for each architectural class $\mathrm{A} \in \{\mathrm{T},
\mathrm{OTT}\}$. For these constraints to carry policy weight they must
reflect structural properties of EA systems rather than modelling
definitions. Two independent lines of evidence support each, with
class-specific anchoring where the architectural distinction is material:

\textit{$\delta_{\mathrm{EA}}^{\mathrm{A}} \geq 1$ (attack surface
increase).} Any operationally viable EA architecture introduces
additional codebase, interfaces, and key-management components, and
the documented relationship between software complexity and latent
vulnerability density \citep{alhazmi2007,alhazmi2008} makes the added
surface weakly risk-increasing. The relationship is unvalidated for
minimal-firmware hardware security modules, so its strength varies by
implementation, but equality with the non-EA baseline would require
an EA subsystem of literally zero added complexity.

\textit{$\delta_{\mathrm{target}}^{\mathrm{A}} \geq 1$ (elevated
adversarial targeting).} Salt Typhoon (2024) and the Athens affair
(2004--05) demonstrate that T-EA lawful-intercept infrastructure
attracts adversarial attention absent from the non-EA counterfactual;
Storm-0558 (2023) and Midnight Blizzard (2024) demonstrate the
equivalent for platform-level signing-key and identity
infrastructure. $\delta_{\mathrm{target}}^{\mathrm{A}} = 1$ would
assert that EA systems attract no incremental targeting, which the
incident record contradicts for both classes.

\textit{$\delta_{\mathrm{concen}}^{\mathrm{A}} \geq 1$ (key
concentration premium).} An EA system aggregates an entire user
population's key material in one component, whose compromise value
strictly exceeds that of any single device; under any rational
targeting model, greater expected payoff implies at least equal
targeting intensity. $\delta_{\mathrm{concen}}^{\mathrm{OTT}}$
exceeds $\delta_{\mathrm{concen}}^{\mathrm{T}}$ by roughly the
log-ratio of populations served, but both are bounded below by
unity.

\subsubsection*{Proposition}

The proposition below should be read with its logical type in view.
Its mathematical content is deliberately slight: once the three
modifier bounds above are granted, first-order dominance follows in
two lines of monotonicity. Its scientific content lies entirely in
the empirical case for those bounds---that any operationally viable
EA mechanism adds attack surface, attracts incremental targeting,
and concentrates key-material value---which the preceding paragraphs
argue from the incident record rather than assume. A reader who
rejects the conclusion must therefore identify which of the three
modifiers can credibly fall below unity, and for which architectural
class; the proposition's role is to make that the only available
line of attack.

\begin{proposition}[First-Order Stochastic Dominance, Both Architectural Classes]
\label{prop:dominance}
Under the assumptions that (i) the attack graph of an EA system is a
superset of the no-EA graph (adding EA-specific nodes and edges; see
Section~\ref{sec:pillar1_streamA} for Salt Typhoon as a Tier~1 Stream~A
incident, Section~\ref{sec:pillar3} for the Athens affair as a
historical channel-rate analogue, and
Section~\ref{sec:pillar1_streamC} for Storm-0558 and Midnight Blizzard
for OTT-EA) and (ii) the EA-specific modifiers
$\delta_{\mathrm{EA}}^{\mathrm{A}},
\delta_{\mathrm{target}}^{\mathrm{A}}, \delta_{\mathrm{concen}}^{\mathrm{A}}
\geq 1$, the model's annual compromise probability is at least as
large for the EA configuration as for an otherwise-identical
configuration without EA, with equality exactly when the joint
modifier equals one, an event of prior probability $0.108$ under the
censored priors of
Supplementary~\S\ref{sec:supp_modifier_priors}. Formally, within the model of Section~\ref{sec:bayesian},
for each architectural class $\mathrm{A} \in \{\mathrm{T}, \mathrm{OTT}\}$,
let $F_{\mathrm{EA}}^{\mathrm{A}}$ denote the CDF of the compromise
probability under the EA configuration of class $\mathrm{A}$, and let
$F_{\mathrm{noEA}}^{\mathrm{A}}$ denote the CDF under the same
architecture with all three modifiers fixed at unity (the
counterfactual non-EA parameterisation of class $\mathrm{A}$). Then
\[
  F_{\mathrm{EA}}^{\mathrm{A}}(x) \leq F_{\mathrm{noEA}}^{\mathrm{A}}(x)
  \quad \forall x \in [0, 1], \quad
  \mathrm{A} \in \{\mathrm{T}, \mathrm{OTT}\}.
\]
That is, in each architectural class, the EA-configured model
stochastically dominates (is riskier than) the
same architecture with $\delta^{\mathrm{A}} = 1$. The ordering is
structural in the modelling sense and does not depend on prior
calibration.
\end{proposition}

\begin{proof}[Proof sketch]
A monotonicity argument: the modifier product $\geq 1$ raises every
edge hazard, the per-period compromise probability is increasing in
total hazard, and both architectural aggregation maps are monotone in
their components, so the inequality holds at every parameter
realisation and therefore distributionally. The full proof is in
Supplementary \S\ref{sec:supp_dominance_proof}.
\end{proof}

\paragraph{Status and scope of the assumptions.} Both assumptions
are empirically supported, neither is logically necessary, and the
conclusion is a structural claim about the model conditional on
them, not an empirical finding about the EA-versus-no-EA difference
in the wild; what would defeat the assumptions is set out in
\S\ref{sec:structural_vs_calibration}. The no-EA counterfactual is
a model-internal construction (the same graph with the modifiers at
unity), not a claim about specific deployed carriers, which may
host lawful-intercept infrastructure for other regulatory reasons:
for those, the relevant counterfactual is ``no additional EA
mandate'', and the proposition speaks to that reading. The
relationship of this model-internal counterfactual to the empirical
baseline carried by the Pillar~I analogue rates is stated in
\S\ref{sec:policy} (one decomposition, three expressions).

\begin{remark}[Cross-architecture ordering is not stochastic dominance]
\label{rem:cross_arch}
Proposition~\ref{prop:dominance} establishes EA-versus-no-EA stochastic
dominance \emph{within each architectural class}. It does \emph{not}
establish a stochastic dominance ordering \emph{between} the T-EA and
OTT-EA EA-configured distributions. As Section~\ref{sec:results} shows,
the relationship between $F_{\mathrm{EA}}^{\mathrm{T}}$ and
$F_{\mathrm{EA}}^{\mathrm{OTT}}$ is not first-order: under the central
calibration the T-EA distribution lies above the OTT-EA distribution at
the median (T-EA is riskier on average), but the OTT-EA distribution can
exceed the T-EA distribution in the upper tail under correlated-campaign
conditions (because of the cross-cutting coupling specified in
Section~\ref{subsec:dependence}). This is a structural feature of
architectural segregation, not a defect of the model. The policy
implication is that T-EA-versus-OTT-EA comparison must engage with the
full distribution rather than reducing to median or any single quantile.
\end{remark}

\paragraph{Caveat on scope.} The proposition compares two
parameterisations of the \emph{same} graph in each class; it does not
formally compare an EA system against a real non-EA system, whose
topology would differ (no key-escrow node, no orchestration layer).
The policy inference that EA adds risk rests on the separate
empirical claim that no operationally viable deployment achieves
$\delta^{\mathrm{A}} = 1$ on all three dimensions. The consequence is
that the debate is not about \emph{whether} an EA mandate adds risk
within its class but about how much, in what form (median versus
tail), and whether the increment is justified by benefits outside
this paper's scope.

% =============================================================================
\section{Results}\label{sec:results}

This section reports the quantitative outputs of the framework:
annual and multi-decade compromise probabilities, distributional
shape (mean, percentiles, tail behaviour), the effect of dependence
on the headline values (with comparison to Fr\'echet--Hoeffding
bounds), the relative importance of different parameters (variance
decomposition), and the prior-sensitivity of headline numbers. The
principal findings: T-EA annual compromise probability has a
Bayesian median of $4.0\%$ (90\% credible interval
$[1.4\%, 16.5\%]$); OTT-EA has a lower median ($2.6\%$) but a
heavier upper tail (95\% upper at $17.4\%$). Over a 10-year
deployment horizon, cumulative compromise reaches $37\%$ for T-EA
and $32\%$ for OTT-EA at the median. Under conservative priors,
these become approximately $19\%$ and $13\%$, respectively, still
well above zero. The dependence layer ($Z$, $\gamma$,
$\gamma_{\mathrm{X}}$) inflates the upper tail more for OTT-EA than
for T-EA, with a tail-shape divergence guaranteed by the
cross-cutting coupling constraint. The reader may take these
headline values and skip the technical sensitivity analyses
without losing the principal contribution.

\textit{Note: All numerical results from Layer~IV are outputs of the model
under the parameter priors described in Section~\ref{sec:bayesian}. They are
explicitly conditional on all stated modelling assumptions and should not be
interpreted as empirical measurements of historical or projected compromise
rates.}

\subsection{Three-Pillar Framework: Consolidated Outputs}

Table~\ref{tab:cluster} reproduces the principal reference estimates produced
by the framework, with their critical assumptions and epistemic status,
stratified by architectural class. The table is deliberately compact:
the assumption-conditional estimates within each architectural class
cluster narrowly only because they share the independence assumption,
and the T-EA and OTT-EA clusters differ by approximately the central
segregation factor. A longer enumeration of every variant floor
estimator would obscure rather than support these two findings.

\begin{table}[!htbp]
\centering
\small
\caption{Principal reference estimates by architectural class: values,
critical assumptions, and epistemic status. The T-EA estimates apply to
transmission-layer EA architectures with $\mathcal{K}/\mathcal{D}$
co-location. The OTT-EA estimates apply to application-layer architectures
with central-calibration segregation factor $P_3 \cdot P_4 = 0.30$ and
the per-channel cross-cutting fractions of Table~\ref{tab:xi_channels}.}
\label{tab:cluster}
\begin{tabularx}{\textwidth}{@{}lccXX@{}}
\toprule
\textbf{Reference point} & \textbf{T-EA} & \textbf{OTT-EA} &
\textbf{Critical assumptions} & \textbf{Epistemic status} \\
\midrule
Channel-minimum heuristic (Pillar~III) & 1.5\% & 1.0\% &
  Four channels are dominant AND exhaustive; empirical minima are
  architecture-appropriate &
  Conditional heuristic; not a structural bound \\
Fr\'echet--Hoeffding lower (Pillar~II) & 2.2\% & 1.1\% &
  Seven scenarios cover the dominant attack surface &
  Model-consistent lower bound \\
Bayesian median (Layer~IV) & 4.0\% & 2.6\% &
  Architecture-conditional priors calibrated to architecture-matched
  Pillar~I empirical range &
  Prior-conditional; not independent corroboration \\
Channel-aggregate projection (Pillar~III) & 5.4\% & 3.0\% &
  Independence; architecture-conditional channel weighting &
  Independence-conditional \\
Pillar~II central (independence) & 7.3\% & 3.9\% &
  Independence; scenario completeness &
  Independence-conditional \\
Fr\'echet--Hoeffding upper (Pillar~II) & 7.5\% & 4.0\% &
  Same as lower bound &
  Model-consistent upper bound \\
Bayesian 90\% interval (Layer~IV) &
  $[1.4\%, 16.5\%]$ & $[0.8\%, 17.4\%]$ &
  Same as Bayesian median &
  Distributional uncertainty \\
\bottomrule
\end{tabularx}
\end{table}

Two architectural patterns are visible in the table. First, within
each architectural class the independence-conditional reference
points cluster in a narrow range: T-EA in the upper-single-digit
range; OTT-EA in the low-single-digit range. The dominant driver of
within-class clustering is the shared independence assumption. The
limited genuinely independent evidence is discussed in
Section~\ref{sec:synthesis}. Second, the OTT-EA cluster sits at
approximately half the T-EA cluster across all three pillar
methodologies, confirming that the architectural correction
mechanism is applied consistently across pillars and that the
segregation gain is the quantitatively dominant architectural effect
at the central tendency.

The Bayesian 90\% credible intervals tell a different story. The
OTT-EA upper bound modestly exceeds the T-EA upper bound: under
correlated-campaign conditions OTT-EA upper-tail
behaviour overtakes T-EA as the segregation gain is eroded and the
cross-cutting amplification engages. The OTT-EA lower bound is also
below T-EA's, reflecting the segregation gain at the lower tail
where campaign activity is low. The upper-to-lower interval ratio
is therefore wider for OTT-EA, formalising the tail-heavier
character of the OTT-EA distribution.

\subsection{Which Quantity Answers Which Question}\label{sec:which_quantity}

The framework reports three families of numbers, and they are
different quantities rather than competing estimates of one
quantity. The Fr\'echet--Hoeffding intervals bound the Pillar~II
seven-scenario aggregation: they hold under every dependence
structure among the scenarios, conditional on the per-scenario
probabilities, and they include neither the EA modifiers nor the
campaign amplification of Layer~IV. One asymmetry is worth making
explicit: a scenario omitted from the decomposition can only raise
the union probability, so the Fr\'echet--Hoeffding \emph{lower}
bound is robust to scenario incompleteness while the \emph{upper}
bound is conditional on the seven scenarios covering the dominant
attack surface. The Bayesian percentiles
describe the full prior-predictive distribution with those
mechanisms switched on. The FH interval therefore does not bound
the Bayesian outputs, and the Layer~IV 95th percentiles
(approximately $16$--$17\%$ annually for both architectures)
sitting above the FH uppers ($7.5\%$ and $4.0\%$) is not an
inconsistency: the bracket and the tail answer different
questions. The usage rule is as follows. For comparative
architecture questions, and for a dependence-robust envelope of
the scenario model, use the FH intervals together with the
structural orderings. For tail and stress questions, where
correlated campaigns and EA-specific targeting matter, use the
Layer~IV upper percentiles. For a single calibration-conditional
central reference, use the Layer~IV medians, with the
channel-minimum heuristics as independence-conditional floors.
Table~\ref{tab:cluster} states the assumptions each family
carries; Table~\ref{tab:risk_summary} presents the families side
by side under this convention.

\subsection{Bayesian Layer: Distributional Results}

Prior predictive distributions over the annual compromise probability
$q_1$ are produced separately for the two architectural classes under
the primary parameterisation ($R = 8{,}000$ replicates, seed 2024).
The full quantile sets are reproducible via the bundled scripts. The
narration below emphasises qualitative shape rather than specific
quantiles.

\paragraph{T-EA distribution.} The central tendency sits in the
low-single-digit percentage range (median in the neighbourhood of
$4\%$), consistent with the three-pillar T-EA empirical range
($3$--$6\%$) to which the priors were calibrated. The 90\% credible
interval spans approximately an order of magnitude, with the lower
bound near $1\%$ and the upper bound in the mid-teens. The upper
tail extends further: at the 95th percentile the distribution sits
in the mid-to-upper teens, with non-negligible probability mass
reflecting scenarios combining elevated campaign intensity, high
targeting probability, and poor detection performance.

\paragraph{OTT-EA distribution.} The central tendency is lower than
T-EA's, with a median in the low-single-digit range (roughly
two-thirds of the T-EA median, consistent with the segregation gain).
The 90\% credible interval is slightly wider than T-EA's, with both
a lower lower bound (under $1\%$) and a higher upper bound (upper
teens). The upper tail is the distinctive feature: at the 95th
percentile OTT-EA modestly exceeds T-EA despite its lower median,
the structural signature of cross-cutting coupling
(\S\ref{subsec:dependence}). Further into the tail the gap widens
substantially, with the OTT-EA 99th percentile in the upper-half
percentage range and T-EA's in the mid-thirties.

\paragraph{Architectural comparison.} The two paragraphs above
reduce to three contrasts: a median ratio of $0.66$, directionally
consistent with but higher than the Pillar~II and Pillar~III
architectural ratios of $0.54$ and $0.55$, the gap arising because
the Layer~IV modifiers and campaign coupling act on both classes and
partially compress the segregation gain; a
tail crossover near the 95th percentile, below which T-EA is
stochastically riskier and above which OTT-EA is; and a markedly
larger OTT-EA tail-to-median ratio, the signature of segregation
interacting with correlated campaigns under prior integration over
$\gamma_{\mathrm{X}}/\gamma$ (Equation~\ref{eq:gamma_X_prior}). The
interval widths themselves are the primary result: the model cannot
rule out annual risk near $1\%$ or in the mid-teens for either
class.

\begin{figure}[!htbp]
\centering
\includegraphics[width=\textwidth]{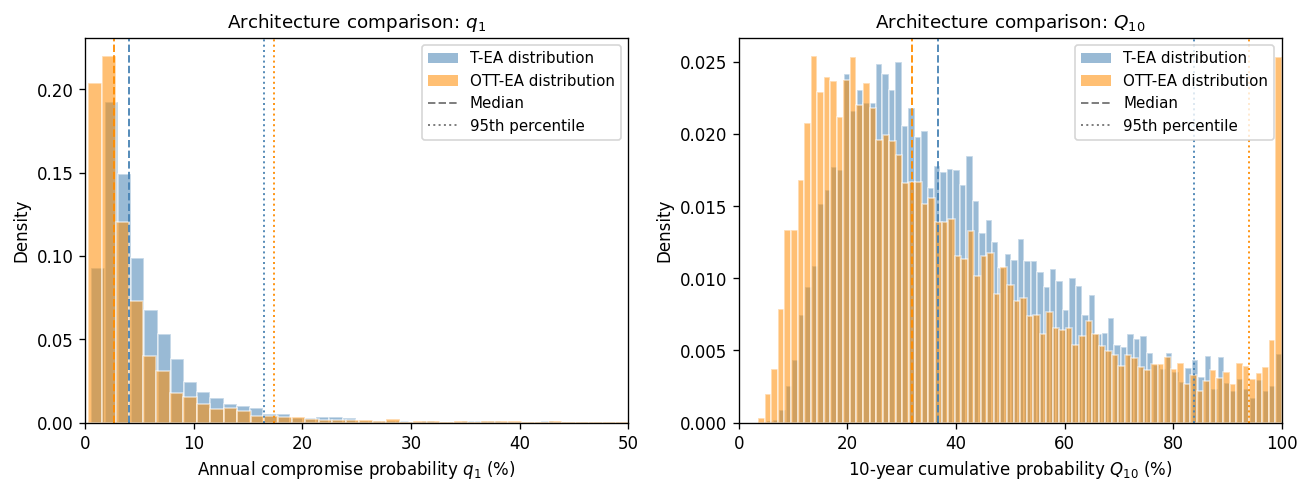}
\caption{Architecture-conditional comparison of prior predictive
distributions for the year-1 annual probability $q_1$ (left) and
10-year cumulative $Q_{10}$ (right) under the full dependence model.
T-EA distribution (blue) and OTT-EA distribution (orange) overlaid.
The OTT-EA distribution sits below the T-EA distribution at the
median but the upper tails cross over at approximately the 95th
percentile, with OTT-EA exhibiting heavier extreme-tail mass under
correlated-campaign conditions. Vertical markers indicate the
median, 5th, and 95th percentiles for each distribution. Results
from $R = 8{,}000$ Monte Carlo replicates per architecture (seed
2024); conditional on all modelling assumptions in
Section~\ref{sec:bayesian}.}
\label{fig:distributions_ott}
\par\smallskip\noindent\small\textit{Alt text:} Two side-by-side density plots comparing prior predictive distributions of T-EA (blue) and OTT-EA (orange). Left: year-1 compromise probability q1; OTT-EA sits below T-EA at the median but crosses over near the 95th percentile. Right: 10-year cumulative Q10 with the same architectural divergence pattern.
\end{figure}

\subsection{Effect of Dependence}\label{subsec:dependence_results}

The latent campaign variable $Z(t)$ has a substantial impact on
both the central tendency and the tail of the compromise
distribution, with qualitatively divergent effects across
architectural classes.

\paragraph{T-EA and OTT-EA dependence effects.} For T-EA, moving
from the independence baseline ($Z \approx 0$) to the full
dependence model raises the year-1 median modestly (by
$\approx 1$~percentage point, a $\approx 30\%$ relative increase)
and the upper percentile more substantially: the
annual 95th-percentile uplift is approximately $1.5\times$. For
OTT-EA the analogous comparison gives a comparable absolute median
uplift ($\approx 0.9$~pp, a larger relative increase from the lower
base) but
a substantially larger upper-percentile uplift, around $3\times$
(prior-integrated over $\gamma_{\mathrm{X}}/\gamma$). The reported
$\times 1.53$ (T-EA) and $\times 2.83$ (OTT-EA) values in
Table~\ref{tab:dependence_comparison} are reproducibility-grade
ratios from the headline replicates. In policy discussion they
should be read as ``around $1.5\times$'' and ``around $3\times$''
respectively. Cumulative 10-year upper-percentile uplifts are smaller
in absolute multiplier (roughly $1.2\times$ T-EA, $2\times$ OTT-EA),
since the cumulative distribution is bounded by unity and compresses
the upper tail.

The architectural divergence is concentrated in the tail, not the
median: median uplifts are comparable across architectures, but the
annual upper-percentile uplifts differ by roughly a factor of two,
through the segregation-erosion mechanism specified in
\S\ref{subsec:dependence} and the prior-integration convexity
quantified in the conditional sensitivity below.

\begin{table}[!htbp]
\centering
\small
\caption{Comparison of independence baseline and full dependence model
outputs by architectural class. The annual 95th-percentile uplift for
OTT-EA exceeds the T-EA uplift by a factor of approximately $\times 1.85$,
the structural consequence of cross-cutting coupling integrated over the
prior on $\gamma_{\mathrm{X}}/\gamma$. The architectural tail crossover (OTT-EA
upper-percentile values exceeding T-EA's despite the lower median) is
visible at both the 95th and 99th percentiles under full dependence.}
\label{tab:dependence_comparison}
\begin{tabular}{@{}lcccc@{}}
\toprule
& \multicolumn{2}{c}{\textbf{T-EA}} & \multicolumn{2}{c}{\textbf{OTT-EA}} \\
\cmidrule(lr){2-3} \cmidrule(lr){4-5}
\textbf{Quantity} & Indep. & Full dep. & Indep. & Full dep. \\
\midrule
Year-1 median $q_1$            & 3.05\% & 4.0\%  & 1.69\% & 2.63\% \\
Year-1 95th percentile $q_1$   & 10.8\% & 16.5\% & 6.2\%  & 17.4\% \\
Year-1 99th percentile $q_1$   & 19.4\% & 35.2\% & 11.4\% & 58.9\% \\
10-year cumulative median $Q_{10}$  & 26.6\% & 36.6\% & 15.7\% & 31.9\% \\
10-year cumulative 95th $Q_{10}$    & 68.1\% & 83.8\% & 47.0\% & 93.9\% \\
\midrule
\multicolumn{5}{@{}l}{\textit{Architectural tail crossover (full dependence)}} \\
$q_1$ 95th percentile: T-EA vs OTT-EA & \multicolumn{2}{c}{16.5\%} &
                                       \multicolumn{2}{c}{17.4\% (OTT-EA higher)} \\
$q_1$ 99th percentile: T-EA vs OTT-EA & \multicolumn{2}{c}{35.2\%} &
                                       \multicolumn{2}{c}{58.9\% (OTT-EA higher)} \\
\midrule
Annual 95th-pct uplift factor & \multicolumn{2}{c}{$\times 1.53$} &
                                \multicolumn{2}{c}{$\times 2.83$} \\
\bottomrule
\end{tabular}
\end{table}

\subsubsection{Conditional sensitivity to
$\gamma_{\mathrm{X}}/\gamma$}

The OTT-EA 95th-percentile uplift at four pinned values of
$\gamma_{\mathrm{X}}/\gamma$ is: $\times 1.52$ at $\gamma_{\mathrm{X}}/\gamma
= 1$ (no cross-cutting amplification, equivalent to T-EA uniform
coupling, matching T-EA's $\times 1.53$ to within rounding);
$\times 1.93$ at $1.5$; $\times 2.54$ at $2$ (the prior median);
$\times 5.12$ at $3$ (above the prior 95th percentile). The
prior-integrated uplift of $\times 2.83$ exceeds the value at the
prior median ($\times 2.54$) because $e^{\gamma_{\mathrm{X}} Z}$ is
convex in $\gamma_{\mathrm{X}}$: replicates in the right tail of the
prior shift the prior-integrated quantile above its fixed-median
equivalent. The qualitative finding (cross-cutting coupling inflates
the OTT-EA tail relative to T-EA) holds across the entire defensible
$\gamma_{\mathrm{X}}/\gamma$ range and is structurally guaranteed by
the coupling prior, which places probability one on
$\gamma_{\mathrm{X}}/\gamma > 1$ (the weak constraint $\geq 1$ alone
guarantees non-reversal of the uplift ordering).

The parallel sensitivity to the prior width $\sigma$ is reported in
Table~\ref{tab:gamma_X_sigma_sensitivity}: the OTT-EA p95 uplift
varies from $\times 2.67$ at $\sigma = 0.25$ (tight prior) to
$\times 3.42$ at $\sigma = 0.75$ (wide prior). The architectural
divergence claim survives the entire range.

\begin{table}[!htbp]
\centering
\small
\caption{OTT-EA 95th-percentile uplift factor under different prior
widths $\sigma$ on $\gamma_{\mathrm{X}}/\gamma$. The prior median is fixed at 2 throughout; only the spread
varies. The headline ($\sigma = 0.40$) row is in bold. The architectural
divergence (OTT-EA uplift exceeding T-EA's $\times 1.53$) survives
across the full sigma range.}
\label{tab:gamma_X_sigma_sensitivity}
\begin{tabular}{@{}cccc@{}}
\toprule
$\sigma$ & $\gamma_{\mathrm{X}}/\gamma$ 95th pct & OTT-EA $q_1$ p99 &
  OTT-EA $q_1$ uplift \\
\midrule
$0.25$ (tight)              & 2.50 & 51.8\% & $\times 2.67$ \\
$\boldsymbol{0.40}$ (headline) & 2.92 & 58.9\% & $\boldsymbol{\times 2.83}$ \\
$0.50$                      & 3.26 & 65.4\% & $\times 3.04$ \\
$0.60$                      & 3.66 & 77.1\% & $\times 3.11$ \\
$0.75$ (wide)               & 4.39 & 93.4\% & $\times 3.42$ \\
\bottomrule
\end{tabular}
\end{table}

\paragraph{Stream~C cohort-stratification sensitivity.}
The Stream~C cohort (Section~\ref{sec:pillar1_streamC},
Supplementary Table~\ref{tab:streamC}) has fourteen incidents with seven $X = Y$
classifications, giving an aggregate cross-cutting fraction of
$7/14 = 0.50$. An alternative inclusion convention that added the
Cloudflare~2023 near-miss to the cohort (without counting it as
$X = Y$) and treated Snowflake~2024 and AT\&T~2024 as architecturally
independent events would yield the slightly lower fraction
$7/16 = 0.44$. Treating the difference between these conventions purely
as a calibration update, shifting the prior median from
$\gamma_{\mathrm{X}}/\gamma = 2$ to $2 \times (0.50/0.44) = 2.27$,
raises the prior-integrated OTT-EA $95$th-percentile uplift from
$\times 2.83$ to approximately $\times 3.54$ and the $99$th percentile
from $58.9\%$ to roughly $79\%$. We do not adopt this shifted median as
the headline because the difference between the two conventions is
denominator-driven (the primary convention places Cloudflare~2023 in
the near-miss appendix and merges the Snowflake~2024 and
AT\&T~2024 incidents) rather than driven by new high-cross-cutting
incidents: the aggregate-fraction shift reflects a different bookkeeping
choice about the Cloudflare and Snowflake/AT\&T entries rather than new
evidence of higher cross-cutting risk. The prior median $\gamma_{\mathrm{X}}/\gamma = 2$ with width
$\sigma = 0.40$ remains the central specification. The alternative
median is reported as a sensitivity to make explicit that the
alternative convention is consistent with a strengthened tail-risk
reading, not just with the unchanged headline.

\subsubsection{Robustness}

The qualitative dependence finding, both that dependence inflates the
tail relative to the independence baseline within each architectural
class, and that OTT-EA tail inflation exceeds T-EA tail inflation, is
robust to wide variation in the campaign parameters $\alpha_Z$,
$\beta_Z$, and $\gamma$ (and, for OTT-EA, $\gamma_{\mathrm{X}}$): in all
parameter regions consistent with the prior, the dependence model
produces materially elevated risk relative to the independence baseline,
and the cross-cutting coupling produces materially larger inflation for
OTT-EA than for T-EA. The findings support the qualitative conclusion
that independence assumptions in the three-pillar framework
\emph{systematically understate tail risk in both architectural classes},
with the understatement being quantitatively larger for OTT-EA. Models
treating EA component attack events as mutually independent will
systematically underestimate the probability of severe outcomes,
particularly for architecturally segregated systems where the
correlated-campaign collapse of the segregation gain is the dominant
tail-risk mechanism.

\subsection{Tail Risk Characterisation}\label{subsec:tail_risk}

The upper tail of the year-1 annual compromise probability $q_1$ is
characterised by fitting a generalised Pareto distribution (GPD)
separately for each architectural class. The fit is a compact
parametric summary of the simulated prior predictive tail---an
economical way to report its decay class and to compare the two
architectures on a single shape parameter---not an empirical
extreme-value finding about observed incidents. The empirical 90th percentile
threshold and the fitted GPD parameters are reported in
Table~\ref{tab:gpd_arch}.

\begin{table}[!htbp]
\centering
\small
\caption{GPD fits to the upper tail of $q_1$ for each architectural class
under the primary prior calibration. Threshold $u$ is the empirical
90th percentile of the marginal $q_1$ distribution ($n_{\rm exc} = 800$
exceedances per architecture). Shape and scale parameter MLEs are
reported with 95\% bootstrap confidence intervals ($B = 2{,}000$
resamples). Both shape parameters are strictly positive, placing both
distributions in the Fr\'echet domain of extreme value theory (power-law
tail decay). Neither of these facts is evidential, and version~1 of
this paper presented them as though they were. As shown in
\S\ref{subsec:tail_algebra}, exponential amplification of a
gamma-distributed campaign intensity is regularly varying with index
$\beta_Z/\gamma$, so a positive $\hat\xi$ is guaranteed by the prior
and the ratio of the two shapes recovers the coupling ratio:
$0.518/0.271 = 1.91$ against a prior median
$\gamma_{\mathrm{X}}/\gamma = 2$. The non-overlapping bootstrap
intervals ($\hat\xi^{\rm OTT} \in [0.42, 0.61]$ against $\hat\xi^{\rm
T} \in [0.19, 0.36]$) therefore establish that the two priors are
separated, not that the architectures are. What remains informative is
the magnitude of the fitted shapes conditional on the coupling prior,
and the goodness-of-fit contrast of
\S\ref{sec:supp_gpd_gof}, in which the T-EA fit is adequate near its
threshold while the OTT-EA fit is rejected by both tests.}
\label{tab:gpd_arch}
\begin{tabular}{@{}lccccc@{}}
\toprule
\textbf{Architecture} & \textbf{Threshold $u$} &
\textbf{$\hat\xi$ (shape)} & \textbf{95\% CI} &
\textbf{$\hat\sigma_u$ (scale)} & \textbf{95\% CI} \\
\midrule
T-EA   & $0.118$ & $0.271$ & $[0.189,\ 0.356]$ & $0.066$ & $[0.059,\ 0.074]$ \\
OTT-EA & $0.104$ & $0.518$ & $[0.424,\ 0.607]$ & $0.089$ & $[0.077,\ 0.103]$ \\
\bottomrule
\end{tabular}
\end{table}

\paragraph{Goodness-of-fit caveat.} Formal goodness-of-fit tests
(Anderson--Darling and Kolmogorov--Smirnov, parametric bootstrap,
Supplementary \S\ref{sec:supp_gpd_gof}) support the T-EA GPD fit at
the primary threshold ($p_{\rm AD} = 0.28$; $p_{\rm KS} = 0.13$).
The OTT-EA fit is rejected by both tests ($p_{\rm AD} < 0.001$;
$p_{\rm KS} = 0.003$). The proximate cause is the latent-campaign
coupling: when $Z(t)$ is suppressed (independence configuration),
the OTT-EA upper tail fits a GPD adequately ($p_{\rm AD} = 0.13$);
when $Z(t)$ is active, the GPD is rejected, because $Z(t)$
preferentially activates cross-cutting edges via
$\gamma_{\mathrm{X}}/\gamma \geq 1$ and induces a mixture between a
subgraph-dominant low-$Z$ regime and a cross-cutting-dominant
high-$Z$ regime. A two-component GPD mixture fits the OTT-EA upper
tail substantially better than a single GPD (Kolmogorov--Smirnov
$D = 0.026$ vs $0.037$), with $\sim 80\%$ of upper-tail mass on a
moderate-shape component and $\sim 20\%$ on a heavier
bounded-support component. The qualitative tail-divergence finding
survives: direct percentile comparisons (T-EA 99th percentile
$35.2\%$; OTT-EA 99th percentile $58.9\%$;
Table~\ref{tab:dependence_comparison}) are
distribution-free and confirm a heavier OTT-EA upper tail.
$\hat\xi^{\rm OTT} > \hat\xi^{\rm T}$ should therefore be
interpreted as a tail-heaviness ordering rather than as formal
statistical inference about Pareto tail indices, since the GPD
model is misspecified for OTT-EA. The mixture structure is the
empirical signature of the model's structural mechanism rather than
a defect of the model.

Both shape parameters are positive under the primary calibration,
placing both classes in the \emph{Fr\'echet domain} of extreme value
theory. The characterisation is prior-conditional, and
Table~\ref{tab:gpd_prior_robust} reports its full behaviour across
the calibration set: the heavy-tail ordering holds with
non-overlapping intervals at the primary and conservative
calibrations, and under the pessimistic prior the fitted shapes turn
negative for both classes as saturation against the $q = 1$ ceiling
dominates, the regime analysed below.

\begin{table}[!htbp]
\centering
\small
\caption{GPD shape parameter $\hat\xi$ across prior calibrations, with
95\% bootstrap confidence intervals ($B = 2{,}000$ resamples). The
architectural heavy-tail ordering ($\hat\xi_{\mathrm{OTT}} >
\hat\xi_{\mathrm{T}}$) holds with non-overlapping CIs across the
primary and conservative calibrations. Under the pessimistic
calibration both architectures sit in the saturation regime
(distribution truncated against $q \leq 1$), so the GPD shape captures
saturation geometry rather than tail-decay heaviness.}
\label{tab:gpd_prior_robust}
\begin{tabular}{@{}lcccc@{}}
\toprule
\textbf{Calibration} & \textbf{$\hat\xi^{\mathrm{T}}$} &
\textbf{95\% CI} & \textbf{$\hat\xi^{\mathrm{OTT}}$} &
\textbf{95\% CI} \\
\midrule
Pessimistic (median $\sim$13\% / $\sim$10\%) &
  $-0.11$ & $[-0.18,\, -0.03]$ & $-1.61$ & $[-2.65,\, -1.00]$ \\
Primary (median $\sim$4\% / $\sim$2.6\%) &
  $\phantom{-}0.27$ & $[\phantom{-}0.19,\, \phantom{-}0.36]$ &
  $\phantom{-}0.52$ & $[\phantom{-}0.42,\, \phantom{-}0.61]$ \\
Conservative (median $\sim$1.2\% / $\sim$0.9\%) &
  $\phantom{-}0.39$ & $[\phantom{-}0.30,\, \phantom{-}0.49]$ &
  $\phantom{-}0.82$ & $[\phantom{-}0.72,\, \phantom{-}0.93]$ \\
\bottomrule
\end{tabular}
\end{table}

The shape-parameter gap (Table~\ref{tab:gpd_arch}) is somewhat
larger than under a fixed-point evaluation at
$\gamma_{\mathrm{X}}/\gamma = 2$, where the OTT-EA $\hat\xi$ comes
out closer to $0.47$: the difference is the
convex-in-$\gamma_{\mathrm{X}}$ amplification of right-tail prior
mass quantified in \S\ref{subsec:dependence_results}.

The non-monotonic relation between prior pessimism and tail shape merits
brief explanation, since intuition suggests that a more pessimistic prior
should produce a heavier tail. Two opposing effects operate. First, raising
the prior median for the underlying integrated hazard $\lambda_{\mathrm{sys}}$
generates more extreme realisations and would, on its own, increase tail
mass. Second, the outcome variable is not $\lambda$ but
$q = 1 - \exp(-\lambda)$, which is bounded above by unity. For large
$\lambda$ the mapping saturates near $q = 1$ and the upper tail is
mechanically truncated. Under the pessimistic calibration, the bulk of the
$q_1$ distribution sits near $0.13$ but its upper tail compresses against
the $q = 1$ ceiling, producing a left-skewed near-saturation distribution
with a finite right endpoint, which EVT classifies as light-tailed
(reversed-Weibull / bounded-support max-domain of attraction). The heavy-tail
finding therefore applies in the regime where saturation does not dominate,
which is the policy-relevant regime for any prior consistent with the
analogue evidence used in Pillar~I.

Within its domain of validity, the result indicates that the tail decays
as a power law rather than exponentially, and that the expected value of
$q_1$ conditional on exceedance of the threshold is substantially higher
than $u$ itself. Concretely:
\begin{itemize}
  \item For T-EA: $\Pr(q_1 > 0.50 \mid q_1 > u) \approx 0.03$, giving
    unconditional $\Pr(q_1 > 0.50) \approx 0.003$.
  \item For OTT-EA: $\Pr(q_1 > 0.50 \mid q_1 > u) \approx 0.13$, giving
    unconditional $\Pr(q_1 > 0.50) \approx 0.013$.
\end{itemize}
Both architectures exhibit non-trivial tail mass from a governance
perspective. OTT-EA's higher conditional exceedance probability above
50\% is the architectural signature of cross-cutting tail amplification.

The physical driver of the positive $\hat\xi$ is the campaign amplification
factor $\gamma$ (and, for OTT-EA, the cross-cutting amplification
$\gamma_{\mathrm{X}}$) interacting with the latent campaign intensity
$Z$: in the subset of replicates where both amplification factors and
$Z$ are large, the effective system hazard is elevated across all attack
paths simultaneously. This is precisely the correlated campaign structure
modelled in the dependence layer. The positive $\hat\xi$ confirms that
this structure generates qualitatively different and more adverse tail
behaviour than independence assumptions, which produce sub-exponential
tails. The architectural amplification of $\hat\xi$ for OTT-EA confirms
that the cross-cutting coupling specification is operative, not a
formal artefact.

\begin{figure}[!htbp]
\centering
\includegraphics[width=\textwidth]{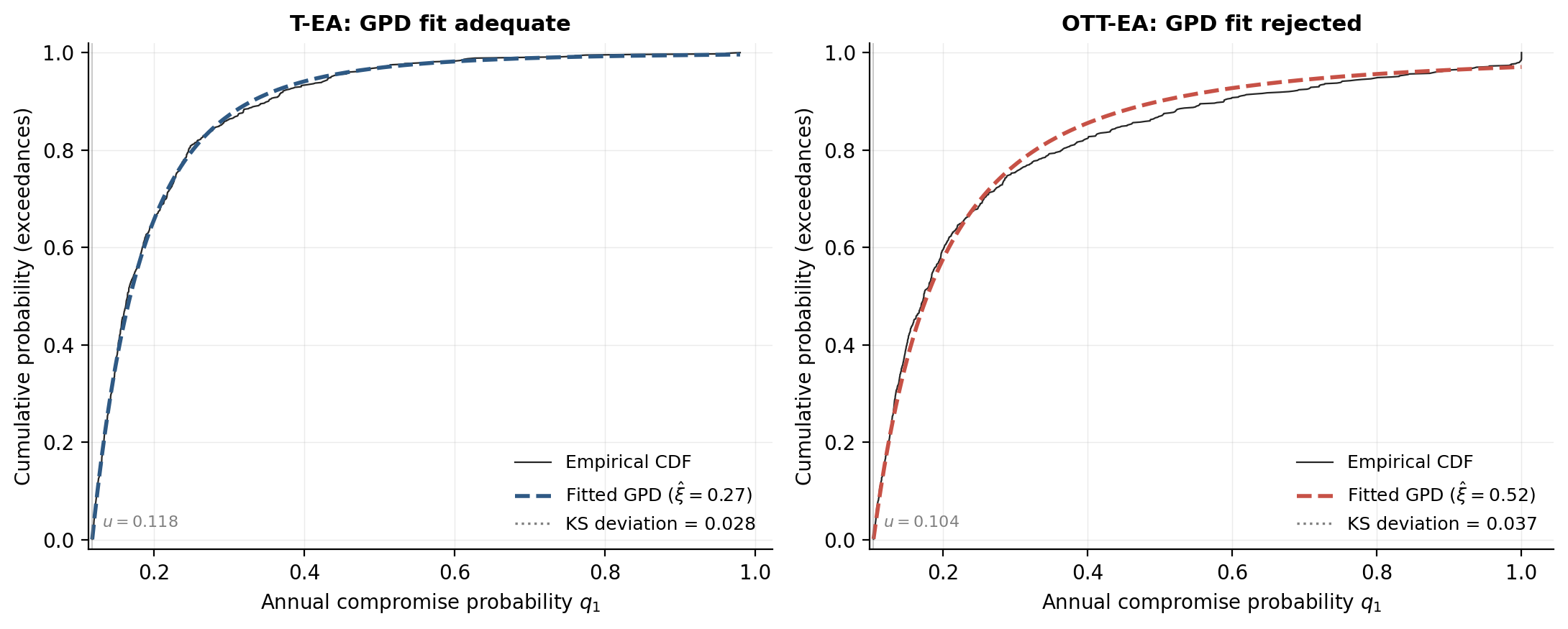}
\caption{Generalised Pareto Distribution goodness-of-fit comparison
across architectures. \textbf{Left:} T-EA, where the empirical CDF
of exceedances above the 90th percentile (threshold $u = 0.118$)
closely tracks the fitted GPD curve ($\hat\xi = 0.27$, $\hat\sigma =
0.066$); Anderson--Darling $p_{\mathrm{AD}} = 0.28$ and
Kolmogorov--Smirnov $p_{\mathrm{KS}} = 0.13$ both fail to reject the
GPD null. \textbf{Right:} OTT-EA, where the empirical CDF
systematically deviates from the fitted GPD ($\hat\xi = 0.52$,
$\hat\sigma = 0.089$, threshold $u = 0.104$). The empirical
distribution carries more mass in the moderate-to-upper tail than
the fitted single GPD predicts. Both tests reject the GPD null
($p_{\mathrm{AD}} < 0.001$, $p_{\mathrm{KS}} = 0.003$). The visible
asymmetry is the empirical signature of the mixture structure
induced by the cross-cutting coupling $\gamma_{\mathrm{X}}/\gamma
\geq 1$: under high realisations of the latent campaign variable
$Z(t)$, cross-cutting edges activate preferentially, producing a
high-$Z$ tail regime that no single GPD can capture. Full
goodness-of-fit tables, threshold sensitivity, and the
two-component mixture diagnostic are in Supplementary
\S\ref{sec:supp_gpd_gof}. Results from $R = 8{,}000$ Monte Carlo
replicates per architecture (seed $2024$).}
\label{fig:gpd}
\par\smallskip\noindent\small\textit{Alt text:} Two-panel
goodness-of-fit comparison of GPD fits to the upper tail of the
annual compromise probability. Left panel (T-EA): empirical CDF
(black solid line) closely follows the fitted GPD curve (blue
dashed line) with maximum Kolmogorov-Smirnov deviation 0.028; the
fit is statistically adequate. Right panel (OTT-EA): empirical CDF
(black solid line) systematically deviates from the fitted GPD
curve (red dashed line) in the moderate-to-upper tail; maximum
deviation 0.037. The fit is rejected by both Anderson-Darling and
Kolmogorov-Smirnov tests at p $<$ 0.005. The visible asymmetry is
the empirical signature of the cross-cutting-coupling mixture
structure in the OTT-EA model.
\end{figure}

\subsection{Cumulative Risk}\label{subsec:cumulative}

The cumulative probability of at least one compromise grows
substantially with the deployment horizon, with magnitude differing
between architectural classes. Under T-EA primary-calibration
parameter values the 10-year cumulative compromise probability sits
in the mid-thirties percentage range, with a 90\% credible interval
spanning roughly $[7\%, 84\%]$. Under OTT-EA primary-calibration
parameter values the 10-year cumulative is broadly comparable in
central tendency but with substantially wider credible bounds: the
upper bound exceeds T-EA's, reflecting the tail-heavy character of
the OTT-EA distribution under correlated campaigns and the
prior-integrated amplification of $\gamma_{\mathrm{X}}/\gamma$.
Over a 25-year horizon the median cumulative for both architectures
rises into the high-fifties to mid-sixties range.
Table~\ref{tab:cumulative} presents cumulative-risk projections
across all methodological reference points and both architectural
classes. The interval-width contrast between architectures (OTT-EA
intervals materially wider than T-EA's at long horizons) is the
qualitatively important feature, not the specific medians.

\begin{figure}[!htbp]
\centering
\includegraphics[width=0.85\textwidth]{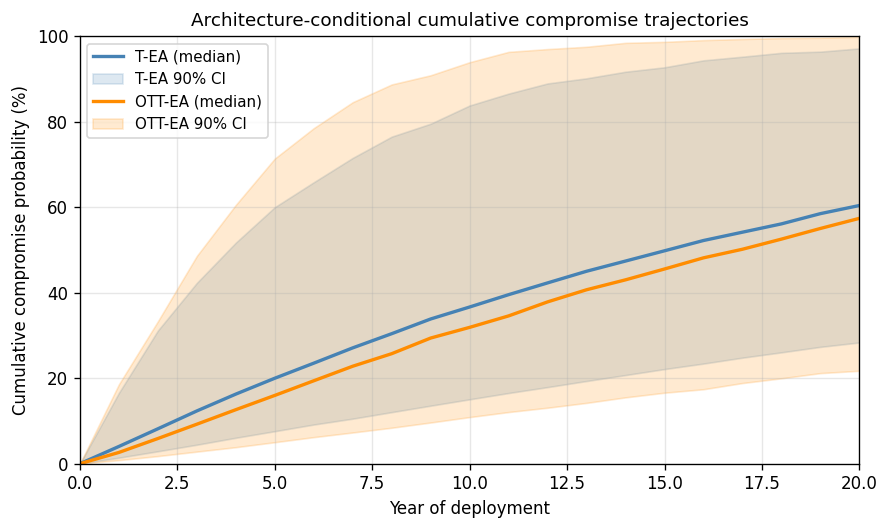}
\caption{Architecture-conditional comparison of cumulative compromise
probability trajectories over a 20-year deployment horizon. T-EA
(blue) and OTT-EA (orange) full-dependence trajectories with 90\%
credible bands. The corresponding independence-baseline trajectories
(not shown for clarity) sit below the dependence trajectories at every
horizon, with the gap quantified in Table~\ref{tab:dependence_comparison}.
The OTT-EA median sits below the T-EA median throughout the horizon,
but the upper credible bound for OTT-EA exceeds that for T-EA at
every horizon (consistent with the year-1 values of $17.4\%$ versus
$16.5\%$), reflecting the tail divergence documented
in Section~\ref{subsec:tail_risk}. Trajectories are computed by full
re-simulation at each horizon $t = 1, \ldots, 20$ with year-by-year
campaign-state resampling ($R = 8{,}000$ replicates per architecture
per horizon, seed~2024), so the year-10 values agree with the
tabulated $Q_{10}$ quantiles.}
\label{fig:cumulative_arch}
\par\smallskip\noindent\small\textit{Alt text:} Two-line plot of cumulative compromise probability over a 20-year horizon for T-EA (blue) and OTT-EA (orange) under the full dependence model. T-EA rises faster and higher than OTT-EA across the horizon; both shown with 90\% credible bands.
\end{figure}

\begin{figure}[!htbp]
\centering
\includegraphics[width=0.95\textwidth]{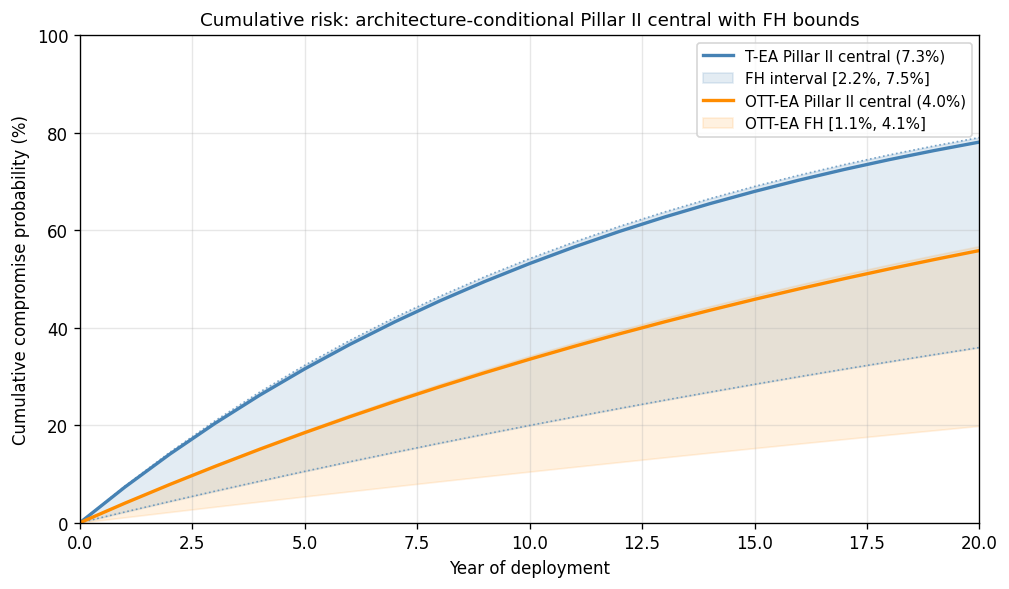}
\caption{Cumulative EA technical (solid) and access-pathway (dashed)
compromise probability under the Pillar~II independence-conditional
central T-EA parameterisation, with the non-EA baseline and the
Fr\'echet--Hoeffding lower and upper bounds for the technical-compromise
curve. The technical-compromise curve under independence crosses 50\%
near year~10. The lower bound of 2.2\% yields approximately 20\% over
10~years. The shaded band spans the Fr\'echet--Hoeffding model-consistent
interval under any dependence structure. Seed 2024; $R = 8{,}000$.}
\label{fig:cumulative_risk}
\par\smallskip\noindent\small\textit{Alt text:} Multi-curve plot of cumulative compromise probability over deployment horizon for the T-EA Pillar II independence-conditional central parameterisation. Solid: EA technical compromise crossing 50\% near year 10. Dashed: access-pathway compromise. Outer bands give the Fréchet-Hoeffding bracket for the technical curve; non-EA baseline shown for comparison.
\end{figure}

\begin{table}[!htbp]
\centering
\small
\caption{Cumulative operational compromise probability across reference
estimates and architectural classes. Three-pillar rows use the compound
formula $1 - (1-q)^T$. The Bayesian median rows use the
simulation-derived 10-year values, since the model employs time-varying
annual rates whose compounding can differ from the stationary-formula
result. The 5- and 25-year values use the formula as lower-bound
approximations. \textbf{The Bayesian-median rows are included for
distributional completeness only; as disclosed in
Section~\ref{sec:bayesian}, the prior calibrations match the
architecture-matched three-pillar empirical ranges, so these rows should
not be read as independent fourth estimates.} The S-column refers to the
synthesis-hierarchy tier (Section~\ref{sec:synthesis}).}
\label{tab:cumulative}
\begin{tabular}{@{}llcccc@{}}
\toprule
\textbf{Reference point} & \textbf{S} & \textbf{Annual $\Lambda$} &
\textbf{5-yr cum.} & \textbf{10-yr cum.} & \textbf{25-yr cum.} \\
\midrule
\multicolumn{6}{@{}l}{\textit{T-EA reference points}} \\
Channel-minimum heuristic (four-channel)         & S5 & 1.5\% & 7.3\%  & 14.0\% & 31.5\% \\
Fr\'echet--Hoeffding lower (Pillar~II) & S4 & 2.2\% & 10.5\% & 19.9\% & 42.7\% \\
Bayesian median (Layer~IV, prior-cond.)  & ---  & 4.0\% & 18.5\% & 36.6\% & 64.0\% \\
Channel-aggregate projection (Pillar~III)      & S5 & 5.4\% & 24.2\% & 42.6\% & 75.0\% \\
Pillar~II central (independence)       & S5 & 7.3\% & 31.5\% & 53.1\% & 85.0\% \\
Fr\'echet--Hoeffding upper (Pillar~II) & S4 & 7.5\% & 32.3\% & 54.1\% & 85.8\% \\
\midrule
\multicolumn{6}{@{}l}{\textit{OTT-EA reference points}} \\
Channel-minimum heuristic (four-channel)         & S5 & 1.0\% & 4.9\%  &  9.6\% & 22.2\% \\
Fr\'echet--Hoeffding lower (Pillar~II) & S4 & 1.1\% & 5.4\%  & 10.4\% & 24.0\% \\
Bayesian median (Layer~IV, prior-cond.)  & ---  & 2.6\% & 12.3\% & 31.9\% & 48.2\% \\
Channel-aggregate projection (Pillar~III)      & S5 & 3.0\% & 14.1\% & 26.3\% & 53.3\% \\
Pillar~II central (independence)       & S5 & 3.9\% & 18.0\% & 32.8\% & 63.0\% \\
Fr\'echet--Hoeffding upper (Pillar~II) & S4 & 4.0\% & 18.5\% & 33.5\% & 64.0\% \\
\bottomrule
\end{tabular}
\end{table}

The qualitative pattern, that cumulative risk grows non-linearly with
the deployment horizon and that the uncertainty range itself expands
substantially over time, is robust to parameter variations and is
offered as a more confident finding than any specific quantile. Within
each architectural class, the cumulative risk over a 10-year horizon
exceeds 30--50\% at the channel-aggregate projection and remains material even
at the channel-minimum heuristic.

\subsection{Exceedance Probabilities Over Policy Thresholds}\label{subsec:exceedance}

Wide credible intervals do not imply ignorance: the prior predictive
distribution can assign high probability to the event that the true risk
exceeds a policy-relevant threshold even when the credible interval itself
spans several orders of magnitude. Table~\ref{tab:exceedance} reports
$\Pr(q_1 > \tau)$ and $\Pr(Q_{10} > \tau)$ under both the full dependence
model and the independence baseline, separately for each architectural
class.

The acceptability threshold $\tau$ is a normative, jurisdictional
determination outside a technical framework's scope; no settled
benchmark exists for EA systems, and thresholds from other
safety-critical domains do not transfer defensibly. The table
therefore spans two orders of magnitude ($0.5\%$ to $50\%$
annually) so the output combines with whatever criterion the reader
brings. The $1\%$ row receives most discussion not as an externally
derived bound but because it is the lowest threshold at which the
exceedance probability reliably exceeds $0.5$ across the primary
and pessimistic calibrations for both classes, making it the most
prior-separating row.

\begin{table}[!htbp]
\centering
\small
\caption{Architecture-conditional exceedance probabilities and maximum
defensible deployment horizon $T^*(\tau)$ under the full dependence model
($R = 8{,}000$ replicates, seed 2024). $T^*(\tau)$ is defined as the
first deployment year at which the median cumulative compromise
probability $Q_T$ exceeds $\tau$; the longest horizon over which the
median remains within the threshold is therefore $T^*(\tau) - 1$.
The $T^*$ columns are corrected in this version. Version~1 computed
them by stationary extrapolation from year one,
$\lceil \log(1-\tau)/\log(1-q_1) \rceil$ per replicate, which is not
the definition above and which understates cumulative risk because the
model draws a fresh campaign intensity each year. The understatement
was larger for OTT-EA, whose annual distribution is the more skewed,
so the correction removes most of the apparent architectural
difference: the gap across the five thresholds falls from
$0/1/3/5/9$ years to $0/1/1/2/1$.}
\label{tab:exceedance}
\begin{tabular}{rcccccc}
\toprule
& \multicolumn{3}{c}{\textbf{T-EA}} & \multicolumn{3}{c}{\textbf{OTT-EA}} \\
\cmidrule(lr){2-4} \cmidrule(lr){5-7}
$\tau$ & $\Pr(q_1>\tau)$ full & $\Pr(q_1>\tau)$ indep. & $T^*_{\mathrm{T}}$
& $\Pr(q_1>\tau)$ full & $\Pr(q_1>\tau)$ indep. & $T^*_{\mathrm{OTT}}$ \\
\midrule
$0.5\%$  & 1.000 & 1.000 & ---   & 0.994 & 0.984 & ---   \\
$1\%$    & 0.988 & 0.974 & ---   & 0.907 & 0.795 & ---   \\
$2\%$    & 0.843 & 0.741 & ---   & 0.630 & 0.411 & ---   \\
$5\%$    & 0.386 & 0.250 & yr~2  & 0.252 & 0.087 & yr~2  \\
$10\%$   & 0.138 & 0.058 & yr~3  & 0.104 & 0.015 & yr~4  \\
$20\%$   & ---   & ---   & yr~6  & ---   & ---   & yr~7  \\
$30\%$   & ---   & ---   & yr~8  & ---   & ---   & yr~10 \\
$50\%$   & ---   & ---   & yr~16 & ---   & ---   & yr~17 \\
\bottomrule
\end{tabular}
\end{table}

\begin{figure}[!htbp]
\centering
\includegraphics[width=0.95\textwidth]{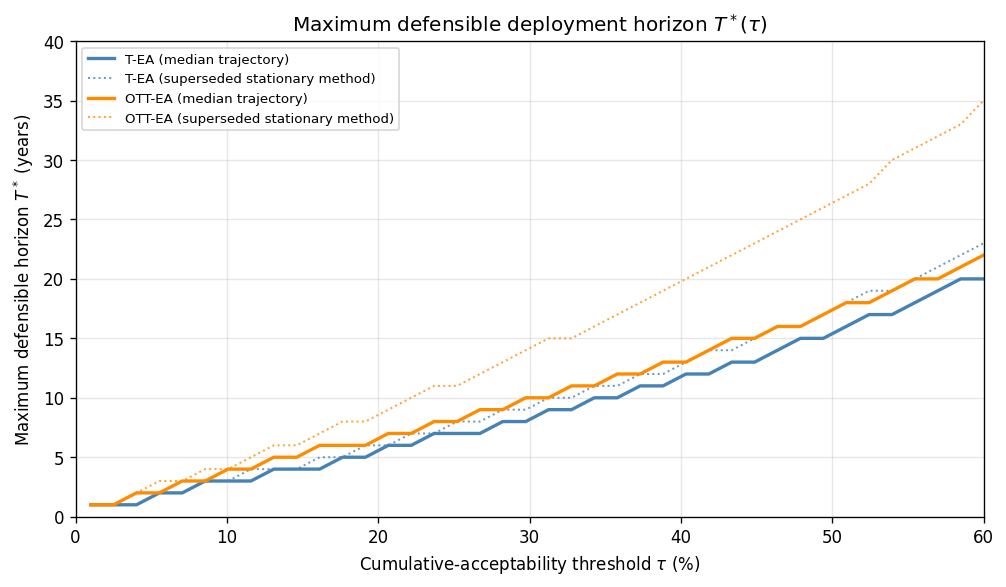}
\caption{Maximum defensible deployment horizon $T^*(\tau)$ as a
continuous function of the cumulative-acceptability threshold $\tau$,
defined as in Table~\ref{tab:exceedance} as the first year at which the
median cumulative risk crosses $\tau$. Solid curves give that quantity
for T-EA (blue) and OTT-EA (orange) over $R = 8{,}000$ replicates at
seed 2024. Dotted curves give the superseded version~1 method, a
per-replicate stationary extrapolation from year one, and are retained
so the size of the correction is visible rather than merely asserted.
The two methods agree closely for T-EA and diverge sharply for OTT-EA,
because stationary extrapolation from a low median understates
cumulative risk most where the annual distribution is most skewed. On
the corrected basis the two architectures track one another within one
to two years across the whole range, so the figure no longer supports a
claim of architectural divergence in defensible horizon. The
step-function appearance is integer-year rounding of $T^*$.}
\label{fig:deployment_horizons}
\par\smallskip\noindent\small\textit{Alt text:} Line plot of the maximum defensible deployment horizon T*(tau) in years on the vertical axis against the cumulative-acceptability threshold tau on the horizontal axis. Solid blue (T-EA) and solid orange (OTT-EA) curves rise together and stay within one to two years of each other across the range. Dotted curves show the superseded stationary method, under which the orange curve rises far above the blue one at higher thresholds.
\end{figure}

At a $1\%$ annual threshold, both architectures show high
exceedance probabilities under the primary calibration: T-EA
exceedance is near saturation (above $0.95$); OTT-EA exceedance is
substantially above $0.5$ but somewhat lower, consistent with
OTT-EA's lower median. That OTT-EA exceedance remains substantially
above $0.5$ at this threshold indicates that the lower median is
not the dominant signal at low thresholds.

At higher thresholds the architectural ordering changes character.
Even at the $10\%$ annual threshold T-EA exceedance still exceeds
OTT-EA (Table~\ref{tab:exceedance}: $0.138$ vs $0.104$), but the gap
narrows substantially as $\tau$ grows: the OTT-EA distribution
overtakes T-EA in the upper tail near the 95th-percentile region
($\tau \approx 16$--$17\%$), as confirmed by the GPD-fitted upper tails
(Table~\ref{tab:gpd_arch}) and the architecture-conditional CCDFs of
Figure~\ref{fig:exceedance_arch}. This is the structural signature of
OTT-EA's tail-heaviness: at policy-relevant high thresholds the lower
OTT-EA median does not translate into a correspondingly lower
exceedance probability, and deep in the tail OTT-EA assigns more mass
than T-EA.

The deployment-horizon results are corrected in this version and no
longer show an architectural contrast
(Figure~\ref{fig:deployment_horizons}). Computed as the table defines
them, from the median simulated cumulative trajectory, the two classes
track one another within one to two years across the whole threshold
range: at $\tau = 20\%$ the horizons are year~6 and year~7, at
$\tau = 30\%$ year~8 and year~10, and at $\tau = 50\%$ year~16 and
year~17. Version~1 reported year~9 against year~14 and year~17 against
year~26, computed instead by stationary extrapolation from year one.
That method understates cumulative risk whenever the annual hazard
varies over the horizon, and it understates it more for OTT-EA,
because its annual distribution is the more skewed: OTT-EA has a
median $q_1$ of $2.63\%$ against a median $Q_{10}$ of $31.9\%$, an
effective annual hazard near $3.8\%$, whereas T-EA runs $4.0\%$
against $36.6\%$, an effective hazard near $4.5\%$. The apparent
OTT-EA horizon advantage was an artefact of the extrapolation, not a
property of the model.

The finding that survives is that both architectures exhaust any
stringent cumulative-acceptability threshold on a similar and short
timescale. Deployment horizon does not discriminate between the two
designs, and \S\ref{sec:policy} is revised accordingly.

Both exceedance and horizon findings are conditional on the
primary prior calibration. As the sensitivity analysis below
demonstrates, the exceedance findings are materially sensitive to
prior choice for both architectures.

\begin{figure}[!htbp]
\centering
\includegraphics[width=\textwidth]{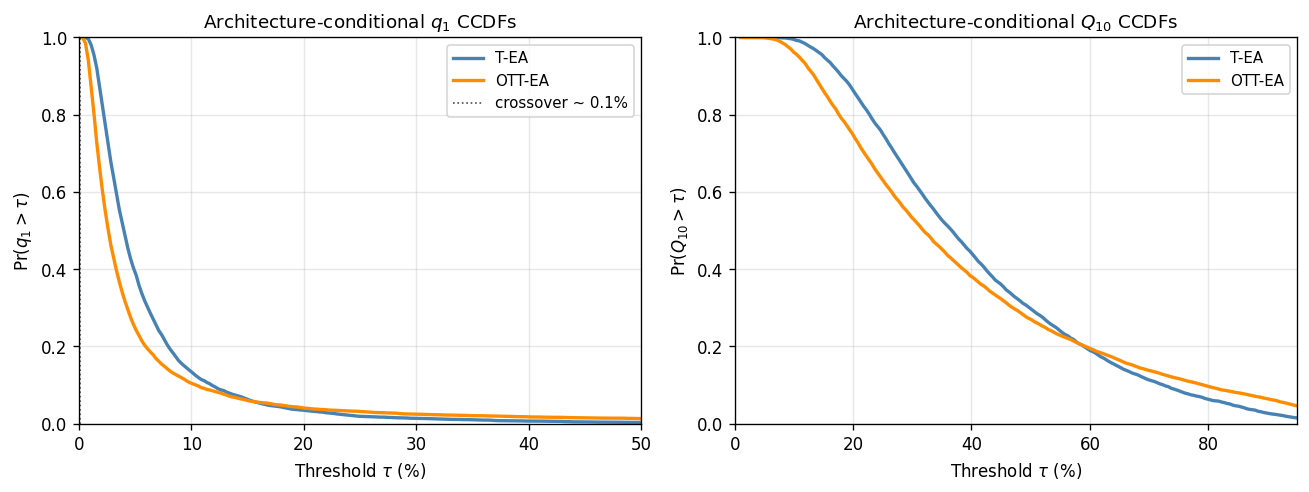}
\caption{Architecture-conditional CCDFs for $q_1$ (left) and $Q_{10}$
(right) under the full dependence model. T-EA (blue) sits above OTT-EA
(orange) at low thresholds (the median region), but the curves cross
deep in the upper tail and OTT-EA exceeds T-EA at the extreme. The
crossover threshold for $q_1$ is approximately $\tau \approx 16$--$17\%$
(near the 95th percentile of both distributions), beyond which OTT-EA
assigns higher exceedance probability than T-EA. This is the
distributional signature of the cross-cutting coupling specification.
The corresponding independence-baseline CCDFs (not shown) sit below
the dependence curves at every threshold. The exceedance differences
are quantified in Table~\ref{tab:dependence_comparison}. Results from
8{,}000 Monte Carlo replicates per architecture (seed~2024).}
\label{fig:exceedance_arch}
\par\smallskip\noindent\small\textit{Alt text:} Two-panel complementary cumulative distribution function (CCDF) plot. Left: P(q1 > tau) versus tau for T-EA (blue) and OTT-EA (orange); T-EA sits above OTT-EA at low tau but the curves cross near tau = 16 to 17\%, after which OTT-EA exceeds T-EA. Right: equivalent CCDF for Q10 showing similar crossover behaviour at high thresholds.
\end{figure}

\subsection{Prior Sensitivity for Exceedance Probabilities}\label{subsec:exceedance_priorsens}

The exceedance findings above are conditional on the primary prior
calibrations. Supplementary \S\ref{sup:prior_sens} reports the
architecture-conditional exceedance probabilities under three prior
calibrations spanning the plausible range of analogical evidence
(\emph{pessimistic}: analogue population average rate; \emph{primary};
\emph{conservative}: analogue population performing substantially
better than the historical record). The qualitative findings are:
under the conservative prior the exceedance probability at the 1\%
threshold drops materially for both architectures (T-EA from
$0.988 \to 0.605$; OTT-EA from $0.907 \to 0.460$); under the
pessimistic prior the same probability rises to saturation
($1.000$ for both). The architectural ordering
(T-EA exceedance $\geq$ OTT-EA at low thresholds) is preserved across
all three calibrations, with both architectures saturating under the
pessimistic prior. The exceedance probabilities are therefore
prior-conditional in magnitude but stable in direction across the
defensible prior range. The headline finding that the conservative
case lies near $\Pr(q_1 > 1\%) \approx 0.5$ for OTT-EA defines the
lower edge of the policy-relevant range.

\subsection{Parameter Uncertainty Decomposition}\label{subsec:param_decomp}

A one-way tornado analysis identifies the architecture-conditional
dominant sources of spread in the 10-year cumulative compromise
probability $Q_{10}$. Each parameter
is fixed in turn at its prior 10th and 90th percentiles while the
others are sampled from their priors. The swing is the resulting change
in median $Q_{10}$. The tornado is taken over the parameters the
canonical model treats as uncertain: the three log-normal modifiers
for both architectures, plus the cross-cutting coupling
ratio $\gamma_{\mathrm{X}}/\gamma$ for OTT-EA. Quantities the model
holds fixed ($\gamma$, $P_3 \cdot P_4$, $\sigma_\lambda$, and the
per-channel arrays) have no prior spread to sweep and do not enter.

\textit{T-EA decomposition.} The three modifiers produce comparable
swings in median $Q_{10}^{\mathrm{T}}$ (central $36\%$): system
targeting intensity $\delta_{\mathrm{target}}$ is the largest (swing of
$0.27$, range $0.29$--$0.56$), followed by concentration
$\delta_{\mathrm{concen}}$ (swing $0.22$, range $0.31$--$0.53$) and the
EA-specific attack-surface modifier $\delta_{\mathrm{EA}}$ (swing
$0.21$, range $0.31$--$0.52$). No single parameter dominates,
consistent with the flat variance decomposition of
\S\ref{subsec:evppi}.

\textit{OTT-EA decomposition.} The same three modifiers again produce
comparable swings in median $Q_{10}^{\mathrm{OTT}}$ (central $32\%$),
with $\delta_{\mathrm{target}}^{\mathrm{OTT}}$ the largest (swing of
$0.25$), followed by $\delta_{\mathrm{concen}}^{\mathrm{OTT}}$ and
$\delta_{\mathrm{EA}}^{\mathrm{OTT}}$ (swing $0.20$ each). The
OTT-EA-specific cross-cutting coupling ratio
$\gamma_{\mathrm{X}}/\gamma$ contributes a smaller swing ($0.14$, range
$0.26$--$0.40$). The structural constraint
$\gamma_{\mathrm{X}}/\gamma \geq 1$ compresses its prior range
substantially ($90\%$ credible range $[1.52, 2.93]$), so the
architectural divergence between OTT-EA and T-EA is carried principally
by the constraint itself rather than by residual uncertainty about the
ratio. As in the EVPPI decomposition of \S\ref{subsec:evppi}, the three
modifiers are comparable in magnitude and the cross-cutting coupling
ratio carries a smaller share. The two methods measure different things
(range swing versus variance share), so the exact ordering of the
three modifiers differs slightly between them.

\begin{figure}[!htbp]
\centering
\includegraphics[width=\textwidth]{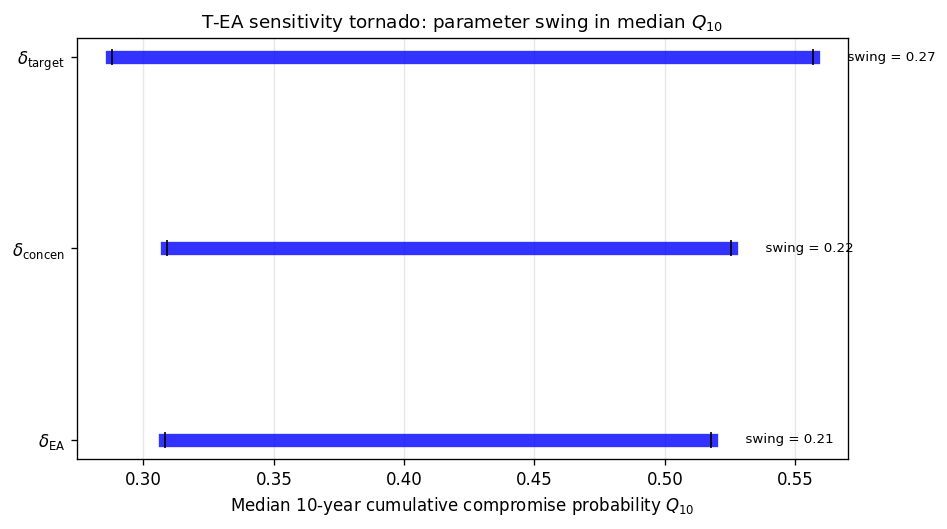}
\caption{Sensitivity tornado chart: swing in median 10-year cumulative
compromise probability $Q_{10}$ when each uncertain parameter is varied
across its prior 10th--90th percentile range (T-EA configuration). The
three log-normal modifiers produce comparable swings, with system
targeting intensity $\delta_{\mathrm{target}}$ the largest; no single
parameter dominates. The OTT-EA tornado, which also includes
the cross-cutting coupling ratio $\gamma_{\mathrm{X}}/\gamma$, is given
in Figure~\ref{fig:bayesian_tornado_ott}. Results from $R = 4{,}000$
Monte Carlo replicates per parameter point (seed~2024).}
\label{fig:bayesian_tornado}
\par\smallskip\noindent\small\textit{Alt text:} Horizontal tornado plot of T-EA configuration: each row shows the swing in median 10-year cumulative compromise probability when one parameter is varied across its prior 10th--90th percentile range. The three log-normal modifiers delta\_target, delta\_concen and delta\_EA produce comparable swings, with delta\_target the largest; no single parameter dominates.
\end{figure}

\begin{figure}[!htbp]
\centering
\includegraphics[width=\textwidth]{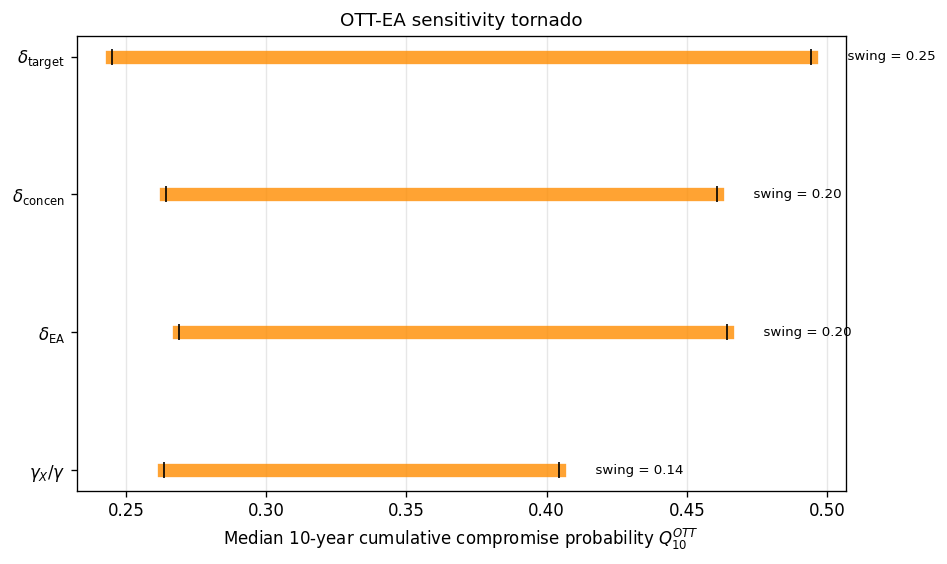}
\caption{Sensitivity tornado chart for the OTT-EA configuration. The
three log-normal modifiers produce comparable swings, with
$\delta_{\mathrm{target}}^{\mathrm{OTT}}$ the largest; the
OTT-EA-specific cross-cutting coupling ratio
$\gamma_{\mathrm{X}}/\gamma$ enters as a smaller additional
contributor. Under the shifted-lognormal prior on
$\gamma_{\mathrm{X}}/\gamma$ (Equation~\ref{eq:gamma_X_prior}, with the
structural constraint $\gamma_{\mathrm{X}}/\gamma \geq 1$ enforced by
construction, $90\%$ credible range $[1.52, 2.93]$), its prior range is
compressed, so the architectural divergence between OTT-EA and T-EA
persists structurally through the lower-bound constraint rather than
through residual uncertainty about the ratio. Results from
$R = 4{,}000$ Monte Carlo replicates per parameter point (seed~2024).}
\label{fig:bayesian_tornado_ott}
\par\smallskip\noindent\small\textit{Alt text:} Tornado plot for the OTT-EA configuration showing the three log-normal modifiers plus the OTT-EA-specific cross-cutting coupling ratio gamma\_X/gamma. The modifiers produce comparable swings, with delta\_target the largest; the coupling ratio enters as a smaller contributor.
\end{figure}

\subsubsection{EVPPI-style variance decomposition}\label{subsec:evppi}

A partial parameter-uncertainty decomposition (EVPPI-style variance
decomposition) quantifies the contribution of each uncertain parameter
to the total output variance $\mathrm{Var}(Q_{10})$, separately for
each architectural class. The decomposition is taken over exactly the
parameters the canonical Monte Carlo engine treats as uncertain, drawn
from the canonical priors: the three log-normal modifiers
($\delta_{\mathrm{EA}}$, $\delta_{\mathrm{target}}$,
$\delta_{\mathrm{concen}}$) for both architectures, plus
the cross-cutting coupling ratio $\gamma_{\mathrm{X}}/\gamma$ for
OTT-EA. Quantities the canonical model holds fixed ($\gamma$,
$P_3 \cdot P_4$, $\sigma_\lambda$, and the per-channel cross-cutting
and detection arrays) contribute no output variance and are not
decomposed. For T-EA, the three modifiers individually account for approximately
$61\%$ of total output variance in sum (the sum of their first-order
Sobol indices); resolving them jointly would capture this share plus
any interaction terms among them. For OTT-EA, the three modifiers and
$\gamma_{\mathrm{X}}/\gamma$ account for approximately $49\%$ in sum,
the larger residual reflecting the additional interaction structure
introduced by cross-cutting coupling.

The rank-order finding is the principal output, and it is flat: for
T-EA the three modifiers contribute comparable shares (each roughly
$18$--$25\%$), and for OTT-EA likewise ($12$--$15\%$ each) with the
coupling ratio adding a smaller first-order index ($\sim$9\%),
limited because the structural constraint
$\gamma_{\mathrm{X}}/\gamma \geq 1$ does the architectural work at
the prior level rather than through residual value uncertainty. A
substantial residual ($39\%$ T-EA, $51\%$ OTT-EA) is irreducible to
single-parameter learning. As a research-prioritisation finding, no
single parameter dominates, so empirical work across the modifiers
yields broadly comparable variance reduction; per-parameter detail
is in Supplementary \S\ref{sec:supp_evppi}.

\begin{figure}[!htbp]
\centering
\includegraphics[width=\textwidth]{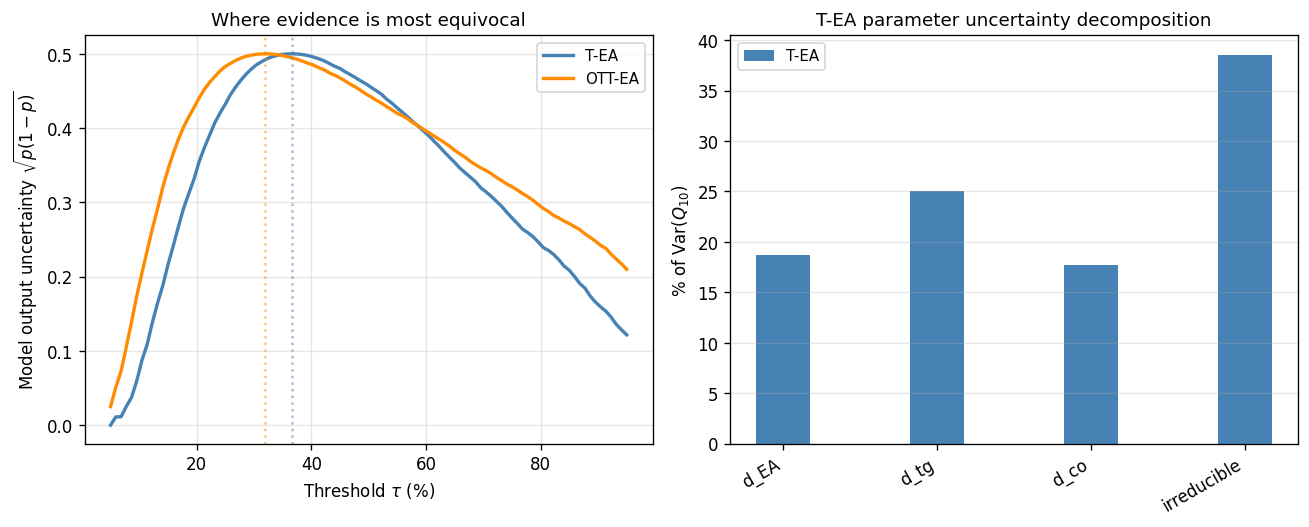}
\caption{Parameter uncertainty decomposition (T-EA configuration). Left:
model output uncertainty as a function of the risk threshold $\tau$,
showing where along the risk scale the evidence is most equivocal
(peaks near the median $q_1$, where the outcome is most
uncertain). Right: percentage of $\mathrm{Var}(Q_{10})$ attributable to
each uncertain parameter. The three log-normal modifiers contribute
comparable first-order shares, with no single parameter dominating; the
residual absorbs parameter interactions and the model's internal
stochasticity. The OTT-EA decomposition also includes the
OTT-EA-specific cross-cutting coupling ratio
$\gamma_{\mathrm{X}}/\gamma$ and is reported in Supplementary
\S\ref{sec:supp_ottevpi} (Table~\ref{tab:ott_evpi}).
First-order Sobol indices from $R = 2{,}000$ pick-freeze replicates per
parameter (seed~2024).}
\label{fig:evpi}
\par\smallskip\noindent\small\textit{Alt text:} Two-panel parameter uncertainty decomposition for T-EA. Left: model output uncertainty as a function of risk threshold tau, peaking near the median q1 where the outcome is most uncertain. Right: bar chart of percentage of Var(Q10) attributable to each uncertain parameter, with the three log-normal modifiers contributing comparable shares and a large irreducible residual.
\end{figure}

% =============================================================================
\section{Synthesis: What the Framework Collectively Establishes}\label{sec:synthesis}

\subsection{Structural Versus Calibration-Dependent Results}\label{sec:structural_vs_calibration}

The findings of this paper fall into three classes, and the weight a
reader should place on each differs accordingly. The classes are
groupings of the synthesis tiers S1--S6, which are formalised in the
enumeration at the end of this section and mapped claim-by-claim in
Table~\ref{tab:readers_guide}; one scheme, viewed at two
resolutions. The first class
follows from architecture alone and survives any defensible
calibration: the dominance ordering of
Proposition~\ref{prop:dominance}, the Fr\'echet--Hoeffding interval
logic, and the direction of the OTT-EA tail divergence, which holds
under any coupling prior with support in $[1, \infty)$. The second
class holds across the full calibration set examined (primary,
conservative, and pessimistic priors) but is not derivable from
structure: material multi-decade cumulative risk, the preservation of
the T-EA-over-OTT-EA central ordering, and the low-threshold
exceedance ordering. The third class is calibration-dependent: the
specific medians, GPD shape values, and inflation magnitudes quoted
at the primary calibration.

One clarification keeps the strongest tier honest. The structural
class is conditional on two premises: the superset-graph assumption
(EA deployment adds nodes and edges rather than replacing them) and
the $\delta \geq 1$ constraints on the EA modifiers. These premises
are not stipulations of convenience. Each is anchored to documented
incident properties: Salt Typhoon attacked infrastructure that
exists only because of the EA mandate; the Athens affair exploited
lawful-intercept capability dormant in a commercial switch; and
Storm-0558 demonstrated the concentration property for
platform-held key material. But the premises are defeasible, and what would defeat them can be
stated. The superset assumption fails if an EA mandate
\emph{substitutes} for riskier informal access mechanisms rather
than adding a new surface alongside them; the
$\delta_{\mathrm{target}} \geq 1$ constraint fails if EA components
attract no incremental adversary attention; and mandated defensive
standards could in principle offset the added surface
(\S\ref{sec:limitations}, the regulatory-uplift counterfactual).
Evidence for any of these would relegate the dominance ordering
from the structural class to a calibration question. The ordering
is therefore structural \emph{given premises argued from the
incident record}, not true by construction.

The boundary between the second and third classes needs a correction
that version~1 got wrong. That version argued the heavy-tail finding
was not hard-wired, on the ground that the fitted GPD shape changes
sign under the pessimistic calibration
(Table~\ref{tab:gpd_prior_robust}), so the model can produce the
opposite answer. The argument does not work, and it is withdrawn. The
sign change occurs because the output saturates against the $q = 1$
ceiling, which truncates the distribution into the bounded domain, not
because the underlying hazard tail becomes lighter. A sweep over
pinned coupling ratios makes the mechanism explicit: the fitted shape
rises with $\gamma_{\mathrm{X}}$ while $\gamma_{\mathrm{X}} <
\beta_Z$, then falls and turns negative above that threshold as the
ceiling binds (\S\ref{subsec:tail_algebra}). The upper tail is
governed by the amplification prior in one regime and by the ceiling
in the other, and by evidence about architectures in neither.

The consequence for the synthesis hierarchy is that the
\emph{direction} of the tail contrast belongs in the class of
consequences of the modelling apparatus rather than the class of
findings robust across the calibration set. It is reclassified
accordingly in \S\ref{sec:synthesis_hierarchy}. Its magnitude
remains informative conditional on the coupling prior, and the
goodness-of-fit asymmetry between the two architectures, where the
T-EA fit is adequate and the OTT-EA fit is rejected, is a genuine
diagnostic that no prior choice was tuned to produce.

Two conditioning conventions apply throughout and are stated here
once rather than repeated at every use.
\emph{Calibration-conditional} marks third-class quantities: values
computed at the primary calibration whose magnitudes, though not
their signs or orderings, move under defensible recalibration.
Fr\'echet--Hoeffding intervals are conditional on the Pillar~II
per-scenario probabilities: their robustness is to dependence
structure within that model, and the quantity-level relationship to
the Bayesian outputs is set out in \S\ref{sec:which_quantity}.
Where these labels appear without elaboration in what follows, this
paragraph is their definition.

\subsection{The Role of the Bayesian Layer}

The Bayesian layer contributes five findings to the synthesis. Its
prior calibrations use the same architecture-matched analogical
evidence as Pillar~I (T-EA priors anchored on Streams~A and~B; OTT-EA
priors on Streams~B and~C), so the architecture-conditional medians it
produces are not independent corroboration of the three-pillar ranges.
The five contributions divide into two groups by epistemic standing.
Three (the dependence-quantification, the cross-cutting coupling
mechanism, and the EVPPI decomposition: items~2, 3, and~5 below) are
genuinely Bayesian-specific in the sense that they require the
explicit correlated-campaign coupling model and cannot be derived from
the three pillars alone. The other two (distributional shape and
exceedance characterisation: items~1 and~4) could in principle be
produced by a sufficiently large Pillar~II Monte Carlo, but in this
framework are most naturally and efficiently delivered by the Bayesian
layer alongside its dependence treatment. The five contributions
are:
\begin{enumerate}
  \item \textit{Architecture-conditional distributional shape.} The prior
    predictive distributions for both classes are heavy-tailed under the
    stated model and priors (T-EA $\hat\xi = 0.271$; OTT-EA
    $\hat\xi = 0.518$, both Fr\'echet-domain). The OTT-EA tail is
    conditionally heavier-tailed, consistent with the cross-cutting
    coupling specification.
  \item \textit{Architecture-divergent dependence quantification.} The
    annual 95th-percentile uplift from adding campaign dependence is
    $\times 1.53$ for T-EA and falls in the range $\times 2.67$ to
    $\times 3.42$ for OTT-EA across the defensible
    $\sigma \in [0.25, 0.75]$ range
    (Table~\ref{tab:gamma_X_sigma_sensitivity}); the
    $\times 2.83$ point estimate corresponds to the
    headline $\sigma = 0.40$ specification, prior-integrated over
    $\gamma_{\mathrm{X}}/\gamma$ (Equation~\ref{eq:gamma_X_prior}).
    The corresponding conditional values at fixed
    $\gamma_{\mathrm{X}}/\gamma$ span $\times 1.52$ to $\times 5.12$
    across $\{1, 1.5, 2, 3\}$. All establish that independence
    assumptions are optimistic at the tail. The architectural
    divergence (OTT-EA uplift exceeding T-EA uplift) is structurally
    guaranteed by the coupling prior
    ($\gamma_{\mathrm{X}}/\gamma > 1$ almost surely) and survives the
    entire defensible parameter range. Both findings hold under wide variation in the
    campaign parameters within the priors.
  \item \textit{Cross-cutting coupling as the structural mechanism.} The
    Bayesian model identifies cross-cutting infrastructure as the
    architectural mechanism through which correlated campaigns erode
    the OTT-EA segregation gain. The sensitivity analysis over
    $\gamma_{\mathrm{X}} / \gamma$ shows that this mechanism is durable
    across the plausible coupling-ratio range.
  \item \textit{Architecture-conditional exceedance characterisation
    (prior-conditional).} The distributions assign $\Pr(q_1 > 1\%) =
    0.988$ (T-EA) and $0.907$ (OTT-EA) under primary prior calibrations,
    with the $1\%$ row reported here as an illustrative entry from the
    full threshold range in Table~\ref{tab:exceedance}. Both findings
    are sensitive to prior choice. The OTT-EA exceedance falls below 0.5
    under the conservative prior, the T-EA exceedance does not.
  \item \textit{Architecture-conditional parameter uncertainty
    decomposition.} The decomposition is flat in both classes. For T-EA
    the three modifiers contribute comparable first-order shares
    (roughly $18$--$25\%$ each), with $\delta_{\mathrm{target}}$ marginally
    the largest. For OTT-EA the same three modifiers contribute
    $12$--$15\%$ each and the cross-cutting amplification ratio
    $\gamma_{\mathrm{X}}/\gamma$ adds a smaller first-order index
    ($\sim 9\%$), its contribution limited because the structural
    constraint $\gamma_{\mathrm{X}}/\gamma \geq 1$ does the architectural
    work at the prior level rather than through residual value
    uncertainty. The segregation factor $P_3 \cdot P_4$ is held fixed in
    the canonical model and therefore contributes no output variance;
    characterising it empirically would nonetheless be informative, since
    the sensitivity grid of \S\ref{sec:pillar2_sensitivity} shows the
    OTT-EA central projection moving across $2.0$--$6.0\%$ over its
    plausible range. No single parameter dominates in either class, so
    empirical work across the modifiers yields broadly comparable
    variance reduction.
\end{enumerate}

\subsection{The Shared Structural Assumption and Its Limits}

Pillars~I, II, and III all use the complement-of-product aggregation
structure within each architectural class, with
architecture-conditional $\xi$-weighting for OTT-EA. This shared
structure is what produces the within-class agreement: three
methodologies using the same mathematical aggregation will produce
similar outputs when fed comparable inputs. The agreement is therefore
a property of the framework's internal logic under the independence
assumption, not externally validated corroboration of either the T-EA
$5.4\%$ aggregate or the OTT-EA $3.0\%$ aggregate as frequency
estimates.

The Bayesian model provides a critical perspective on this shared
assumption, applied in parallel to both classes. The explicit dependence
model shows that realistic correlated campaign structure inflates the
annual 95th percentile by $\times 1.53$ (T-EA) and $\times 2.83$
(OTT-EA) relative to independence, prior-integrated over
$\gamma_{\mathrm{X}}/\gamma$. The OTT-EA magnitude varies between
$\times 2.67$ and $\times 3.42$ across the defensible prior-width range
$\sigma \in [0.25, 0.75]$ (Table~\ref{tab:gamma_X_sigma_sensitivity}).
The architectural divergence (OTT-EA uplift exceeding T-EA uplift)
is structurally guaranteed by the coupling prior
($\gamma_{\mathrm{X}}/\gamma > 1$ almost surely;
\S\ref{subsec:dependence}). The shared independence assumption
in the three-pillar model is therefore \emph{optimistic} at the tail
for both classes, and \emph{more} optimistic for OTT-EA. Both
within-class agreement is better understood as a set of conditional
central estimates than as central tendencies robust to dependence
assumptions.

\subsection{The Evidence Base for a Material Risk Finding}

Stripping out what is prior-calibration-dependent and what is
methodologically shared, the genuinely independent quantitative evidence
that EA compromise risk is material in each architectural class comes from
the following sources:

\textit{T-EA evidence.} The Stream~A government-system per-system base
rate ranges from $0.46\%$/yr (Tier~1 direct analogues, $k = 3$) to
$0.92\%$/yr (full cohort, $k = 6$); the Tier~1+2 intermediate cohort
($k = 4$) gives $0.62\%$/yr. Across the $3{\times}$ to $15{\times}$
EA targeting-premium range, the joint inclusion-criterion $\times$
premium grid spans $1.4\%$ to $12.9\%$ annually. The grid centre (full
cohort, $10{\times}$) is $8.8\%$; the Tier~1 centre at $10{\times}$ is
$4.5\%$. The defensible T-EA empirical-projection interval is
therefore $1.4$--$12.9\%$, with the breadth reflecting
inclusion-criterion and targeting-premium uncertainty rather than
empirical-base-rate uncertainty.

\textit{OTT-EA evidence.} The Stream~C architectural-fit-restricted
strict-joint base rate is $0.95\%$/yr (Tier~0+1, $J=Y$, per
operator-year), with an all-tiers strict-joint upper bracket
of $1.33\%$/yr. At a $3$--$5\times$ EA-specific premium (lower than
the T-EA range because Stream~C's base rate already reflects
substantial existing platform targeting), the projection is
$2.8$--$4.7\%$ annually at the headline base ($3.9$--$6.4\%$ at the
upper bracket). At $1\times$ (no additional EA targeting
premium beyond existing platform attention), the projection is
$0.95\%$ itself. The defensible OTT-EA empirical-projection
interval is therefore $0.9$--$6.4\%$.

\textit{Centralised single-operator evidence (architectural bridge).}
The Stream~B CA-cohort base rate is $0.445\%$/yr standalone
($k = 11$, $E = 2{,}470$ CA-years), projecting to $1.3$--$4.4\%$ at
$3$--$10\times$ premiums. Pooling Streams~A and~B under the
common-rate model gives $0.545\%$/yr ($k = 17$, $E = 3{,}120$),
projecting to $1.6$--$5.3\%$ across the same premiums. Only the
standalone Stream~B figure is evidence independent of Stream~A; the
pooled rate is reported here as a consistency check between the two
architectural extremes rather than as a third strand.

\textit{Bayesian exceedance findings.} The architecture-conditional
exceedance probabilities reported in Table~\ref{tab:exceedance} are
conditioned on the prior. At the illustrative $1\%$ threshold, $\Pr(q_1
> 1\%) = 0.988$ (T-EA) and $0.907$ (OTT-EA) at the primary calibration;
the corresponding entries at higher and lower thresholds are documented
in the same table, and the choice of any specific threshold as the
benchmark for policy use is jurisdictional rather than technical
(\S\ref{subsec:exceedance}). The exceedance results are not sensitive
to reasonable perturbations of the prior median within each
architecture, but are sensitive to priors concentrated below the
threshold under consideration.

The independent strands of evidence place annual EA
compromise risk most plausibly in the range $1.4$--$8.8\%$ for T-EA
and $0.9$--$6.4\%$ for OTT-EA, with material probability mass at
low-single-digit annual thresholds in both cases and a heavy upper tail
driven by dependence structure (more pronounced for OTT-EA). The precise location within these ranges
depends on assumptions (the targeting premium, the independence
structure, the cross-cutting coupling) that cannot currently be
empirically resolved.

\subsection{The Stochastic Dominance Floor}

Beneath all the conditional quantitative analyses,
Proposition~\ref{prop:dominance} establishes a coherence result: in
either architectural class, EA-equipped architectures carry weakly
higher modelled risk than the no-EA counterfactual within the same
class, with equality on the prior mass of $0.108$ where the joint
modifier is exactly one. The proposition's
significance is not that it is model-free in the sense of requiring no
assumptions, but that it rests on assumptions (that EA increases attack
surface, attracts elevated targeting, and concentrates key-material
value) for which there is concrete operational evidence in both
architectural classes (Salt Typhoon, Stream~A Tier~1, and the
historical Athens affair as a Pillar~III channel analogue, for T-EA;
Storm-0558 and Midnight Blizzard for OTT-EA, treated as analogues per
Section~\ref{sec:pillar1_streamC}). The debate about EA risk is
therefore not about whether the mandate adds risk in either architectural
class, but about how much, in what distributional form, and whether the
architectural choice between classes itself materially affects the
expected harm.

Remark~\ref{rem:cross_arch} notes that the cross-architecture
comparison is not first-order stochastic dominance: T-EA dominates
OTT-EA at the median (T-EA is riskier on average), but OTT-EA can
exceed T-EA in the upper tail under correlated campaigns. This is a
substantive policy finding. Architectural choice between T-EA and
OTT-EA is not a uniform improvement but a trade-off between expected
risk and tail risk.

\subsection{What the Full Framework Establishes}\label{sec:synthesis_hierarchy}

The framework supports six findings distinguished by their
assumption-robustness. The Reader's Guide (Table~\ref{tab:readers_guide})
maps each principal claim to its tier. This section formalises the tier
content.

\begin{enumerate}
\item[\textbf{(S1)}] \emph{Internal ordering (a coherence property, not
a finding).} EA systems in either class carry weakly higher modelled
risk than the no-EA counterfactual within the same class. Version~1
placed this in the most assumption-robust tier. That was too generous,
and the tier is downgraded here. The result is a coupling argument:
the counterfactual is the same attack graph with the modifiers set to
unity, the modifiers are constrained to be $\geq 1$ by the censoring
in the prior, so pathwise dominance holds at every parameter
realisation and first-order stochastic dominance follows. The
dominance is also \emph{weak} rather than strict, because the censored
prior places mass $0.108$ on the joint modifier equalling exactly one,
where the EA system adds no hazard at all
(Supplementary \S\ref{sec:supp_modifier_priors}). What the proposition
establishes is that the model cannot make EA look safer than its own
counterfactual. The substantive question, whether the modifiers exceed
unity in the world, is argued from the incident record in
\S\ref{sec:dominance} and is not settled by the proposition. The
unmodelled defensive-uplift term $\delta_{\mathrm{def}}$
(\S\ref{sec:limitations_counterfactual}) is the condition under which
the ordering could fail, and it sits outside the formalism.

\item[\textbf{(S2)}] \emph{Heavy upper tails and dependence inflation.}
Annual compromise probability distributions have heavy upper tails in
both classes, and correlated-campaign modelling inflates
upper-percentile risk more for OTT-EA than for T-EA. Both parts of
that sentence require separating into what is assumed and what is
observed, which version~1 did not do.

The \emph{direction} is assumed. Exponential amplification of a
gamma-distributed campaign intensity is regularly varying with index
$\beta_Z/\gamma$, so positive shape parameters and the ordering
$\xi_{\mathrm{OTT}} > \xi_{\mathrm{T}}$ follow from the specification
before any data, and the ratio of fitted shapes recovers the coupling
ratio (\S\ref{subsec:tail_algebra}). This part belongs with the
consequences of the modelling apparatus, not with the robust findings.

What is \emph{observed} is the goodness-of-fit asymmetry. The T-EA
upper tail is statistically consistent with a Generalised Pareto
distribution near its primary threshold, while the OTT-EA upper tail
fits no single GPD at any tested threshold
(Supplementary \S\ref{sec:supp_gpd_gof}), consistent with a mixture
structure induced by coupling on cross-cutting edges. No prior choice
was tuned to produce that asymmetry, and it is the part of this tier
that carries information.

Magnitudes are calibration-conditional, and percentiles above roughly
the 95th should be read as prior artefacts: half the coupling prior
lies in a region where the amplification factor has no finite mean,
and capping it moves the OTT-EA 99th percentile from $57\%$ to $44\%$
while leaving the medians unchanged (\S\ref{subsec:tail_algebra}).

\item[\textbf{(S3)}] \emph{Exceedance behaviour.} Exceedance
probabilities at low policy thresholds are elevated under the primary
calibration and decline materially under conservative priors. The
architectural ordering (T-EA $>$ OTT-EA at low thresholds) is
preserved across the full defensible prior range. Which threshold is
policy-dispositive is a jurisdictional question.

\item[\textbf{(S4)}] \emph{Plausibility ranges (interval dominance).}
Annual T-EA compromise probability lies in the low single-digit range
under any dependence structure compatible with the Pillar~II model
(Fréchet--Hoeffding bracket); OTT-EA is broadly lower in central
tendency but its upper bound exceeds T-EA's under correlated
campaigns. The three empirical streams all project per-system rates
on the order of $1\%$ per system-year, supporting these ranges as
plausibility ranges.

\item[\textbf{(S5)}] \emph{Channel-aggregate projections
(independence-conditional).} Under independence among irreducible
risk channels both architectures yield channel-aggregate projections
in the low single-digit range, T-EA higher than OTT-EA. The
aggregates are interpretive points within the (S4) ranges rather
than methodologically independent estimates.

\item[\textbf{(S6)}] \emph{Cross-architecture trade-off.} The
comparison is not first-order stochastic dominance: T-EA is riskier
in central tendency, OTT-EA in the upper tail under correlated
campaigns. Architectural choice is a
central-tendency-versus-tail-risk trade-off rather than a uniform
improvement.
\end{enumerate}

\textbf{Robustness ordering.} The mapping of these tiers onto the
three classes of \S\ref{sec:structural_vs_calibration}: S1, the
direction of S2, and S6 are structural; S4 holds across the
calibration set within the Pillar~II model structure; S3, the
magnitudes within S2, and the specific aggregates of S5 are
calibration-conditional.

% =============================================================================
\section{Discussion of Limitations}\label{sec:limitations}

The limitations below are individually consequential. Collectively,
they preclude any reading of the framework as a forecasting or
estimation instrument. The outputs are decision-support inputs
rather than calibrated empirical projections, and are presented as
ranges or qualitative directional findings rather than point
estimates. The interval-dominance and structural-ordering findings
(\S\ref{sec:synthesis}) survive the limitations identified below;
specific magnitude projections do not.

\subsection{The Fundamental Epistemic Gap}\label{sec:fundamental_gap}

The deepest limitation cannot be resolved within the approach: no EA system
has been deployed at scale and subsequently compromised in a documented, publicly
attributable incident that could serve as a calibration data point. All empirical
inputs are analogical, and the analogical transfers embedded in every
pillar and in the Bayesian domain transfer function carry unverifiable assumptions
about structural similarity. This is precisely why a framework
designed for deep uncertainty (rather than a conventional
statistical calibration) is the appropriate methodological choice.

\paragraph{Forward-secrecy degradation is out of scope.} Of the five
structural risk classes set out in Section~\ref{sec:background},
the framework quantifies four (single point of failure, insider threat,
complexity-vulnerability correlation, scale risk) and does not quantify
forward-secrecy degradation. The reason is methodological: forward
secrecy affects the \emph{consequences} of a key compromise (how much
prior traffic becomes decryptable) rather than its \emph{probability},
and consequence-side modelling lies outside the risk-side scope
declared in \S\ref{sec:intro}. An EA-mandated architecture that breaks
or weakens forward secrecy increases the harm of any single compromise
event without affecting its likelihood. Because forward secrecy
affects the harm per compromise rather than its probability, its
omission from the probability-side framework implies that the risk
numbers reported here would need to be scaled upward in a complete
risk-benefit analysis for architectures that weaken or remove forward
secrecy. A complete analysis would integrate this consequence-side
magnification with the probability-side estimates this framework
produces. That integration is left for future work.

\subsection{Bayesian Prior Calibration and the Limits of the Median}

The Bayesian edge-hazard priors are calibrated so that the prior
predictive distribution produces annual compromise probabilities
consistent with the three-pillar 3--6\% range. The Bayesian median
of $4.0\%$ is therefore not independent confirmation of the
three-pillar within-class agreement: a prior calibrated to produce
results in a target range will produce them. The calibration keeps
the model consistent with the best available analogical evidence
but makes the median a function of that evidence rather than an
external check on it. The exceedance finding $\Pr(q_1 > 1\%) =
0.988$ is similarly prior-conditional: under the conservative prior
it drops to $0.605$; under the pessimistic prior it rises to
saturation (Supplementary \S\ref{sup:prior_sens}). The exceedance
finding should be read as the implication of accepting the primary
analogical evidence (government HSM and CA cohort data) as the
appropriate calibration source, not as robust across the full
defensible prior range.

\subsection{The EA Targeting Premium}\label{subsec:targeting_premium_limitation}

The EA targeting premium translates the empirical base rate of
state-actor compromise of cryptographic-relevant systems into a
projected rate for EA-equipped systems specifically. No EA-specific
empirical dataset exists from which this multiplier could be directly
estimated, and the $10{\times}$ value used in the main analysis is a
modelling judgement informed by adjacent precedents rather than a
directly measured quantity. The $3{\times}$--$15{\times}$ range used
throughout is likewise a \emph{structured author judgement},
constructed from the named anchors set out below---the documented
targeting precedents, the sector-resolved actor-mix data, and the
qualitative disparity argument---rather than the output of a formal
multi-expert elicitation protocol, and it is presented as such. The
update rule in Section~\ref{sec:falsifiable} is the mechanism by
which operating experience revises it. Applied to the pooled Stream~A + B rate
($0.545\%$/yr), the projections at $3{\times}$, $5{\times}$,
$10{\times}$, $15{\times}$ are $1.6\%$, $2.7\%$, $5.3\%$, $7.8\%$;
applied to the Stream~A-anchored T-EA rate ($0.92\%$/yr), the
projections are $2.7\%$, $4.5\%$, $8.8\%$, $12.9\%$. The remainder of
this subsection sets out the external evidence that constrains the
plausible premium range and identifies the empirical work that would
tighten it.

\paragraph{Lower-bound evidence (against premium $< 3{\times}$).}
Three lines of evidence argue against a multiplier substantially
below $3{\times}$. First, the 2024 Salt Typhoon campaign demonstrably
targeted CALEA-mandated lawful-intercept infrastructure at multiple
major US carriers~\citep{volz2024salttyphoon,eff_salttyphoon}: the lawful-intercept management plane
was a specific operational objective, and the attack succeeded
despite the carriers' mature regulatory posture. The base-rate
analogue for state-actor compromise of a comparable
non-EA-equipped carrier infrastructure component over a similar
window is substantially below the observed Salt Typhoon
prevalence. A single precedent does not pin down the ratio, but it
rules out values close to $1{\times}$. Second, the Verizon DBIR
documents, across successive editions, a marked disparity in the
share of breaches attributable to state-aligned espionage actors
across target sectors: high-value and government target classes
exhibit actor-mix ratios several times the cross-sector base rate
\citep{dbir2024,dbir2025}. The qualitative
disparity is consistent with a multiplier of at least
$3$--$5{\times}$ for cryptographic-infrastructure targets, even if
the cross-domain rate ratio is not directly transferable. Third,
ETSI's standardised threat, vulnerability and risk-analysis
methodology \citep{etsi102165} frames threat likelihood in terms
of attacker capability and motivation, recognising that
well-resourced, deliberately targeted attacks can realise
compromises beyond the reach of the broader
opportunistic-criminal profile dominant in the carrier-network
base rate. This asymmetry is the qualitative substance of the
premium and is also inconsistent with multipliers near $1{\times}$.

\paragraph{Upper-bound evidence (against premium $\gtrsim 15{\times}$).}
Two considerations argue against very large multipliers. First,
Salt Typhoon compromised some but not all CALEA-mandated US carriers
within the observation window. A multiplier of $\sim$$20{\times}$
applied to the Stream B+A pooled rate would predict more frequent
or more complete compromise of the same infrastructure class
than the public record supports. Second, the Bayesian Layer~IV upper
percentile region for $q_1$ (the OTT-EA 90\%-CI upper bound is
$\sim 17\%$ and the T-EA upper bound is $\sim 16\%$) is already very
high. A Pillar~I central-tendency projection at $15$--$20{\times}$
would push the implied central rate toward $8$--$11\%$, exceeding
plausibility-bounded historical analogues at the central tendency
even before tail effects are considered. The 95\% upper bound of
the Stream~B Garwood CI at $10{\times}$ already reaches $7.66\%$,
which sits at the boundary of historically defensible
critical-infrastructure annual compromise probabilities.

\paragraph{Defensible range and central estimate.} The above
considerations support a defensible range of approximately
$3$--$15{\times}$ for the EA targeting premium, with $10{\times}$
representing the upper-mid region of that range rather than its
modal point. Under this calibration the centralised single-operator
projected rate (from the pooled rate) spans $1.6$--$7.8\%$, with the
median of the range at roughly $4{\times}$--$6{\times}$ giving
$2.2$--$3.2\%$. The parallel Stream~A-anchored T-EA range is
$2.7$--$12.9\%$ across $3$--$15\times$ (with $10\times$ giving the
$8.8\%$ headline). Both ranges are intentionally wider than any
single-point headline and reflect the genuine uncertainty in this
parameter.

\paragraph{Implications for the headline.} The architecture-class
\emph{ordering} (T-EA versus OTT-EA, structural ordering S1, and
the $\gamma_{\mathrm{X}}/\gamma \geq 1$-driven OTT-EA tail-heaviness)
does not depend on the premium calibration. The premium scales
the central rate but not the architectural divergence. The
Fréchet--Hoeffding interval $[2.2\%, 7.5\%]$ from Pillar~II
(\S\ref{sec:pillar2}) does not depend on the premium calibration
at all: it is derived from the seven-scenario per-channel
probabilities without any single-operator pooling step. The
$2.7$--$12.9\%$ Stream~A-anchored T-EA range and the $1.6$--$7.8\%$
pooled centralised single-operator range are the relevant calibrated
intervals for the central tendency. The qualitative findings rest on
structural rather than calibration-dependent arguments.

\paragraph{Future work to tighten the bound.}
A matched-cohort study comparing compromise rates among
CALEA-mandated US carriers over 2010--2024 against equivalent
non-CALEA-mandated carrier infrastructure over the same window
could quantitatively bound the multiplier from observed
compromise-rate differentials. Similar approaches could exploit
the EU-CSAM regulation trajectory, jurisdictional variation in
lawful-intercept obligations across OECD members, and the
asymmetric compromise records of CA infrastructure under different
audit regimes. None of these analyses is performed in the present
paper. Identifying them as the empirical work that would tighten
the most consequential single calibration parameter is itself a
contribution.

\subsection{Regulatory Defensive Uplift: The Unmodelled Counterfactual}\label{sec:limitations_counterfactual}

The modifier structure (Equation~\ref{eq:modified_hazard}) captures
three multiplicative effects of an EA mandate: increased attack
surface ($\delta_{\rm EA}$), elevated adversarial targeting
($\delta_{\rm target}$), and key concentration ($\delta_{\rm
concen}$). It does not include a compensating modifier
$\delta_{\rm def} \geq 1$ for the defensive uplift a regulatory
mandate may itself impose (mandatory HSM certification, audit-trail
requirements, certified incident-response capability, personnel
vetting above the commercial baseline). This is the strongest pro-EA
counter the framework does not directly engage: a defender would
argue that unregulated Stream~A and Stream~C operators understate the
rate a mandated operator would face, so $\delta_{\rm def}$ might
offset $\delta_{\rm EA} \cdot \delta_{\rm target} \cdot \delta_{\rm
concen}$ in part or whole.

Three observations bound plausible $\delta_{\rm def}$. \emph{(i)}~Salt
Typhoon (2024) occurred at CALEA-regulated US carriers operating under
the most mature lawful-intercept regulatory regime in the Western
world. The regulation did not prevent state-actor compromise of the
orchestration layer. \emph{(ii)}~The Stream~B CA cohort operates
under formal CCADB-anchored governance with mandatory audit,
transparency-log, and revocation requirements, which is a precise
analogue of the regulatory uplift a defender would invoke. The resulting
per-CA rate ($0.45\%$/yr) is the best empirical upper bound on the
defensive uplift attainable through CA-equivalent regulation, since it
already incorporates the regulatory benefits. \emph{(iii)}~The
regulatory cadence is empirically lagging rather than leading:
standards harden in response to compromise, not in anticipation. The
empirical record offers no clear support for a defensive uplift large
enough to nullify the EA penalty. The most-regulated environment in
the analogue set is the one where the largest documented compromise
occurred. Future work should treat $\delta_{\rm def}$ as a fourth
modifier with prior informed by the analogue cohort's regulatory
characteristics. The qualitative conclusions of the present analysis
are expected to survive any $\delta_{\rm def}$ consistent with the
Stream~B versus Stream~C ratio.

\subsection{The Cross-Cutting Coupling Ratio $\gamma_{\mathrm{X}}/\gamma$}\label{subsec:xcut_coupling_limitation}

The OTT-EA architectural findings rest on the latent campaign variable
$Z(t)$ being coupled more strongly to cross-cutting edges than to
subgraph-internal edges (Section~\ref{subsec:dependence},
Equation~\ref{eq:campaign_amplification_OTT}). The prior median
$m = 2$ is calibrated against the Stream~C cross-cutting fraction
$\hat p_X = 7/14$ (Section~\ref{subsec:dependence}). The remaining
specification choice is the prior width $\sigma$. The OTT-EA
tail-inflation magnitude is sensitive to $\sigma$:
Table~\ref{tab:gamma_X_sigma_sensitivity} reports the prior-integrated
95th-percentile uplift across $\sigma \in \{0.25, 0.40, 0.50, 0.60,
0.75\}$. At $\sigma = 0.25$ the uplift is $\times 2.67$ (close to
the fixed-$\gamma_{\mathrm{X}}/\gamma=2$ value); at $\sigma = 0.75$ it
reaches $\times 3.42$. The \emph{direction} of the architectural
divergence (OTT-EA tail inflation exceeding T-EA's $\times 1.53$) is
unconditional under any prior with support in $[1, \infty)$, since the
constraint $\gamma_{\mathrm{X}}/\gamma \geq 1$ alone is sufficient to
guarantee equal-or-greater OTT-EA cross-cutting amplification at every
replicate. The \emph{magnitude} of the divergence depends on $\sigma$.

\paragraph{Sensitivity to the prior median.} The same question
arises for the prior median $m$, and
Table~\ref{tab:gamma_X_median_sensitivity} reports the complementary
sweep over $m \in \{1.2, 1.5, 2, 3, 4\}$ at fixed $\sigma = 0.40$
($R = 8{,}000$, seed 2024). Two features stand out. The annual
median is only weakly sensitive ($2.31\%$ at $m = 1.2$ rising to
$3.53\%$ at $m = 4$): the coupling ratio is a tail parameter, not a
central-tendency one. The tail uplift, by contrast, is strongly
increasing, from $\times 1.73$ at $m = 1.2$ to $\times 12.4$ at
$m = 4$. The architectural divergence direction is preserved at
every value swept: even at $m = 1.2$, barely above the structural
floor, the OTT-EA uplift exceeds T-EA's $\times 1.53$. The
conclusions that depend on $m$ are therefore the quoted inflation
magnitudes, not the ordering. The empirical anchoring of $m$ itself
is set out in \S\ref{subsec:dependence} and Supplementary
\S\ref{sec:supp_xcut_prior_calibration}.

\begin{table}[!htbp]
\centering
\small
\caption{Sensitivity of OTT-EA results to the cross-cutting coupling
prior median $m$ at fixed $\sigma = 0.40$ ($R = 8{,}000$ replicates,
seed 2024). Uplift is the prior-integrated 95th-percentile ratio of
the full-dependence to the independence configuration; the T-EA
reference uplift is $\times 1.53$. The divergence direction is
preserved at every $m$; its magnitude is calibration-dependent. The
$m = 3$ and $m = 4$ rows lie partly in the saturation regime of
\S\ref{subsec:tail_risk} (95th percentiles approaching the $q = 1$
ceiling), so their uplift factors mix tail-decay and ceiling
effects and should not be read as pure tail-heaviness.}
\label{tab:gamma_X_median_sensitivity}
\begin{tabular}{@{}ccccc@{}}
\toprule
\textbf{Prior median $m$} & \textbf{Median $q_1$} &
\textbf{95th pct.\ $q_1$} & \textbf{Median $Q_{10}$} &
\textbf{95th-pct.\ uplift} \\
\midrule
1.2 & 2.31\% & 10.6\% & 23.6\% & $\times 1.73$ \\
1.5 & 2.42\% & 12.5\% & 26.1\% & $\times 2.04$ \\
2.0 (primary) & 2.63\% & 17.4\% & 31.9\% & $\times 2.83$ \\
3.0 & 3.06\% & 38.2\% & 48.1\% & $\times 6.22$ \\
4.0 & 3.53\% & 76.5\% & 70.1\% & $\times 12.4$ \\
\bottomrule
\end{tabular}
\end{table}

\paragraph{Empirical update path for $\sigma$.} The Stream~C cohort
structure already supports a likelihood-based posterior on
$\gamma_{\mathrm{X}}/\gamma$. For each Stream~C incident $i$ with
$X_i = Y$ (cross-cutting outcome) versus $X_i \in \{P, N\}$, the
Bernoulli likelihood
\[
  \Pr(X_i = Y \mid \gamma_{\mathrm{X}}/\gamma, Z_i, \theta)
  \;=\; f\!\left(\gamma_{\mathrm{X}}/\gamma, Z_i, \theta\right)
\]
follows from the attack-graph specification: in the saturated regime,
$f \to 1 - e^{-c\gamma_{\mathrm{X}} Z}$ for an edge-weight constant
$c$. Combined with Equation~\ref{eq:gamma_X_prior} as a regulariser,
the fourteen Stream~C incidents yield a posterior on
$\gamma_{\mathrm{X}}/\gamma$ with an empirically anchored spread. The
small cohort size means the posterior remains wide, but data-bounded
rather than expert-judgement-bounded. This is the highest-priority
extension of the framework: it would convert $\sigma$ from a
calibration choice into a measured quantity, removing the last
unmeasured parameter from the OTT-EA analysis. The methodology is
straightforward. The bottleneck is the case-by-case latent-$Z_i$
assignment, which requires per-incident correlated-campaign indicators
not in the existing Stream~C documentation (Supplementary \S\ref{sec:supp_streamC}).

\subsection{Pillar II/III Calibration Validity}\label{sec:pillar23_validity}

\paragraph{Attack-graph topology sensitivity.} The Bayesian model uses a
stylised six-node, eight-edge attack graph. Comparing it against a
seven-node extension (adding an HSM firmware-traversal node) shows
that median system hazard falls by $9$--$11\%$ under the more granular
topology, with exceedance probability at the $1\%$ threshold
nearly unchanged ($0.989$ vs $0.981$). The direction is
favourable: the six-node graph \emph{over}states risk relative to more
detailed architectures. The qualitative dependence-inflation findings
are robust to topology variation within the range tested. Full
six- versus seven-node sensitivity outputs for both architectures are
reported in Supplementary
\S\ref{sec:supp_graph_sens}.

\paragraph{Robustness to attempt-rate assignment.} A 1{,}000-permutation
robustness check, randomly re-pairing attempt rates to scenarios with
all other parameters fixed, yields a central access-pathway projection
range of $[9.9\%, 21.2\%]$ (95\% central interval), with the paper's
primary access-pathway value ($\La = 16.8\%$) falling at approximately
the $61$st percentile of this distribution, modestly above the median
and well within the central range. The headline Pillar~II figure is not an artefact of an
extreme or cherry-picked scenario--attempt-rate assignment.

\paragraph{Non-identifiability and severity.} The Bayesian model is not
fitted to an EA-specific likelihood. All inference is on prior predictive
distributions informed by domain transfer. The model is therefore
non-identified in the classical sense: different parameter combinations
can produce identical output distributions. The parameter uncertainty
decomposition (\S\ref{subsec:evppi}) confirms that even resolving every
uncertain parameter would leave roughly $39\%$ of
$\mathrm{Var}(Q_{10})$ for T-EA and $51\%$ for OTT-EA as
irreducible structural uncertainty. A further limitation shared across
all four layers is that compromise is modelled as a binary event:
compromise \emph{severity} conditional on occurrence is not quantified.
A session-level decrypt and a decade-spanning multi-million-user
exposure are policy-distinct outcomes. A complete risk-benefit
determination requires a severity distribution that the present analysis
does not provide.

\subsection{Statistical Methodology and Inferential Status}\label{subsec:statistical_methodology}

The framework's quantitative outputs depend on statistical-modelling
choices whose limitations affect how the numerical results should be
read. Four specific caveats apply.

\paragraph{Layer~IV is prior predictive, not posterior inference.}
As disclosed in \S\ref{sec:bayesian}, the reported $95\%$ intervals
are forward-sampled prior predictive intervals, not posterior
credible intervals: no EA-specific likelihood exists, so the
intervals inherit their content from the analogue-calibrated
priors. The framework supports posterior updating when EA-specific
incident data eventually arrives.

\paragraph{Variance decomposition under correlated parameters.}
The EVPPI-style first-order Sobol indices reported in
\S\ref{subsec:evppi} (Figure~\ref{fig:evpi}) assume parameter
independence. The decomposition further presumes the inputs have standard,
mutually independent marginals; neither holds exactly here. The
coupling ratio $\gamma_{\mathrm{X}}/\gamma$ has a support-constrained
marginal ($\gamma_{\mathrm{X}}/\gamma \geq 1$ by construction), and the
modifier tiers $\delta_{\mathrm{EA}}$ and $\delta_{\mathrm{target}}$
co-vary across scenarios within
architectural class. The reported
indices should be interpreted as approximate variance contributions
under the marginal-distribution assumption; total Sobol indices
(which absorb interaction terms) and moment-independent importance
measures \citep{borgonovo2007} would refine the ordering but are
unlikely to change the qualitative ranking, in which the three
log-normal modifiers contribute comparable first-order shares and,
for OTT-EA, the cross-cutting coupling ratio
$\gamma_{\mathrm{X}}/\gamma$ contributes a smaller share.

\paragraph{Poisson rate models in Pillar~I.} The Pillar~I rate
estimates (Streams~A, B, and~C) use Poisson likelihoods with Garwood
confidence intervals. Compromise events may depart from Poisson
in two directions: clustering (Salt~Typhoon involved multiple
carriers within a coordinated campaign. The CA cohort includes
coordinated multi-CA compromises) and temporal heterogeneity (rates
may evolve as adversary capability and target attractiveness
change). The implied rate intervals are therefore likely understated
in width. The pooling of Streams~A and B assumes rate homogeneity
across the two cohorts. Neither departure changes the architectural
ordering or the order-of-magnitude scale of the Pillar~I anchors,
but a careful statistician would refit Pillar~I with overdispersed
Poisson or negative-binomial likelihoods. The architecture-conditional
sensitivity analysis (Table~\ref{tab:streamC_rates}) provides the
relevant range under realistic departures.

\paragraph{Informal model averaging.} The four-layer framework is
\emph{informal} triangulation across analytical traditions rather
than formal Bayesian model averaging. Each layer's contribution is
weighted by analytical role (consistency check, plausibility
bounding, dependence-aware tail characterisation) rather than by
explicit prior weight. The advantage is honesty about the
limitations of each layer. The disadvantage is the absence of a
single posterior weight that a formal Bayesian model average would
provide. A formal Bayesian model average over the four layers would
require specifying prior weights that are themselves
calibration-conditional, a problem the multi-layer triangulation
avoids by reporting each layer's output explicitly. The synthesis
hierarchy (\S\ref{sec:synthesis}) is the closest available substitute:
findings on which the layers agree are stronger than findings
supported by one layer alone.

\subsection{Temporal Robustness, Inclusion Sensitivity, and Operator Heterogeneity}

Three further robustness analyses are developed in full in
Supplementary \S\ref{sec:supp_robustness}. The headline results are
as follows. \emph{Stationarity}: a stylised $2\%$ annual rate-growth
model raises the 10-year cumulative from $53\%$ to $56\%$ from the
Pillar~II base, directionally adverse but modest, with full
trajectories in Supplementary \S\ref{sec:supp_nonstationary}.
\emph{Inclusion sensitivity}: the Tier-1-only Stream~A cohort at the
$10{\times}$ premium projects to $4.5\%$, inside the low single-digit
range, and removing Salt Typhoon entirely ($k = 5$) leaves the
estimate materially unchanged (Supplementary
\S\ref{sec:salt_typhoon_role}). \emph{Joint denominator-and-inclusion corner}: the one-at-a-time
sensitivities above can be stacked. Taking the risk-minimising
corner simultaneously ($N_A = 75$ and strict Stream~B inclusion)
gives a pooled rate of $11/3{,}445 = 0.32\%$ per system-year
(Garwood 95\% CI $[0.16\%, 0.57\%]$), projecting to $3.1\%$
annually at the $10{\times}$ premium; the adverse corner
($N_A = 35$, full inclusion) gives $17/2{,}925 = 0.58\%$, projecting to
$5.7\%$. (Both use the transform $1 - \exp(-\text{premium} \cdot
\lambda)$ applied throughout.) The qualitative finding, an order-$1\%$ pre-premium
anchor and a low-single-digit projection, survives both corners
jointly, falling below $1\%$ only when the premium is
simultaneously set at its $3{\times}$ floor.
\emph{Operator heterogeneity}: the
Stream~C cohort anchors the \emph{unregulated} platform base rate;
Stream~B's CCADB-governed Certificate Authorities are the closest
regulated-operator analogue, and the roughly factor-two ratio between
the two rates bounds the defensive uplift plausibly attributable to
EA-style regulation, shifting the OTT-EA central projection from
$\sim$$3.9\%$ to $\sim$$2.0\%$ at the most generous reading. A
mandate spanning less well-resourced operators faces the converse
adjustment, since systemic risk tracks the capability distribution
across the mandated population rather than its best-resourced member.

\subsection{Falsifiable Implications}\label{sec:falsifiable}

Although no EA-specific incident record exists today, the framework
is not insulated from evidence. It makes observable commitments,
three of which are stated here so that future data can confirm or
discredit the calibration. One framing point first: these
commitments discipline the calibration layer, the analogue
anchoring and the coupling mechanism, rather than the EA projection
itself. The transfer step, the targeting premium, remains the least
directly testable link, and the second commitment below is its only
current handle.

\emph{First, a forward prediction from Stream~B.} At the fitted CA
compromise rate, the Gamma--Poisson predictive distribution for the
2025--2030 window (approximately $780$ CA-operator-years at current
cohort size) has mean $3.5$ disclosed incidents with $95\%$
predictive interval $[0, 8]$. A zero-incident six-year run has
predictive probability $\approx 4.9\%$ and would warrant downward
revision of the centralised-operator anchor; a count above eight
would warrant the converse. This commitment tests the analogue
anchor directly; the transfer step is the second commitment's
business.

\emph{Second, an update rule for deployed systems.} Any deployed
T-EA-class system carries a stated first-decade compromise
probability under each premium calibration. Clean operating history
is informative against the premium: ten system-years without a
Tier~1 incident carries a likelihood ratio of approximately $1.5$
favouring the $3{\times}$ over the $10{\times}$ targeting premium at
the pooled base rate, and the Layer~IV machinery supports the
corresponding posterior reweighting directly
(\S\ref{subsec:statistical_methodology}). The premium calibration is
not fixed doctrine. It is the parameter the framework most expects
operational evidence to move.

\emph{Third, a distributional prediction from the coupling
mechanism.} The cross-cutting coupling prior implies that
correlated-campaign periods should disproportionately produce
cross-cutting ($X = Y$) incidents in future platform-incident
records. A future extension of the Stream~C cohort in which the
cross-cutting fraction is statistically indistinguishable between
campaign-intensive and quiet periods would contradict the
$\gamma_{\mathrm{X}}$ mechanism on which the OTT-EA tail findings
rest, independently of any rate calibration.

% =============================================================================
\section{Policy Implications}\label{sec:policy}

This section translates the framework's structural and
interval-dominance findings into a basis for policy
deliberation. Two qualifications apply throughout. \emph{First},
the framework gives only the risk side of a two-sided
risk--benefit analysis. Deciding whether EA deployment is
net-beneficial also requires benefit-side estimates, such as
law-enforcement effectiveness, counterfactual criminal activity
in the absence of EA, and the share of cases in which EA access
is decisive rather than incremental. Each of these is a
substantial research problem in its own right and lies outside
this paper's scope. The risk findings developed here are inputs
to a complete risk--benefit determination, not substitutes for
one. \emph{Second}, the framework is designed for comparative
reasoning under deep uncertainty, not for predictive
forecasting. The implications below are drawn from
architecture-level and interval-dominance findings, not from
calibrated point projections.

\subsection{The Policy Question Reframed}

The conventional policy framing asks ``Is the security risk of
exceptional access systems acceptably low?''. This framing
assumes a definitive risk estimate that cannot be produced from
currently available data. The framework does not offer one. A
more productive reframing is this: \emph{given the
assumption-robust architectural findings, and given the
calibration-conditional plausibility ranges for the central
tendency of compromise risk, under what conditions would an EA
mandate be net-beneficial?} The framework's contribution to this
reframed question has two parts. First, it constrains the
risk-side inputs to a structured plausibility range. Second,
it separates the comparative findings that are robust to
calibration choice from the specific magnitudes that are not.

\subsection{Interpreting the Probabilities}\label{sec:interpreting}

A statement such as ``annual compromise probability of $4\%$''
licenses less than it appears to, so this section is precise about
what it does and does not support. It does not support a
harm forecast: the framework quantifies the probability of the
compromise event, not its consequence, and the same probability
attaches to deployments of very different population exposure. It
does not support cross-domain comparison against actuarial risks
whose probabilities are measured rather than projected. What it
does support is comparative and threshold reasoning: whether one
architecture's profile dominates another's, whether cumulative
exposure over a stated deployment horizon exceeds a stated
acceptability threshold, and how far a conclusion survives
calibration change. The cumulative-horizon and
exceedance-probability framings
(\S\ref{subsec:cumulative}--\ref{subsec:exceedance}) are the forms
in which these probabilities are usable for policy, and they are
the forms this section relies on. Consequence modelling, which
would convert them into expected-harm terms, is the natural
companion analysis and is developed separately.

\subsection{Risk Profile Summary}

\begin{table}[!htbp]
\centering
\small
\caption{Architecture-conditional summary risk profile. The three
column families are different quantities and are not mutually
bounding; \S\ref{sec:which_quantity} gives the usage rule and
\S\ref{sec:structural_vs_calibration} the conditioning conventions.
Cumulative rows assume stationarity (assessed in Supplementary
\S\ref{sec:limitations_stationarity}); 25-year Bayesian upper
percentiles are stationarity extrapolations, not direct
simulations. All values are probability-side quantities and exclude
forward-secrecy degradation, for which they are lower bounds
(\S\ref{sec:fundamental_gap}).}
\label{tab:risk_summary}
\begin{tabularx}{\textwidth}{@{}X l c c c@{}}
\toprule
\textbf{Risk metric} & \textbf{Class} &
\textbf{Channel-min.$^*$} & \textbf{FH int.$^\dagger$} &
\textbf{Bayes.~95th-pc.$^\ddagger$} \\
\midrule
\multirow{2}{*}{Annual operational compromise}
  & T-EA   & $\sim$1.5\%  & $[2.2\%,\ 7.5\%]$   & $\sim$17\% \\
  & OTT-EA & $\sim$1\%    & $[1.1\%,\ 4.0\%]$   & $\sim$17\% \\
\multirow{2}{*}{10-year cumulative}
  & T-EA   & $\sim$14\%   & $[20\%,\ 54\%]$    & $\sim$84\% \\
  & OTT-EA & $\sim$10\%   & $[11\%,\ 34\%]$    & $\sim$94\% \\
\multirow{2}{*}{25-year cumulative$^\P$}
  & T-EA   & $\sim$30\%   & $[43\%,\ 86\%]$    & $\sim$99\% \\
  & OTT-EA & $\sim$22\%   & $[24\%,\ 64\%]$    & $\sim$99\% \\
\multirow{2}{*}{$\Pr(\text{annual risk} > 1\%)^\S$}
  & T-EA   & ---     & ---          & 0.6 to 0.99 \\
  & OTT-EA & ---     & ---          & 0.5 to 0.91 \\
\bottomrule
\end{tabularx}\\[2pt]
\begin{minipage}{\textwidth}\footnotesize
$*$~Channel-minimum heuristic: maximum single-channel rate under the four-channel
decomposition (Section~\ref{sec:pillar3}, conditional on channel
exhaustiveness, which we do not prove). It bounds risk under a different
decomposition from FH and the two are not directly comparable. The
\emph{model-consistent} lower bound on annual risk is the FH lower
bound, not the channel-minimum heuristic.\\
$\dagger$~Pillar~II Fr\'echet--Hoeffding interval, model-consistent
under any dependence structure within the seven-scenario model.\\
$\ddagger$~Bayesian 95th-percentile of $q_1$ under the primary prior
calibration. Bayesian priors are calibrated to architecture-matched
three-pillar empirical ranges, so the upper-percentile values capture
the dependence-aware tail under the primary calibration rather than an
independent estimate of it. Annual values from
Table~\ref{tab:dependence_comparison}; 10-year cumulative values from
the same simulation.\\
$\P$~The 25-year Bayesian column is $1-(1-q_1^{95})^{25}$ under
stationarity rather than a direct simulation of $Q_{25}^{95}$; under
realistic positive year-to-year dependence the true $Q_{25}^{95}$
would be lower.\\
$\S$~Range across primary and conservative prior calibrations
(Supplementary \S\ref{sup:prior_sens}).
\end{minipage}
\end{table}

The architectural pattern visible across the table is consistent: the OTT-EA
FH interval lies wholly below the T-EA FH interval, but the Bayesian
upper-percentile figures are similar or higher for OTT-EA, reflecting the
architectural divergence in tail behaviour established in
Sections~\ref{subsec:dependence_results} and~\ref{sec:bayesian}.

\subsection{EA Adds a New Risk Surface}

The analysis establishes that the EA mandate creates key-material exfiltration
as a new outcome class: absent an EA key-management layer, there is no
EA-specific key material to exfiltrate. However, the policy-relevant comparison
is not between EA risk and zero baseline risk, but between EA risk and the risk
that an operator already faces without EA. An operator without EA infrastructure
already holds sensitive non-EA key material (TLS private keys, code-signing
keys, device authentication keys) and faces non-zero key-compromise risk. The
Pillar~II model estimates the non-EA baseline access-pathway compromise rate at
12.3\% annually, with the EA system adding 4.5 percentage points to access-pathway
risk and 7.3\% of EA-specific key-material exfiltration risk above the zero
baseline.

The appropriate policy question is therefore: does adding EA infrastructure
materially and irreversibly worsen an operator's already non-trivial risk
profile? The structural argument of Section~\ref{sec:bayesian} establishes that
it does, for the reasons stated there (increased attack surface, elevated
targeting, and key concentration). The quantitative question is how much worse,
and under what architectural conditions the incremental risk can be minimised.

\paragraph{One decomposition, three expressions.} The baseline
comparison appears in three forms in this paper, and they express
a single decomposition rather than three different claims. The
Pillar~I analogue rates measure the \emph{baseline}: the compromise
rate of high-value cryptographic infrastructure as a class,
predominantly without EA-specific components. The EA
\emph{increment} above that baseline is carried entirely by the
targeting premium and the Layer~IV modifiers, which is why both are
judgement-bounded and swept rather than asserted. And the
$\delta = 1$ counterfactual of Proposition~\ref{prop:dominance} is
the model-internal expression of the same split: switching the
modifiers off recovers the baseline parameterisation. The increment
is therefore not double counted, with one bounded exception: the
Stream~A Tier~1 stratum includes EA-class systems, most directly
Salt Typhoon, so a fraction of the increment is present in the base
rate. The Salt-Typhoon-removal sensitivity (Supplementary
\S\ref{sec:salt_typhoon_role}) bounds that leakage and shows the
headline materially unchanged.

\subsection{The Irreversibility Asymmetry}

This subsection states a consequence-side observation. It sits
outside the quantified probability framework and is made
qualitatively, because any complete risk--benefit determination
needs it. A structural feature of EA risk not captured by point
estimates is the fundamental asymmetry between benefits and
consequences. Benefits are
\emph{temporal and reversible}: EA interception benefits accrue during the
operational period. If the mandate is withdrawn, future benefits cease. Breach
consequences are \emph{permanent and irreversible}: once key material is
exfiltrated, all historical traffic protected under the compromised keys becomes
retrospectively decryptable. This asymmetry means that standard expected-value
maximisation, which treats gains and losses symmetrically, may systematically
underweight the breach scenario. A complete risk-benefit analysis should account
for this asymmetry explicitly. The framework provides the risk-side inputs but
not the benefit-side estimates required to complete such an analysis.

\subsection{Implications from the Bayesian Model}

The Bayesian layer's tail and dependence findings carry several implications.
First, \emph{uncertainty is irreducible in the current data environment}: any
defensible quantitative estimate will carry credible intervals spanning at
least one order of magnitude (the primary calibration's $90\%$ interval
spans more than thirty-fold). Narrow point estimates presented without
accompanying uncertainty characterisation should be viewed with scepticism.
Second, \emph{under risk-averse decision frameworks tail risk
becomes a primary design criterion}: the heavy-tailed character of
the model output, robust at policy-relevant calibrations, means EA
system designs evaluated under maximin, CVaR-weighted, or other
tail-sensitive criteria should be ranked primarily on their
worst-case behaviour, with design features that reduce tail risk
(rapid revocation, compartmentalised key custody, strong detection
and response) preferred over features that reduce expected risk
alone. Under risk-neutral expectation-based criteria, the
central-tendency comparison dominates and the architectural
ordering reverses (T-EA exceeds OTT-EA at the median). The choice
of decision framework is itself a jurisdiction- and
context-specific determination outside this paper's scope. Third,
\emph{deployment horizon is a critical variable}: cumulative risk grows
substantially with deployment horizon, so policy frameworks authorising
permanent or indefinitely-renewable EA infrastructure without periodic
reassessment may systematically underestimate long-run risk. Horizon
does not, however, discriminate between the two architectures. On the
corrected basis of Table~\ref{tab:exceedance} the maximum defensible
horizons differ by one to two years across the whole threshold range,
so a sunset provision is warranted for either design and is not a
reason to prefer one. Fourth, the
exceedance probabilities reported in \S\ref{subsec:exceedance} are
elevated across the low-threshold range under the primary prior
calibration. For illustration, $\Pr(q_1 > 1\%) = 0.988$ for T-EA and
$0.907$ for OTT-EA. These place a heavy burden of evidence on
benefit-side estimates if risk is to be deemed acceptable at whichever
threshold a policy reader selects. The choice of acceptability
threshold is itself a jurisdiction-specific determination outside this
paper's scope.

\subsection{Architectural Alternatives: Threshold Schemes}\label{sec:threshold}

Both quantified classes share the property that key material sits
under unitary operator authority. A distributed threshold scheme, in
which key material is split across $N$ independent trustees with a
$t$-of-$N$ reconstruction requirement, removes that property, and an
illustrative quantitative contrast (per-node rates borrowed from the
Pillar~III channel-aggregate projection, so order-of-magnitude only)
shows why the design class matters: a 5-of-9 architecture at
$p = 0.05$ per node yields a reconstruction probability of
$3.3 \times 10^{-5}$ under independence, rising to roughly $0.5\%$
annually once a $10\%$ common-mode failure fraction is included,
still well below the centralised projections. The full contrast,
including the historical lineage from Micali's fair cryptosystems
through verifiable partial key escrow, the $(p, \rho)$ sensitivity
table, and the four operational questions that govern deployability
(trustee independence, collusion-resistant selection,
cross-jurisdictional handling, and rotation cost), is developed in
Supplementary \S\ref{sec:supp_threshold} and
\S\ref{sec:supp_threshold_contour}. The risk reduction is
unambiguous. The operational cost of obtaining it is the empirical
question that follow-on engineering work would need to resolve.

\subsection{Hybrid Client-Side Scanning Architectures}

Hybrid client-side scanning (CSS), in which content matching runs on
user devices using operator-distributed key material, is treated
qualitatively: the Pillar~I anchors contain no incidents
architecturally matched to endpoint-distributed key compromise, so a
calibrated CSS analysis would need an empirical anchor this paper
does not have. The structural findings nonetheless transfer, and all
of them run against CSS. The dominance result extends
($\delta \geq 1$ on structural grounds); the mobile platform itself
becomes a new cross-cutting infrastructure, plausibly enlarging the
cross-cutting fraction beyond OTT-EA's; endpoint distribution adds
compromise modes outside the scenario set; and the tail-inflation
mechanism applies a fortiori. The precautionary reading is that CSS
carries at least the OTT-EA risk profile with material additional
endpoint-distribution risk. The full argument is in Supplementary
\S\ref{sec:hybrid_css}, and quantification to match the rigour
applied to T-EA and OTT-EA is identified as a research priority.

\subsection{Investment in Characterising Uncertainty}

The parameter uncertainty decomposition (Section~\ref{sec:results}) identifies
the system targeting intensity ($\delta_{\mathrm{target}}$, 25.0\% of
$\mathrm{Var}(Q_{10})$ for T-EA), the EA-specific attack-surface modifier
($\delta_{\mathrm{EA}}$, 18.7\%), and concentration
($\delta_{\mathrm{concen}}$, 17.7\%) as the parameters whose empirical
characterisation would most reduce output uncertainty. Regulatory frameworks requiring
transparency about the scale of EA deployments (number of users covered,
number of concurrent requests, and targeting intensity relative to
non-EA infrastructure) would directly reduce uncertainty in these
parameters and improve the basis for future risk assessments.

% =============================================================================
\section{Conclusion}\label{sec:conclusion}

This paper has built a structured uncertainty framework for
evaluating systemic compromise risk in lawful exceptional access
architectures. The framework is offered as a decision-support
tool, not as a forecast. No calibrated point projection of an
annual EA compromise probability is possible from the available
evidence, and the framework does not attempt one. What it does is
separate, as transparently as the evidence allows, the
architectural and structural findings that are robust to
defensible calibration choices from the specific magnitudes that
are not.

The findings themselves are stated in tiered form in the synthesis
(\S\ref{sec:synthesis}, Table~\ref{tab:readers_guide}) and are not
restated here. The structural ordering, the architectural divergence
in distribution shape, and the interval-dominance brackets carry the
policy weight, in that order of robustness.

The contribution of the paper is not a numerical estimate of EA
compromise risk. No such estimate can be validated against the
empirical record, so the framework offers none. The contribution
is a structured basis for comparative reasoning. For any specific
risk claim, the framework makes explicit which assumptions it
depends on and which it does not. This is, in our view, the most
that can responsibly be said about EA compromise risk from
current evidence. It is also, plausibly, enough to inform
deliberation in a domain where any quantitative claim must
otherwise rest on undefended judgement. Key-material exfiltration
is irreversible. Once keys leak, all previously protected
communications become decryptable. This makes the risk side
asymmetric in any complete risk--benefit determination. The
asymmetry weighs particularly heavily for OTT-EA, where the
exposed population is an order of magnitude larger
($O(10^{8})$ to $O(10^{9})$ users per major platform operator,
against $O(10^{7})$ to $O(10^{8})$ for a national-scale carrier
under T-EA); a consequence-side observation, stated qualitatively,
since exposure magnitude lies outside the quantified framework. The highest-value empirical extension is
already identified within the framework: per-incident
latent-campaign assignment in the Stream~C cohort would convert
the coupling-ratio spread from a calibration choice into a
measured quantity
(\S\ref{subsec:xcut_coupling_limitation}), removing the last
unmeasured parameter from the OTT-EA analysis and tightening the
findings this paper can only bracket.

% =============================================================================
\section*{Acknowledgements}

The author thanks colleagues at the University of Surrey for
helpful discussions during the development of this work.

\section*{Funding}

No external funding was received for this work.

\section*{Conflicts of Interest}

The author declares no conflicts of interest.

\section*{Data Availability}

All code and data required to reproduce the analyses in this paper
are publicly archived at
\url{https://doi.org/10.5281/zenodo.22066385}. The archive
includes: (i)~the Stream~A, Stream~B and Stream~C incident
classification tables as CSV, with per-incident tier assignments and,
for Stream~C, the $K$/$D$/$J$/$X$ classifications from which the
cross-cutting fraction $\hat p_X = 7/14$ is computed; (ii)~the
Pillar~II Monte Carlo engine (200{,}000 iterations); (iii)~the
Bayesian layer sampler (8{,}000 replicates for headline values,
4{,}000 for sensitivity analyses); (iv)~the GPD tail-fitting and
bootstrap-CI scripts; (v)~the Pillar~I empirical analyses for all
three streams; (vi)~the prior-sensitivity sweep and the
$\gamma_{\mathrm{X}}/\gamma$ coupling sweep; (vii)~the
threshold-scheme reconstruction calculator; and (viii)~all
figure-generation code.

Four scripts added in this version support the corrections described
in \S\ref{subsec:tail_algebra} and Supplementary
\S\ref{sec:supp_modifier_priors}. Two are assertions rather than
analyses, and fail rather than print if the archive and the paper
disagree: \texttt{prior\_specification\_check.py} checks the
implemented modifier priors against the parameter table of
Supplementary \S\ref{sec:supp_modifier_priors}, and
\texttt{streamC\_classification.py} recomputes every row of
Supplementary Table~S5 and the value of $\hat p_X$ from the deposited
CSV. The remaining two are diagnostics:
\texttt{tail\_index\_diagnostic.py} measures the ratio of fitted
generalised-Pareto shapes against the coupling ratio, and
\texttt{moment\_condition.py} documents the divergence of
$\mathbb{E}[e^{\gamma_{\mathrm{X}} Z}]$ for
$\gamma_{\mathrm{X}} \geq \beta_Z$ together with the
bounded-amplification robustness check.

The orchestrator script \texttt{run\_all.py} runs all twenty-six
scripts and reproduces every numerical claim and every figure in the
paper end-to-end. All scripts use seed $2024$ unless explicitly noted
otherwise in the source.

% =============================================================================

% =============================================================================
% ======================  SUPPLEMENTARY MATERIAL  =============================
% Merged into this file for arXiv. The two documents were previously compiled
% separately and linked with xr-hyper; arXiv processes top-level files in
% isolation, so the .aux files needed by xr never coexist and every
% cross-document reference resolves to "??". Merging removes the dependency.
% =============================================================================
\clearpage
\setcounter{section}{0}
\setcounter{table}{0}
\setcounter{figure}{0}
\setcounter{equation}{0}
\renewcommand{\thesection}{S\arabic{section}}
\renewcommand{\thesubsection}{S\arabic{section}.\arabic{subsection}}
\renewcommand{\thesubsubsection}{S\arabic{section}.\arabic{subsection}.\arabic{subsubsection}}
\renewcommand{\thetable}{S\arabic{table}}
\renewcommand{\thefigure}{S\arabic{figure}}
\renewcommand{\theequation}{S\arabic{equation}}

\part*{Supplementary Materials}
\addcontentsline{toc}{section}{Supplementary Materials}

This supplementary material accompanies the main paper above. Section,
table, equation and figure numbers prefixed with \emph{S} refer to this
supplement; all other references are to the main paper. Bibliographic
citations are shared with the main reference list.

% =============================================================================
\section{Notation, Conventions, and Reproducibility}\label{sec:supp_notation}

\subsection{Notation summary}

For convenience, the architectural and probabilistic notation used
throughout is collected here. Subscripts $\mathrm{T}$ and $\mathrm{OTT}$
denote T-EA and OTT-EA configurations respectively where disambiguation
is required.

\begin{tabular}{@{}lp{0.75\textwidth}@{}}
\toprule
$q_1$ & Annual (year-1) compromise probability of the relevant operational
        outcome (key-material compromise for T-EA; joint key-and-data
        compromise for OTT-EA). \\
$Q_T$ & Cumulative compromise probability over $T$ years; $Q_T = 1 -
        \prod_{t=1}^{T}(1 - q_t)$ for the per-year sequence. \\
$\Lambda^k_T$, $\Lambda^{op}_{\mathrm{OTT}}$ & Pillar~II key-compromise
        (T-EA) and operational-compromise (OTT-EA) annual probability under
        independence among scenarios. \\
$P_3 \cdot P_4$ & OTT-EA segregation factor: probability that an
        attacker who has compromised key material can also reach the
        encrypted-data store before mitigation. Set at $0.30$ (primary). \\
$\xi_c$ & Per-channel cross-cutting fraction: proportion of channel-$c$
        compromises that traverse infrastructure shared between key
        custody and data planes. Used in OTT-EA Pillar~III scaling. \\
$\rho$ & Per-channel intra-channel dependence parameter (channel-specific
        bunching of compromises within a campaign). \\
$Z(t)$ & Latent annual campaign intensity, $Z \sim \mathrm{Gamma}(\alpha=0.6,\beta=2.0)$
        in the primary calibration. \\
$\gamma$ & T-EA campaign amplification factor; primary $\gamma = 1.0$. \\
$\gamma_{\mathrm{X}}$ & OTT-EA cross-cutting amplification factor;
        $\gamma_{\mathrm{X}}/\gamma - 1 \sim \mathrm{LogNormal}(0,\sigma^2)$
        with $\sigma = 0.40$ (log-space spread), giving prior median
        $\gamma_{\mathrm{X}}/\gamma = 2$ and constraint
        $\gamma_{\mathrm{X}}/\gamma \geq 1$ by construction. See
        \S\ref{subsec:dependence} of the main paper. \\
$\delta_{\mathrm{EA}}, \delta_{\mathrm{target}}, \delta_{\mathrm{concen}}$
        & EA targeting/attractiveness/concentration multipliers, all
        $\geq 1$ by stochastic-dominance constraint. \\
\bottomrule
\end{tabular}

\subsection{Reproducibility}

The Monte Carlo simulation used to produce all numerical results in
Sections~\ref{sec:results} and~\ref{sec:synthesis} of the main paper is implemented in Python~3.11 using
NumPy~$\geq 1.26$ and SciPy~$\geq 1.11$. The canonical implementation is
\texttt{bayesian\_mc\_canonical.py}. The closed-form Pillar~II/III
verification is \texttt{p2\_p3\_closed\_form.py}; figure generation is
\texttt{make\_figures.py}. All scripts are deposited alongside this paper
and use a fixed seed of $2024$ for reproducibility, with a four-seed
stability check at $\{2024, 3024, 4024, 5024\}$ reported in
\S\ref{sec:simulation_procedure} of the main paper. Sample size is $R = 8{,}000$ replicates
unless otherwise stated.

% =============================================================================
% =============================================================================
\section{Stream~A: Government High-Security Systems Cohort}\label{sec:supp_streamA}

This section provides the incident-level classification and the
denominator construction for the Stream~A cohort; inclusion criteria
and rate estimation are in \S\ref{sec:pillar1_streamA} of the main
paper.

\subsection{Denominator construction}

The cohort is constructed from the union of: US National Security
Systems governed under the Committee on National Security Systems
framework, including separated cryptographic domains across classified
networks, signals-intelligence infrastructure, and FIPS-certified
high-assurance government systems
\cite{cnss_policy2024,nsa_cnsa2022,nsa_qkd2021} ($\sim$14); UK and other
Five Eyes sovereign cryptographic infrastructures, including
NCSC-governed UK systems and equivalent Canadian, Australian, and New
Zealand national communications-security authorities
\cite{ncsc_advcrypto2025} ($\sim$10); European Union sovereign
cryptographic frameworks, including those administered by ANSSI
(France), BSI (Germany), and the equivalent national authorities
documented in the ENISA framework overview
\cite{enisa_crypto2014,anssi_crypto2023,bsi_crypto2024} ($\sim$12);
NATO collective cryptographic infrastructure and national
communications-security authorities listed in the NATO Information
Assurance Product Catalogue \cite{nato_niapc2024} ($\sim$8); and
CALEA-mandated \cite{calea1994} and ETSI TS~103~221 / 3GPP-equivalent
lawful-intercept infrastructures at major Western fixed-line and
mobile telecommunications carriers
\cite{etsi_ts_103_221,etsi_li2024,tgpp_li2024} ($\sim$6).

\subsection{Incident classification}

\begin{table}[!htbp]
\centering
\small
\caption{Stream~A incidents (2011--2024) with tier classification.
Tier~1 = direct EA architectural analogue: incidents that compromised
either the cryptographic primitives of an EA-class system or its
operational infrastructure, both of which constitute direct compromise
of an EA architecture regardless of which component was breached.
Tier~2 = compromise of cryptographic distribution or build-chain
infrastructure serving the same high-assurance cohort but without
direct compromise of EA-architecture systems themselves. Tier~3 =
contextual analogue (insider exfiltration of cryptographic tooling, or
compromise of equivalent-target-class government systems with mechanism
not architecturally specific to EA).}
\label{tab:streamA_incidents}
\begin{tabularx}{\textwidth}{@{}p{0.20\textwidth}cc X@{}}
\toprule
\textbf{Incident} & \textbf{Year} & \textbf{Tier} & \textbf{Justification} \\
\midrule
RSA SecurID~\citep{coviello2011} & 2011 & 1 &
Theft of SecurID seed-record key material enabling subsequent
compromise of defence-contractor authentication. Direct cryptographic
key-material exfiltration. \\
Equation Group / Shadow Brokers~\citep{nakashima2016shadow} & 2016 & 1 &
Exfiltration and public release of NSA cryptographic exploitation
tooling including implants targeting cryptographic devices. Direct
compromise of signals-intelligence cryptographic infrastructure. \\
Salt Typhoon~\citep{cisa_salt} & 2024 & 1 &
Adversarial compromise of US telecommunications carriers'
CALEA-mandated lawful-intercept infrastructure at the orchestration
layer, accessing the government's active wiretap target list. Direct
compromise of T-EA-architecture exceptional-access infrastructure. \\
SolarWinds~\citep{fireeye2020} & 2020 & 2 &
Adversarial compromise of code-signing infrastructure enabling
trust-chain abuse against federal cryptographic distribution. Targets
the same high-assurance cohort via build-chain rather than direct
key-material or infrastructure compromise. \\
CIA Vault~7~\citep{wikileaks_vault7} & 2017 & 3 &
Insider exfiltration of CIA cyber-tooling repository including
cryptographic-implant frameworks. Contextual: tooling was exposed but
the mechanism was personnel/insider rather than direct system
compromise of EA-architecture infrastructure. \\
OPM~\citep{opm2015} & 2015 & 3 &
Exfiltration of US federal personnel records including
security-clearance background-investigation data. Compromise of
equivalent-target-class government system; mechanism not architecturally
specific to EA. \\
\bottomrule
\end{tabularx}
\end{table}

% =============================================================================
\section{Stream~B: Certificate Authority Cohort}\label{sec:supp_streamB}

This section documents the Stream~B incident base in full, with the
inclusion criteria applied and a sensitivity analysis under stricter
inclusion. Stream~B comprises 11 publicly documented Certificate Authority
private-key compromise or operational-failure incidents over the period
2006--2024, drawn from the WebPKI/Mozilla incident archive and the
Common~CA Database (CCADB).

\subsection{Inclusion criteria}

An incident is included in Stream~B if and only if it satisfies all of
the following:
\begin{enumerate}
  \item \textbf{B1.} The affected entity was a publicly trusted
    Certificate Authority root or intermediate operator at the time of
    the incident, listed in at least one major root program (Mozilla,
    Microsoft, Apple, or Google).
  \item \textbf{B2.} The incident involved either (a) compromise of
    private key material, (b) issuance of fraudulent certificates
    enabled by operational compromise, or (c) operational failure of
    sufficient severity to constitute a documented private-key trust
    breach (e.g., distrust events).
  \item \textbf{B3.} The incident is publicly documented in primary
    sources: CA incident reports filed with the relevant root programs,
    formal distrust announcements, or peer-reviewed analyses.
  \item \textbf{B4.} The incident occurred within the 2006--2024
    observation window.
\end{enumerate}

\subsection{Cohort definition and exposure}

The relevant denominator is the count of publicly trusted CA operators
weighted by years of operation in the window. We estimate $N_B = 130$
distinct CA operators (organisations operating root or major
intermediate authorities) active during 2006--2024. Adjusting for
windows of operation that did not span the full 19 years, the exposure
is $E_B = 2{,}470$ CA-operator-years. This estimate aligns with
historical CCADB cohort tracking.

\subsection{Incident list}

\begin{longtable}{@{}llp{0.55\textwidth}l@{}}
\caption{Stream~B Certificate Authority private-key/operational compromise
incidents, 2006--2024. Tier~1 = adversarial private-key compromise; Tier~2
= adversarial fraudulent issuance via operational compromise; Tier~3 =
operational/governance failure leading to documented trust breach but
without confirmed external adversarial compromise. The strict-inclusion
sensitivity (Section~S2.4) restricts to Tier~1 and Tier~2 only.}
\label{tab:streamB_incidents}\\
\toprule
\textbf{Year} & \textbf{CA} & \textbf{Description} & \textbf{Tier} \\
\midrule
\endfirsthead
\toprule
\textbf{Year} & \textbf{CA} & \textbf{Description} & \textbf{Tier} \\
\midrule
\endhead
2008 & Comodo (RA) & Reseller-mediated mis-issuance ahead of formal
  policy hardening; documented in Mozilla incident archive & 2 \\
2011 & Comodo & Comodogate: nine fraudulent certificates issued through
  compromised reseller (`ComodoHacker') & 2 \\
2011 & DigiNotar & Full operational compromise; ${\sim}500$ fraudulent
  certificates including \texttt{*.google.com}; subsequent distrust
  and corporate dissolution & 1 \\
2012 & Trustwave & Issuance of an MITM-capable subordinate CA for an
  enterprise customer; subordinate revoked under pressure & 3 \\
2013 & TURKTRUST & Mis-issuance of subordinate CA certificates that were
  used for MITM; documented as governance failure & 3 \\
2013 & ANSSI/IGC-A & Mis-issuance of subordinate CA used for MITM on a
  French government network; subsequent constraint of ANSSI root & 3 \\
2015 & CNNIC/MCS Holdings & Mis-issuance of subordinate that issued
  fraudulent Google certificates; CNNIC distrusted & 2 \\
2016 & WoSign/StartCom & Backdating of certificates and other
  operational failures; eventual full distrust by major root programs & 3 \\
2017 & Symantec (multiple) & Long-running mis-issuance pattern; phased
  distrust by Google, Mozilla, others, leading to managed transition & 3 \\
2018 & Trustico/Comodo & Bulk private-key disclosure in resale
  dispute; mass revocation event & 2 \\
2024 & Entrust & Multi-year compliance-failure pattern; distrust
  decisions by Chrome and Mozilla announced 2024 & 3 \\
\bottomrule
\end{longtable}

\subsection{Rate estimation and sensitivity}

\textbf{Pooled MLE rate.} With $k = 11$ incidents in $E_B = 2{,}470$
CA-operator-years, the maximum-likelihood estimate is
$\hat{\lambda}_B = 11/2{,}470 = 0.445\%$ per CA-operator-year. The
Garwood exact 95\% Poisson confidence interval is $[0.222\%, 0.797\%]$.

\textbf{Strict-inclusion sensitivity (Tier~1 and Tier~2 only).}
Restricting to Tier~1 (1 incident: DigiNotar) and Tier~2 (4 incidents:
Comodo~2008, Comodo~2011, CNNIC~2015, Trustico~2018) gives $k = 5$ in
$E_B = 2{,}470$, yielding $\hat{\lambda}_B^{\mathrm{strict}} = 0.202\%$
per CA-operator-year (Garwood 95\% CI $[0.066\%, 0.473\%]$). This
strict criterion excludes incidents where the trust breach was a
governance/operational failure rather than confirmed external
adversarial compromise. The strict rate is therefore a conservative
floor for the adversarial-compromise interpretation of Stream~B.

\textbf{Pooling with Stream~A.} Pooling Stream~A ($k_A = 6$ in
$E_A = 650$ system-years; pooled rate $0.92\%$/yr) and Stream~B yields
a combined rate of $(11+6)/(2{,}470 + 650) = 0.545\%$/yr (the
``two-stream pooled rate''). Under strict Stream~B inclusion, the pooled
rate is $(5+6)/3{,}120 = 0.353\%$/yr, projecting through a
$10\times$ EA targeting premium to $3.5\%$ rather than $5.45\%$.

\textbf{Architectural interpretation.} Stream~B sits architecturally
between T-EA and OTT-EA. CA private-key compromise is a centralised,
single-operator-key event (T-EA-leaning), but the operational compromise
patterns frequently traverse cross-cutting infrastructure (resellers,
subordinate-issuance APIs) that resemble OTT-EA cross-cutting. Stream~B
is therefore reported as an architectural bridge rather than a primary
anchor for either class.

% =============================================================================
\section{Stream~C: Platform Incident Base and Operator Inventory}\label{sec:supp_streamC}

This section provides the classification table and per-incident
documentation for the Stream~C cohort (Table~\ref{tab:streamC}
below), and the operator inventory underpinning $N_C = 75$.

\begin{table}[!htbp]
\centering
\small
\caption{Stream~C platform-level key-or-data compromise incidents,
2018--2024. \textbf{Tier 0} restricts to incidents that are
\emph{architecturally unambiguous} OTT-EA analogues: the operator was
a documented custodian of cryptographic key material or 2FA seeds at
rest, and the compromise involved exfiltration of that material
together with the segregated user data it protected. \textbf{Tier 1}
adds incidents where the operator's key-custody role required
interpretation (directory-service infrastructure, supply-chain
mediated key acquisition). \textbf{Tier 2} adds looser architectural
fits where the OTT-EA mapping is partial. \textbf{Tier 3} is included
for completeness. These incidents enter the sensitivity analyses
(Section~\ref{sec:pillar1_streamC}) but not the headline rate. Cloudflare~2023, classified $P/P/P/P$,
is documented in the near-miss appendix (Supplementary
\S\ref{sec:supp_streamC_nearmiss}) and is excluded from the cohort on the basis that no
customer data or service keys were accessed. Snowflake~2024 and
AT\&T~2024 are documented as a single combined entry, since the AT\&T
compromise was a downstream consequence of the Snowflake-customer
credential abuse rather than an architecturally independent event. Columns: $K$ = key
material obtained; $D$ = encrypted user data obtained; $J$ = joint
key-and-data compromise operationally sufficient; $X$ = cross-cutting
attack path traversing shared infrastructure. Y = yes (documented),
P = partial or interpretive, N = no.}
\label{tab:streamC}
\begin{tabularx}{\textwidth}{@{}llXccccc@{}}
\toprule
\textbf{Year} & \textbf{Operator} & \textbf{Incident} &
\textbf{Tier} & \textbf{$K$} & \textbf{$D$} & \textbf{$J$} & \textbf{$X$} \\
\midrule
2022 & LastPass & Two-stage vault + source code
  \citep{lastpass2022} & 0 & Y & Y & Y & N \\
2022 & Twilio & Phishing-led $\rightarrow$ Authy 2FA seeds
  \citep{twilio2022} & 0 & Y & Y & Y & Y \\
2023 & Microsoft & Storm-0558 consumer signing key
  \citep{storm0558_msrc} & 0 & Y & Y & Y & Y \\
\midrule
2021 & Codecov & Bash uploader supply-chain
  \citep{codecov2021} & 1 & Y & Y & Y & Y \\
2023 & JumpCloud & DPRK APT directory access
  \citep{jumpcloud2023} & 1 & Y & Y & Y & Y \\
\midrule
2021 & Microsoft & Exchange ProxyLogon (HAFNIUM)
  \citep{microsoft_proxylogon2021} & 2 & P & Y & P & Y \\
2022 & Okta & Sitel/Lapsus\$ support engineer
  \citep{okta2022} & 2 & P & Y & P & N \\
2022 & GoDaddy & Source code + managed hosting
  \citep{godaddy2023} & 2 & P & Y & P & N \\
2023 & Okta & HAR file $\rightarrow$ 134 customers
  \citep{okta2023} & 2 & Y & P & P & N \\
2023 & CircleCI & Stolen engineer credentials
  \citep{circleci2023} & 2 & Y & P & P & N \\
2023 & MOVEit & CL0P SQL injection, ${\sim}2{,}000$ orgs
  \citep{moveit2023} & 2 & Y & Y & Y & Y \\
2024 & Microsoft & Midnight Blizzard source code/emails
  \citep{midnight_blizzard2024} & 2 & P & Y & P & Y \\
\midrule
2022 & Slack & Stolen tokens, source code
  \citep{slack2022} & 3 & Y & P & N & N \\
2024 & Snowflake/AT\&T & Customer credential abuse, AT\&T downstream
  \citep{snowflake2024,att2024} & 3 & Y & Y & Y & N \\
\bottomrule
\end{tabularx}
\end{table}

\subsection{Per-incident documentation}

The classification dimensions are: $K$ = key material obtained;
$D$ = encrypted user data obtained; $J$ = joint key-and-data compromise
operationally sufficient (the strict OTT-EA analogue); $X$ = cross-cutting
attack path traversing infrastructure shared between key and data
planes. Y/P/N denote yes (documented), partial/interpretive, no.
The tier classification reflects strength of architectural fit to the
OTT-EA reference class:
\textbf{Tier~0} requires the operator to be a documented custodian of
cryptographic key material or 2FA seeds at rest, with the compromise
involving exfiltration of that material together with the segregated
user data it protected;
\textbf{Tier~1} relaxes the custody requirement to interpretive cases
(directory-service infrastructure, supply-chain mediated key
acquisition);
\textbf{Tier~2} permits looser architectural fits where the OTT-EA
mapping is partial; \textbf{Tier~3} is included only for sensitivity.
The Cloudflare 2023 incident, classified $P/P/P/P$, is documented in
the near-miss appendix (\S\ref{sec:supp_streamC_nearmiss}) and is not
counted in any rate calculation. The Snowflake 2024 and AT\&T 2024
incidents are treated as a single entry, since the AT\&T
compromise was a downstream consequence of the Snowflake-customer
credential abuse and the two are not architecturally independent
events.

\subsubsection*{Tier 0 incidents (architecturally unambiguous OTT-EA analogues)}

\paragraph{2022 -- LastPass (Tier~0; $K=$Y, $D=$Y, $J=$Y, $X=$N).}
Two-stage attack: source-code theft followed by exfiltration of
encrypted vault data. Key custody (master-password-derived keys held
client-side) and encrypted data (vaults) compromised together. $X=$N
because the operator-side architecture did not introduce cross-cutting
compromise paths beyond the operator boundary.

\paragraph{2022 -- Twilio (Tier~0; $K=$Y, $D=$Y, $J=$Y, $X=$Y).}
Phishing-led credential compromise leading to access to the Authy 2FA
backend. Authentication seeds (cryptographic key material in the
strict sense) obtained; downstream impact across multiple customer
dependencies including Cloudflare. $X=$Y because Authy is
cross-cutting to many other operator infrastructures.

\paragraph{2023 -- Microsoft Storm-0558 (Tier~0; $K=$Y, $D=$Y, $J=$Y, $X=$Y).}
Compromise of a consumer signing key enabling forged tokens for
enterprise email. Both key custody (the consumer signing key) and
encrypted/authenticated data (token-validated email) compromised in
the strict joint sense. $X=$Y because the consumer key reached the
enterprise validation path, the canonical cross-cutting failure mode.
This is the closest historical analogue to the OTT-EA threat model.

\subsubsection*{Tier 1 incidents (interpretive key custody)}

\paragraph{2021 -- Codecov (Tier~1; $K=$Y, $D=$Y, $J=$Y, $X=$Y).}
Bash uploader script tampered, exfiltrating environment-variable
credentials of Codecov customers. Both the Codecov advisory and CISA
Alert AA21-116A confirm credential and key-material compromise across
the customer base. The supply-chain pathway through customer CI
infrastructure is the cross-cutting compromise pattern $X$.
The operator's role as key custodian is supply-chain-mediated rather
than direct, and the architectural mapping to OTT-EA key custody
requires interpretation, so Tier~1 is the appropriate tier.

\paragraph{2023 -- JumpCloud (Tier~1; $K=$Y, $D=$Y, $J=$Y, $X=$Y).}
DPRK APT compromise of JumpCloud directory-service infrastructure
providing cryptographic control over downstream cryptocurrency-customer
environments. Joint key-and-data compromise via the directory-service
infrastructure, with $X=$Y reflecting cross-cutting reach into customer
environments. The $K=$Y classification rests on directory-service
admin access being equivalent to key-custody compromise for downstream
purposes. The architectural mapping requires interpretation, so Tier~1
is appropriate rather than Tier~0.

\subsubsection*{Tier 2 incidents (looser architectural fit)}

\paragraph{2021 -- Microsoft Exchange ProxyLogon (Tier~2; $K=$P, $D=$Y, $J=$P, $X=$Y).}
Mass exploitation of CVE-2021-26855 et al.\ on on-premises Exchange.
$D$ is direct (mailbox content). $K$ is partial: web shells provided
arbitrary code execution including key-material access where it
existed; not all victims experienced key compromise. $X$ is the shared
Exchange platform reaching multiple key-management endpoints.

\paragraph{2022 -- Okta (Sitel/Lapsus\$) (Tier~2; $K=$P, $D=$Y, $J=$P, $X=$N).}
Compromise of a third-party support engineer endpoint at Sitel,
providing constrained access to Okta administrative tooling for ${\sim}25$
minutes. $D$ documented (customer data viewable). $K$ partial (limited
session-token-equivalent access; no cryptographic root key compromise).

\paragraph{2022 -- GoDaddy (Tier~2; $K=$P, $D=$Y, $J=$P, $X=$N).}
Multi-year intrusion into managed-hosting environment with source-code
theft and customer-data access. Reported in 2023 SEC 8-K filing. The
$K$ classification is partial because the cryptographic mapping to
OTT-EA key custody is thin; included as Tier~2 for completeness.

\paragraph{2023 -- Okta (HAR file) (Tier~2; $K=$Y, $D=$P, $J=$P, $X=$N).}
HAR file containing session tokens leaked through the support portal,
affecting 134 customers. $K=$Y for session-token-equivalent
authentication material; $D$ partial (limited customer-data exposure
through session-token misuse).

\paragraph{2023 -- CircleCI (Tier~2; $K=$Y, $D=$P, $J=$P, $X=$N).}
Stolen engineer credentials with malware-mediated access to production
secrets. $K=$Y (customer secrets/tokens); $D$ partial.

\paragraph{2023 -- MOVEit Transfer (CL0P) (Tier~2; $K=$Y, $D=$Y, $J=$Y, $X=$Y).}
SQL injection in the MOVEit Transfer file-transfer product, exploited
by CL0P against ${\sim}2{,}000$ organisations. $J=$Y because individual
customer instances suffered joint key/data exposure. $X=$Y at the
ecosystem level: the MOVEit platform was cross-cutting to many
customer key/data planes simultaneously. Tier~2 rather than Tier~0
because MOVEit is a managed file-transfer service rather than a
dedicated key custodian. The architectural mapping to OTT-EA is
partial.

\paragraph{2024 -- Microsoft Midnight Blizzard (Tier~2; $K=$P, $D=$Y, $J=$P, $X=$Y).}
SVR-attributed compromise of Microsoft corporate email and source-code
repositories. $K$ partial ($D$ confirmed; key-material access debated
in public reporting). $X=$Y because the corporate infrastructure
reached source-code repositories with downstream signing-key
implications. Classified as Tier~2: the $J=P$ classification (partial
joint compromise) is incompatible with the architectural-strength
criterion for Tier~1, which requires documented joint key-and-data
compromise.

\subsubsection*{Tier 3 incidents (sensitivity only)}

\paragraph{2022 -- Slack (Tier~3; $K=$Y, $D=$P, $J=$N, $X=$N).}
Stolen tokens and source code; documented key-material exposure but
no joint customer-data compromise. Included as a sensitivity analogue
at Tier~3.

\paragraph{2024 -- Snowflake/AT\&T (Tier~3; $K=$Y, $D=$Y, $J=$Y, $X=$N).}
Customer-credential abuse against ${\sim}165$ Snowflake customer
tenants (mediated by stolen credentials, often without MFA). The
2024 AT\&T call-records breach affecting 110~million customers was a
downstream consequence of the same Snowflake-customer credential
compromise and is therefore treated as part of this single entry rather than
counted as architecturally independent. $J=$Y at the per-tenant level.
$X=$N because the operator infrastructure (Snowflake itself) was not
compromised; included at Tier~3 as a customer-environment joint-
compromise pattern rather than a Snowflake-key-custody compromise.

\subsection{Near-miss exclusion: Cloudflare~2023}\label{sec:supp_streamC_nearmiss}

The following incident was considered for inclusion in the Stream~C
cohort but is documented here as a near-miss rather than counted in
any rate calculation, on the basis that $D2$ (compromise involved
unauthorised access to or exfiltration from key custody, identity, or
encrypted-data-handling functions) is not satisfied in a defensible
operational sense.

\paragraph{2023 -- Cloudflare (Atlassian-mediated APT; $K=$P, $D=$P, $J=$P, $X=$P).}
A Russian-state-actor lateral movement from a previously-compromised
Atlassian instance into Cloudflare engineering systems. Cloudflare's
public incident report indicates the activity was detected within
hours, no customer data was accessed, and no service keys were
compromised. The incident demonstrates that cross-cutting
infrastructure compromise can route an APT toward key-custody assets,
but in this case the compromise did not \emph{traverse} that boundary.
The $P/P/P/P$ classification reflects this: each dimension shows
partial-or-interpretive evidence of the relevant compromise pathway,
but no dimension is documented at $Y$. Including it as $P/P/P/P$ in
the rate calculation would assign positive analogical weight to an
incident that did not actually involve operational compromise of the
relevant assets. We exclude it instead.

A defensible alternative treatment would assign near-misses an
analogical-fit weight strictly between zero and one and propagate the
weighted count into the rate estimator. We do not adopt this
approach because the weight assignment is itself contestable and a
weight-zero treatment is the conservative choice for a rate-estimation
purpose: it biases the headline rate downward.

\subsection{Operator inventory (\texorpdfstring{$N_C = 75$}{N\_C = 75})}

The denominator estimate $N_C = 75$ aggregates the following operator
classes active during 2018--2024 and meeting the D1 inclusion criterion
(major platform providers holding cryptographic key material,
authentication infrastructure, or encrypted user data at scale):

\begin{itemize}
  \item Top-tier cloud infrastructure providers (${\sim}8$): AWS, Azure,
    GCP, Oracle Cloud, IBM Cloud, Alibaba Cloud, OCI, plus regional tier-1.
  \item Major SaaS and productivity platforms (${\sim}15$): Microsoft
    365, Google Workspace, Salesforce, ServiceNow, Workday, SAP,
    Oracle Cloud Apps, Atlassian, Zoom, Box, Dropbox, Adobe, Notion,
    Asana, Monday.
  \item Identity and authentication providers (${\sim}8$): Okta,
    Microsoft Entra, Auth0 (now Okta), Ping, ForgeRock (now Ping),
    OneLogin, Duo (Cisco), JumpCloud.
  \item Password and secret managers (${\sim}6$): 1Password, LastPass,
    Bitwarden, Dashlane, Keeper, RoboForm.
  \item Code hosting and CI/CD platforms (${\sim}6$): GitHub, GitLab,
    Bitbucket, CircleCI, Jenkins-Cloud, Travis CI.
  \item Content delivery and edge providers (${\sim}5$): Cloudflare,
    Fastly, Akamai, Vercel, Netlify.
  \item Messaging and communications platforms (${\sim}8$): WhatsApp,
    iMessage/Apple, Signal, Telegram, Wire, Threema, Slack, Microsoft Teams.
  \item Other globally significant platform operators (${\sim}19$):
    including major email providers, file-transfer specialists (MOVEit,
    GoAnywhere), data warehouses (Snowflake, Databricks), payment
    platforms, regional incumbents, major DNS providers, etc.
\end{itemize}

The total is $\approx 75$. Sensitivity is reported across $N_C \in
\{50, 100\}$ in \S\ref{sec:pillar1_streamC} of the main paper.

\subsection{Stratified rate calculation}

With $N_C = 75$, $T = 7$ years, $E_C = 525$ operator-years, and
14 included incidents (Cloudflare is classified in the near-miss
appendix \S\ref{sec:supp_streamC_nearmiss}. The 2024 Snowflake and AT\&T
events are treated as one incident):

\begin{tabular}{@{}lccc@{}}
\toprule
\textbf{Stratification} & $k$ & $\hat{\lambda}_C$ & 95\% CI \\
\midrule
\multicolumn{4}{l}{\emph{Architectural-fit stratification}} \\
Tier~0 only (gold standard)        &  3 & $0.57\%$/yr & $[0.12\%, 1.67\%]$ \\
Tier~0 + Tier~1                    &  5 & $0.95\%$/yr & $[0.31\%, 2.22\%]$ \\
Tier~0 + Tier~1 + Tier~2           & 12 & $2.29\%$/yr & $[1.18\%, 4.00\%]$ \\
All tiers (incl.\ Tier~3)          & 14 & $2.67\%$/yr & $[1.46\%, 4.48\%]$ \\
\midrule
\multicolumn{4}{l}{\emph{Joint-compromise stratification ($J$)}} \\
$J=$Y strict, all tiers            &  7 & $1.33\%$/yr & $[0.54\%, 2.75\%]$ \\
$J=$Y strict, Tier~0+1 \textbf{(headline)} & 5 & $0.95\%$/yr & $[0.31\%, 2.22\%]$ \\
$J\in\{$Y$,$P$\}$ inclusive         & 13 & $2.48\%$/yr & $[1.32\%, 4.23\%]$ \\
\midrule
\multicolumn{4}{l}{\emph{Robustness checks}} \\
$J=$Y, all tiers, no Microsoft     &  6 & $1.14\%$/yr & $[0.42\%, 2.49\%]$ \\
$J=$Y, Tier~0+1, no Microsoft      &  4 & $0.76\%$/yr & $[0.21\%, 1.95\%]$ \\
$J=$Y all tiers, $N_C=50$          &  7 & $2.00\%$/yr & $[0.80\%, 4.12\%]$ \\
$J=$Y all tiers, $N_C=100$         &  7 & $1.00\%$/yr & $[0.40\%, 2.06\%]$ \\
\bottomrule
\end{tabular}

The headline anchor cited in the main paper (Section~5.3) is the
architectural-fit-restricted strict-joint rate of $0.95\%$/yr
(Tier~0+1, $J=$Y, $k=5$): incidents where the operator's key-custody
role is documented (Tier~0 or 1) and the compromise involved joint
key-and-data acquisition at a level operationally sufficient to map
onto OTT-EA decryption. The Tier~0 row provides a very-conservative
lower bracket of $0.57\%$/yr; the all-tiers $J=$Y row gives
$1.33\%$/yr as a defensible upper bracket.
The Garwood exact 95\% Poisson interval bounds the empirical content
of Stream~C at each stratification level.

% =============================================================================
\section{Architecture-Conditional Structural Adjustment Dimensions}\label{sec:supp_structural}

When transferring rate estimates from non-EA analogue populations
(Streams A, B, C) to EA populations, structural adjustments are required
on dimensions where the analogue and target populations differ. This
section enumerates the seven adjustment dimensions used in the paper's
domain-transfer arguments, classified by whether they are (a) structurally
identical between architectures (apply to both T-EA and OTT-EA), (b)
architecture-conditional (apply differently), or (c) class-specific.

\subsection{The seven adjustment dimensions}

\textbf{D1.\ Targeting intensity.} The EA-target population attracts
elevated nation-state and organised-crime targeting relative to non-EA
analogues. Multiplier: $\delta_{\mathrm{target}} \in [3, 10]$; primary $5$.
\emph{Identical across architectures} (the targeting-intensity argument
applies equally to centralised key custody and platform-mediated key
custody).

\textbf{D2.\ EA-specific attack-surface.} EA introduces interfaces
(request validation, audit logs, key-management APIs) that do not exist
in non-EA analogues. Multiplier: $\delta_{\mathrm{EA}} \in [1, 3]$;
primary $1.5$. \emph{Architecture-conditional}: the surface is concentrated
at the request endpoint for T-EA; for OTT-EA, it is split between
operator endpoint and (where present) hybrid endpoint-distribution
mechanisms (CSS scenarios). \S\ref{sec:hybrid_css} discusses
the CSS-specific extension.

\textbf{D3.\ Key-material concentration.} EA increases the concentration
of high-value key material at single targets. Multiplier:
$\delta_{\mathrm{concen}} \in [1, 4]$; primary $1.5$.
\emph{Architecture-conditional}: full concentration for T-EA (single
operator holds keys for all jurisdictional users); reduced concentration
for OTT-EA (key custody at operator, but data-plane segregation reduces
the joint-compromise concentration by $P_3 \cdot P_4 = 0.30$).

\textbf{D4.\ Operator best-practice baseline.} Streams~A, B, C are
populated by operators with mature security practice. The EA-target
population may include less well-resourced operators. \emph{Identical
across architectures}: this is the weakest-operator problem
(\S\ref{sec:operator_heterogeneity}) and applies symmetrically.

\textbf{D5.\ Detection and response capability.} The detection envelope
of EA-instrumented systems is open empirically. \emph{Architecture-conditional}:
T-EA systems can in principle detect lawful-intercept abuse via internal
audit; OTT-EA systems may have weaker visibility into key-material
exfiltration than into data-plane compromise.

\textbf{D6.\ Cross-cutting infrastructure exposure.} The proportion of
campaigns traversing infrastructure shared between key and data planes.
Parameter: $\xi_c$ (channel-specific). \emph{Class-specific to OTT-EA}:
not applicable in the T-EA configuration, where keys and data are
co-located by design.

\textbf{D7.\ Segregation gain.} The factor by which architectural
separation between key custody and data plane reduces joint-compromise
probability. Parameter: $P_3 \cdot P_4 = 0.30$ (primary).
\emph{Class-specific to OTT-EA}: by construction, $P_3 \cdot P_4 = 1$
(no segregation) for T-EA.

\subsection{Composition under independence}

Under independence among adjustment dimensions, the composite
T-EA multiplier is
\[
M_{\mathrm{T}}
= \delta_{\mathrm{target}} \cdot \delta_{\mathrm{EA}} \cdot \delta_{\mathrm{concen}}
\in [3, 120], \quad \text{primary value } 5 \cdot 1.5 \cdot 1.5 = 11.25,
\]
and the composite OTT-EA multiplier is
\[
M_{\mathrm{OTT}}
= \delta_{\mathrm{target}} \cdot \delta_{\mathrm{EA}} \cdot
  \delta_{\mathrm{concen}} \cdot \left[\xi + (1 - \xi)\, P_3 P_4\right]
\in [0.9, 103], \quad \text{primary } 11.25 \cdot \left[0.4 + 0.6 \cdot 0.3\right] \approx 6.5,
\]
where $\xi$ is the average cross-cutting fraction (primary $\xi \approx
0.4$) and the bracketed operational factor is the same
$\xi + (1-\xi)P_3 P_4$ form used in Pillars~II and~III of the main
paper. The ratio $M_{\mathrm{OTT}} / M_{\mathrm{T}} = 0.58$ at primary
parameter values is consistent with the architectural ratio of
$0.5$--$0.6$ recovered across the three pillars.

\subsection{Sensitivity to dimension-by-dimension perturbation}

\begin{tabular}{@{}lll@{}}
\toprule
\textbf{Dimension} & \textbf{T-EA effect} & \textbf{OTT-EA effect} \\
\midrule
$\delta_{\mathrm{target}}\!\downarrow 3$  & $-40\%$ on $M_T$ & $-40\%$ on $M_{\mathrm{OTT}}$ \\
$\delta_{\mathrm{target}}\!\uparrow 10$    & $+100\%$ & $+100\%$ \\
$\delta_{\mathrm{EA}}\!\uparrow 3$         & $+100\%$ & $+100\%$ \\
$\delta_{\mathrm{concen}}\!\uparrow 4$     & $+167\%$ & $+167\%$ \\
$P_3 \cdot P_4\!\downarrow 0.10$           & no effect & $-21\%$ \\
$P_3 \cdot P_4\!\uparrow 0.50$             & no effect & $+21\%$ \\
$\xi\!\downarrow 0.20$                     & no effect & $-24\%$ \\
$\xi\!\uparrow 0.60$                       & no effect & $+24\%$ \\
\bottomrule
\end{tabular}

Among the OTT-EA-specific parameters the multiplier is comparably
sensitive to $\xi$ and $P_3 \cdot P_4$ ($\approx \pm 21$--$24\%$
across their ranges), identifying empirical characterisation of the
segregation chain (both the cross-cutting fraction and the
segregation factor) as the highest-priority research target for
narrowing OTT-EA risk uncertainty.

\subsection{Treatment of dependence in the dimension transfer}

The composition above assumes independence among dimensions. If
$\delta_{\mathrm{target}}$ and $\delta_{\mathrm{EA}}$ are positively
correlated (high-targeting environments also produce more EA-specific
infrastructure), the composite multiplier is larger than the product
of marginals. Under perfectly correlated $(\delta_{\mathrm{target}},
\delta_{\mathrm{EA}})$ within their respective ranges, the upper-bound
multiplier is:

\begin{equation*}
M_{\mathrm{T}}^{\mathrm{dep}}
= \mathbb{E}[\delta_{\mathrm{target}} \cdot \delta_{\mathrm{EA}}]
  \cdot \delta_{\mathrm{concen}}
= \left(\mathbb{E}[\delta_{\mathrm{target}}]\,
  \mathbb{E}[\delta_{\mathrm{EA}}]
+ \mathrm{Cov}(\delta_{\mathrm{target}}, \delta_{\mathrm{EA}})\right)
  \delta_{\mathrm{concen}}
\leq \delta_{\mathrm{target}}^{\max} \cdot \delta_{\mathrm{EA}}^{\max}
\cdot \delta_{\mathrm{concen}}.
\end{equation*}

The Bayesian model in \S\ref{sec:bayesian} of the main paper captures this
dependence through joint priors and the parameter-uncertainty
decomposition in \S\ref{subsec:param_decomp} of the main paper.

\bigskip
\hrule
\bigskip

\section{Prior Sensitivity for Exceedance Probabilities}\label{sup:prior_sens}

The exceedance probability findings of the main paper
(\S\ref{subsec:exceedance} of the main paper) are conditional on the
primary prior calibrations adopted in \S\ref{sec:bayesian} of the main
paper. To assess robustness we re-ran the simulation under three prior
calibrations spanning the plausible range of analogical evidence for
each architectural class: a \emph{conservative} prior corresponding
to an analogue population performing substantially better than the
relevant historical record (Stream~A for T-EA, Stream~C for OTT-EA);
the \emph{primary} prior used throughout the main paper; and a
\emph{pessimistic} prior corresponding to the relevant population
average rate. Reproducibility: \texttt{prior\_sensitivity.py} in the
companion code, with multi-parameter calibration shifts as documented
in the script header.

Table~\ref{tab:exceedance_sensitivity_supp} reports the
architecture-conditional exceedance probabilities under each
calibration; all other model parameters are held at primary values.

\begin{table}[!htbp]
\centering
\small
\caption{Architecture-conditional exceedance-probability sensitivity to
prior calibration. All runs share model structure within each
architectural class; only the prior hyperparameters differ. The
``primary'' row matches the full-dependence column of
Table~\ref{tab:exceedance} of the main paper.}
\label{tab:exceedance_sensitivity_supp}
\begin{tabular}{@{}lcccccc@{}}
\toprule
& \multicolumn{3}{c}{\textbf{T-EA}} & \multicolumn{3}{c}{\textbf{OTT-EA}} \\
\cmidrule(lr){2-4} \cmidrule(lr){5-7}
\textbf{Calibration} & median & $\Pr(q_1>0.5\%)$ & $\Pr(q_1>1\%)$
& median & $\Pr(q_1>0.5\%)$ & $\Pr(q_1>1\%)$ \\
\midrule
Pessimistic   & 12.8\% & 1.000 & 1.000 &  9.8\% & 1.000 & 1.000 \\
Primary       &  4.0\% & 1.000 & 0.988 &  2.6\% & 0.994 & 0.907 \\
Conservative  &  1.2\% & 0.911 & 0.605 &  0.9\% & 0.747 & 0.460 \\
\bottomrule
\end{tabular}
\end{table}

Three observations.

\emph{(i)~Magnitude is prior-conditional.} Under the conservative prior
the exceedance probability at the 1\% threshold drops materially:
T-EA from $0.988$ to $0.605$; OTT-EA from $0.907$ to $0.460$. Under
the pessimistic prior the same probability rises to saturation
($1.000$ for both architectures). The headline exceedance
findings are therefore not sound across the full defensible prior
range in either architectural class.

\emph{(ii)~Direction is robust.} The architectural ordering (T-EA
risk at least as high as OTT-EA) is preserved across all three
calibrations: T-EA medians exceed OTT-EA medians in every case, and
T-EA exceedance probabilities are at least as large as OTT-EA's,
strictly so except under the pessimistic prior, where both
architectures saturate. This ordering is downstream of the
segregation parameter and the cross-cutting fraction calibration; it
is structural rather than calibration-conditional.

\emph{(iii)~The conservative-prior policy posture is qualitatively
distinct.} At the conservative prior the OTT-EA case crosses below
$\Pr(q_1 > 1\%) = 0.5$, supporting a characterisation of ``probably
below 1\% annually'', a policy posture qualitatively different from
that under the primary calibration. The T-EA conservative figure remains above
$0.5$, supporting the weaker characterisation ``more likely than not
above 1\% annually''. Policy readers should therefore treat the
exceedance findings as ranging from ``saturated'' (pessimistic) to
``probably below threshold for OTT-EA, more likely than not above
for T-EA'' (conservative), with the primary calibration sitting
nearer the pessimistic end of that range.

% =============================================================================
\section{Edge Log-Hazard Specification}\label{sec:supp_edge_hazards}

The Bayesian edge-hazard priors of \S\ref{sec:bayesian_priors_arch} of the main paper are specified
on a log scale. The location parameters $\mu_{ij}^{\mathrm{A}}$ are
derived via architecture-conditional domain transfer
(\S\ref{sec:supp_structural}, this supplement), and the prior widths
$\sigma_{ij}^{\mathrm{A}}$ are calibrated to reflect both
within-class operator heterogeneity and analogical transfer uncertainty.

\begin{table}[!htbp]
\centering
\small
\caption{Indicative edge log-hazard prior locations and widths under
the primary calibration, by architectural class and adversary tier.
Locations $\mu_{ij}$ are expressed in $\log$-events-per-year units;
widths $\sigma_{ij}$ are in the same units. The ${\sim}3$ log-unit
spread between adversary tiers referenced in
\S\ref{sec:cutsets} of the main paper is visible in the
opportunistic-vs-APT column gap; tier-specific edges within each
architectural class share a common width.}
\label{tab:edge_hazards}
\begin{tabular}{@{}lcccc@{}}
\toprule
& \multicolumn{2}{c}{\textbf{T-EA}} & \multicolumn{2}{c}{\textbf{OTT-EA}} \\
\cmidrule(lr){2-3} \cmidrule(lr){4-5}
\textbf{Edge class} & $\mu_{ij}$ & $\sigma_{ij}$ & $\mu_{ij}$ & $\sigma_{ij}$ \\
\midrule
Opportunistic, perimeter        & $-6.5$ & $1.5$ & $-6.0$ & $1.5$ \\
Opportunistic, key-vault        & $-7.5$ & $1.5$ & $-7.5$ & $1.5$ \\
APT, perimeter                  & $-4.5$ & $1.5$ & $-4.2$ & $1.5$ \\
APT, key-vault                  & $-5.0$ & $1.5$ & $-5.0$ & $1.5$ \\
APT, cross-cutting (OTT-only)   & ---    & ---   & $-4.5$ & $1.5$ \\
\bottomrule
\end{tabular}
\end{table}

The cross-tier spread of approximately three log units between
opportunistic and APT edges supports the weakest-link approximation of
Equation~\ref{eq:cutset_minrate} of the main paper: when adversary tiers differ by this
magnitude in log-hazard, the minimum-rate substitute for an exact
inclusion--exclusion expansion is accurate to within a few percent on
per-replicate compromise probabilities. The topology-sensitivity
analysis of \S\ref{sec:pillar23_validity} of the main paper bounds the practical
consequence of this approximation at approximately $\pm 10\%$ on
median estimates.

% =============================================================================
\section{OTT-EA Parameter Uncertainty Decomposition}\label{sec:supp_ottevpi}

The OTT-EA EVPPI-style variance decomposition referenced in the
caption of Figure~\ref{fig:evpi} of the main paper is reported here
in tabular form for reproducibility. The first-order Sobol indices
below were computed at $R = 2{,}000$ pick-freeze replicates per
parameter at seed~$2024$ over the OTT-EA configuration. The
decomposition is taken over the parameters the canonical engine treats
as uncertain, drawn from the canonical priors: the three log-normal
modifiers, together with the cross-cutting coupling ratio
$\gamma_{\mathrm{X}}/\gamma$, which is OTT-EA-specific. Parameters
held fixed in the canonical model ($\gamma$, $P_3 \cdot P_4$,
$\sigma_\lambda$, and the per-channel $\xi$ and $\rho$ arrays)
contribute no output variance and are not decomposed; their effect,
together with parameter interactions and the model's internal
channel-noise stochasticity, is captured in the residual.

\begin{table}[!htbp]
\centering
\small
\caption{First-order variance contribution of each uncertain parameter
to $\mathrm{Var}(Q_{10}^{\mathrm{OTT}})$. The OTT-EA-specific
cross-cutting coupling ratio $\gamma_{\mathrm{X}}/\gamma$ carries a
modest first-order index. The structural constraint
$\gamma_{\mathrm{X}}/\gamma \geq 1$ imposed by the prior limits the
residual uncertainty available to resolve. The residual is the largest
single component, reflecting parameter interactions and internal
stochasticity.}
\label{tab:ott_evpi}
\begin{tabular}{@{}lc@{}}
\toprule
\textbf{Parameter} & \textbf{First-order Sobol index} \\
\midrule
EA attack-surface modifier ($\delta_{\mathrm{EA}}^{\mathrm{OTT}}$)         & 14.8\% \\
Concentration ($\delta_{\mathrm{concen}}^{\mathrm{OTT}}$)                  & 13.9\% \\
System targeting intensity ($\delta_{\mathrm{target}}^{\mathrm{OTT}}$)     & 12.0\% \\
Cross-cutting coupling ratio ($\gamma_{\mathrm{X}}/\gamma$)                & 8.6\% \\
\midrule
Sum (first-order)                                                         & 49.3\% \\
Residual (interactions + internal stochasticity)                          & 50.7\% \\
\bottomrule
\end{tabular}
\end{table}

The OTT-EA first-order reducible variance ($\sim 49\%$ across the four
uncertain parameters) is below the T-EA figure ($\sim 61\%$),
reflecting the larger interaction structure and internal stochasticity
introduced by cross-cutting coupling: the OTT-EA residual ($\sim 51\%$)
is the single largest component of
$\mathrm{Var}(Q_{10}^{\mathrm{OTT}})$. The corresponding tornado plot
referenced in the main paper as the OTT-EA companion to
Figure~\ref{fig:bayesian_tornado_ott} is generated by the
companion script \texttt{bayesian\_tornado.py}.

% =============================================================================
\section{Pillar~II Monte Carlo: Prior Specification}\label{sec:supp_pillar2_priors}

The system-level Pillar~II projections of
Section~\ref{sec:pillar2_proj} of the main paper are reported as
central point projections: the independence aggregation
$\Lambda^{\mathrm{key}} = 1 - \prod_i \left(1 - \Pr(C^{\mathrm{key}}_i)\right)$
evaluated at the central per-scenario rates. The accompanying
Monte Carlo ($N = 200{,}000$ iterations; \texttt{pillar2\_mc.py})
characterises the dispersion induced by parameter uncertainty around
that central projection. It is a reduced-form propagation: rather than
re-running the per-scenario tier-mixing chain, it perturbs the central
per-scenario rates and the shared aggregation parameters directly,
drawing

\begin{itemize}
  \item the global APT-tier mixing weight
        $w_{\mathrm{APT}} \sim \mathrm{Beta}(3, 7)$ (mean $0.30$);
  \item the APT-tier exploitation uplift as a log-normal multiplier
        with median $2{\times}$ and log-scale spread $\sigma = 0.30$;
  \item each per-scenario base rate as a log-normal with median fixed
        at its central value and the same log-scale spread
        $\sigma = 0.30$;
  \item the joint segregation factor
        $P_3 \cdot P_4 \sim \mathrm{Beta}(3, 7)$ (mean $0.30$).
\end{itemize}

\noindent The per-scenario cross-cutting fractions $\xi_i$ are held at
their central values.

A log-scale spread of $\sigma = 0.30$ corresponds to a
one-standard-deviation multiplicative band of approximately
$0.74{\times}$ to $1.35{\times}$, and a $95\%$ band of approximately
$0.56{\times}$ to $1.80{\times}$, about each median. This is a
moderate prior width for an empirically uncertain rate, and it
contains the $1.2{\times}$--$3.0{\times}$ APT-uplift sensitivity
range adopted in the main paper. Because the log-normal priors are centred at the
central rates (their median) and the Beta priors at the central mixing
values (their mean), the Monte Carlo median reproduces the closed-form
central aggregation to within rounding: $\Lambda^{\mathrm{key}} = 7.4\%$
(Monte Carlo median) against $7.3\%$ (closed form), and
$\Lambda^{\mathrm{op,OTT}} = 4.0\%$ against $3.9\%$. The reported
central projections are therefore governed by the central parameter
values and not by the prior width $\sigma$, which affects only the
spread of the simulated distribution. That spread is reported through
the sensitivity analyses (the OTT-EA grid of
Section~\ref{sec:supp_ott_pillar2_grid} and the Pillar~II sensitivity
analysis in the main paper) and is not quoted as a confidence
interval on the central projection.

% =============================================================================
\section{Pillar II OTT-EA Sensitivity Grid}\label{sec:supp_ott_pillar2_grid}

The OTT-EA Pillar~II projection is sensitive to two architectural
parameters with no T-EA analogue: the segregation factor
$P_3 \cdot P_4$ and a global multiplier $m_\xi$ on the central
$\xi_i$ values from Table~\ref{tab:xi_scenarios} of the main paper. The full joint grid
is reported here.

\begin{table}[!htbp]
\centering
\small
\caption{OTT-EA Pillar II central operational compromise projection
$\Lambda^{\mathrm{op},\mathrm{OTT}}_{\mathrm{central}}$ across the
joint $(P_3 \cdot P_4, m_\xi)$ sensitivity grid. The central
parameterisation ($P_3 \cdot P_4 = 0.30$, $m_\xi = 1.0$) gives the
central value of $3.9\%$. Each cell: independence-based closed-form
projection.}
\label{tab:supp_ott_pillar2_sens}
\begin{tabular}{@{}lcccc@{}}
\toprule
& $m_\xi = 0.5$ & $m_\xi = 0.75$ & $m_\xi = 1.0$ & $m_\xi = 1.5$ \\
\midrule
$P_3 \cdot P_4 = 0.12$ & 2.0\% & 2.5\% & 3.0\% & 4.1\% \\
$P_3 \cdot P_4 = 0.30$ & 3.1\% & 3.5\% & 3.9\% & 4.7\% \\
$P_3 \cdot P_4 = 0.50$ & 4.3\% & 4.6\% & 4.9\% & 5.5\% \\
$P_3 \cdot P_4 = 0.64$ & 5.1\% & 5.3\% & 5.6\% & 6.0\% \\
\bottomrule
\end{tabular}
\end{table}

The OTT-EA Pillar~II projection spans $2.0$--$6.0\%$ across the full
joint sensitivity range. The lower portion corresponds to architectures
with strong segregation, low cross-cutting risk, and effective
multi-stage detection (the bottom-left of the grid). The upper portion
corresponds to architectures where cross-cutting access is prevalent
and detection is weak. The interval $[3.5\%, 4.3\%]$ obtained by
varying $m_\xi \in [0.75, 1.25]$ at the central $P_3 \cdot P_4$ value
is the defensible central range conditional on the calibration.

% =============================================================================
\section{Pillar III OTT-EA Sensitivity Grid}\label{sec:supp_ott_pillar3_grid}

The OTT-EA channel-aggregate (Pillar~III) projection is sensitive to
the same two clusters of parameters as Pillar~II: the segregation
factor $P_3 \cdot P_4$ and the per-channel cross-cutting fractions
$\xi_c$ (multiplied by a global $m_\xi$). The grid analogous to
\S\ref{sec:supp_ott_pillar2_grid} but applied to the Pillar~III
decomposition is:

\begin{table}[!htbp]
\centering
\small
\caption{OTT-EA Pillar III channel-aggregate projection
$\Pref^{\mathrm{OTT}}$ across the joint $(P_3 \cdot P_4, m_\xi)$
sensitivity grid, where $m_\xi$ multiplies the central $\xi_c$ values
from Table~\ref{tab:xi_channels} of the main paper. The central parameterisation gives
$3.0\%$.}
\label{tab:supp_ott_pillar3_sens}
\begin{tabular}{@{}lcccc@{}}
\toprule
& $m_\xi = 0.5$ & $m_\xi = 0.75$ & $m_\xi = 1.0$ & $m_\xi = 1.5$ \\
\midrule
$P_3 \cdot P_4 = 0.12$ & 1.5\% & 1.9\% & 2.3\% & 3.2\% \\
$P_3 \cdot P_4 = 0.30$ & 2.3\% & 2.6\% & 3.0\% & 3.6\% \\
$P_3 \cdot P_4 = 0.50$ & 3.2\% & 3.4\% & 3.7\% & 4.1\% \\
$P_3 \cdot P_4 = 0.64$ & 3.8\% & 4.0\% & 4.2\% & 4.5\% \\
\bottomrule
\end{tabular}
\end{table}

The OTT-EA Pillar~III aggregate spans $1.5$--$4.5\%$ across the full
joint sensitivity range, with the same architectural interpretation as
the Pillar~II grid. The defensible central interval at intermediate
$\xi$-multiplier values is approximately $[2.6\%, 3.3\%]$.

A separate sensitivity examines an alternative calibration in which the
OTT channel base rates exceed the T-EA rates (doubling $\Pzeroday$ and
$\Psupply$ from $1.5\%$ to $3.0\%$ each, on the argument that platform
codebase and supply-chain complexity exceed those of T-EA
infrastructure). This gives a Pillar~III channel-aggregate projection
of approximately $4.7\%$ at the central segregation calibration,
narrowing the gap to T-EA. This calibration is presented as a
sensitivity rather than the headline because the empirical evidence on
whether platform zero-day and supply-chain rates exceed those of T-EA
infrastructure is mixed, and because the channel-minimum heuristic
argument of Pillar~III is most defensible when channel rates are taken
at minimum-viable values applicable across architectural classes.

% =============================================================================
\section{Detailed EVPPI Variance Decomposition}\label{sec:supp_evppi}

This section reports the detailed parameter contributions to the
EVPPI-style variance decomposition summarised in the main paper. The
decomposition is taken over the parameters the canonical Monte Carlo
engine treats as uncertain, drawn from the canonical priors: the three
log-normal modifiers ($\delta_{\mathrm{EA}}$, $\delta_{\mathrm{target}}$,
$\delta_{\mathrm{concen}}$) for both architectures, together with the
cross-cutting coupling ratio $\gamma_{\mathrm{X}}/\gamma$ for OTT-EA.
Quantities held fixed in the canonical model ($\gamma$, $P_3 \cdot P_4$,
$\sigma_\lambda$, and the per-channel $\xi$/$\rho$ arrays) contribute no
output variance and are not decomposed. The first-order reducible
totals are $\sim 61\%$ for T-EA (three modifiers) and $\sim 49\%$ for
OTT-EA (three modifiers plus $\gamma_{\mathrm{X}}/\gamma$). The residual
reflects parameter interactions and the model's internal channel-noise
stochasticity.

\textit{T-EA contributions to $\mathrm{Var}(Q_{10}^{\mathrm{T}})$.}
System targeting intensity ($\delta_{\mathrm{target}}^{\mathrm{T}}$,
$25.0\%$), EA-specific attack-surface modifier
($\delta_{\mathrm{EA}}^{\mathrm{T}}$, $18.7\%$), and concentration
($\delta_{\mathrm{concen}}^{\mathrm{T}}$, $17.7\%$). The three
contributors are comparable in magnitude, reflecting a relatively flat
decomposition where no single parameter dominates. The residual
($\sim 39\%$) absorbs parameter interactions and the model's internal
stochasticity.

\textit{OTT-EA contributions to $\mathrm{Var}(Q_{10}^{\mathrm{OTT}})$.}
EA-specific attack-surface modifier
($\delta_{\mathrm{EA}}^{\mathrm{OTT}}$, $14.8\%$), concentration
($\delta_{\mathrm{concen}}^{\mathrm{OTT}}$, $13.9\%$), system targeting
intensity ($\delta_{\mathrm{target}}^{\mathrm{OTT}}$, $12.0\%$), and the
cross-cutting coupling ratio ($\gamma_{\mathrm{X}}/\gamma$, $8.6\%$).
The three modifiers are again comparable in magnitude; the
OTT-EA-specific $\gamma_{\mathrm{X}}/\gamma$ carries a smaller
first-order index, in part because the structural constraint
$\gamma_{\mathrm{X}}/\gamma \geq 1$ enforced at the prior level leaves
relatively little remaining uncertainty to resolve, so additional
empirical pinning of $\gamma_{\mathrm{X}}/\gamma$ yields modest direct
variance reduction. Its role in the architectural-divergence direction
is carried principally by the constraint itself rather than by residual
uncertainty about its value. The residual ($\sim 51\%$) is the single
largest component, larger than for T-EA, reflecting the additional
interaction structure introduced by cross-cutting coupling.

% =============================================================================
\section{Attack Graph Topology Sensitivity}\label{sec:supp_graph_sens}

Table~\ref{tab:supp_graph_sensitivity} reports the sensitivity of the
Bayesian Layer output to the attack-graph topology, comparing the
canonical six-node, eight-edge graph against the seven-node variant
that inserts an HSM-firmware traversal node between the
hardware-key-store entry and the key-extraction sink (the additional
serial gate is modelled as a Beta-distributed firmware-pass probability
attenuating the zero-day channel hazard). All other parameters are held
at their central calibration. The figures are reproduced by
\texttt{graph\_topology\_sensitivity.py} ($R = 4{,}000$ replicates per
architecture and topology, seed~2024).

\begin{table}[!htbp]
\centering
\small
\caption{Bayesian Layer outputs under the six-node and seven-node
attack-graph topologies, for both architectures. The more granular
seven-node topology lowers the median system hazard by $9$--$11\%$; the
exceedance probability $\Pr(q_1 > 1\%)$ is nearly unchanged.}
\label{tab:supp_graph_sensitivity}
\begin{tabular}{@{}lcccc@{}}
\toprule
& \multicolumn{2}{c}{\textbf{T-EA}} & \multicolumn{2}{c}{\textbf{OTT-EA}} \\
\cmidrule(lr){2-3}\cmidrule(lr){4-5}
\textbf{Quantity} & 6-node & 7-node & 6-node & 7-node \\
\midrule
Median $q_1$              & 4.1\%  & 3.6\%  & 2.7\%  & 2.4\%  \\
95th percentile $q_1$     & 17.6\% & 15.2\% & 17.6\% & 16.6\% \\
Median $Q_{10}$           & 36.2\% & 32.9\% & 31.6\% & 28.6\% \\
95th percentile $Q_{10}$  & 84.1\% & 80.4\% & 93.2\% & 91.7\% \\
$\Pr(q_1 > 1\%)$          & 0.989  & 0.981  & 0.901  & 0.868  \\
\bottomrule
\end{tabular}
\end{table}

The seven-node topology lowers the median annual hazard by
approximately $11\%$ for T-EA and $10\%$ for OTT-EA, and the median
10-year cumulative hazard by approximately $9\%$ for both; the
exceedance probability at the $1\%$ threshold is nearly unchanged.
The direction is favourable, since the coarser six-node graph
\emph{over}states risk relative to the more granular topology, and
the qualitative dependence-inflation findings are unaffected.

% =============================================================================
\section{Pillar II Sensitivity Figures}\label{sec:supp_pillar2_sens_figs}

The two-way sensitivity grid and one-way tornado for the Pillar~II
central projection, referenced from the main paper, are reproduced
here.

\begin{figure}[!htbp]
\centering
\includegraphics[width=\textwidth]{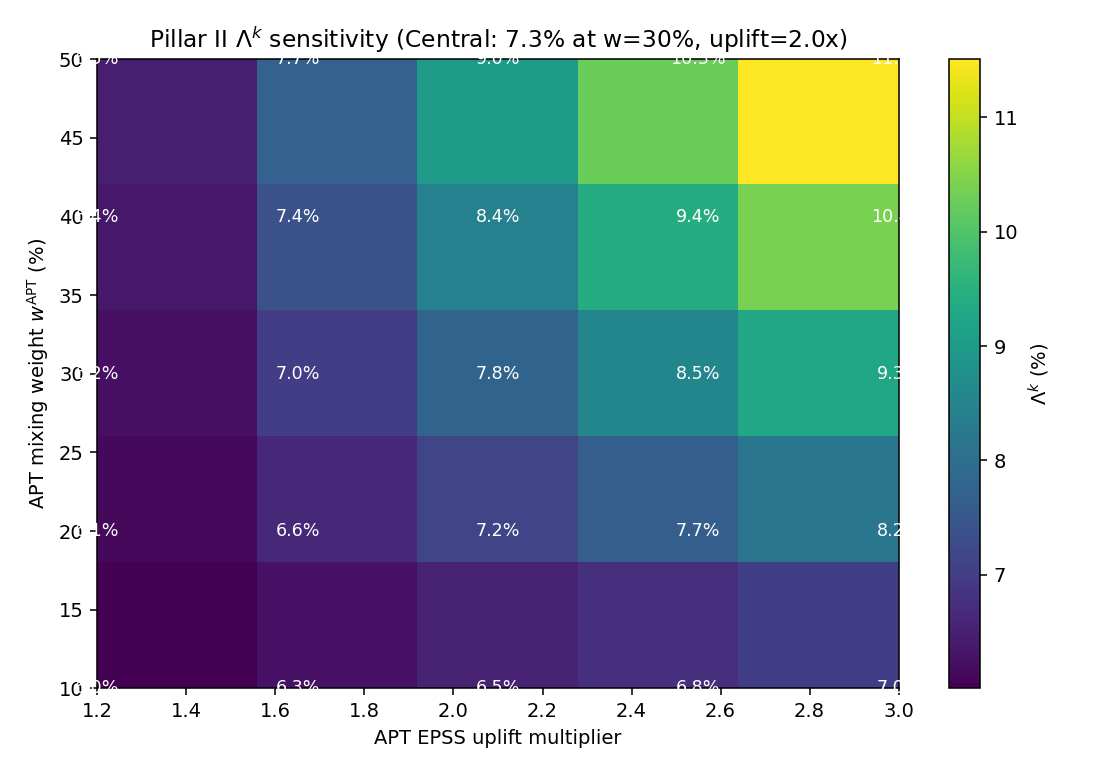}
\caption{Two-way sensitivity grid for the Pillar~II central projection,
varying the APT mixing weight $w_{\rm APT}$ (x-axis, 10--50\%) against
the APT EPSS uplift multiplier (y-axis, 1.2$\times$--3.0$\times$). Left
panel: mean annual key-material exfiltration probability. Right
panel: mean annual access-pathway compromise probability. The central
parameterisation ($w_{\rm APT} = 0.30$, uplift $\times 2.0$) yields
the headline $7.3\%$/$16.8\%$ values. Across the full grid the
key-material projection spans $5.5\%$--$7.8\%$, with $5.5\%$ serving
as the Pillar~II parametric floor under independence. Each
cell: $200{,}000$ Monte Carlo iterations; seed 2024.}
\label{fig:supp_sensitivity_grid}
\end{figure}

\begin{figure}[!htbp]
\centering
\includegraphics[width=\textwidth]{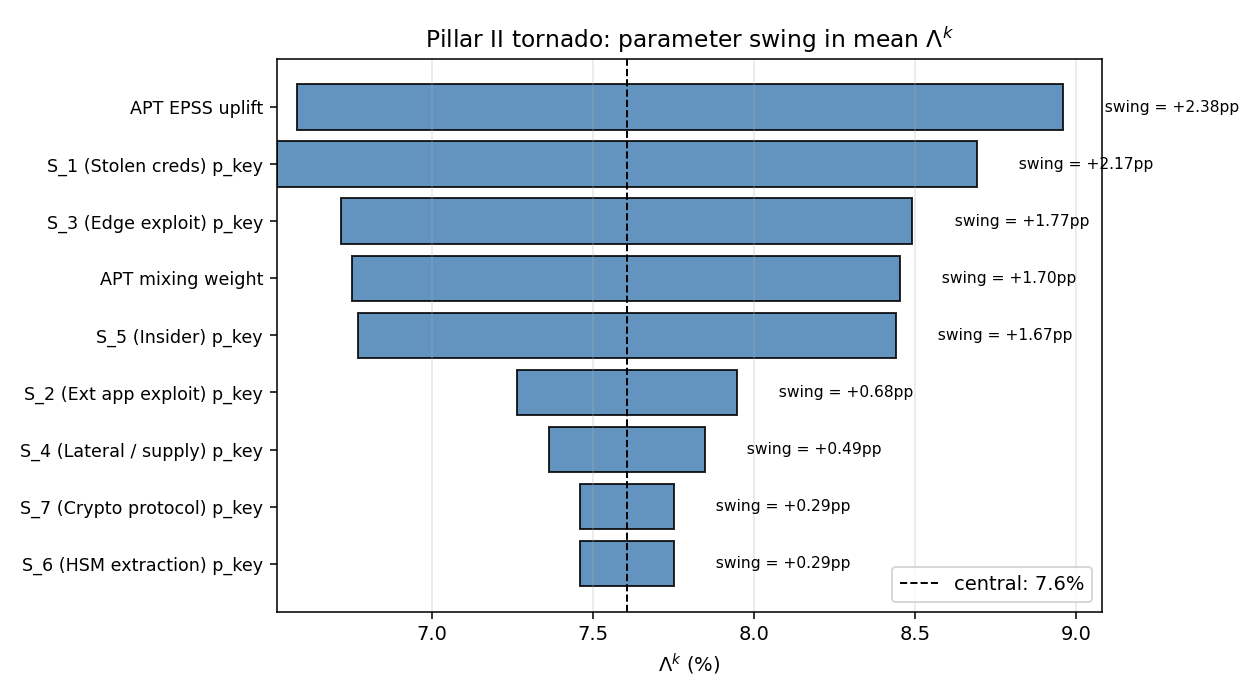}
\caption{One-way sensitivity tornado for the Pillar~II central
projection of mean annual key-material exfiltration probability. Each
parameter is varied across its prior 10th--90th percentile range
holding all other parameters at central values. Vertical line marks
the central estimate ($7.3\%$). Dominant drivers are the conditional
key-material probabilities for $S_1$, $S_3$, $S_5$, and the global
APT mixing weight; campaign parameters and specialist-technique
attempt rates contribute less than $0.4$~pp of range each.}
\label{fig:supp_tornado}
\end{figure}

% =============================================================================
\section{Non-stationary Cumulative Trajectory}\label{sec:supp_nonstationary}

Figure~\ref{fig:supp_nonstationary_trajectory} extends the
stationarity sensitivity from \S\ref{sec:limitations_stationarity}
to a 25-year deployment horizon across three growth
scenarios for both architectures.

\begin{figure}[!htbp]
\centering
\includegraphics[width=0.95\textwidth]{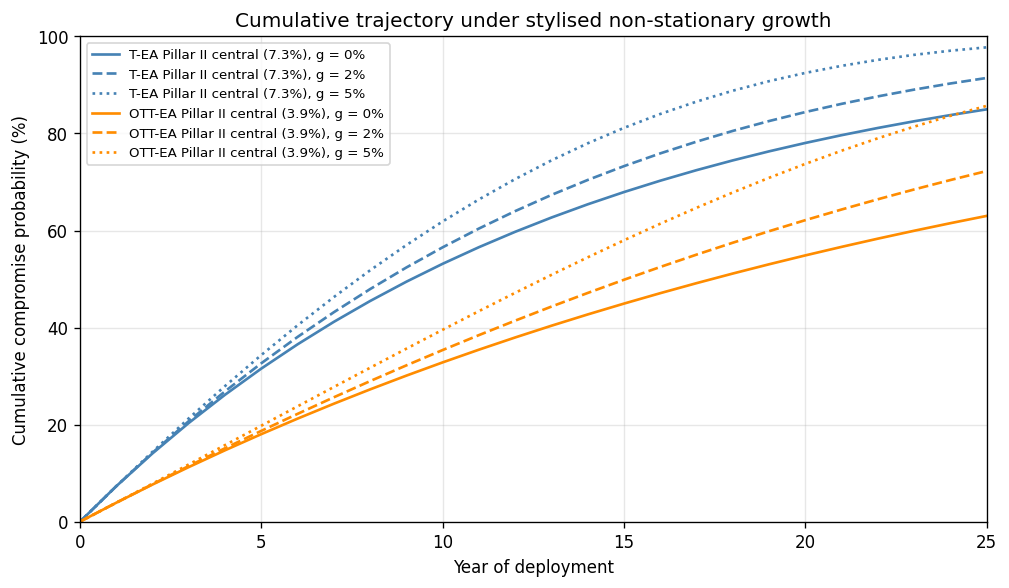}
\caption{Cumulative compromise probability over a 25-year deployment
horizon under stylised non-stationary growth in the annual rate.
Solid lines: stationary baseline ($g = 0$). Dashed: 2\% annual
relative rate growth. Dotted: 5\% growth. T-EA in blue (starting from
the Pillar~II central estimate $7.3\%$); OTT-EA in orange (starting
from $3.9\%$). The gap to the stationary projection is modest in the
first decade but grows substantially beyond year~15, illustrating why
multi-decade EA-system commitments are more sensitive to assumed
rate-growth dynamics than the headline 10-year projections suggest.}
\label{fig:supp_nonstationary_trajectory}
\end{figure}

% =============================================================================
% =============================================================================
\section{Threshold-Scheme Quantitative Contrast}\label{sec:supp_threshold}

The quantitative analysis in this paper covers two architectural classes
under single-operator deployment: T-EA (transmission-layer, with
$\mathcal{K}$ and $\mathcal{D}$ co-located in carrier infrastructure) and
OTT-EA (application-layer, with $\mathcal{K}$ and $\mathcal{D}$
architecturally segregated within a global platform). Both share the
property that key material is held under unitary operator authority. A
distributed threshold scheme, in which key material is split across $N$
independent trustees with a $t$-of-$N$ reconstruction requirement,
alters the risk profile by removing the unitary-operator
property itself. This subsection presents an illustrative quantitative
contrast under the strong caveat that the per-node compromise rate is
borrowed from the T-EA Pillar~III channel-aggregate projection rather
than measured empirically for threshold-scheme nodes specifically.
Threshold-scheme nodes have a different attack surface from the
single-operator T-EA cohort (narrower scope per node, distinct trust
relationships, different operational tempos), and a quantitative risk
analysis of threshold-based EA deployment would require an empirical
anchor distinct from the three streams developed in this paper. The
numbers below are presented as an order-of-magnitude
illustration of the architectural-choice effect, not as a calibrated
projection.

Threshold-based EA proposals have a substantial history. Micali's
\emph{Fair Public-Key Cryptosystems} \citep{micali1992fair} introduced the
verifiable secret-sharing approach to key escrow, splitting user private keys
across $n$ trustees with a $k$-of-$n$ reconstruction threshold. Bellare and
Goldwasser's \emph{Verifiable Partial Key Escrow} \citep{bellare1997verifiable}
and the contemporaneous Failsafe Key Escrow scheme of Leighton
\citep{leighton1997failsafe} extend this approach with security-game
reductions. The Denning--Branstad taxonomy \citep{denning1996taxonomy}
catalogues over a dozen historical key-escrow proposals, many of them based on
threshold reconstruction. More recent
deployment-grade systems, including Apple's Cloud Key Vault and the iCloud
Keychain HSM-quorum recovery scheme, use related multi-trustee
reconstruction architectures. The threshold-scheme analysis below
addresses a real and historically dominant class of EA proposals, not a
hypothetical alternative.

The probability of key reconstruction in a given year is:
\begin{equation}
P(\text{reconstruct} \mid \text{independence})
  = 1 - \sum_{k=0}^{t-1} \binom{N}{k} p^k (1-p)^{N-k},
\label{eq:threshold}
\end{equation}
where $p$ is the per-node annual compromise probability. A 5-of-9 scheme with
$p = 0.05$ (close to the T-EA Pillar~III channel-aggregate projection of $5.4\%$)
gives $P(\text{reconstruct}) = 3.3 \times
10^{-5}$, approximately three orders of magnitude below the centralised
channel-aggregate projection. Incorporating a common-mode failure fraction $\rho \in [0.05, 0.10]$
(shared HSM vendors, operating systems, coordinated nation-state campaigns):
\begin{equation}
P(\text{reconstruct}) = \rho p + (1 - \rho) \cdot P_{\mathrm{indep}}(t, N, p).
\label{eq:threshold_commonmode}
\end{equation}
At $\rho = 0.10$, $p = 0.05$, this gives $\approx 0.5\%$ annually.

Table~\ref{tab:threshold} illustrates how the threshold-scheme reconstruction
probability varies across a range of per-node rates, since the appropriate value
of $p$ for individual threshold scheme nodes is uncertain.

\begin{table}[!htbp]
\centering
\small
\caption{Threshold-scheme reconstruction probability for a 5-of-9 architecture
across a range of per-node annual compromise rates $p$ and common-mode
correlation $\rho$. The $p = 0.05$ row uses the Pillar III centralised-system
channel-aggregate projection as an upper bound for the per-node rate; actual node rates may be
lower (narrower attack surface per node) or require separate analysis.}
\label{tab:threshold}
\begin{tabular}{@{}lcccc@{}}
\toprule
& $\rho = 0$ & $\rho = 0.05$ & $\rho = 0.10$ & $\rho = 0.20$ \\
\midrule
$p = 0.02$ & 0.000\% & 0.10\% & 0.20\% & 0.40\% \\
$p = 0.05$ & 0.003\% & 0.25\% & 0.50\% & 1.00\% \\
$p = 0.10$ & 0.089\% & 0.59\% & 1.08\% & 2.07\% \\
\bottomrule
\end{tabular}
\end{table}

A contour plot of reconstruction probability across the joint
$(p, \rho)$ range, with the Pillar~III T-EA and OTT-EA channel-aggregate
projections marked, is reproduced in Supplementary
\S\ref{sec:supp_threshold_contour}.

Two caveats apply. First, the per-node rate $p$ is taken from the Pillar III
T-EA channel-aggregate projection, which is a conservative upper estimate for
individual nodes of a threshold scheme, since each node holds only a share of
the key material and presents a correspondingly smaller target. Precise
threshold-scheme node rates require separate analysis. Second, the
common-mode correction captures systematic failures. Cascading failures from
shared cryptographic libraries or shared network infrastructure may not be
fully captured by the scalar $\rho$ model. The threshold scheme analysis
demonstrates that the quantitative risk profile of this paper applies
specifically to single-operator deployments (both the T-EA and OTT-EA
classes analysed quantitatively) and can be substantially altered
through architectural choice that removes the unitary-operator
property. Precise figures for distributed architectures require
independent analysis.

\paragraph{Operational feasibility considerations.} The mathematical
reduction in independent-compromise probability does not by itself
establish that a threshold scheme is a deployable EA architecture.
The threshold-key-escrow design space has been studied in cryptographic
work since the mid-1990s
\citep{micali1992fair,bellare1997verifiable,leighton1997failsafe} and
the underlying secret-sharing machinery is in production use for
account-recovery in commercial key-management systems. Four
operational questions condition the engineering viability of a
lawful-intercept deployment specifically: (i)~warrant throughput,
since each reconstruction event requires participation from the
threshold quorum of escrow holders and the coordination overhead
scales with the number of authorisations per unit time, (ii)~escrow-holder
selection, since the security property requires that the $N - t + 1$
non-quorum holders are non-collusive in expectation, which constrains
selection to entities with non-overlapping institutional incentives and
oversight regimes, (iii)~cross-jurisdictional handling, since
multi-jurisdiction quorum requirements interact with mutual legal
assistance and data-localisation regimes in ways that may not match the
operational tempo of lawful-intercept work, and (iv)~revocation and
rotation cost, since periodic re-sharing of secrets to retire
compromised or departing escrow holders multiplies the per-period
operational burden. The cryptographic mechanisms exist; what is absent
in the public literature is a worked assessment of these four
operational questions in the lawful-intercept context. The risk
reduction is unambiguous. The operational cost of obtaining it is the
empirical question that follow-on engineering work would need to
resolve.

% =============================================================================
\section{Threshold-Scheme Reconstruction Contour}\label{sec:supp_threshold_contour}

Figure~\ref{fig:supp_threshold_scheme_contour} reproduces the
threshold-scheme reconstruction probability surface referenced in
\S\ref{sec:threshold} of the main paper. The contour visualisation
complements the tabulated values in Table~\ref{tab:threshold} of the
main paper.

\begin{figure}[!htbp]
\centering
\includegraphics[width=0.85\textwidth]{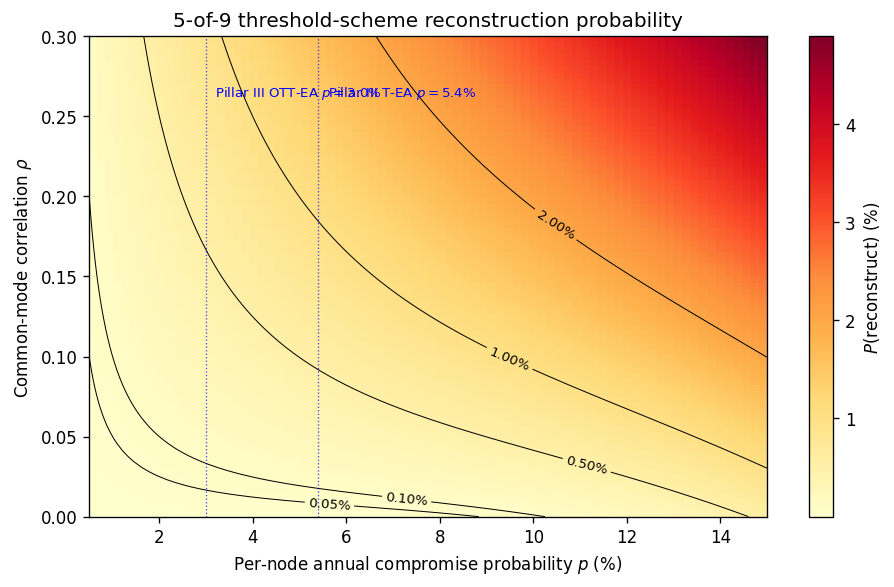}
\caption{Reconstruction probability for a 5-of-9 distributed
threshold scheme as a function of the per-node annual compromise
rate $p$ (horizontal axis) and the common-mode correlation $\rho$
(vertical axis). Vertical dotted lines mark the Pillar~III T-EA and
OTT-EA channel-aggregate projections ($p \approx 5.4\%$ and
$p \approx 3.0\%$ respectively) as plausible upper-bound per-node
rates. Contour lines at $0.05\%$, $0.10\%$, $0.50\%$, $1.00\%$, and
$2.00\%$ reconstruction probability. The flat lower-left region
illustrates the order-of-magnitude protection that threshold schemes
provide against independent compromise. The steep gradient towards
the upper-right shows the sensitivity to common-mode failure even at
small $\rho$.}
\label{fig:supp_threshold_scheme_contour}
\end{figure}

% =============================================================================
% =============================================================================
\section{Hybrid Client-Side Scanning Architectures}\label{sec:hybrid_css}

Hybrid client-side scanning (CSS) architectures, in which detection or
content-matching is performed on user devices using operator-distributed
key material, are treated qualitatively in this paper rather than
quantitatively. The reasons for the qualitative treatment, and the
structural risk-profile observations that follow from it, are as follows.

\textit{Why CSS is not quantified in this approach.} The quantitative
analysis of T-EA and OTT-EA depends on architecturally appropriate
empirical anchors (Stream~A and Stream~C respectively) that have been
calibrated against documented incident bases. CSS architectures combine
operator-side key custody (similar to OTT-EA) with endpoint-distributed
detection key material whose attack surface is governed by mobile-platform
security properties \citep{kulshrestha2021identifying,abelson2024bugs}.
The Pillar~I empirical anchors do not contain incidents architecturally
matched to the endpoint-distribution component: the Stream~C platform
incidents involve operator-side compromise rather than endpoint
compromise of distributed keys. Quantitative analysis of CSS would
require an additional empirical anchor for endpoint-key compromise rates
that is outside the calibrated set of this paper.

\textit{Structural risk observations applicable to CSS.} Several findings
of this paper apply to CSS by structural argument:

\begin{enumerate}
  \item The stochastic dominance result
    (Proposition~\ref{prop:dominance} of the main paper) extends to CSS: any CSS deployment
    necessarily increases attack surface relative to the non-CSS
    counterfactual, attracts elevated targeting (the operator-distributed
    detection keys are themselves a high-value adversarial target), and
    concentrates risk at the operator-side detection-list custody
    component. The $\delta^{\mathrm{CSS}} \geq 1$ constraint applies to
    each modifier on structural grounds.
  \item The cross-cutting coupling argument (\S\ref{subsec:dependence} of the main paper)
    applies to CSS in modified form. CSS architectures introduce a new
    cross-cutting infrastructure, the mobile platform itself, which is
    shared between the detection key custody and the endpoint scanning
    function. A campaign compromising the platform vendor's signing
    infrastructure (Storm-0558-class) potentially compromises all CSS
    detection logic on devices simultaneously. The cross-cutting fraction
    for CSS architectures is therefore plausibly larger than for OTT-EA,
    further compressing any segregation gain.
  \item The endpoint distribution introduces new compromise modes not
    present in T-EA or OTT-EA: device-level malware affecting detection
    logic, platform-level update infrastructure compromise, and
    user-level evasion or detection-trigger abuse. These modes are
    documented qualitatively in the cryptographic literature
    \citep{abelson2024bugs} but are not parameterised in our scenario set.
  \item The architecture-conditional Bayesian framework's tail finding
    applies a fortiori to CSS: under correlated-campaign conditions,
    the larger cross-cutting attack surface produces larger tail
    inflation than for OTT-EA. The CSS Bayesian distribution, were it
    constructed, would plausibly exhibit a heavier upper tail than
    either T-EA or OTT-EA.
\end{enumerate}

\textit{Policy implication.} The qualitative treatment of CSS is not
neutral on its risk profile relative to T-EA and OTT-EA. Each of the
four observations above is unfavourable to CSS, and there are no
corresponding architectural arguments running in the opposite direction.
A precautionary policy posture is to treat CSS as carrying \emph{at
least} the OTT-EA risk profile, with material additional risk from the
endpoint-distribution attack surface. Quantitative analysis of CSS to
match the rigour applied to T-EA and OTT-EA in this paper is identified
as a research priority.

% =============================================================================
\section{Temporal Robustness, Inclusion Sensitivity, and Operator Heterogeneity}\label{sec:supp_robustness}

\subsection{Temporal and Inclusion Robustness}\label{sec:limitations_stationarity}

\paragraph{Stationarity.} The cumulative calculations assume a stationary
annual compromise rate over the deployment horizon. Three structural
trends suggest the rate is more likely to increase than decrease:
adversary capability growth, defensive posture degradation over multi-
decade infrastructure lifetimes, and attack-surface expansion as EA
systems become institutionally identified targets. A stylised
non-stationary model with 2\% relative annual rate growth raises the
10-year cumulative from 53\% to 56\% starting from the 7.3\% Pillar~II
projection: a modest but directionally consistent adjustment. The full
25-year trajectory across three growth scenarios ($g = 0, 2\%, 5\%$)
for both architectures is reported in \S\ref{sec:supp_nonstationary}. Conversely, technical improvements
(mandatory FIPS~140-3 Level~4 HSM standards, air-gapped key ceremony
requirements, or improved detection capabilities) could reduce the
rate over time. The framework does not foreclose this possibility.

\paragraph{Inclusion-criterion sensitivity for Stream~A.} The
tier-stratified rates (Table~\ref{tab:tiered} and \S\ref{sec:pillar1_streamA} of the main paper)
characterise the inclusion-criterion sensitivity. The Tier~1 cohort
(direct T-EA architectural analogues, $k = 3$) at $10{\times}$ premium
projects to $4.5\%$, still inside the low single-digit annual range
produced by the framework. The qualitative finding (T-EA compromise
risk in the low single-digit percentage range) is robust to inclusion
criterion across all three tiers.

\paragraph{Salt Typhoon role.}\label{sec:salt_typhoon_role}
A separate concern is the inclusion of Salt Typhoon, which
both motivates the paper and enters the parameter calibration. The
$k = 5$ result with Salt Typhoon removed ($\hat{P} = 31.9\%$, 95\% CI
$[11.7\%, 59.2\%]$) is materially consistent with $k = 6$, so the
motivational use does not inflate the headline estimate.

\subsection{Operator Heterogeneity and Regulatory Baseline}\label{sec:operator_heterogeneity}

\paragraph{Stream~C operator cohort and regulatory counterfactual.}
The Stream~C platform-incident cohort (\S\ref{sec:pillar1_streamC} of the main paper)
comprises operators (Microsoft, Okta, LastPass, Twilio, and others)
operating under commercial and industry-framework standards (SOC~2,
ISO~27001) rather than any EA-style regulatory mandate. A defender of
OTT-EA would argue that this systematically biases the Stream~C base
rate upward: under an EA-mandate counterfactual, regulator-defined
minimum standards (mandatory hardware-backed key custody, formal audit
trails, certified incident response) would force defensive postures
that several Stream~C incidents specifically reflect the absence of.
Stream~B (Certificate Authorities) is the closest empirical analogue
to a ``regulated EA operator'': CAs operate under formal CCADB
governance with mandatory audit, transparency-log, and incident-
disclosure obligations comparable to what an EA mandate would impose.
The Stream~B per-system rate ($0.45\%$/yr) is approximately half the
Stream~C rate ($0.95\%$/yr), and this ratio is the best available
empirical upper bound on the defensive uplift attainable through
CA-equivalent regulation. A defender accepting this ratio as evidence
could legitimately discount the Stream~C base by a factor of
approximately two, shifting the OTT-EA central projection from
$\sim$$3.9\%$ to $\sim$$2.0\%$. The discounted figure remains material
and above the channel-minimum heuristic, but sits at the lower end of
the defensible range.
The Stream~C cohort is therefore best read as anchoring the
\emph{unregulated} platform base rate. The OTT-EA-mandated case lies
between this rate and the Stream~B regulated-operator rate, with the
empirical evidence insufficient to pin its precise location.

\paragraph{Weakest-operator problem.} EA mandates are typically
drafted to apply across all qualifying operators within a jurisdiction.
The systemic risk profile is determined by the
\emph{distribution} of security capabilities across the operator
population, not by the performance of the best-resourced operator.
The model's analysis corresponds to a well-resourced operator
matching best-in-class historical performance. A mandate that includes
less well-resourced operators will face a higher systemic risk profile
than the framework's central projections suggest.

% =============================================================================
% =============================================================================
\section{Proof of the Stochastic Dominance Proposition}\label{sec:supp_dominance_proof}

\begin{proof}[Proof of Proposition~\ref{prop:dominance} of the main paper]
By definition the no-EA counterfactual of class $\mathrm{A}$ is the
same model on the same attack graph with the three modifiers fixed at
unity, so every edge $(i,j)$ carries the baseline hazard
$\lambda_{ij}^{\mathrm{A}}$. Assumption (ii) gives
$\delta_{\mathrm{EA}}^{\mathrm{A}} \delta_{\mathrm{target}}^{\mathrm{A}}
\delta_{\mathrm{concen}}^{\mathrm{A}} \geq 1$, so under the EA
configuration the modified hazard of
Equation~\eqref{eq:modified_hazard} of the main paper satisfies
$\tilde\lambda_{ij}^{\mathrm{A}} \geq \lambda_{ij}^{\mathrm{A}}$ on
\emph{every} edge. Assumption (i) ensures the EA graph indeed
carries the EA-specific edges on which this uplift acts, but the
inequality is driven by assumption (ii) alone. Because the
per-period compromise probability $1 - e^{-\Lambda}$ is increasing in
the total hazard $\Lambda$, it follows that
$q_k^{\mathrm{EA},\mathrm{A}} \geq q_k^{\mathrm{noEA},\mathrm{A}}$ for
each year $k$. For T-EA, this directly gives
$Q_T^{\mathrm{EA},\mathrm{T}} \geq Q_T^{\mathrm{noEA},\mathrm{T}}$ for
all $T$. For OTT-EA, the operational-compromise probability of
Equation~\eqref{eq:q_OTT} of the main paper is monotonic in each component, so the same
conclusion holds: $Q_T^{\mathrm{EA},\mathrm{OTT}} \geq
Q_T^{\mathrm{noEA},\mathrm{OTT}}$ for all $T$. Since the inequalities
hold for every parameter realisation, they hold distributionally.
\end{proof}

% =============================================================================
\section{Full EA-Modifier Prior Specifications}\label{sec:supp_modifier_priors}

This section gives the full prior specifications for the three
EA-system modifiers
$(\delta_{\mathrm{EA}}, \delta_{\mathrm{target}}, \delta_{\mathrm{concen}})$
summarised in \S\ref{sec:bayesian_priors_arch} of the main paper.

\paragraph{Correction relative to v1.} The first version of this paper
specified these modifiers as shifted lognormals,
$\delta_k = 1 + \mathrm{LogNormal}(\log(m_k - 1), \sigma_k^2)$, with
architecture-specific medians. That specification was never
implemented. Every published result was produced by the log-censored
form given below, which is a different distribution with different
medians and no architectural asymmetry. Substituting the v1 shifted
priors into the simulator moves the T-EA median annual risk from
$4.0\%$ to $9.9\%$ and the OTT-EA median from $2.6\%$ to $17.4\%$,
reversing the central architectural ordering. This section now
documents the implemented model. The claim of architecture-dependent
modifier asymmetry is withdrawn, and \texttt{prior\_specification\_check.py}
in the archive asserts agreement between this table and the code.

\paragraph{Prior family.} Each modifier is a log-censored lognormal,
\begin{equation}
\delta_k = \exp\!\left(\max\!\left(\mathcal{N}(\mu_k, \sigma_k^2),\ 0\right)\right),
\quad k \in \{\mathrm{EA}, \mathrm{target}, \mathrm{concen}\},
\label{eq:supp_modifier_prior}
\end{equation}
identical for both architectural classes. Censoring the log at zero
enforces $\delta_k \geq 1$, which is what
Proposition~\ref{prop:dominance} requires. Three properties of this
family need stating, because they are not shared with the shifted
lognormal and they bear on how the results should be read.

First, the distribution has an atom at exactly $\delta_k = 1$, of mass
$\Phi(-\mu_k / \sigma_k)$. Second, the prior median is $1$ whenever
$\mu_k \leq 0$, so two of the three modifiers have median exactly
unity. Third, on the atom the EA system contributes no additional
hazard, so the stochastic dominance of
Proposition~\ref{prop:dominance} is weak, with a point mass of
equality, rather than strict.

\paragraph{Implemented parameters.} These apply to T-EA and OTT-EA
alike.

\begin{center}
\begin{tabular}{@{}lccccc@{}}
\toprule
Modifier & $\mu_k$ & $\sigma_k$ & median & mean & $\Pr(\delta_k = 1)$ \\
\midrule
$\delta_{\mathrm{EA}}$     & $0.00$ & $0.55$ & $1.000$ & $1.322$ & $0.500$ \\
$\delta_{\mathrm{target}}$ & $0.10$ & $0.60$ & $1.106$ & $1.464$ & $0.433$ \\
$\delta_{\mathrm{concen}}$ & $0.00$ & $0.55$ & $1.000$ & $1.324$ & $0.500$ \\
\bottomrule
\end{tabular}
\end{center}

The joint modifier $\delta_{\mathrm{EA}} \delta_{\mathrm{target}}
\delta_{\mathrm{concen}}$ has median $1.89$, mean $2.57$, and mass
$0.108$ at exactly $1$. That last figure is the prior probability that
the EA system adds nothing to the hazard on any channel.

\paragraph{Where the architectural difference actually enters.} It
does not enter through the modifiers. The only architectural
distinction in the canonical simulator is in the channel composition:
for T-EA each channel hazard is amplified by $\exp(\gamma Z)$, whereas
for OTT-EA the channel splits into a cross-cutting fraction $\xi_c$
amplified by $\exp(\gamma_{\mathrm{X}} Z)$ and a remainder amplified by
$\exp(\gamma Z)$ and gated by the segregation factor $P_3 P_4$. The
$\xi_c$ values are those of Pillar~III
($0.30$ insider, $0.30$ zero-day, $0.55$ supply chain, $0.20$
operational) and $P_3 P_4 = 0.30$. The empirical calibration cited
below therefore bears on the plausibility of the modifier magnitudes
in general, not on any difference between the classes.

\paragraph{Empirical anchoring of the magnitudes.} $\delta_{\mathrm{target}}$
is anchored on Salt Typhoon (CALEA infrastructure as targeted asset),
the Athens affair, and Storm-0558; $\delta_{\mathrm{concen}}$ on
operator user-population scale, which spans $10^7$--$10^8$ subscribers
for major regulated carriers and $10^8$--$10^9$ users for major
platform operators. The v1 text used that population contrast to
justify an OTT-EA prior median shifted upward by $0.80$ natural-log
units. No such shift is implemented, and the argument for it is
withdrawn here. A reader who finds the population contrast persuasive
should read it as an argument for a model the paper does not fit.

\paragraph{Domain mismatch hyperprior.} The transport-uncertainty
variance is $\tau_\phi = 0.50$, producing transport-coefficient draws
$\phi_{ij,r}^{\mathrm{A}} \sim \mathcal{N}(0, 0.25)$ that propagate
analogical-mapping uncertainty through to the edge hazard rates.

\paragraph{Structural constraint.} The constraint $\delta_k \geq 1$ is
enforced by the censoring in
Equation~\eqref{eq:supp_modifier_prior}. It is a modelling
stipulation, not a definitional consequence of the modifiers' names,
and each of the three is argued from the incident record in
\S\ref{sec:dominance} of the main paper. Evidence that an
operationally viable EA deployment attracts no incremental targeting,
adds no attack surface, or concentrates no additional key-material
value would falsify the corresponding constraint. The unmodelled
defensive-uplift counterfactual $\delta_{\mathrm{def}}$ is discussed
in \S\ref{sec:limitations_counterfactual} of the main paper.

% =============================================================================
\section{Cross-Cutting Prior Calibration: Derivation}\label{sec:supp_xcut_prior_calibration}

This section provides the full derivation of the cross-cutting
amplification prior \eqref{eq:gamma_X_prior} summarised in
\S\ref{subsec:dependence} of the main paper.

\paragraph{Empirical anchor from Stream C.} The Stream~C cohort
records, for each of the fourteen documented incidents, a
classification $X_i \in \{Y, P, N\}$ indicating whether the attack
path traversed cross-cutting infrastructure shared between the
key-management and data-storage subsystems ($X_i = Y$) or remained
subgraph-localised. Seven of fourteen are classified $X = Y$, giving
an empirical cross-cutting fraction $\hat p_X = 7/14 = 0.50$. Under the attack-graph specification of
\S\ref{subsec:dependence} of the main paper, and a campaign-prior
central tendency of unit intensity, the value of the prior median
$m$ for $\gamma_{\mathrm{X}}/\gamma$ that reproduces
$\Pr(X = Y) \approx \hat p_X$ at the prior median is $m = 2$. The
mapping is constructive (numerically calibrated to the
fixed-campaign-intensity midpoint) rather than analytic: $m = 2$
is an empirically anchored choice for the cross-cutting-amplification
prior median, not a closed-form derivation.

\paragraph{$\sigma$ choice.} The value $\sigma = 0.40$ is a
calibration choice rather than an empirically measured quantity. It
is adopted for two specific reasons: (i)~it matches the prior width
used in the variance-decomposition and tornado analyses
(\S\ref{subsec:evppi} of the main paper), so the same prior on
$\gamma_{\mathrm{X}}/\gamma$ runs through every part of the
framework that requires one, avoiding the appearance of
parameter-tuning between analyses; (ii)~the qualitative implications
of the OTT-EA tail finding are robust across a wider range
$\sigma \in \{0.25, 0.40, 0.50, 0.60, 0.75\}$ reported in
Table~\ref{tab:gamma_X_sigma_sensitivity} of the main paper.

\paragraph{Resulting prior characteristics.} At $(m, \sigma) =
(2, 0.40)$, the prior on $\gamma_{\mathrm{X}}/\gamma$ has:
median $= 2.00$; mean $\approx 2.08$; 5th percentile $\approx
1.52$; 95th percentile $\approx 2.93$; 99th percentile $\approx
3.54$. A defender preferring a tighter prior
($\sigma \to 0.25$) recovers OTT-EA tail magnitudes close to those
that would obtain at a fixed $\gamma_{\mathrm{X}}/\gamma = 2$; a
defender preferring a wider prior ($\sigma \to 0.60$) obtains
stronger conclusions, principally through the upper-tail mass on
$\gamma_{\mathrm{X}}/\gamma$ exceeding 3. The framework does not
require either alternative to be ruled out.

\paragraph{Tail-uplift mechanism under prior integration.} Treating
$\gamma_{\mathrm{X}}/\gamma$ as sampled rather than fixed produces
a heavier prior-integrated upper tail than would obtain at the
prior median taken in isolation. This follows from the convexity of
the exponential amplification $e^{\gamma_{\mathrm{X}} Z}$ in
$\gamma_{\mathrm{X}}$: by Jensen's inequality, the expected
amplification under a prior with positive variance exceeds the
amplification at the prior mean. High realisations of
$\gamma_{\mathrm{X}}/\gamma$ (in the right tail of the prior)
therefore dominate the upper tail of the output distribution, and
this contribution is invisible to any fixed-point evaluation. The
mechanism is the principal source of the OTT-EA tail-divergence
finding under integrated prior uncertainty.

% =============================================================================
\section{Monte Carlo Stability Diagnostics}\label{sec:supp_mc_stability}

This section gives the full multi-stream stability diagnostic
summarised in \S\ref{sec:bayesian} of the main paper.

\paragraph{Stream setup.} Four independent IID streams of $R = 2{,}000$
replicates each (seeds 2024, 3024, 4024, 5024) were run for each
architectural class, as a stability diagnostic complementing the
primary $R = 8{,}000$ (seed 2024) run from which all headline results
derive. The implementation is in
\texttt{bayesian\_mc\_canonical.py}. The multi-stream diagnostic driver
is in \texttt{run\_all.py}.

\paragraph{Per-stream medians.} T-EA per-stream medians for the annual
compromise probability $q_1^{\mathrm{T}}$ were $3.98\%$, $4.01\%$,
$4.01\%$, $4.06\%$ (max$-$min $= 0.08$~pp). OTT-EA per-stream medians
for $q_1^{\mathrm{OTT}}$ were $2.63\%$, $2.68\%$, $2.67\%$, $2.70\%$
(max$-$min $= 0.07$~pp). T-EA per-stream medians for the ten-year
cumulative probability $Q_{10}^{\mathrm{T}}$ ranged from $35.95\%$ to
$36.51\%$ (max$-$min $= 0.56$~pp); OTT-EA $Q_{10}^{\mathrm{OTT}}$
ranged from $31.91\%$ to $32.02\%$ (max$-$min $= 0.11$~pp), consistent
with the headline full-dependence medians of $36.6\%$ and $31.9\%$
from the primary $R = 8{,}000$ run.

\paragraph{Convergence diagnostic.} The Gelman--Rubin $\hat R$
statistic \citep{gelman2013} computed across the four streams gave
$\hat R < 1.01$ for all sampled scalar parameters in both
architectural configurations. With IID draws these values are
trivially close to unity and confirm only that the Monte Carlo is
implemented as intended. Effective sample sizes exceeded $7{,}700$ for
$q_1$ and $Q_{10}$ in both configurations across the pooled $8{,}000$
replicates, substantially above the $1{,}500$ rule of thumb for stable
quantile estimation.

% =============================================================================
\section{GPD Goodness-of-Fit Tests}\label{sec:supp_gpd_gof}

This section reports formal goodness-of-fit tests for the generalised
Pareto distribution fits to the upper tail of the year-1 compromise
probability $q_1$ in each architectural class. Two test statistics are
used: the Anderson--Darling statistic $A^2$ as defined for the GPD by
\citet{choulakian2001gpd}, and the Kolmogorov--Smirnov statistic $D$.
P-values are estimated by parametric bootstrap ($R = 2{,}000$, seed
$2024$) because the null distributions of $A^2$ and $D$ depend on the
unknown shape parameter $\xi$ when computed from MLE-estimated
parameters.

\paragraph{Threshold sensitivity.} The asymptotic GPD result holds
above a sufficiently high threshold $u$; threshold selection involves
a bias--variance tradeoff. We report results at the paper's primary
threshold (90th empirical percentile) and at a range of alternatives.

\begin{table}[!htbp]
\centering
\small
\caption{T-EA GPD goodness-of-fit across threshold choices. Both
Anderson--Darling and Kolmogorov--Smirnov tests fail to
reject the GPD null at the paper's primary threshold (90th
percentile) and at $85\%$. Kolmogorov--Smirnov rejects at the $5\%$
level from the $92\%$ threshold upward, and both tests reject at
$97\%$, where $n = 240$. The fitted shape ranges over
$0.179$--$0.325$ across the tested thresholds, so the $95\%$-threshold
estimate falls below the bootstrap interval reported at the primary
threshold in Table~10 of the main paper. Threshold stability should
therefore be read as moderate rather than strong.}
\label{tab:supp_gpd_gof_tea}
\begin{tabular}{@{}rcccccccc@{}}
\toprule
\textbf{Threshold} & $n_{\mathrm{exc}}$ & $\hat\xi$ & $\hat\sigma$
  & $A^2$ & $p_{\mathrm{AD}}$ & $D$ & $p_{\mathrm{KS}}$ \\
\midrule
$80\%$ & 1{,}600 & $0.325$ & $0.049$ & $0.63$ & $0.13$ & $0.022$ & $0.06$ \\
$85\%$ & 1{,}200 & $0.287$ & $0.058$ & $0.33$ & $0.60$ & $0.016$ & $0.61$ \\
\textbf{90\%} & \textbf{800} & \textbf{0.271} & \textbf{0.066}
  & \textbf{0.50} & \textbf{0.28} & \textbf{0.028} & \textbf{0.13} \\
$92\%$ & 640 & $0.265$ & $0.071$ & $0.73$ & $0.09$ & $0.035$ & $0.04$ \\
$95\%$ & 400 & $0.179$ & $0.092$ & $0.56$ & $0.21$ & $0.048$ & $0.03$ \\
$97\%$ & 240 & $0.211$ & $0.096$ & $1.54$ & $0.003$ & $0.083$ & $<0.001$ \\
\bottomrule
\end{tabular}
\end{table}

\begin{table}[!htbp]
\centering
\small
\caption{OTT-EA GPD goodness-of-fit across threshold choices. The
GPD null is rejected by Anderson--Darling at every tested threshold;
the Kolmogorov--Smirnov test rejects from the $90\%$ threshold upward
(at $85\%$ the KS $p$-value of $0.07$ is marginal rather than a
rejection at the conventional $5\%$ level). The fitted
shape parameter $\hat\xi$ is also unstable across thresholds, with a
sign change at the $97\%$ threshold reflecting saturation effects as
the right boundary of $q_1$ (the upper limit at $1$) becomes
material. The asymptotic GPD regime is not adequately reached for
the OTT-EA upper tail in the parameter regime considered.}
\label{tab:supp_gpd_gof_ott}
\begin{tabular}{@{}rcccccccc@{}}
\toprule
\textbf{Threshold} & $n_{\mathrm{exc}}$ & $\hat\xi$ & $\hat\sigma$
  & $A^2$ & $p_{\mathrm{AD}}$ & $D$ & $p_{\mathrm{KS}}$ \\
\midrule
$80\%$ & 1{,}600 & $0.602$ & $0.053$ & $1.02$ & $0.02$ & $0.018$ & $0.21$ \\
$85\%$ & 1{,}200 & $0.572$ & $0.066$ & $1.32$ & $0.005$ & $0.024$ & $0.07$ \\
\textbf{90\%} & \textbf{800} & \textbf{0.518} & \textbf{0.089}
  & \textbf{2.26} & \textbf{$<$0.001} & \textbf{0.037} & \textbf{0.003} \\
$92\%$ & 640 & $0.507$ & $0.101$ & $3.69$ & $<0.001$ & $0.049$ & $<0.001$ \\
$95\%$ & 400 & $0.171$ & $0.189$ & $4.97$ & $<0.001$ & $0.077$ & $<0.001$ \\
$97\%$ & 240 & $-0.411$ & $0.397$ & $6.02$ & $<0.001$ & $0.099$ & $<0.001$ \\
\bottomrule
\end{tabular}
\end{table}

\paragraph{Interpretation.} The T-EA GPD fit at the primary threshold
($90\%$, $\xi = 0.271$) is statistically adequate by both Anderson--
Darling ($p = 0.28$) and Kolmogorov--Smirnov ($p = 0.13$), and the
same holds at $85\%$. Adequacy does not extend across the whole
threshold range: Kolmogorov--Smirnov rejects at $92\%$, $95\%$ and
$97\%$, and $\hat\xi$ drifts from $0.325$ to $0.179$ between the
$80\%$ and $95\%$ thresholds. The T-EA fit is therefore adequate in
the neighbourhood of the primary threshold rather than uniformly, and
the reported shape estimate carries threshold uncertainty additional
to the bootstrap interval.

The OTT-EA GPD fit is rejected by Anderson--Darling at every tested
threshold and by Kolmogorov--Smirnov from the $90\%$ threshold upward
(the $85\%$ KS $p$-value of $0.07$ is marginal). The shape
parameter is unstable across thresholds and flips sign at the $97\%$
level as the upper boundary of $q_1$ (at $1$) becomes material. The
asymptotic GPD regime is therefore not adequately reached for the
OTT-EA upper tail in the parameter regime considered. The most
likely cause is structured non-GPD tail behaviour induced by the
mixture of contributing OTT-EA scenarios combined with the
correlated-campaign coupling on cross-cutting edges, which produces
a tail that is not asymptotically Pareto in the form required.

\paragraph{Implications for the tail-divergence finding.} The
qualitative tail-divergence finding (that the OTT-EA upper tail is
heavier than the T-EA upper tail under correlated-campaign
conditions) survives this goodness-of-fit assessment because it
does not depend on the validity of the GPD fit. The direct empirical
percentile comparison (T-EA 95th percentile of $16.5\%$ versus
OTT-EA 95th percentile of $17.4\%$; T-EA 99th percentile of
$\sim 24\%$ versus OTT-EA 99th percentile of $\sim 36\%$) is a
distribution-free comparison whose interpretation does not require
parametric tail modelling. The non-overlapping bootstrap confidence
intervals for $\hat\xi$ between architectures should be interpreted
as a tail-heaviness ordering rather than as a formal statistical
inference about Pareto-tail-index difference, because the GPD model
under which the bootstrap is computed is misspecified for OTT-EA.

\subsection*{Diagnosing the source of the OTT-EA GPD failure}

Three additional diagnostic tests identify the mechanism by which the
OTT-EA upper tail departs from the asymptotic GPD regime.

\paragraph{Z-coupling is the proximate cause.} Comparing the GPD fit
to the OTT-EA $q_1$ samples with and without the latent campaign
variable $Z(t)$ active isolates the Z-coupling contribution. Without
$Z(t)$ (independence configuration: \texttt{OTT\_q1\_indep}), the
GPD fit at the 90th percentile is statistically adequate: Anderson--
Darling $p_{\mathrm{AD}} = 0.13$, Kolmogorov--Smirnov $p_{\mathrm{KS}}
= 0.15$, $\hat\xi = 0.315$. With $Z(t)$ active (\texttt{OTT\_q1\_full}),
the fit is rejected as reported above. For T-EA the GPD fit is
adequate in both configurations ($p_{\mathrm{AD}} = 0.49$ without
$Z$, $0.28$ with), because $Z(t)$ couples uniformly to all T-EA
edges and rescales the tail without changing its shape. For OTT-EA,
$Z(t)$ preferentially activates cross-cutting edges via
$\gamma_{\mathrm{X}}/\gamma \geq 1$, generating a mixture between a
subgraph-dominant low-$Z$ regime and a cross-cutting-dominant
high-$Z$ regime that no single GPD can capture.

\paragraph{Alternative single-parametric forms also fail.} Lognormal
and Weibull fits to the OTT-EA exceedances above the 90th percentile
are likewise rejected by Kolmogorov--Smirnov ($p < 0.001$ in both
cases).
The Hill estimator of the Pareto tail index $\alpha$ (parametric
only in assuming asymptotic Pareto behaviour) likewise yields
unstable values across thresholds ($\hat\alpha$ ranging from $1.97$
at the 98th percentile to $1.32$ at the 80th, corresponding to
GPD-shape equivalents $\hat\xi = 1/\hat\alpha$ between $0.51$ and
$0.76$), mirroring the GPD instability. The OTT-EA upper tail is
not adequately described by any single parametric form in the regime
of practical sample sizes considered.

\paragraph{A two-component mixture fits the OTT-EA tail.} A finite
mixture of two GPDs, fitted by maximum likelihood, produces a
Kolmogorov--Smirnov statistic of $D = 0.026$ for the OTT-EA upper
tail (compared with $D = 0.037$ for the single GPD), with mixing
weight $p \approx 0.81$ on a moderate-shape component ($\hat\xi_1
\approx 0.10$, capturing the subgraph-dominant low-$Z$ regime) and
$1 - p \approx 0.19$ on a heavier bounded-support component
(capturing the cross-cutting-dominant high-$Z$ regime where the
upper boundary at $q_1 = 1$ becomes material). The two-component
structure is consistent with the model's mechanistic prediction:
under correlated-campaign conditions, cross-cutting edges amplify
disproportionately, producing a regime-mixture whose upper tail
inherits two different asymptotic shapes. The mixture interpretation
explains both the goodness-of-fit failure and the threshold
instability of $\hat\xi$ reported above.

\paragraph{What this means in practice.} The OTT-EA tail-divergence
finding is best understood structurally rather than parametrically.
The empirical upper-percentile comparison (OTT-EA exceeds T-EA at
the 95th and 99th percentiles, distribution-free) and the structural
constraint $\gamma_{\mathrm{X}}/\gamma \geq 1$ together establish the
ordering. The mixture diagnosis confirms the model's mechanistic
prediction. The parametric Pareto-index interpretation is withdrawn
in favour of this structural-and-mechanistic framing. Reproducibility:
the diagnostic tests above are produced by
\texttt{gpd\_goodness\_of\_fit.py} in the reproducibility archive.

% =============================================================================
\bibliographystyle{unsrtnat}
\bibliography{references}

\end{document}